\documentclass[11pt]{article}
\usepackage{jheppub}
\notoc
\usepackage{graphicx} 
\usepackage{graphbox}
\usepackage{slashbox}
\usepackage[usenames,dvipsnames]{xcolor} 
\usepackage{tikz}	
\usepackage{dirtytalk}
\usepackage{dsfont}
\usepackage{amsmath}
\usepackage{nccmath}
\usepackage{amsthm}
\usepackage{float}
\usepackage{bbm, dsfont}
\usepackage[normalem]{ulem}
\usepackage{slashed}
\usepackage{mathtools}
\usepackage{cancel}
\usepackage[normalem]{ulem}
\usepackage{cleveref}
\usepackage{orcidlink}
\usepackage{enumerate}
\usepackage{fontawesome}
\usepackage{adjustbox}
\usepackage{tcolorbox}
\usepackage{ulem}
\usepackage{dsfont}
\usepackage{shuffle}

\usetikzlibrary{arrows.meta}

\newtheorem{theorem}{Theorem}[section]
\newtheorem{corollary}[theorem]{Corollary}
\newtheorem{lemma}[theorem]{Lemma}
\newtheorem{prop}[theorem]{Proposition}
\newtheorem{remark}[theorem]{Remark}
\newtheorem{definition}[theorem]{Definition}

\def\beq{\begin{equation}}
\def\eeq{\end{equation}}
\def\ie{\begin{equation}\begin{aligned}}
\def\fe{\end{aligned}\end{equation}}

\newcommand{\Umod}{U_\text{mod}}

\newcommand{\ep}{\epsilon}
\newcommand{\est}{\text{G}}

\newcommand{\Phinew}{\Phi}

\newcommand{\conpl}{\mathbbm{\Delta}}

\def\ad{\mathrm{ad}}
\def\dd{\text{d}}
\def\mot{\mathfrak{m}}
\def\BF{{\rm BF}}

\newcommand{\dplus}[1]{{\cal D}^+ \! \left[\begin{smallmatrix}#1\end{smallmatrix}\right]\!}
\newcommand{\cplus}[1]{{\cal C}^+ \! \left[\begin{smallmatrix}#1\end{smallmatrix}\right]\!}

\newcommand{\beqvY}[1]{\beta^{\rm eqv}\!\left[\begin{smallmatrix}#1\end{smallmatrix};\tau\right]}
\newcommand{\bplusY}[1]{\beta_+\!\left[\begin{smallmatrix}#1\end{smallmatrix};\tau\right]}
\newcommand{\bminusY}[1]{\beta_-\!\left[\begin{smallmatrix}#1\end{smallmatrix};\tau\right]}

\newcommand{\xcomp}[2]{\chi\! \left[\begin{smallmatrix}#1\\#2\end{smallmatrix}\right]}

\newcommand{\ee}[3]{{\cal E}\! \left[\begin{smallmatrix}#1\\#2\end{smallmatrix};#3\right]}
\newcommand{\eeqqvv}[3]{{\cal E}^{\rm eqv}\! \left[\begin{smallmatrix}#1\\#2\end{smallmatrix};#3\right]}
\newcommand{\gab}[3]{\Gamma^{\rm sv}\! \left[\begin{smallmatrix}#1\\#2\end{smallmatrix};#3\right]}
\newcommand{\lab}[3]{\Lambda^{\rm sv}\! \left[\begin{smallmatrix}#1\\#2\end{smallmatrix};#3\right]}
\newcommand{\eez}[4]{{\cal E}\! \left[\begin{smallmatrix}#1\\#2\\#3\end{smallmatrix};#4\right]}
\newcommand{\eezbno}[3]{\overline{{\cal E}\! \left[\begin{smallmatrix}#1\\#2\\#3\end{smallmatrix}\right]}}
\newcommand{\eezno}[3]{{\cal E}\! \left[\begin{smallmatrix}#1\\#2\\#3\end{smallmatrix}\right]}

\newcommand{\beqvtau}[3]{\beta^{\rm eqv}\! \left[\begin{smallmatrix}#1\\#2\end{smallmatrix};#3\right]}
\newcommand{\bsvtau}[3]{\beta^{\rm sv}\! \left[\begin{smallmatrix}#1\\#2\end{smallmatrix};#3\right]}
\newcommand{\beqvztau}[4]{\beta^{\rm eqv}\! \left[\begin{smallmatrix}#1\\#2\\#3\end{smallmatrix} ;#4\right]}
\newcommand{\beqveqvtau}[7]{\beta^{\rm eqv}\! \left[\begin{smallmatrix}#1\\#2\\#3\end{smallmatrix} \left| \begin{smallmatrix}#4\\#5\\#6\end{smallmatrix}  \right.;#7\right]}

\newcommand{\wpz}[2]{\omega_{+}\! \left[\begin{smallmatrix}#1\end{smallmatrix};\tau,#2\right]}
\newcommand{\wmz}[2]{\omega_{-}\! \left[\begin{smallmatrix}#1\end{smallmatrix};\tau,#2\right]}
\newcommand{\wpmz}[2]{\omega_{\pm}\! \left[\begin{smallmatrix}#1\end{smallmatrix};\tau,#2\right]}
\newcommand{\bpmz}[4]{\beta_{\pm} \! \left[\begin{smallmatrix}#1\\#2\\#3\end{smallmatrix};#4\right]}
\newcommand{\bplus}[3]{\beta_+ \! \left[\begin{smallmatrix}#1\\#2\end{smallmatrix};#3\right]}
\newcommand{\bplusz}[4]{\beta_+ \! \left[\begin{smallmatrix}#1\\#2\\#3\end{smallmatrix};#4\right]}

\newcommand{\bminus}[3]{\beta_- \! \left[\begin{smallmatrix}#1\\#2\end{smallmatrix};#3\right]}
\newcommand{\bminusz}[4]{\beta_- \! \left[\begin{smallmatrix}#1\\#2\\#3\end{smallmatrix};#4\right]}

\def\Re{{\rm Re \,}}
\def\Im{{\rm Im \,}}

\usepackage[bbgreekl]{mathbbol}
\DeclareSymbolFontAlphabet{\mathbbm}{bbold}
\DeclareSymbolFontAlphabet{\mathbb}{AMSb}

\title{Elliptic modular graph forms, equivariant iterated integrals and single-valued elliptic polylogarithms}

\author{Oliver Schlotterer$^{a,b}$, Yoann Sohnle$^{a}$ and Yi-Xiao Tao$^{c,d}$}

\vspace{8mm}
\affiliation[a]{Department of Physics and Astronomy, Uppsala University, Box 516, 75120 Uppsala, Sweden}
\affiliation[b]{Department of Mathematics, Centre for Geometry and Physics, Uppsala University, Box 480, 75106 Uppsala, Sweden}
\affiliation[c]{Department of Mathematical Sciences, Tsinghua University, Beijing 100084, China}
\affiliation[d]{Nordita, KTH Royal Institute of Technology and Stockholm University, Hannes Alf\'{v}ens v\"{a}g 12, SE-106 91 Stockholm, Sweden}

\emailAdd{oliver.schlotterer@physics.uu.se}
\emailAdd{yoann.sohnle@physics.uu.se}
\emailAdd{taoyx21@mails.tsinghua.edu.cn}

\vspace{5mm}

\abstract{The low-energy expansion of genus-one string amplitudes produces infinite families of non-holomorphic modular forms after each step of integrating over a point on the torus worldsheet which are known as elliptic modular graph forms (eMGFs). We solve the differential equations of eMGFs depending on a single point $z$ and the modular parameter $\tau$ via iterated integrals over holomorphic modular forms which individually transform inhomogeneously under ${\rm SL}_2(\mathbb Z)$. Suitable generating series of these iterated integrals over $\tau$, their complex conjugates and single-valued multiple zeta values (svMZVs) are combined to attain equivariant transformations under ${\rm SL}_2(\mathbb Z)$ such that their components are modular forms. 

Our generating series of equivariant iterated integrals for eMGFs is related to elliptic multiple polylogarithms (eMPLs) through a gauge transform of the flat Calaque-Enriquez-Etingof connection. By converting iterated $\tau$-integrals to iterated integrals over points on a torus, we arrive at an explicit construction of single-valued eMPLs where all the monodromies in the points cancel. 
Each single-valued eMPL depending on a single point $z$ is found to be a finite combination of meromorphic eMPLs, their complex conjugates, svMZVs and equivariant iterated Eisenstein integrals.
Our generating series determines the latter two admixtures via so-called zeta generators and Tsunogai derivations which act on the two generators $x$, $y$ of a free Lie algebra and where the coefficients of words in $x,y$ define the single-valued eMPLs.}

\preprint{UUITP--33/25}

\begin{document}

\maketitle

\newpage
\setcounter{tocdepth}{2}
\tableofcontents

\vspace{5mm}
\hrule
\vspace{5mm}

%%%%%%%%%%%%%%%%%%%%%%%%%%%%%%%%%%%%%%%%%%%%%%%%%%
%%%%%%%%%%%%%%%%%%%%%%%%%%%%%%%%%%%%%%%%%%%%%%%%%%
%%%%%%%%%%%%%%%%%%%%%%%%%%%%%%%%%%%%%%%%%%%%%%%%%%
%%%%%%%%%%%%%%%%%%%%%%%%%%%%%%%%%%%%%%%%%%%%%%%%%%
\section{Introduction}
\label{sec:1}

Scattering amplitudes in string theory have a history of fruitful symbiosis with modular forms going back to the early 70s~\cite{Shapiro:1972ph}. As a central physics implication of the modular invariance of closed-string amplitudes, gravitational loop diagrams of string theories are free of ultraviolet divergences, even in ten or more spacetime dimensions. This work is dedicated to the mathematical implications of modularity in genus-one string amplitudes, with a focus on the contributions of a fixed torus worldsheet with two marked points.

For the modular group ${\rm SL}_2(\mathbb Z)$ under investigation, holomorphic and meromorphic modular forms are well-understood textbook material \cite{Zagier:123, DHoker:2024book}. In the realm of non-holomorphic modular forms, however, it was a more recent investigation of string amplitudes which added so-called {\it modular graph forms (MGFs)} to the subject of modularity \cite{DHoker:2015gmr, DHoker:2015wxz, DHoker:2016mwo}. 
Already at genus one, MGFs generalize non-holomorphic Eisenstein series to an infinite function space selected by low-energy expansions of string amplitudes \cite{Green:1999pv, Green:2008uj, DHoker:2019blr, Claasen:2024ssh} with interesting number-theoretic properties \cite{Zerbini:2015rss, DHoker:2015wxz, DHoker:2019xef, Zagier:2019eus} and an intriguing web of algebraic and differential relations \cite{DHoker:2015gmr, DHoker:2016mwo, DHoker:2016quv, Kleinschmidt:2017ege, Basu:2019idd}.\footnote{See \cite{Gerken:review, Berkovits:2022ivl, DHoker:2024book} for overview references and \cite{Gerken:2020aju, Claasen:2025vcd} for {\sc Mathematica} packages on MGFs.}

The independent representatives of MGFs under their wealth of relations are distilled by solving their differential equations \cite{DHoker:2016mwo, Gerken:2018zcy} in terms of iterated integrals over holomorphic Eisenstein series \cite{DHoker:2015wxz, Broedel:2018izr, Panzertalk, Gerken:2020yii}.
While integration usually spoils the ${\rm SL}_2(\mathbb Z)$ properties of holomorphic modular forms, Brown identified a systematic way of producing non-holomorphic forms from iterated Eisenstein integrals, their complex conjugates and multiple zeta values (MZVs) \cite{brown2017multiple, Brown:2017qwo, Brown:2017qwo2}. These modular combinations are known as {\it equivariant iterated Eisenstein integrals} and indeed offer a streamlined description of MGFs \cite{Dorigoni_2022, Dorigoni:2024oft, Claasen:2025vcd} as conjectured by Brown.

In this work, we extend the construction of equivariant iterated integrals to more general modular forms in the integrands, namely the coefficients $f^{(k)}$ of the doubly-periodic Kronecker-Eisenstein series
\beq
f^{(k)}(u\tau{+}v,\tau) = - \sum_{m,n \in \mathbb Z\atop {(m,n)\neq (0,0) }} \frac{e^{2\pi i (nu-mv) }}{(m\tau{+}n)^k} \, , \ \ \ \ k\geq 2
\label{intro.01}
\eeq
that depend on a marked point $z=u\tau{+}v$ on a torus through its co-moving coordinates $u,v \in \mathbb R$ and reduce to holomorphic Eisenstein series as $u,v\rightarrow 0$ if $k\geq 4$. Our first main result is a generating-series construction of non-holomorphic forms by combining iterated integrals of the above $f^{(k)}$ and holomorphic Eisenstein series over modular parameters $\tau$ with their complex conjugates and MZVs.

The new $z$-dependent class of equivariant iterated integrals in this work is inspired by and constructed to describe so-called {\it elliptic modular graph forms} (eMGFs) from the string-theory literature. Non-separating degenerations of genus-two MGFs \cite{DHoker:2017pvk} give rise to eMGFs at genus one where one of the genus-two moduli takes the role of two marked points modulo translation invariance of the torus \cite{DHoker:2018mys}. More general eMGFs capture intermediate steps of successively integrating over the marked points on a torus in genus-one string amplitudes. Similar to the case of genus-one MGFs, their $z$-dependent generalizations to eMGFs obey a web of algebraic and differential relations \cite{DHoker:2018mys, DHoker:2020tcq, Basu:2020pey, Basu:2020iok, Dhoker:2020gdz, Hidding:2022vjf} which are exposed in their alternative description via equivariant iterated integrals in this~work.

The simplest class of eMGFs matches Zagier's single-valued elliptic polylogarithms \cite{Ramakrish} which connects to the second main result of this work. Iterated integrals over points on a torus are conveniently organized via {\it elliptic multiple polylogarithms} (eMPLs) which close under integration over the points \cite{BrownLev}. Moreover, eMPLs are at the heart of multidisciplinary exchange since their mathematical investigations \cite{Lev, Levrac, BrownLev, Enriquez:2023emp} find a variety of physics applications in string amplitudes \cite{Broedel:2014vla, Broedel:2017jdo, Panzertalk, Kaderli:2021kqn} and Feynman integrals \cite{Broedel:2017kkb, Bourjaily:2022bwx, Weinzierl:2022}. Similar to their genus-zero counterparts on the sphere, eMPLs are multi-valued functions on the torus and exhibit monodromies as the points of the torus wind around its homology cycles (global monodromies) or other points (local monodromies).

The second main result of this work is a generating-series construction of single-valued eMPLs
{\it in one variable}, i.e.\ that depend on one point $z=u\tau{+}v$ on the torus besides its modular parameter\footnote{Strictly speaking, eMGFs and single-valued eMPLs are functions on the moduli space $\mathfrak{M}_{1,2}$ of twice-punctured tori. We are not counting in the second marked point which is fixed to the origin by translation invariance on the torus. Similarly, our terminology on genus-zero multiple polylogarithms in $n$ variables does not account for the three marked points on the sphere which have been fixed to $(0,1,\infty)$ and does not expose that they are (multi-valued or single-valued) functions on the moduli space $\mathfrak{M}_{0,n+3}$ of  $(n{+}3)$-punctured~spheres.} without any global or local monodromies in $z$. We derive an infinite family of single-valued eMPLs from the equivariant iterated integrals over modular parameters of the Kronecker-Eisenstein kernels $f^{(k)}(u\tau{+}v,\tau)$ in (\ref{intro.01}). The latter are at the same time a possible choice of integration kernels for eMPLs when integrating over the point $z=u\tau{+}v$ instead of $\tau$.\footnote{More precisely, the integration kernels in the Brown-Levin formulation \cite{BrownLev} of eMPLs in one variable are $\dd \bar z$, $\dd z$ and $\dd z \, f^{(k)}(z,\tau)$ with $k\geq 1$, combined into the flat connection reviewed in section \ref{sec:2.1.2}.} On these grounds, the link between eMGFs and single-valued eMPLs is not surprising and amounts to a change of fibration basis of the underlying iterated integrals -- conversions between integrals over $z$ and those over $\tau$. We will implement the change of fibration basis at the level of generating series using a flat connection which is gauge-equivalent to the one of Calaque-Enriquez-Etingof  \cite{KZB}.

A central number-theoretic feature of (e)MGFs, equivariant iterated integrals and the single-valued eMPLs to be constructed in this work is the appearance of single-valued MZVs \cite{Schnetz:2013hqa, Brown:2013gia} in their expansions around the cusp $\tau \rightarrow i \infty$. This is still conjectural for MGFs (see \cite{Zerbini:2015rss, DHoker:2015wxz, DHoker:2019xef, Zagier:2019eus, Gerken:2020yii, Vanhove:2020qtt} for supporting examples and arguments) and hard-wired in Brown's equivariant iterated Eisenstein integrals \cite{brown2017multiple, Brown:2017qwo, Brown:2017qwo2}. The fully explicit formulation of Brown's construction \cite{Dorigoni_2022, Dorigoni:2024oft} makes use of the $f$-alphabet description of (single-valued) MZVs \cite{Brown:2011ik, BrownTate} and links the appearance of higher-depth instances or products to that of odd Riemann zeta values. These links rely on Lie-algebraic tools dubbed {\it zeta generators} which act on generators of the fundamental group of the punctured sphere \cite{Levine, DG:2005, Brown:anatomy, BrownTate, Brown:depth3} and torus \cite{EnriquezEllAss, Hain:KZB, Schneps:2015mzv, hain_matsumoto_2020, Dorigoni:2024iyt} that enter the generating series of this work.

At genus zero, zeta generators can be used to reformulate \cite{Frost:2023stm, Frost:2025lre} the motivic coaction \cite{Goncharov:2001iea, Goncharov:2005sla, BrownTate, Brown:2011ik} and single-valued map \cite{svpolylog, Broedel:2016kls, DelDuca:2016lad, Brown:2018omk} of multiple polylogarithms (MPLs) in any number of variables. The recent description of Brown's single-valued iterated Eisenstein integrals \cite{Brown:2017qwo2} via genus-one incarnations of zeta generators \cite{Dorigoni:2024oft} mirrors the structure of single-valued genus-zero MPLs in one variable \cite{Frost:2023stm}. A key result of this work is another example where zeta generators offer a genus-agnostic description of single-valued generating series that depend on the same number of variables: our generating-series construction of single-valued eMPLs in one variable -- functions of two complex moduli $z,\tau$ -- closely resembles the construction of single-valued genus-zero MPLs in two variables $y,z \in \mathbb C$ in \cite{Frost:2023stm}.

As a practical appeal of constructing single-valued generating series from zeta generators, the fibration bases of all the contributing iterated integrals are preserved. At genus zero, for instance, single-valued MPLs in multiple variables require admixtures of MPLs in fewer variables (besides MZVs) to cancel all monodromies, and the construction of \cite{Frost:2023stm} generates these MPLs with coherent choices of integration endpoints and fibration bases. In this work, the cancellation of global and local monodromies from single-valued eMPLs relies on admixtures of single-valued MZVs and equivariant iterated Eisenstein integrals. The latter can be viewed as a fibration basis for modular combinations of Enriquez' elliptic multiple zeta values \cite{Enriquez:Emzv, Broedel:2015hia, Matthes:Thesis, Zerbini:2018sox} with their complex conjugates and MZVs.

In summary, the two main results of this work are
the construction of infinite families of
\begin{itemize}
\item equivariant iterated integrals over combinations of holomorphic Eisenstein series and Kronecker-Eisenstein kernels $f^{(k)}(u\tau{+}v,\tau)$ accompanied by their complex conjugates and single-valued MZVs as coefficients, see (\ref{intr.80}) or Theorem \ref{3.thm:2};
\item single-valued eMPLs from eMPLs (either in their meromorphic formulation \cite{Broedel:2017kkb} or their Brown-Levin counterparts \cite{BrownLev}) and their complex conjugates, with equivariant iterated Eisenstein integrals and single-valued MZVs as coefficients, see (\ref{onemain}) or Theorem \ref{3.cor:1}.
\end{itemize}
For both of these families, each member is a finite combination of iterated integrals meromorphic in $\tau$ and their complex conjugates. The respective generating series will be given in fully explicit form, and the extraction of individual equivariant iterated integrals and single-valued eMPLs requires iterative applications of Lie-algebraic commutation relations that will be described in detail and can be computed to any desired order.
Further results on single-valued eMPLs concerning their complex-conjugation properties and their asymptotics $\tau \rightarrow i\infty$ 
are previewed in the summary section \ref{sec:1.1} below. 

Throughout this work, we will rely on the $f$-alphabet description and the single-valued map of MZVs (see appendix \ref{sec:A.2}) which are strictly speaking only well-defined for motivic MZVs \cite{Goncharov:2001iea, Goncharov:2005sla, BrownTate, Brown:2011ik}. Accordingly, the theorems of this work are built on the widely trusted assumption that MZVs as real numbers are isomorphic to their motivic counterparts.

\subsection*{Note added}

Before completion of this work, we were informed by Konstantin Baune, Johannes Broedel and Yannis M\"ockli about their upcoming work \cite{Baune:2025svempl} on closely related results, to appear on the arXiv on the day before this preprint. According to an inspiring exchange with the authors, their work provides an alternative construction of single-valued eMPLs in one variable where the monodromies are directly shown to cancel. For the single-valued eMPLs in our work, the absence of monodromies follows in an indirect way, namely through the link with equivariant iterated integrals and modularity properties. The two preprints are expected to offer complementary perspectives on single-valued eMPLs. We are grateful to Konstantin Baune, Johannes Broedel and Yannis M\"ockli for notifying us of their upcoming work and smooth coordination on the arXiv submissions.

%%%%%%%%%%%%%%%%%%%%%%%%%%%%%%%%%%%%%%%%%%%%%%%%%%
%%%%%%%%%%%%%%%%%%%%%%%%%%%%%%%%%%%%%%%%%%%%%%%%%%
%%%%%%%%%%%%%%%%%%%%%%%%%%%%%%%%%%%%%%%%%%%%%%%%%%
%%%%%%%%%%%%%%%%%%%%%%%%%%%%%%%%%%%%%%%%%%%%%%%%%%

\subsection*{Outline}

After a brief overview of the main results of this work in section \ref{sec:1.1}, we start its main part in section \ref{sec:2} by reviewing background material on eMPLs, eMGFs and the construction of single-valued MPLs from zeta generators. Section \ref{sec:key} then introduces the key ingredients of our construction including Lie-algebra structures for genus one, equivariant iterated Eisenstein integrals and the generating series of $z$-dependent iterated integrals. Section \ref{sec:3} is then dedicated to our main results -- first and foremost Theorems \ref{3.thm:2} and \ref{3.cor:1} on the advertised equivariant and single-valued generating series, followed by a discussion of their behavior under complex conjugation and $\tau \rightarrow i\infty$ degeneration in sections \ref{sec:3.cc} and \ref{sec:3.3}, respectively. In section \ref{sec:4}, we spell out the proof of Theorem \ref{3.thm:1} which leads to the change of fibration basis between iterated integrals over $\tau$ or $z$ and is used in proving our main theorems. Section \ref{sec:5} then connects the earlier results to eMGFs and the counting of their independent representatives under algebraic relations, followed by a conclusion and outlook in section~\ref{sec:outopen}. 

Several appendices provide additional review material, display lengthy expressions, and contain several proofs of results in the main body. Moreover, numerous examples of bracket relations among the Lie-algebra generators as well as translations between eMGFs and iterated integrals can be found in the {\sc Mathematica} notebooks {\tt Brackets.nb} and {\tt eMGFdata.nb} included as ancillary files to the arXiv submission of this work.

%%%%%%%%%%%%%%%%%%%%%%%%%%%%%%%%%%%%%%%%%%%%%%%%%%
%%%%%%%%%%%%%%%%%%%%%%%%%%%%%%%%%%%%%%%%%%%%%%%%%%
%%%%%%%%%%%%%%%%%%%%%%%%%%%%%%%%%%%%%%%%%%%%%%%%%%
%%%%%%%%%%%%%%%%%%%%%%%%%%%%%%%%%%%%%%%%%%%%%%%%%%

\subsection{Overview of main results}
\label{sec:1.1}

This section aims to give an overview of the main results of this work at an intermediate level of detail: We provide a gentle amount of background information, motivation and intuition for the ingredients of our generating-series constructions, with the goal of making this 11-page preview of a $> 100$-page paper suitable for standalone reading.

%%%%%%%%%%%%%%%%%%%%%%%%%%%%%%%%%%%%
%%%%%%%%%%%%%%%%%%%%%%%%%%%%%%%%%%%%
%%%%%%%%%%%%%%%%%%%%%%%%%%%%%%%%%%%%
\subsubsection{Generating series and path-ordered exponentials}
\label{int.1.g}

The different types of iterated integrals to be encountered in this work will all be organized into generating series obtained from path-ordered exponentials
\beq
{\rm Pexp}\bigg( \int^a_b \conpl(t) \bigg) =  1 + \int^a_b \conpl(t) + \int^a_b\conpl(t_1) \int^{t_1}_b \conpl(t_2) + \ldots 
\label{intro.02}
\eeq
of Lie-algebra valued one-form connections $\conpl(t)$, see appendix \ref{sec:A.4} for further details. The integration variable $t$ and the endpoints $a,b$ will either take the role of a marked point on the sphere or the torus, or they will refer to its modular parameter $\tau$ in the upper half plane.

The connections $\conpl(t)$ under investigation will be composed of non-commuting variables -- generators of certain Lie algebras -- multiplying one-forms in $t$. At genus zero, for instance, we will encounter the Knizhnik-Zamolodchikov (KZ) connection $\conpl(t) \rightarrow \dd t \, (\frac{e_0}{t}+ \frac{e_1}{t{-}1})$ with logarithmic one-forms and two braid operators $e_0$, $e_1$ that generate a free Lie algebra. The path-ordered exponential (\ref{intro.02}) will then take the form of ($z\in \mathbb C$)
\begin{align}
\mathbb G_{e_0,e_1}(z) &=
{\rm Pexp}\bigg({-} \int^0_z \dd t \,\bigg(\frac{e_0}{t}+ \frac{e_1}{t{-}1} \bigg) \bigg) \label{intro.03} \\
&= 
\sum_{r=0}^\infty \sum_{a_1,\ldots,a_r \atop {\in \{ 0,1\}}} G(a_r,\ldots,a_2,a_1;z) e_{a_1}  e_{a_2 }\ldots e_{a_r}
\notag
\end{align}
and generate MPLs, i.e.\ iterated integrals $G(a_1,\ldots,a_r;z)$ of $\dd \log(t{-}a)$ forms with $a\in \mathbb C$ as defined in (\ref{appA.01}). Moreover, each word $e_{a_1}e_{a_2}\ldots$ in the infinite series expansion of (\ref{intro.02}) will then be accompanied by a different MPL, and the freeness of ${\rm Lie}[e_0,e_1]$ guarantees that the coefficients of such words are well-defined. At genus one, the formulation of single-valued eMPLs through the series $\mathbbm{\Gamma}^{\rm sv}_{x,y}(z,\tau)$ in (\ref{cor.3.4a}), (\ref{agrteq.02}) will also take values in a freely generated Lie algebra, but this is not the case for the composing equivariant series $\mathbb I^{\rm eqv}_{\epsilon^{\rm TS}}(\tau)$, $\mathbb I^{\rm eqv}_{\epsilon,b}(u,v,\tau)$ in (\ref{lieg1.51}), (\ref{lieg1.61}) whose generators obey certain relations reviewed in section \ref{sec:2.4.1}.

The example of MPLs illustrates that $\conpl(t_1)$ and  $\conpl(t_2)$ in (\ref{intro.02}) generically do not commute at $t_1\neq t_2$ such that double integrals do not factorize in general.
The connections $\conpl(t)$ in this work will all be flat\footnote{Flatness of the connection is automatic for $\conpl(t) = \dd t \,\mathbb F(t)$ for meromorphic functions $\mathbb F(t)$ as is the case in (\ref{intro.03}), but we will also encounter the non-meromorphic Brown-Levin connection (\ref{notsec.06}) which generates eMPLs and whose flatness hinges on mixing $\dd t$ and $\dd \bar t$.} and thus ensure that the integrals in (\ref{intro.02}) only depend on the endpoints $a,b$ and the homotopy class of the integration path between them, but not on small path deformations away from singular points of the integrand. 

A key advantage of generating iterated integrals via path-ordered exponentials is that many of their key properties are encoded in concatenations of series. In the case of (\ref{intro.03}), the differential equations of individual MPLs take the generating-series form $\partial_z\mathbb G_{e_0,e_1}(z) = \mathbb G_{e_0,e_1}(z) \, (\frac{e_0}{z}+ \frac{e_1}{z{-}1})$.
More importantly, the monodromies of MPLs as the integration path towards $z$ loops around the origin (to be informally denoted by $z \rightarrow e^{2\pi i}z$) take the considerably more compact form via generating series,
\beq
\mathbb G_{e_0,e_1}(e^{2\pi i}z) = e^{2\pi i e_0}\mathbb G_{e_0,e_1}(z)
\label{intro.05}
\eeq
as compared to the monodromies of individual MPLs.\footnote{The coefficient of $e_0 e_0 e_1 e_0$ in (\ref{intro.05}) for instance encodes that 
\[
G(0,1,0,0;e^{2\pi i}z) = G(0,1,0,0;z) + 2\pi i G(0,1,0;z) + \frac{1}{2} (2\pi i)^2 G(0,1;z)
\]
and illustrates the unpacking of the concatenation $e^{2\pi i e_0}\mathbb G_{e_0,e_1}(z)$ of series.}
Accordingly, generating series of single-valued MPLs are constructed by concatenating $\mathbb G_{e_0,e_1}(z)$
from the left with other series that produce the inverse monodromy $e^{-2\pi i e_0}$ as a right-multiplicative factor in the case of $z \rightarrow e^{2\pi i}z$. Loops around different singular points produce more general monodromy factors, and the single-valued completion of $\mathbb G_{e_0,e_1}(z)$ is engineered to ensure that the appropriate inverse factors are produced in each case. The same mechanisms apply to genus-one integrals where the monodromies of eMPLs series $\mathbbm{\Gamma}_{x,y}(z,\tau)$ in (\ref{notsec.19}) or the obstruction of the iterated $\tau$ integrals $\mathbb I_{\epsilon,b}(u,v,\tau)$ in (\ref{notsec.17}) to transform equivariantly under ${\rm SL}_2(\mathbb Z)$ arise via left-multiplicative and invertible factors analogous to $e^{2\pi i e_0}$ in (\ref{intro.05}) which will be canceled in the constructions of this work.

In fact, as will be used below, the path-ordered exponentials themselves are invertible as series, see (\ref{invpath}). The simple rule for spelling out the inverse w.r.t.\ the concatenation product is to invert the integration order and to insert a minus sign for each integration,~e.g.
\beq
\mathbb G_{e_0,e_1}(z)^{-1} = \sum_{r=0}^\infty  (-1)^r \sum_{a_1,\ldots,a_r \atop {\in \{ 0,1\}}} G(a_1,a_2,\ldots,a_r;z) e_{a_1}  e_{a_2 }\ldots e_{a_r}
\label{intro.07}
\eeq
as the inverse of (\ref{intro.03}). This inversion prescription hinges on the shuffle product universal to any convergent iterated integral produced by (\ref{intro.02}), even to the regularized values of formally divergent integrals encountered in this work.

%%%%%%%%%%%%%%%%%%%%%%%%%%%%%%%%%%%%
%%%%%%%%%%%%%%%%%%%%%%%%%%%%%%%%%%%%
%%%%%%%%%%%%%%%%%%%%%%%%%%%%%%%%%%%%
\subsubsection{Overview of connections and Lie algebras}
\label{int.1.a}

In this section and the overview table \ref{tab:A}, we present the schematic form of the genus-zero and genus-one connections as well as the associated iterated integrals and Lie algebras encountered in this work. 

At genus zero, the second and third row of table \ref{tab:A} recaps the series $\mathbb G_{e_0,e_1}(z)$ in (\ref{intro.03}) and indicates the modifications for MPLs in two variables $y,z \in \mathbb C$ with generating series $\mathbb G_{e_0',e_1',e_y'}(y,z)$ in (\ref{svmpl.01}), respectively. The three generators of the freely generated algebra ${\rm Lie}[e_0',e_1',e_y']$ in the two-variable case will be distinguished from their one-variable counterparts $e_0,e_1$ through a prime. Flatness of the underlying KZ connection in two variables (see table \ref{tab:A} for one of its components) implies a first {\it normalization condition}: Commutators that mix primed and unprimed generators $e_i$ and $e_j'$ from the generating series of MPLs in one and two variables, respectively, always lie in the Lie algebra ${\rm Lie}[e_0',e_1',e_y']$ of the latter,
\beq
[e_i,e_j'] \in {\rm Lie}[e_0',e_1',e_y'] \, , \ \ \ \ \ \ i \in \{0,1\} \, , \ \ \ \ j \in \{0,1,y \}
\label{intro.08}
\eeq
i.e.\ ${\rm Lie}[e_0,e_1]$ normalizes ${\rm Lie}[e_0',e_1',e_y']$.

%\small
\begin{table}[h]
\centering
\begin{tabular}{c||c | c | c| c}
$\!\!$gen.\ series$\!\!$&type of it.\ integrals &schematic form of $\conpl(t)$ &Lie-alg.\ generator &def.
\\\hline\hline
$\mathbb G_{e_i}(z)$
&MPLs in 1 var.\ $z$
&$ \displaystyle \Bigg. \dd t \, \bigg( \frac{e_0}{t} +\frac{e_1}{t{-}1} \bigg)$ 
&braid operators $e_i$
&$\!$(\ref{svmpl.00})$\!$
\\\hline
$\mathbb G_{e_i'}(y,z)$
&MPLs in 2 var.\ $y,z$
&$ \displaystyle \Bigg. \dd t \, \bigg( \frac{e'_0}{t} +\frac{e'_1}{t{-}1} +\frac{e'_y}{t{-}y} \bigg)$ 
&braid operators $e'_i$
&$\!$(\ref{svmpl.01})$\!$
\\\hline
\hline
$\mathbb I_{\epsilon^{\rm TS}}(\tau)$
&it.\ Eisenstein int.\
&$\bigg. \displaystyle \dd t \, \bigg( \ep_0 {+} \sum_{k=4}^\infty {\rm G}_k(t) \ep_k^{\rm TS}\bigg) \Big.$ 
&$\!$Tsunogai derivation $\epsilon_k^{\rm TS}\!$
&$\!$(\ref{notsec.18})$\!$
\\\hline
%%%%%%
$\!\!\mathbb I_{\epsilon,b}(u,v,\tau)\!\!$
&$\begin{array}{c} \!\!\textrm{iterated Kronecker-}\!\! \\ \textrm{Eisenstein int.}
\end{array}$
&$\begin{array}{c}  \dd t \, \Big( \ep_0 {+} \sum_{k=4}^\infty {\rm G}_k(t) \ep_k 
\\
{+} \sum_{k=2}^\infty f^{(k)}(ut{+}v,t) b_k \Big) 
\end{array} \Bigg.$ 
&$\begin{array}{c} \textrm{derivation} \ \epsilon_k \\ \textrm{Lie polynomial} \ b_k \end{array}$
&$\!$(\ref{notsec.17})$\!$
\\\hline
$\mathbbm{\Gamma}_{x,y}(z,\tau)$
&eMPL in 1 var.\ $z$
&\small $\displaystyle  \! \!  \dd t \sum_{k=0}^\infty f^{(k)}(t,\tau) {\rm ad}_x^k y
{-}\frac{\pi \,\dd \bar t}{\Im \tau} x \! \! $ 
&free generators $x,y$
&$\!$(\ref{notsec.19})$\!$ 
\end{tabular}
\caption{\label{tab:A}\textit{Overview of different path-ordered exponentials for iterated integrals at genus zero (second and third line with $\mathbb G_{e_i}(z) = \mathbb G_{e_0,e_1}(z)$ and $\mathbb G_{e_i'}(y,z) = \mathbb G_{e_0',e_1',e_y'}(y,z) $) and genus one (last three lines). The expressions for the connection one-forms $\conpl(t)$ are understood to be schematic, dropping rational prefactors and even gauge transformations via $e^{2\pi i \tau \ep_0}$ in case of the $\ep_0$-dependent connections, see (\ref{sc5.15}). The precise formulae can be found by following the equation numbers in the rightmost column where the respective generating series are defined.}}
\end{table}
%\normalsize

The last three lines of table \ref{tab:A} gather the cast of characters at genus one involving Kronecker-Eisenstein kernels
$f^{(k)}(z,\tau)$ in (\ref{intro.01}) and holomorphic Eisenstein series ${\rm G}_k(\tau)= - f^{(k)}(0,\tau)$ at $k\geq 4$ as integration kernels. 
The simplest connection and path-ordered exponential only depends on the modular parameter $\tau$ which appears as the endpoint of the integration path $\int^{i\infty}_\tau$ defining the series $\mathbb I_{\epsilon^{\rm TS}}(\tau)$ of iterated Eisenstein integrals in (\ref{notsec.18}). This can to some extent be viewed as an analogue of the MPLs in one variable $z$, but there are two notable differences between the Lie-algebra generators: In contrast to the braid operators $e_0$, $e_1$ in the genus-zero case, the {\it Tsunogai derivations} $\ep_0$ and $\ep^{\rm TS}_k$ with even $k\geq 4$ in the connection of $\mathbb I_{\epsilon^{\rm TS}}(\tau)$ (i) furnish an infinite set of generators and (ii) do not form a free algebra in view of relations below like $[\ep_0,[\ep_0,[\ep_0,\ep^{\rm TS}_4]]]=0$ or $[\ep^{\rm TS}_4,\ep^{\rm TS}_{10}]=3 [\ep^{\rm TS}_6,\ep^{\rm TS}_8]$ \cite{LNT, Pollack, Broedel:2015hia, WWWe}. 

The two main results of this work concern iterated integrals at genus one that depend on two variables: $\tau$ in the upper half plane and $z= u\tau{+}v$ with co-moving coordinates $u,v \in \mathbb R$. 
\begin{itemize}
\item A first incarnation of $(z,\tau)$-dependent iterated integrals is sketched in the second line from below of table \ref{tab:A}: simply adjoin the additional integration kernels $\dd t \, f^{(k)}(ut{+}v,t)$ with respect to modular parameters at fixed $u,v$ to the Eisenstein integrals of the one-variable case of $\mathbb I_{\epsilon^{\rm TS}}(\tau)$. The generating series $\mathbb I_{\epsilon,b}(u,v,\tau)$ of such Kronecker-Eisenstein integrals defined in (\ref{notsec.17}) features new Lie-algebra generators $b_k$ with $k\geq 2$ and a variant of the Tsunogai derivations
\beq
\ep_k = \ep^{\rm TS}_k - b_k \,  , \ \ \ \ \ \ k\geq 4 \ {\rm even}
\label{intr.68}
\eeq
\item The last line of table \ref{tab:A} displays the Brown-Levin connection \cite{BrownLev} for eMPLs in one variable, namely homotopy-invariant iterated integrals of $\dd t\,f^{(k)}(t,\tau)$ over points. Their generating series $\mathbbm{\Gamma}_{x,y}(z,\tau)$ defined by (\ref{notsec.19}) takes values in the free Lie algebra with two generators $x,y$. As will be detailed in section \ref{sec:4.3}, eMPLs are expressible in terms of the iterated $\tau$ integrals in the expansion of the series $\mathbb I_{\epsilon,b}(u,v,\tau)$, see \cite{Broedel:2018iwv} for earlier discussions of this change of fibration basis.
\end{itemize}
The two Lie algebras in the bullet points are linked by bracket relations such as ($k\geq 2$)
\beq
b_k = -{\rm ad}_x^{k-1}(y) \, , \ \ \ \  [\ep_0,x]= y \, , \ \ \ \ [\ep_0,y] = 0 
\label{intr.69}
\eeq
which follow from the flatness of the Calaque-Enriquez-Etingof connection \cite{KZB} that carries different entries of table \ref{tab:A} in its $\dd \tau$ and $\dd z$ components. In preparation for compactly stating the normalization conditions on these generators, we introduce the shorthand
\beq
b_k^{(j)} = {\rm ad}^j_{\ep_0}(b_k) \, , \ \ \ \ \ \
\ep_k^{(j)} = {\rm ad}^j_{\ep_0}(\ep_k) \, , \ \ \ \ \ \
\ep_k^{(j){\rm TS}} = \ep_k^{(j)} + b_k^{(j)}
\label{intro.21}
\eeq
and note the following relations for nested commutators with $\epsilon_0$
\beq
{\rm ad}^{k-1}_{\ep_0}(b_k) = 0
\, , \ \ \ \ \ \
{\rm ad}^{k-1}_{\ep_0}(\ep_k) = 0
\label{intro.22}
\eeq
Denoting the Lie algebras of non-vanishing nested brackets in (\ref{intro.21}) by
\begin{align}
\mathfrak{L}_b &= {\rm Lie}\big[ \, \big\{b_k^{(j)}, \ k\geq 2 , \ 0 \leq j \leq k{-}2 \big\} \,\big]  \label{intro.23} \\
\mathfrak{L}_{\ep^{\rm TS}} &= {\rm Lie}\big[ \, \big\{\ep_k^{(j){\rm TS}}, \ k\geq 4 \ {\rm even} , \ 0 \leq j \leq k{-}2 \big\} \,\big]
\notag
\end{align}
then freeness of ${\rm Lie}[x,y]$ implies that the non-zero Lie polynomials $b^{(j)}_k$ in $x,y$ of degree $k \geq 2$ also form a free algebra $\mathfrak{L}_b$. The algebra $\mathfrak{L}_{\ep^{\rm TS}}$ of Tsunogai derivations, by contrast, is not free by Pollack relations starting from the aforementioned $[\ep^{\rm TS}_4,\ep^{\rm TS}_{10}]=3 [\ep^{\rm TS}_6,\ep^{\rm TS}_8]$.

The normalization conditions among the generators in the last three lines of table \ref{tab:A} are given by (see Lemma \ref{braklem})
\beq
[x,b^{(j)}_k] ,\, [y,b^{(j)}_k] , \, [x,\ep^{(j)}_k] ,\, [y,\ep^{(j)}_k] \in \mathfrak{L}_b
\label{intro.24}
\eeq
for any $k\geq 2$ and $0\leq j\leq k{-}2$
as well as ($k_i\geq 2$ and $0\leq j_i\leq k_i{-}2$ for $i=1,2$)
\beq
[\ep^{(j_1)}_{k_1},b^{(j_2)}_{k_2}] \in \mathfrak{L}_b
\ \ \ \Rightarrow \ \ \ [\ep^{(j_1) {\rm TS}}_{k_1},b^{(j_2)}_{k_2}] \in 
\mathfrak{L}_b
\label{intro.25}
\eeq
In particular, the methods of appendix \ref{sec:D} allow us to express the brackets of (\ref{intro.24}) and $[\ep^{(j_1)}_{k_1},b^{(j_2)}_{k_2}] $ as a Lie polynomial in $b_k^{(j)}$ to any desired degree $k_1{+}k_{2}$, also see the ancillary file {\tt Brackets.nb} of the arXiv submission for all brackets up to and including $k_1{+}k_2=10$.

%%%%%%%%%%%%%%%%%%%%%%%%%%%%%%%%%%%%
%%%%%%%%%%%%%%%%%%%%%%%%%%%%%%%%%%%%
%%%%%%%%%%%%%%%%%%%%%%%%%%%%%%%%%%%%
\subsubsection{Zeta generators}
\label{int.1.z}

While the Lie-algebra generators $x$, $y$, $\ep^{(j)}_{k}$, $b^{(j_2)}_{k_2}$ in table \ref{tab:A} control the generating series of functions of $z$ or $\tau$, their interplay with MZVs is described by so-called zeta generators. The latter are defined as Ihara derivations acting on the Lie algebra of the fundamental group of the thrice-punctured sphere \cite{Levine, DG:2005, Brown:anatomy, BrownTate, Brown:depth3} and the once-punctured torus \cite{EnriquezEllAss, Hain:KZB, Schneps:2015mzv, hain_matsumoto_2020, Dorigoni:2024iyt}. Denoting the genus-zero and genus-one incarnations of zeta generators by $M_w$ and $\sigma_w$, respectively, with $w\geq 3$ odd, then both of ${\rm Lie}[M_3, M_5,\ldots]$ and ${\rm Lie}[\sigma_3, \sigma_5,\ldots]$ are freely generated. At a geometric level, the genus-one zeta generators $\sigma_w$ capture the MZVs from the nodal sphere obtained from pinching one of the homology cycles of the torus, see e.g.\ appendix A of \cite{Dorigoni:2024iyt}.

The bracket relations among zeta generators and the generators in table \ref{tab:A} can be computed explicitly to any desired order (see sections \ref{sec:2.3.1}, \ref{sec:2.4} and e.g.\ \cite{Frost:2023stm, Dorigoni:2024iyt} for a recent computation-oriented account) and line up with the following normalization conditions:
\begin{itemize}
\item At genus zero, both ${\rm Lie}[e_0,e_1]$ and ${\rm Lie}[e_0',e_1',e_y']$ are separately normalized by zeta generators, with bracket relations
\beq
[e_0,M_w] = 0 \,, \ \ \ \ \ \
[e_1,M_w] = \big[P_w(e_0,e_1), e_1 \big] \, , \ \ \ \ \ \
w\geq 3 \ {\rm odd}
\label{intro.31}
\eeq
for MPLs in one variable as well as $[e_0',M_w] = 0$ and similar combinations of $P_w(\cdot,\cdot)$ for $[e_1',M_w],\, [e_y',M_w] \in {\rm Lie}[e_0',e_1',e_y']$ given by (\ref{svmpl.08}) for MPLs in two variables: The bracket relations for $[e_1,M_w]$, $[e_1',M_w]$, $[e_y',M_w]$ feature combinations of the Lie polynomials $P_w(e_0,e_1)$ of degree $w$ determined by the Drinfeld associator via (\ref{notsec.12}), e.g.\ $P_3(e_0,e_1)= [e_0{+}e_1,[e_1,e_0]]$ and $P_5(e_0,e_1)$ given by~(\ref{notsec.13}).
\item At genus one, the adjoint action of zeta generators $\sigma_w$ on $\ep_k^{(j){\rm TS}}$, $\ep_k^{(j)}$, $b_k^{(j)}$ produces infinite Lie series of unbounded degrees $\geq w{+}k{+}1$ subject to the normalization conditions
\beq
[\sigma_w , \ep_k^{(j){\rm TS}}] \in \mathfrak{L}_{\epsilon^{\rm TS}} \, , \ \ \ \ k\geq 4 \ {\rm even}, \, 0\leq j\leq k{-}2
\label{intro.32}
\eeq
for iterated Eisenstein integrals as well as
\beq
[\sigma_w , x], \, [\sigma_w , y] \in \mathfrak{L}_{b} \ \ \ \Rightarrow \ \ \ 
[\sigma_w , b_k^{(j)}] \in \mathfrak{L}_{b} \, , \ \ \ \ k\geq 2 , \, 0\leq j\leq k{-}2
\label{intro.33}
\eeq
The explicit Lie-series form of the brackets in (\ref{intro.32}) and (\ref{intro.33}) can be computed from the methods of section \ref{sec:2.4.2} and \cite{Dorigoni:2024iyt}.
\end{itemize}
Zeta generators play a key role in the construction of single-valued versions of MPLs, iterated Eisenstein integrals and eMPLs. They will enter through the following type of group-like series with single-valued MZVs as coefficients
\begin{align}
\mathbb M^{\rm sv}_{\sigma} &= \sum_{r=0}^{\infty} \sum_{i_1,\ldots,i_r \atop {\in 2\mathbb N+1} } \rho^{-1}\big({\rm sv}(f_{i_1} \ldots f_{i_r})\big) \, \sigma_{i_1}\ldots  \sigma_{i_r}
\label{intro.34} \\
&=1 +2 \sum_{i_1 \in 2\mathbb N+1} \zeta_{i_1} \sigma_{i_1}
+ 2 \sum_{i_1,i_2 \in 2\mathbb N+1} \zeta_{i_1} \zeta_{i_2}\sigma_{i_1} \sigma_{i_2} + \ldots
\notag
\end{align}
where the ellipsis features (conjecturally) indecomposable higher-depth MZVs in their $f$-alphabet representation reviewed in appendix \ref{sec:A.2}. We will also encounter variants of the series (\ref{intro.34}) with genus-zero generators $M_w$ or new objects $z_w$, $\Sigma_w(u)$ (see section \ref{int.1.b}) in the place of $\sigma_w$ which will be denoted by $\mathbb M^{\rm sv}_{0}$, $\mathbb M^{\rm sv}_{z}$ and $\mathbb M^{\rm sv}_{\Sigma(u)}$, respectively. In all the single-valued series in section \ref{int.1.d} below, the path-ordered exponentials of table \ref{tab:A} will be sandwiched between series of zeta generators $\mathbb M^{\rm sv}_{\bullet} $ and their inverses $(\mathbb M^{\rm sv}_{\bullet} )^{-1}$. As exemplified~by
\beq
(\mathbb M_\sigma^{\rm sv})^{-1} \, \mathbb X \,
\mathbb M_\sigma^{\rm sv} = \sum_{r=0}^{\infty} \sum_{i_1,i_2,\ldots,i_r \atop {\in 2\mathbb N+1} } \rho^{-1}\big({\rm sv}(f_{i_1} f_{i_2}\ldots f_{i_r})\big) \, \big[[ \ldots [[ \mathbb X ,\sigma_{i_1}] , \sigma_{i_2}], \ldots ], \sigma_{i_r} \big]
\label{intro.35}
\eeq
the conjugations by series in zeta generators of arbitrary series $\mathbb X$ will lead to infinite series in nested commutators which can be sequentially evaluated using the bracket relations of zeta generators. Since the Lie algebras where the series of table \ref{tab:A} take values are all normalized by the zeta generators $M_w$ and $\sigma_w$ at the corresponding genus, all orders of (\ref{intro.35}) preserve the respective algebras of $\mathbb X  \in \big\{\mathbb G_{e_0,e_1}(z),\,
\mathbb G_{e_0',e_1', e_y'}(y,z),\,
\mathbb I_{\epsilon^{\rm TS}}(\tau) ,\,
\mathbb I_{\epsilon,b}(u,v,\tau),\,
\mathbbm{\Gamma}_{x,y}(z,\tau) \big\}$.

%%%%%%%%%%%%%%%%%%%%%%%%%%%%%%%%%%%%
%%%%%%%%%%%%%%%%%%%%%%%%%%%%%%%%%%%%
%%%%%%%%%%%%%%%%%%%%%%%%%%%%%%%%%%%%
\subsubsection{Overview of single-valued generating series}
\label{int.1.d}

With the path-ordered exponentials in table \ref{tab:A} and the series (\ref{intro.34}) in zeta generators in place, we can here preview one of the two main results of this work concerning the generating series of single-valued eMPLs in one variable $z$. The final prerequisite needed is the transposition operation 
\beq
(l_1 l_2\ldots l_r)^T = (l_r)^T \ldots (l_2)^T (l_1)^T
\label{intro.36}
\eeq
on concatenations of arbitrary Lie-algebra generators $l_i \in \{e_j, e_j', x,y,b_k^{(j)},\ep_k^{(j)} \}$ of table \ref{tab:A}. The transposition $(\ldots)^T$ is defined to leave individual braid operators $e_i$, $e_j'$ at genus zero invariant and to act on the genus-one generators with alternating signs, see section \ref{sec:2.4.4}:
\beq
\big( x^T, \, y^T \big) = ({-}x,y) \, , \ \ \ \
(\ep^{(j)}_{k})^T = (-1)^j \ep^{(j)}_{k} \, , \ \ \ \ 
(b^{(j)}_{k})^T = (-1)^j b_{k}^{(j)}
\label{intro.37}
\eeq
The significance of the transposition for the construction of single-valued generating series can be seen from the left-multiplicative monodromy $e^{2\pi i e_0}$ of $\mathbb G_{e_0,e_1}(z)$ in (\ref{intro.05}) which is compensated by the right-multiplicative monodromy $e^{-2\pi i e_0}$ of $\overline{ \mathbb G_{e_0,e_1}(z)^T }$. However, general monodromies of genus-zero MPLs (from looping the integration contour of their Pexp around singular points $\neq 0$) involve MZVs and MPLs in fewer variables. Accordingly, the single-valued completions of the MPLs series $\mathbb G_{e_0,e_1}(z)$ and $\mathbb G_{e_0',e_1',e_y'}(y,z)$ necessitate additional series beyond the naive composition with their complex conjugate transposes like $\overline{\mathbb G_{e_0,e_1}(z)^T}\mathbb G_{e_0,e_1}(z)$. Indeed, the constructions of single-valued MPLs in \cite{svpolylog, Broedel:2016kls, DelDuca:2016lad, Brown:2018omk} are equivalent to the following compositions of series 
\begin{align}
\mathbb G^{\rm sv}_{e_0,e_1}(z) &= (\mathbb M_0^{\rm sv})^{-1} \, \overline{ \mathbb G_{e_0,e_1}(z)^T } \,
\mathbb M_0^{\rm sv} \,
\mathbb G_{e_0,e_1}(z) 
\label{intro.39}
\\
%%%%%
\mathbb G_{e_0',e_1',e_y'}^{\rm sv}(y,z) &= \mathbb G^{\rm sv}_{e_0,e_1}(y)^{-1}\, (\mathbb M_0^{\rm sv})^{-1}  \,
\overline{  \mathbb G_{e_0',e_1',e_y'}(y,z) ^T } \,\mathbb M_0^{\rm sv} \,  \mathbb G^{\rm sv}_{e_0,e_1}(y) \,
 \mathbb G_{e_0',e_1',e_y'}(y,z)
 \notag
\end{align}
where the series $\mathbb G^{\rm sv}_{e_0,e_1}$ constructed in the first line is needed to make the second one fully explicit. These formulae are also noted in the overview table \ref{tab:ovsv} of single-valued generating series and discussed in more detail in section \ref{sec:2.3.3}.

The mechanism of cancelling monodromies by compositions of suitable generating series carries over to path-ordered exponentials for eMPLs at genus one where global and local monodromies result in different invertible left-multiplicative factors analogous to $e^{2\pi i e_0}$ in (\ref{intro.05}). The global monodromies as $z$ is moved around the $A$- and $B$-cycles of the torus to $z\rightarrow z{+}1$ and $z \rightarrow z{+}\tau$, respectively, produce Enriquez' elliptic multiple zeta values (eMZVs) \cite{Enriquez:Emzv, Broedel:2015hia, Matthes:Thesis, Zerbini:2018sox}. Since eMZVs only depend on $\tau$ and are expressible as iterated Eisenstein integrals with $\mathbb Q[(2\pi i)^{\pm 2\pi i}]$-linear combinations of MZVs as coefficients \cite{Broedel:2015hia, Lochak:2020}, they can (within this section) be viewed as genus-one analogues of the MPLs in one variable $y$ generated by $\overline{\mathbb G_{e_0,e_1}(y)^T}$ and $\mathbb G_{e_0,e_1}(y)$ in (\ref{intro.39}). In the same way as only the single-valued combinations of one-variable MPLs were needed for the single-valued completion (\ref{intro.39}) of 
$\mathbb G_{e_0',e_1',e_y'}^{\rm sv}(y,z)$ in two variables, it is the series \cite{Brown:2017qwo2, Dorigoni:2024oft}
\beq
\mathbb I^{\rm sv}_{\ep^{\rm TS}}(\tau)  = 
(\mathbb M^{\rm sv}_{\sigma})^{-1}  \, \overline{ \mathbb I_{\ep^{\rm TS}}(\tau)^T} \, \mathbb M^{\rm sv}_{\sigma}\,
\mathbb I_{\ep^{\rm TS}}(\tau)
\label{intro.41}
\eeq
in single-valued iterated Eisenstein integrals and thereby single-valued eMZVs that enters our construction of single-valued MPLs. The indirect analysis of monodromy cancellations in section \ref{sec:3.2} leads to the following main result of this work in Theorem \ref{3.cor:1}: 
\beq
\mathbbm{\Gamma}^{\rm sv}_{x,y}(z,\tau) =  \mathbb I^{\rm sv}_{\ep^{\rm TS}}(\tau)^{-1} \,
(\mathbb M^{\rm sv}_\sigma)^{-1} \,
\overline{\mathbbm{\Gamma}_{x,y}(z,\tau)^T} \,
  \mathbb M^{\rm sv}_\sigma \,
\mathbb I^{\rm sv}_{\ep^{\rm TS}}(\tau)  \,
\mathbbm{\Gamma}_{x,y}(z,\tau) 
\label{onemain}
\eeq
Its coefficients of different words in $b_k^{(j)}$ (and thus in $x,y$) produce single-valued eMPLs in one variable. They can furthermore be presented in a basis of non-holomorphic modular forms by passing to the series $\mathbbm{\Lambda}^{\rm sv}_{x,y}(z,\tau)$ in (\ref{cor.3.4c}), with the simple conversion formula (\ref{gavsla}) for the individual single-valued eMPLs.

\begin{table}[h]
\centering
\begin{tabular}{c||c||c}
\# vars.\ &genus zero&genus one
\\\hline\hline
1& $ \displaystyle \Bigg. \mathbb G^{\rm sv}_{e_i}(z) = (\mathbb M_0^{\rm sv})^{-1} \, \overline{ \mathbb G_{e_i}(z)^T } \,
\mathbb M_0^{\rm sv} \,
\mathbb G_{e_i}(z)  $ 
%%%%%
%%%%%
%%%%%
&$\displaystyle \Bigg.\mathbb I^{\rm sv}_{\ep^{\rm TS}}(\tau)  = 
(\mathbb M^{\rm sv}_{\sigma})^{-1}  \, \overline{ \mathbb I_{\ep^{\rm TS}}(\tau)^T} \, \mathbb M^{\rm sv}_{\sigma}\,
\mathbb I_{\ep^{\rm TS}}(\tau) $
\\\hline
2&$\bigg. \! \! \! \! \! \! \! \! \! \begin{array}{c}\mathbb G_{e_i'}^{\rm sv}(y,z) = \mathbb G^{\rm sv}_{e_i}(y)^{-1}\, (\mathbb M_0^{\rm sv})^{-1} \Big. \\
\phantom{\times\times}\times 
\overline{  \mathbb G_{e_i'}(y,z) ^T } \,\mathbb M_0^{\rm sv} \,  \mathbb G^{\rm sv}_{e_i}(y) 
 \mathbb G_{e_i'}(y,z) \Big. \end{array} $
 &$\! \! \! \! \!
 \begin{array}{c} \!\!\!  \mathbbm{\Gamma}^{\rm sv}_{x,y}(z,\tau) =  \mathbb I^{\rm sv}_{\ep^{\rm TS}}(\tau)^{-1} \,
(\mathbb M^{\rm sv}_\sigma)^{-1} \Big. \ \ \\
\phantom{x\times}
\times
\overline{\mathbbm{\Gamma}_{x,y}(z,\tau)^T} \,
  \mathbb M^{\rm sv}_\sigma \,
\mathbb I^{\rm sv}_{\ep^{\rm TS}}(\tau)  \,
\mathbbm{\Gamma}_{x,y}(z,\tau) \Big. \end{array} $
\end{tabular}
\caption{\label{tab:ovsv}\textit{Overview of single-valued generating series of MPLs at genus zero, see (\ref{svmpl.12}) and (\ref{svmpl.15}), and single-valued iterated Eisenstein integrals (\ref{lieg1.61}) or eMPLs (\ref{agrteq.02}) at genus one. The formulae in the same line exhibit an identical structure and showcase that the formulation via zeta generators exposes close analogies between genus zero and one. We are using shorthands $\mathbb G_{e_i}(z) = \mathbb G_{e_0,e_1}(z)$ and $\mathbb G_{e_i'}(y,z) = \mathbb G_{e_0',e_1',e_y'}(y,z)$ and similarly in their single-valued versions.}}
\end{table}

The right-hand side of (\ref{onemain}) is a series in $b_k^{(j)}$ since the conjugation by the zeta generators of $(\mathbb M^{\rm sv}_\sigma)^{\pm 1}$ can be iteratively carried out via (\ref{intro.35}) and preserves the Lie algebra of the series $\overline{\mathbbm{\Gamma}_{x,y}(z,\tau)^T} $ in antiholomorphic eMPLs defined by (\ref{ahologam}). The same is true for the additional conjugation by Tsunogai derivations coming from $\mathbb I^{\rm sv}_{\ep^{\rm TS}}(\tau)^{\pm 1}$ in (\ref{onemain}). 
Since the evaluation of commutators  $\sigma_w$ and $\ep_k^{(j){\rm TS}}$ is algorithmic to any desired order, the main result (\ref{onemain}) can be translated into explicit formulae for the individual single-valued eMPLs in one variable: finite combinations of eMPLs and their complex conjugates (from $\mathbbm{\Gamma}_{x,y}(z,\tau)$, $\overline{\mathbbm{\Gamma}_{x,y}(z,\tau)^T} $) with $\mathbb Q$-linear combinations of single-valued MZVs (from $(\mathbb M^{\rm sv}_\sigma)^{\pm 1}$) and single-valued iterated Eisenstein integrals (from $\mathbb I^{\rm sv}_{\ep^{\rm TS}}(\tau)^{\pm 1}$) as coefficients.

As detailed in table \ref{tab:ovsv} and Remark \ref{3.rmk:1}, the genus-one formulae (\ref{intro.41}) and (\ref{onemain}) for single-valued generating series depending on one and two complex variables mirror the structure of (\ref{intro.39}) for single-valued genus zero MPLs in one and two variables, respectively. This correspondence between genus zero and genus one brings the series $\mathbb G_{e_0,e_1}(z)  
\leftrightarrow
\mathbb I_{\ep^{\rm TS}}(\tau)$ and $\mathbb G_{e_0',e_1',e_y'}(y,z) \leftrightarrow \mathbbm{\Gamma}_{x,y}(z,\tau)$ into direct analogy which propagates to their single-valued counterparts in table \ref{tab:ovsv}, see section 3.2.4 of \cite{Dorigoni:2024oft} for the analogies $\mathbb G^{\rm sv}_{e_0,e_1}(z)\leftrightarrow
\mathbb I^{\rm sv}_{\ep^{\rm TS}}(\tau)$ in the one-variable case. As a genus-agnostic commonality of the single-valued series in table \ref{tab:ovsv}, the antiholomorphic series are conjugated by concatenation products of single-valued series in fewer variables -- MZVs with no moduli-dependence and potentially $\mathbb G^{\rm sv}_{e_0,e_1}(y)$ or $\mathbb I^{\rm sv}_{\ep^{\rm TS}}(\tau)$ in the two-variable case.

In particular, we find a genus-one echo of the fact that the genus-zero formulae for single-valued generating series in terms of zeta generators preserve the fibration bases of the contributing MPLs \cite{Frost:2023stm}: The ingredients $\overline{ \mathbb I_{\ep^{\rm TS}}(\tau)^T} $ and $\mathbb I_{\ep^{\rm TS}}(\tau)$ of $\mathbb I^{\rm sv}_{\ep^{\rm TS}}(\tau)$ in (\ref{intro.41}) present eMZVs and their complex conjugates in a canonical form as iterated Eisenstein integrals where all their algebraic relations are already incorporated \cite{Broedel:2015hia, Lochak:2020}.

Note that the conjugation by the composite series $\mathbb M^{\rm sv}_\sigma 
\mathbb I^{\rm sv}_{\ep^{\rm TS}}(\tau)$ in (\ref{onemain}) can be formulated as a change of alphabet of the letters of the series $\overline{\mathbbm{\Gamma}_{x,y}(z,\tau)^T} $
\beq
x \rightarrow  
\big(\mathbb I^{\rm sv}_{\ep^{\rm TS}}(\tau) \big)^{-1}\,
(\mathbb M^{\rm sv}_\sigma )^{-1}
\, x\, \mathbb M^{\rm sv}_\sigma \,
\mathbb I^{\rm sv}_{\ep^{\rm TS}}(\tau) \, , \ \ \ \ 
y \rightarrow  
\big(\mathbb I^{\rm sv}_{\ep^{\rm TS}}(\tau) \big)^{-1}\,
(\mathbb M^{\rm sv}_\sigma )^{-1}
\, y\, \mathbb M^{\rm sv}_\sigma \,
\mathbb I^{\rm sv}_{\ep^{\rm TS}}(\tau)
\eeq
in close analogy with the changes of alphabet that drive the construction of single-valued MPLs in \cite{svpolylog, Broedel:2016kls, DelDuca:2016lad} and the single-valued eMPLs of Baune, Broedel and M\"ockli \cite{Baune:2025svempl}.

%%%%%%%%%%%%%%%%%%%%%%%%%%%%%%%%%%%%
%%%%%%%%%%%%%%%%%%%%%%%%%%%%%%%%%%%%
%%%%%%%%%%%%%%%%%%%%%%%%%%%%%%%%%%%%
\subsubsection{Equivariant iterated integrals}
\label{int.1.b}

In the case of iterated integrals over modular forms, the analogues of monodromies of (e)MPLs are inhomogeneous terms under ${\rm SL}_2(\mathbb Z)$ transformations which enter through left-multiplicative and invertible series and are referred to as cocycles. However, while the KZ or Brown-Levin connections gathering the integration kernels for (e)MPLs are single-valued, the connections for iterated $\tau$ integrals in the fourth and fifth line of table \ref{tab:A} are not modular invariant. Instead, the gauge transformed version 
\begin{align}
 {\mathbb D}_{\ep^{\rm TS}}(\tau) &=    2\pi i\, \dd \tau
\sum_{k=4}^\infty \frac{ (k{-}1) }{(2\pi i)^{k}} \sum_{j=0}^{k-2}\dfrac{(-1)^j}{j!}(2\pi i \tau)^j  {\rm G}_k(\tau)\epsilon_k^{{ (j)\rm TS}}  
\label{intr.71}\\
%%%
\mathbb D_{\ep,b}(u,v,\tau) &=  2\pi i \, \dd \tau
\sum_{k=2}^\infty \frac{ (k{-}1) }{(2\pi i)^{k}} \sum_{j=0}^{k-2}\dfrac{(-1)^j}{j!}\,(2\pi i \tau)^j\big[ {\rm G}_k(\tau)\epsilon_k^{(j)}
 - f^{(k)}(u\tau{+}v,\tau)b_k^{(j)} \big]  
\notag
\end{align}
of the connections in table \ref{tab:A} are equivariant under ${\rm SL}_2(\mathbb Z)$ in the sense that modular transformations mix finite-dimensional multiplets of kernels $\dd \tau \, \tau^j {\rm G}_k(\tau)$ or $\dd \tau \, \tau^j f^{(k)}(u\tau{+}v,\tau)$. These mixings will be expressed through the transformation of the above generators $b_k^{(j)}$, $\ep_k^{(j)}$  in (\ref{intro.21}) under the $\mathfrak{sl}_2$ algebra generated by $\ep_0$ in section \ref{int.1.a} and another derivation $\ep_0^\vee$ via
\begin{align}
[ \ep_0, b_{k}^{(j)} ] &= b_{k}^{(j+1)} \, , &[ \ep_0, \ep_{k}^{(j)} ] &= \ep_{k}^{(j+1)}  \label{intr.72} \\
[ \ep^\vee_0, b_{k}^{(j)} ] &= j(k{-}j{-}1)b_{k}^{(j-1)} \, , &[ \ep^\vee_0, \ep_{k}^{(j)}   ] &= j(k{-}j{-}1)\ep_{k}^{(j-1)} 
\notag 
\end{align}
which follow from the $\ep_0$ action on $x,y$ in (\ref{intr.69}) as well as $[\ep_0^\vee,x] = 0$ and $[\ep^\vee_0,y] = x$.
Modular transformations $T$ and $S$ of the kernels $\dd \tau \, \tau^j {\rm G}_k(\tau)$ or $\dd \tau \, \tau^j f^{(k)}(u\tau{+}v,\tau)$ in the connection (\ref{intr.71}) can then be absorbed into conjugations by 
\begin{align}
U_T = e^{2\pi i \ep_0} \, , \ \ \ \ \ \ U_S 
= e^{\ep_0^\vee} e^{-\ep_0} e^{\ep_0^\vee} (2\pi i )^{-[\ep_0,\ep_0^\vee]}  = e^{ \ep_0^\vee/(2\pi i)} e^{ -2\pi i \ep_0} e^{\ep_0^\vee/ (2\pi i)}
\label{intr.73}
\end{align}
in the Lie group of $\mathfrak{sl}_2$. The mixing (\ref{uschoice.01}) under $U_T$, $U_S$ of $b_k^{(j)}$ and $\ep_k^{(j)}$ with different $j$ but fixed $k$ identifies both of $\{ b_k^{(j)}: \, 0\leq j \leq k{-}2 \}$ and $\{ \epsilon_k^{(j)}: \, 0\leq j \leq k{-}2 \}$ as $(k{-}1)$-dimensional representations of $\mathfrak{sl}_2$ and implies the following {\it equivariant} transformation law of (\ref{intr.71}) 
\begin{align}
{\mathbb D}_{\ep^{\rm TS}}  (\gamma \cdot \tau)   = U_\gamma^{-1}\, {\mathbb D}_{\ep^{\rm TS}}(\tau)  \,U_\gamma 
\, , \ \ \ \
{\mathbb D}_{\ep,b}  \big(\gamma\cdot (u,v,\tau)\big)   = U_\gamma^{-1}\, {\mathbb D}_{\ep,b}  (u,v, \tau)  \,U_\gamma 
\label{intr.74}
\end{align}
where the ${\rm SL}_2(\mathbb Z)$ action on the moduli $(z,\tau)$ and the co-moving coordinates $u,v$ are given by
\begin{align}
\gamma\cdot (z,\tau) = \bigg( \frac{z}{c\tau {+} d},\frac{a\tau{+}b}{c\tau {+} d}\bigg) \, , \ \ \ \ \gamma\cdot ( \smallmatrix v \\ -u \endsmallmatrix )
= ( \smallmatrix a &b \\ c &d \endsmallmatrix ) ( \smallmatrix v \\ -u \endsmallmatrix ) \, , \ \ \ \  \gamma = ( \smallmatrix a &b \\ c &d \endsmallmatrix ) \in {\rm SL}_2(\mathbb Z)
\label{intr.75}
\end{align}
The $\exp(\mathfrak{sl}_2)$ element $U_\gamma$ in (\ref{intr.74}) associated with a general modular transformation $\gamma \in {\rm SL}_2(\mathbb Z)$ is understood to decompose into $U_T$, $U_S$ in (\ref{intr.73}) according to the decomposition of $\gamma$ into the generators $T$, $S$. 

The ${\rm SL}_2(\mathbb Z)$ transformations of the path-ordered exponentials $\mathbb I_{\ep^{\rm TS}}(\tau)$ and $\mathbb I_{\ep,b}(u,v,\tau)$ of ${\mathbb D}_{\ep^{\rm TS}}(\tau)$ and $\mathbb D_{\ep,b}(u,v,\tau)$ in (\ref{intr.71}) depart from the equivariant behavior (\ref{intr.74}) of the connections by
left-multiplicative cocycles such as $e^{2\pi i N^{\rm TS}}  e^{2\pi i \ep_0}$ under the modular $T$ transformation 
\beq
\mathbb I_{\ep^{\rm TS}}(\tau{+}1) = e^{2\pi i N^{\rm TS}} \, e^{2\pi i \ep_0} \,U_T^{-1}\, \mathbb I_{\ep^{\rm TS}}(\tau)\, U_T \, , \ \ \ \ \ \
N^{\rm TS} = -\ep_0 + \sum_{k=4}^\infty (k{-}1)   \frac{B_k}{k!} \, \ep^{\rm TS}_k
\label{intr.76}
\eeq
see appendix \ref{app:thm21} and \cite{brown2017multiple, Brown:2017qwo2, saad2020multiple, Dorigoni:2024oft, Kleinschmidt:2025dtk} for discussions of the analogous $S$-cocycle. Hence, the analogues of single-valued $z$-dependent generating series in the realm of iterated $\tau$ integrals of modular forms are equivariant generating series subject to
\begin{align}
{\mathbb I}^{\rm eqv}_{\ep^{\rm TS}}  (\gamma \cdot \tau)   = U_\gamma^{-1} \, {\mathbb I}^{\rm eqv}_{\ep^{\rm TS}}(\tau) \, U_\gamma 
\, , \ \ \ \
{\mathbb I}^{\rm eqv}_{\ep,b}  \big(\gamma\cdot (u,v,\tau)\big)   = U_\gamma^{-1} \, {\mathbb I}^{\rm eqv}_{\ep,b}  (u,v, \tau) \,  U_\gamma 
\label{intr.78}
\end{align}
i.e.\ that mirror the equivariant transformation (\ref{intr.74}) of the connections. In the solely $\tau$-dependent case, the equivariant series ${\mathbb I}^{\rm eqv}_{\ep^{\rm TS}}(\tau)$ with the same holomorphic $\tau$-derivative as $\mathbb I_{\ep^{\rm TS}}(\tau)$ was found in Brown's work \cite{Brown:2017qwo2} and presented in the following  explicit form in \cite{Dorigoni:2024oft}
\begin{align}
\mathbb I^{\rm eqv}_{\ep^{\rm TS}}(\tau)  = 
(\mathbb M^{\rm sv}_{z})^{-1}  \, \overline{ \mathbb I_{\ep^{\rm TS}}(\tau)^T} \, \mathbb M^{\rm sv}_{\sigma}\,
\mathbb I_{\ep^{\rm TS}}(\tau)
\label{intr.79}
\end{align}
This merely differs from the series $\mathbb I^{\rm sv}_{\ep^{\rm TS}}(\tau) $ in single-valued iterated Eisenstein integrals by the left-multiplicative factor $(\mathbb M^{\rm sv}_{z})^{-1}$ in the place of $(\mathbb M^{\rm sv}_{\sigma})^{-1}$ in (\ref{intro.41}). According to the discussion below (\ref{intro.34}), the expansion of $\mathbb M^{\rm sv}_{z}$ is obtained by replacing the genus-one zeta generators $\sigma_w$ in that equation by their arithmetic parts $z_w$. The latter are $\mathfrak{sl}_2$ singlets in the sense that $[\ep_0,z_w]= [\ep^\vee_0,z_w]= 0$ for any odd $w \geq 3$ and are the only contributions to $\sigma_w = z_w +\ldots$ which are not expressible in terms of  $\ep_k^{(j){\rm TS}}$,\footnote{A preferred way of splitting $\sigma_w = z_w +\ldots$ into an arithmetic part $z_w$ and the {\it geometric} terms in the ellipsis given by a Lie series in $\ep_k^{(j){\rm TS}}$ can be found in Theorem 5.4.1 of \cite{Dorigoni:2024iyt}.} see section \ref{sec:2.4.2} for further details. The necessity of the left-multiplicative $(\mathbb M^{\rm sv}_{z})^{-1}$ for equivariance of (\ref{intr.79}) can be seen from the ${\rm SL}_2(\mathbb Z)$ transformation $ \overline{ \mathbb I_{\ep^{\rm TS}}(\gamma\cdot\tau)^T} = U_\gamma^{-1}  \overline{ \mathbb I_{\ep^{\rm TS}}(\tau)^T}  U_\gamma  \overline{\mathbb S^T_\gamma}$ (with some cocycle $\mathbb S_\gamma$) of the adjacent factor since $(\mathbb M^{\rm sv}_{z})^{-1} U_\gamma^{-1} = U_\gamma^{-1} (\mathbb M^{\rm sv}_{z})^{-1} $ while $(\mathbb M^{\rm sv}_{\sigma})^{-1} U_\gamma^{-1} \neq U_\gamma^{-1} (\mathbb M^{\rm sv}_{\sigma})^{-1}$.

The explicit realization of the equivariant series ${\mathbb I}^{\rm eqv}_{\ep,b}  (u,v, \tau) $ in (\ref{intr.78}) depending on both $\tau$ and $z=u\tau{+}v$ is the second main result of this work in Theorem \ref{3.thm:2}:
\beq
\mathbb I^{\rm eqv}_{\ep,b}(u,v,\tau)  = 
(\mathbb M^{\rm sv}_{z})^{-1}\, \overline{ \mathbb I_{\ep,b}(u,v,\tau)^T} \, \mathbb M^{\rm sv}_{\Sigma(u)}\,
\mathbb I_{\ep,b}(u,v,\tau) 
\label{intr.80}
\eeq
The series $\mathbb M^{\rm sv}_{\Sigma(u)}$ in single-valued MZVs in the middle involves augmented zeta generators
\beq
\Sigma_w(u) = e^{u x} \big( P_w( t_{12}, t_{01})+\sigma_w  \big) e^{-ux} \, , \ \ \ \ \ \ w\geq 3 \ {\rm odd}  
\label{intr.81}
\eeq
instead of the $\sigma_w$ in the solely $\tau$-dependent case (\ref{intr.79}). By Lemma \ref{zwprop} and (\ref{intr.81}), $\mathbb M^{\rm sv}_{\Sigma(u)}$ is an explicitly known series in arithmetic zeta generators $z_w$ and the generators $b_k^{(j)}$, $\ep_k^{(j)}$ of the series $\mathbb I_{\ep,b}(u,v,\tau)$. The Lie polynomials $P_w$ are the same as in (\ref{intro.31}) at genus zero and produce Lie series in $b_k^{(j)}$ by the expressions (\ref{lieg1.17}) for the letters $t_{12}$, $t_{01}$ in terms of $x,y$.

In both of (\ref{intr.79}) and (\ref{intr.80}), all the arithmetic $z_w$ from the middle series $\mathbb M^{\rm sv}_{\sigma}$ and $\mathbb M^{\rm sv}_{\Sigma(u)}$ eventually cancel against those in the leftmost factors $(\mathbb M^{\rm sv}_{z})^{-1}$ once they are pulled through the series $\overline{ \mathbb I_{\ep^{\rm TS}}(\tau)^T}$ and $\overline{ \mathbb I_{\ep,b}(u,v,\tau)^T}$ antiholomorphic in $\tau$. This process results in commutators of $z_w$ with the generators of $\overline{ \mathbb I_{\ep^{\rm TS}}(\tau)^T}$ and $\overline{ \mathbb I_{\ep,b}(u,v,\tau)^T}$ which preserve the respective algebras since the zeta generators do. In this way, one finds an infinite series of equivariant iterated integrals of $\dd \tau \, \tau^j {\rm G}_k(\tau)$ and $\dd \tau \, \tau^j f^{(k)}(u\tau{+}v,\tau)$ by organizing the expansion of (\ref{intr.80}) into words in $b_k^{(j)}$, $\ep_k^{(j)}$. While the Pollack relations among $\ep_k^{(j){\rm TS}}$ and thus $\ep_k^{(j)}$ prevent the coefficients in this expansion from being well-defined at degrees $k_1{+}k_2{+}\ldots \geq 14$, the ratio
\begin{align}
\mathbbm{\Gamma}^{\rm sv}_{x,y}(z,\tau) &= 
\big(\mathbb I^{\rm eqv}_{\ep^{\rm TS}}(\tau)\big)^{-1} \,
\mathbb I^{\rm eqv}_{\ep,b}(u,v,\tau)  
\label{intr.83} \\
&= \big(\mathbb I^{\rm eqv}_{\ep^{\rm TS}}(\tau)\big)^{-1} \,
(\mathbb M^{\rm sv}_{z})^{-1}\, \overline{ \mathbb I_{\ep,b}(u,v,\tau)^T} \, \mathbb M^{\rm sv}_{\Sigma(u)}\,
\mathbb I_{\ep,b}(u,v,\tau)
\notag
\end{align}
is a series solely in the free-Lie-algebra generators $b_k^{(j)}$ with well-defined coefficients. In fact, the ratio (\ref{intr.83}) matches the series (\ref{onemain}) in single-valued eMPLs by Theorem \ref{3.thm:1}. Hence, our results on equivariant iterated integrals offer an alternative viewpoint on single-valued eMPLs and implement a change of fibration basis on the iterated integrals.

Note that the expansion coefficients of the equivariant series (\ref{intr.79}) and (\ref{intr.80}) do not yet furnish modular forms as one can see from the non-trivial $T$ transformation (\ref{intr.73}). However, one can convert linear combinations of equivariant iterated integrals into modular forms \cite{Brown:2017qwo} which arise as coefficients after conjugating $\mathbb I^{\rm eqv}_{\ep^{\rm TS}}(\tau)$, $\mathbb I^{\rm eqv}_{\ep,b}(u,v,\tau)$, $\mathbbm{\Gamma}^{\rm sv}_{x,y}(z,\tau)$ by \cite{Dorigoni:2024oft}
\beq
\Umod(\tau) = \exp \bigg( \frac{-\epsilon_0^\vee}{4\pi \Im \tau} \bigg)  \exp(2\pi i \bar \tau \epsilon_0)
\label{prevumod}
\eeq
in the Lie group of $\mathfrak{sl}_2$. The resulting generating series in {\it modular frame} are for instance discussed in section \ref{sec:2.5.3} and throughout section \ref{sec:3}. As detailed in Theorem \ref{3.cor:1} and Corollary \ref{corof44}, this implies that single-valued eMPLs can be brought into a basis of non-holomorphic modular forms which in turn establishes their link to eMGFs.

%%%%%%%%%%%%%%%%%%%%%%%%%%%%%%%%%%%%
%%%%%%%%%%%%%%%%%%%%%%%%%%%%%%%%%%%%
%%%%%%%%%%%%%%%%%%%%%%%%%%%%%%%%%%%%
\subsubsection{Further results}
\label{int.1.c}

While the discussion of this section was tailored to state our two main results on equivariant iterated integrals and single-valued eMPLs in Theorems \ref{3.thm:2} and \ref{3.cor:1}, the main body of this work offers various follow-up results and links with eMGFs to be summarized here. 

Already at genus zero, single-valued MPLs can be engineered to either preserve the holomorphic differential equations or the antiholomorphic ones of the associated multi-valued MPLs, but not both at the same time. Our construction of single-valued eMPLs has the same feature for the differential equations in $\tau$ at fixed $u,v$ and, with a minor caveat (see section \ref{sec.4.diff}) for their $z$-derivatives. The additional contaminations in the antiholomorphic differential equations are pinpointed by relating our generating series of eMPLs to their complex conjugates in section \ref{sec:3.cc}. The net effect of complex conjugation boils down to conjugations by generating series in equivariant iterated Eisenstein integrals and MZVs, a phenomenon that is well-known from their genus-zero counterparts and discussed from an algebraic-geometry perspective in \cite{brown2017multiple, Brown:2017qwo2}.

We furthermore establish a direct link between single-valued (e)MPLs at genus zero and genus one by investigating the expansion of single-valued eMPLs around the cusp $\tau \rightarrow i\infty$ in section \ref{sec:3.3}. By tracking subleading orders in $q=e^{2\pi i \tau}$ at fixed values of the co-moving coordinates $u,v\in \mathbb R$ of $z=u\tau{+}v$, we show that the asymptotics of single-valued eMPLs at the cusp is expressible in terms of Laurent polynomials in $\pi \Im \tau$ whose coefficients are $\mathbb Q$-linear combinations of (products of) (i) single-valued MZVs as it is conjecturally the case for MGFs \cite{Zerbini:2015rss, DHoker:2015wxz} and (ii) single-valued MPLs in one variable $e^{2\pi i z}$.

A major line of motivation and valuable guidance for this work comes from the eMGFs in genus-one closed-string amplitudes reviewed in section \ref{sec:2.2}. Their differential equations in $z$ and $\tau$ \cite{Dhoker:2020gdz} and the resulting iterated-integral representations \cite{Hidding:2022vjf} were important hints for the existence and the properties of the equivariant generating series (\ref{onemain}) and (\ref{intr.80}). The implications of our main results for eMGFs, in particular the counting of their independent representatives, are discussed in section \ref{sec:5}.

%%%%%%%%%%%%%%%%%%%%%%%%%%%%%%%%%%%%%%%%%%%%%%%%%%
%%%%%%%%%%%%%%%%%%%%%%%%%%%%%%%%%%%%%%%%%%%%%%%%%%
%%%%%%%%%%%%%%%%%%%%%%%%%%%%%%%%%%%%%%%%%%%%%%%%%%
%%%%%%%%%%%%%%%%%%%%%%%%%%%%%%%%%%%%%%%%%%%%%%%%%%

\section*{Acknowledgments}

We are grateful to Francis Brown, Emiel Claasen, Marco David, Eric D'Hoker, Daniele Dorigoni, Mehregan Doroudiani, Claude Duhr, Benjamin Enriquez, Hadleigh Frost, Martijn Hidding, Deepak Kamlesh, Axel Kleinschmidt, Franca Lippert, Franziska Porkert, Carlos Rodriguez, Leila Schneps, Bram Verbeek and Federico Zerbini for combinations of inspiring discussions and collaboration on related topics. We furthermore thank Eric D'Hoker for helpful comments on a draft. YS is particularly grateful to Benjamin Enriquez for some enlightening discussions and feedback at the IRMA in Strasbourg. We would like to thank Konstantin Baune, Johannes Broedel and Yannis M\"ockli for notifying us of their closely related results \cite{Baune:2025svempl} and coordinating the arXiv submissions of both preprints.
YT would like to thank Nordita for hosting him during several stages of the project. YT is supported by the National Key R\&D Program of China (NO.\ 124B2094 and NO.\ 2020YFA0713000). The research of OS and YS is funded by the European Union under ERC Synergy Grant MaScAmp 101167287. Views and opinions expressed are however those of the author(s) only and do not necessarily reflect those of the European Union or the European Research Council. Neither the European Union nor the granting authority can be held responsible for them.

%%%%%%%%%%%%%%%%%%%%%%%%%%%%%%%%%%%%%%%%%%%%%%%%%%
%%%%%%%%%%%%%%%%%%%%%%%%%%%%%%%%%%%%%%%%%%%%%%%%%%
%%%%%%%%%%%%%%%%%%%%%%%%%%%%%%%%%%%%%%%%%%%%%%%%%%
%%%%%%%%%%%%%%%%%%%%%%%%%%%%%%%%%%%%%%%%%%%%%%%%%%
\section{Background material}
\label{sec:2}

The goal of this section is to review background material needed to motivate and state the results in later sections: elliptic multiple polylogarithms in section \ref{sec:2.1}, elliptic modular graph forms in section \ref{sec:2.2} and Lie-algebra techniques for single-valued genus-zero polylogarithms in section \ref{sec:2.3}.

%%%%%%%%%%%%%%%%%%%%%%%%%%%%%%%%%%%%%%%%%%%%%%%%%%
%%%%%%%%%%%%%%%%%%%%%%%%%%%%%%%%%%%%%%%%%%%%%%%%%%
\subsection{Elliptic multiple polylogarithms and their generating series}
\label{sec:2.1}

The space of rational functions can be extended to close under integration
by adjoining multiple polylogarithms (MPLs) $G(a_1,\ldots,a_r;z)$ reviewed in appendix \ref{sec:A.1}
\cite{LapDan:1953rq, GONCHAROV1995197, Brown:2009qja}.
At a geometric level, MPLs capture iterated integrals on the Riemann sphere with marked points, 
i.e.\ Riemann surfaces of genus zero. MPLs are constructed from meromorphic and 
single-valued one-forms or {\it integration kernels} 
$\frac{\dd z}{z{-}a}$ with $z,a \in \mathbb C$ and at most simple poles in $z$.

At genus one, the analogous function space on the torus $T^2 = \mathbb C/ (\mathbb Z {+} \tau \mathbb Z)$ 
with modular parameter $\tau \in \mathbb H = \{ \tau \in \mathbb C\, | \ \Im \tau > 0\}$ that closes under integration 
over marked points is based on elliptic multiple polylogarithms (eMPLs) \cite{BrownLev}. In contrast 
to the integration kernels of MPLs at genus zero, the integration kernels of eMPLs cannot be made to 
simultaneously enjoy single-valuedness and meromorphicity in $z \in T^2$ 
without introducing double poles. Instead, eMPLs can be constructed in three different
but eventually equivalent ways from
\begin{itemize}
\item[(i)] meromorphic integration kernels $g^{(n)}(z{-}a,\tau)$ with $n\in \mathbb N_0$ which have no poles other than $(z{-}a)^{-1}$ at $n=1$ for $z,a\in T^2$ and are multi-valued in both, see section \ref{sec:2.1.1};
\item[(ii)] single-valued integration kernels $f^{(n)}(z{-}a,\tau)$ with $n\in \mathbb N_0$ which have no poles other than $(z{-}a)^{-1}$ at $n=1$ and are non-meromorphic in $z,a\in T^2$, see section \ref{sec:2.1.2};
\item[(iii)] meromorphic and single-valued integration kernels with poles of order $\geq 2$, see \cite{Levrac, BEFZ:2110} and \cite{BEFZ:2212, Enriquez:2023emp} for two different approaches.
\end{itemize}
The eMPLs obtained from iterated integrals of (i), (ii) and (iii) all exhibit monodromies 
when the endpoint $z$ of the integration path is shifted around
the homology cycles of the torus ($z \rightarrow z{+}1$ or $z \rightarrow z{+}\tau$)
or around the singular points of the integrand (e.g.\ $z\in a{+}\mathbb Z {+} \tau \mathbb Z$ in case of $g^{(n)}(z{-}a,\tau)$ or $f^{(1)}(z{-}a,\tau)$).
One of the key results of this work is the construction
of single-valued eMPLs by combining iterated integrals of kernels 
(i) and (ii) with their complex conjugates and suitable $\tau$-dependent 
coefficients that result in good modular properties.

%%%%%%%%%%%%%%%%%%%%%%%%%%%%%%%%%%%%%%%%%%%%%%%%%%
%%%%%%%%%%%%%%%%%%%%%%%%%%%%%%%%%%%%%%%%%%%%%%%%%%
\subsubsection{eMPLs from meromorphic integration kernels}
\label{sec:2.1.1}

A first class (i) of integration kernels for eMPLs are the meromorphic
Kronecker-Eisenstein coefficients $g^{(n)}(z,\tau)$ with $z \in \mathbb C$ and
$n\in \mathbb N_0$ defined by (\ref{appA.10}), with $B$-cycle monodromies
given by (\ref{bmon}). For $r\in \mathbb N$, $n_1,\ldots,n_r \in \mathbb N_0$ and
$z,a_1,\ldots,a_r\in \mathbb C$, we recursively define meromorphic 
eMPLs by~\cite{Broedel:2014vla, Broedel:2017kkb}
\begin{align}
\tilde \Gamma\big( \smallmatrix n_1 &n_2 &\ldots &n_r \\ a_1 &a_2 &\ldots &a_r \endsmallmatrix ;z,\tau\big) = \int^z_0  \dd t  \, g^{(n_1)}(t{-}a_1,\tau) \, \tilde \Gamma\big( \smallmatrix n_2 &\ldots &n_r \\ a_2 &\ldots &a_r \endsmallmatrix ;t,\tau\big)
\, , \ \ \ \  \tilde \Gamma\big( \smallmatrix \emptyset \\ \emptyset  \endsmallmatrix ;z,\tau\big) =1
 \label{defempl}
\end{align}
Similar to the shuffle regularization of MPLs in (\ref{appA.02}), we shuffle regularize the endpoint 
divergence of eMPLs with $n_r=1$ at $t\rightarrow 0$ by assigning 
\cite{Broedel:2014vla, Broedel:2018iwv, Broedel:2019tlz}
\begin{align}
\tilde \Gamma\big( \smallmatrix 1 \\ 0 \endsmallmatrix ;z,\tau\big) 
&= \lim_{\varepsilon \rightarrow 0} \bigg( \int^z_\varepsilon  \dd t \, g^{(1)}(t,\tau)  + \log(1{-}e^{2\pi i \varepsilon}) \bigg)
\label{reggt} \\
&= \log(1{-}e^{2\pi i z}) - i\pi z + 2 \sum_{m,n=1}^\infty \frac{1}{m}\, \big( 1{-} \cos(2\pi m z) \big) \, q^{mn}
\notag
\end{align}
Throughout this work, we will restrict our discussion to eMPLs with $a_i=0$ for $i=1,\ldots,r$
which are referred to as {\it eMPLs in one variable}. Further specialization of eMPLs in
one variable to $z=1$, i.e.\ to the $A$-cycle $t \in (0,1)$ of the torus as an integration path, gives rise to
Enriquez' $A$-elliptic multiple zeta values \cite{Enriquez:Emzv}
\beq
\omega(n_1,\ldots,n_r ; \tau) = \tilde \Gamma\big( \smallmatrix n_r  &\ldots &n_1 \\ 0 &\ldots &0 \endsmallmatrix ;1,\tau\big) 
\label{defemzv}
\eeq
see \cite{Matthes:Thesis, Lochak:2020, Zerbini:2018sox} for further studies in the mathematics literature
and \cite{Broedel:2014vla, Broedel:2015hia, Broedel:2017jdo} from a string-theory perspective.

We shall perform the construction of single-valued eMPLs in one variable at the level of the following generating series:
Consider the free Lie algebra ${\rm Lie}[x,y]$ in two generators $x,y$ and the $\dd z$ component of the Calaque-Enriquez-Etingof (CEE) connection \cite{KZB}
\beq
\tilde{\mathbb K}_{x,y}(z,\tau) = \dd z \, {\rm ad}_x F(z,{\rm ad}_{ \frac{x}{2\pi i} },\tau) y  = \dd z \sum_{n=0}^\infty (2\pi i)^{1-n} g^{(n)}(z,\tau) \ad_x^{n}(y)
\label{defKcon}
\eeq
where the meromorphic Kronecker-Eisenstein series $F(z,\alpha,\tau) $ is defined
by (\ref{appA.10}) and\\  $\ad_x^k(y) = [x,\ad_x^{k-1}(y)]$ with $\ad_x(y) = [x,y]$.
Then, the following path-ordered exponential in the conventions of appendix \ref{sec:A.4}
generates eMPLs (\ref{reggt}) in one variable according to
\begin{align}
{\rm Pexp} \biggl(\int^0_z
 \tilde{\mathbb K}_{x,y}(z_1,\tau)\biggr) 
 &= 1 + \sum_{r=1}^\infty (-1)^r  \! \! \! \! \!\sum_{n_1,\ldots,n_r=0}^\infty \! \! \! \! \!
 (2\pi i)^{ \sum_{i=1}^r (1-n_i) }
 \tilde \Gamma\big( \smallmatrix n_1  &\ldots &n_r \\ 0  &\ldots &0 \endsmallmatrix ;z,\tau\big)
 \ad_x^{n_r}(y) \ldots  \ad_x^{n_1}(y)
 \notag \\
 &= 1- \sum_{n_1=0}^{\infty} (2\pi i)^{1-n_1}  \tilde \Gamma\big( \smallmatrix n_1 \\ 0 \endsmallmatrix ;z,\tau\big) \ad_x^{n_1}(y)
\notag \\
 &\quad + \sum_{n_1,n_2=0}^{\infty} (2\pi i)^{2-n_1-n_2}  \tilde \Gamma\big( \smallmatrix n_1 &n_2 \\ 0 &0 \endsmallmatrix ;z,\tau\big)
 \ad_x^{n_2}(y) \ad_x^{n_1}(y) - \ldots
 \label{genKempl}
\end{align}
with $\tilde \Gamma\big( \smallmatrix n_1  &\ldots &n_r \\ 0  &\ldots &0 \endsmallmatrix ;z,\tau\big)$ at $r\geq 3$ in the ellipsis, see (\ref{cctilK}) for the complex conjugate series.

%%%%%%%%%%%%%%%%%%%%%%%%%%%%%%%%%%%%%%%%%%%%%%%%%%
%%%%%%%%%%%%%%%%%%%%%%%%%%%%%%%%%%%%%%%%%%%%%%%%%%
\subsubsection{eMPLs from single-valued integration kernels}
\label{sec:2.1.2}

The second class (ii) of integration kernels for eMPLs in the above enumeration is
furnished by the doubly-periodic
Kronecker-Eisenstein coefficients $f^{(n)}(z,\tau)$ with $z \in \mathbb C$ and
$n\in \mathbb N_0$ defined by (\ref{appA.11}). By their antiholomorphic derivatives
in (\ref{pbarf}), a naive replacement $g^{(n)}(z,\tau) \rightarrow f^{(n)}(z,\tau)$ in the connection
$\tilde{\mathbb K}_{x,y}(z,\tau)$ in (\ref{defKcon}) would lead to a non-zero exterior derivative $\dd_z = \dd z \, \partial_z+ \dd \bar z \, \partial_{\bar z}$ and spoil its flatness.
This can be fixed by the additional $\dd \bar z$ component in the Brown-Levin connection \cite{BrownLev}\footnote{Our conventions for the generators $x,y$ can be mapped to those of \cite{BrownLev} by rescaling $(x,y) \rightarrow (2\pi i x, y/(2\pi i))$.}
\begin{align}
\mathbb J^{\rm BL}_{x,y}(z,\tau) &=  \frac{x}{2i \Im \tau}\, (\dd z - \dd \bar z)  + \dd z \, {\rm ad}_x \Omega(z,{\rm ad}_{ \frac{x}{2\pi i} },\tau) y 
\label{notsec.06} \\
&=  {-} \frac{x\,\dd \bar z }{2i \Im \tau} + \dd z \bigg\{ \frac{x}{2i \Im \tau}+2\pi i y + \sum_{n=1}^\infty (2\pi i)^{1-n} f^{(n)}(z,\tau) \ad_x^{n}(y)  \bigg\} \notag
\end{align}
The doubly-periodic Kronecker-Eisenstein series $\Omega(z,\alpha,\tau)$ is defined in (\ref{appA.11}),
and its antiholomorphic derivative in (\ref{pbarf}) implies the flatness condition away from $z=0$
\beq
\dd_z \mathbb J^{\rm BL}_{x,y}(z,\tau) = \mathbb J^{\rm BL}_{x,y}(z,\tau) \wedge \mathbb J^{\rm BL}_{x,y}(z,\tau) + \pi \, \dd z \wedge \dd \bar z \, [x,y] \, \delta^2(z)
\label{jblflt}
\eeq
Hence, the following path-ordered exponential only depends on the endpoints and the homotopy class
of the integration path,
\beq
{\rm Pexp} \biggl(\int^0_z \mathbb J^{\rm BL}_{x,y}(z_1,\tau)  \biggr) = \sum_{W \in \{ x,y\}^\times} \Gamma_{\rm BL}(W;z,\tau)
\label{jbexp}
\eeq
and defines an eMPL $\Gamma_{\rm BL}(W;z,\tau)$ for each word $W$ in the letters $x,y$. The Brown-Levin eMPLs are expressible in terms of the meromorphic eMPLs 
$\tilde \Gamma(\smallmatrix \ldots \\ \ldots \endsmallmatrix ;z,\tau)$ in (\ref{defempl}):
introducing the real ``co-moving'' coordinates $u,v$ of the marked point $z \in T^2$,
\beq
z = u \tau {+} v \, , \ \ \ \ \ \ u,v  \in \mathbb R
\label{comov}
\eeq
Brown-Levin eMPLs $\Gamma_{\rm BL}(W;z,\tau)$ are $\mathbb Q[(2\pi i)^{\pm 1}, u]$-linear combinations
of $\tilde \Gamma(\smallmatrix \ldots \\ \ldots \endsmallmatrix ;z,\tau)$. This can be seen from the
fact that the CEE and Brown-Levin connections (\ref{defKcon}) and (\ref{notsecCEE}) are related 
by the gauge transformation \cite{DHoker:2025szl}
\begin{align}
 \mathbb J^{\rm BL}_{x,y}(z,\tau)  =
e^{u x} \tilde{\mathbb K}_{x,y}(z,\tau) e^{-u x}  + (\dd e^{ux}) e^{-ux} 
\label{notsecCEE}
\end{align}
As a consequence, the path-ordered exponentials (\ref{genKempl}) and (\ref{jbexp})
are related by
\beq
{\rm Pexp} \biggl(\int^0_z \mathbb J^{\rm BL}_{x,y}(z_1,\tau)  \biggr)  = {\rm Pexp} \biggl(\int^0_z \tilde{\mathbb K}_{x,y}(z_1,\tau)  \biggr)  e^{-ux}
\label{gaugepexp}
\eeq
which generates the relations between the respective eMPLs upon expansion in words in $x,y$.

It is worth noting that the variant of the Brown-Levin connection
\begin{align}
\mathbb J^{\rm mod}_{x,y}(z,\tau) &=
\mathbb J^{\rm BL}_{x,y}(z,\tau) - \frac{x \, \dd z }{2i \Im \tau} = \mathbb J^{\rm BL}_{x,y+\frac{x}{4\pi \Im \tau}}(z,\tau)  \notag \\
&=  - \frac{x \,  \dd \bar z }{2i \Im \tau} + \dd z \, {\rm ad}_x \Omega(z,{\rm ad}_{ \frac{x}{2\pi i} },\tau) y 
\label{notsec.07}
\end{align}
exhibits the modularity property \cite{DHoker:2023vax}\footnote{The connection in Theorem 3.2 of the reference specializes to $\mathbb J^{\rm mod}_{x,y}(z,\tau)$ in (\ref{notsec.07}) at genus one upon setting $\hat a = 2\pi i y$ and $b = \frac{x}{2\pi i}$.}
\beq
\mathbb J^{\rm mod}_{x,y}\bigg(\frac{z}{c\tau{+}d},\frac{a\tau{+}b}{c\tau{+}d} \bigg) = 
\mathbb J^{\rm mod}_{(c\tau{+}d)x,\frac{y}{c\tau{+}d}}(z,\tau) \, , \ \ \ \ \ \
\big( \smallmatrix  a &b \\ c &d \endsmallmatrix \big) \in {\rm SL}_2(\mathbb Z)
\label{jmodtrf}
\eeq
This can be traced back to the modularity of the doubly-periodic Kronecker-Eisenstein kernels
\beq
 f^{(n)} \bigg(\frac{z}{c\tau {+}d},\frac{a\tau{+}b}{c\tau {+}d} \bigg) = 
(c\tau{+}d)^n\,  f^{(n)}(z,\tau)
 \label{appA.13}
\eeq
and implies that after the redefinition $y \rightarrow y+\frac{x}{4\pi \Im \tau}$, 
the combinations of Brown-Levin eMPLs in the expansion of
(\ref{jbexp}) are modular forms. However, this modularity argument
needs to be refined (see appendix \ref{app:gmod}) to also take the regularization
of endpoint divergences through the $\varepsilon$ cutoff of (\ref{reggt})
into account.

%%%%%%%%%%%%%%%%%%%%%%%%%%%%%%%%%%%%%%%%%%%%%%%%%%
%%%%%%%%%%%%%%%%%%%%%%%%%%%%%%%%%%%%%%%%%%%%%%%%%%
\subsection{Elliptic modular graph forms}
\label{sec:2.2}

The function space of single-valued eMPLs to be constructed in this work 
aims to capture a class of modular $T^2$ integrals in the string-theory literature
dubbed {\it elliptic modular graph forms} (eMGFs) 
\cite{DHoker:2018mys, DHoker:2020tcq, Basu:2020pey, Basu:2020iok, Dhoker:2020gdz}.
They generalize Zagier's single-valued eMPLs \cite{Ramakrish} beyond depth one and
were firstly related to single-valued MPLs at genus zero in \cite{DHoker:2015wxz}.
The differential equations of eMGFs \cite{Dhoker:2020gdz}
were solved in terms of iterated integrals over modular parameters for a variety of cases 
\cite{Hidding:2022vjf}. A key result of this work is to establish this iterated-integral 
structure at the level of generating series, namely to construct an infinite family
of non-holomorphic modular forms with the properties of eMGFs from iterated
integrals and their complex conjugates. The link of the generating series in section \ref{sec:3} to eMGFs is discussed in section \ref{sec:5}.

%%%%%%%%%%%%%%%%%%%%%%%%%%%%%%%%%%%%%%%%%%%%%%%%%%
%%%%%%%%%%%%%%%%%%%%%%%%%%%%%%%%%%%%%%%%%%%%%%%%%%
\subsubsection{Definition of dihedral eMGFs}
\label{sec:2.2.1}

Simple prototypes of eMGFs are furnished by the meromorphic lattice sums that produce holomorphic 
Eisenstein series ${\rm G}_k(\tau)$ and Kronecker-Eisenstein coefficients $f^{(k)}(z,\tau)$
at fixed comoving coordinates $u,v\in \mathbb R$ of $z=u\tau{+}v \in\mathbb C$,
\begin{align}
{\rm G}_k(\tau) &=  \sum_{(m,n) \in \Lambda'_\tau} \frac{1}{(m\tau{+}n)^k} \, , \ \ \ \ k\geq 3
\label{latsumgf} \\
f^{(k)}(z,\tau) &= - \sum_{(m,n) \in \Lambda'_\tau} \frac{e^{2\pi i (nu-mv) }}{(m\tau{+}n)^k} \, , \ \ \ \ k\geq 2
\notag 
\end{align}
where $z\neq \mathbb Z{+}\tau \mathbb Z$ in case of $k=2$ and the summation range for the integers $m,n$ is given by
\beq
\Lambda'_\tau = \mathbb Z^2 \setminus \{(0,0)\}
\label{deflat}
\eeq
Generic eMGFs generalize these lattice sums to feature additional
anti-holomorphic factors $(m\bar \tau{+}n)^{-1}$ in the summand and/or multiple pairs $(m_1,n_1), (m_2,n_2),\ldots \in \Lambda'_\tau$ 
of summation variables. Zagier's single-valued elliptic polylogarithms \cite{Ramakrish}
\beq
\dplus{a \\ b}(z,\tau) = \frac{(\Im \tau)^a}{\pi^b} \sum_{(m,n) \in \Lambda'_\tau} \frac{e^{2\pi i (nu-mv) }}{(m\tau{+}n)^a (m\bar \tau{+}n)^b}
\label{defZag}
\eeq
for instance are defined for arbitrary $z = u\tau{+}v$ and integer exponents $a,b\in \mathbb N_0$ with $a{+}b\geq 2$
(using the second Kronecker limit formula for $a{+}b= 2$) and related to meromorphic eMPLs (\ref{defempl})
and their complex conjugates in \cite{Broedel:2019tlz}. They are special cases of the dihedral eMGFs \cite{Dhoker:2020gdz}
\beq
\cplus{a_1 &\ldots &a_r \\ b_1 &\ldots &b_r \\ z_1 &\ldots &z_r}(\tau) = \frac{(\Im \tau)^{a_1+\ldots+a_r}}{\pi^{b_1+\ldots+b_r}} \! \!  \! \!  \! \!   \sum_{(m_1,n_1),\ldots,(m_r,n_r) \in \Lambda'_\tau}   \! \!  \! \!  \! \! \delta\bigg(\sum_{j=1}^r (m_j, n_j) \bigg) \prod_{i=1}^r  \frac{e^{2\pi i (n_i u_i -m_i v_i) }}{(m_i\tau{+}n_i)^{a_i} (m_i\bar \tau{+}n_i)^{b_i}}
\label{defdiC}
\eeq
with $r\in \mathbb N$, $a_i,b_i \in \mathbb N_0$ and $z_i = u_i \tau{+} v_i \in \mathbb C$
with $u_i,v_i \in \mathbb R$ for $i=1,2,\ldots,r$.
By the Kronecker delta that enforces $\sum_{j=1}^r(m_j,n_j) = (0,0)$, the multiple sums in 
(\ref{defdiC}) are absolutely convergent if $a_i{+}a_j{+} b_i {+}b_j > 2$ for any pair $(i,j)$ 
in the range $1\leq i<j\leq r$. Cases
with $a_i{+}a_j{+} b_i {+}b_j = 2$ for some of these pairs $(i,j)$ in turn
can be dealt with through the second Kronecker limit formula.
Note that (\ref{latsumgf}), (\ref{defZag}) and (\ref{defdiC}) are connected through the special cases
\begin{align}
\dplus{a \\ 0}(z,\tau) &= - (\Im \tau)^a  f^{(a)}(z,\tau) \label{speceMGF}
\\
\cplus{a_1 &a_2 \\ b_1 &b_2 \\ z_1 &z_2}(\tau)  &= (-1)^{a_2 + b_2} \dplus{a_1{+}a_2 \\ b_1{+}b_2}(z_1{-}z_2,\tau)
\notag
\end{align}
and eMGFs (\ref{defdiC}) with a single column $r=1$ vanish due to their empty summation range.
The dependence of (\ref{defZag}) and (\ref{defdiC}) on $z_j=u_j\tau{+}v_j$ solely enters
through the Fourier modes $\sim e^{2\pi i (n_j u_j -m_j v_j)}$ with $n_j,m_j \in \mathbb Z$.
This manifests doubly-periodicity of  both ${\cal D}^+ \! [\ldots](z,\tau)$ and ${\cal C}^+ \! [\ldots](\tau)$
under $z_j \rightarrow z_j{+}1$ and $z_j \rightarrow z_j{+}\tau$, i.e.\ single-valuednesss in all points $z_j \in T^2$.
Moreover, the lattice sums expose the modular properties analogous to those of $f^{(n)}$ in (\ref{appA.13}),
\begin{align}
\dplus{a \\ b}\bigg(\frac{z}{\gamma\tau{+}\delta} ,\frac{\alpha\tau{+} \beta}{\gamma\tau {+}\delta} \bigg) &= (\gamma\bar \tau{+}\delta)^{b-a} \dplus{a \\ b}(z,\tau)  \, , \ \ \ \ \big( \smallmatrix \alpha & \beta \\  \gamma & \delta \endsmallmatrix \big) \in {\rm SL}_2(\mathbb Z)
\notag \\
\cplus{a_1 &\ldots &a_r \\ b_1 &\ldots &b_r \\ z_1/(\gamma\tau{+}\delta) &\ldots &z_r/(\gamma\tau{+}\delta)}\bigg(\frac{\alpha\tau{+} \beta}{\gamma\tau {+}\delta} \bigg) &=(\gamma\bar \tau{+}\delta)^{b_1+\ldots+b_r-a_1-\ldots-a_r}  \cplus{a_1 &\ldots &a_r \\ b_1 &\ldots &b_r \\ z_1  &\ldots &z_r}(\tau)
 \label{modeMGF}
\end{align}
The powers of $\Im \tau$ in our
normalization conventions for (\ref{defZag}) and (\ref{defdiC}) are chosen to ensure
that modular transformations only feature antiholomorphic factors of 
$(\gamma\bar \tau{+}\delta)$ as opposed to holomorphic ones $(\gamma\tau{+}\delta)$, i.e.\
that eMGFs are non-holomorphic modular forms of weight $(0, \sum_{i=1}^r (b_i{-}a_i))$.

%%%%%%%%%%%%%%%%%%%%%%%%%%%%%%%%%%%%%%%%%%%%%%%%%%
%%%%%%%%%%%%%%%%%%%%%%%%%%%%%%%%%%%%%%%%%%%%%%%%%%
\subsubsection{General definition of eMGFs}
\label{sec:2.2.2}

In order to motivate the word {\it graph} in the terminology for eMGFs, we note that
their dihedral instances (\ref{defdiC}) admit representations as convolution integrals over the torus \cite{Dhoker:2020gdz}
\beq
\cplus{a_1 &\ldots &a_r \\ b_1 &\ldots &b_r \\ z_1 &\ldots &z_r}(\tau) = \int_{T^2} \frac{\dd^2 z}{\Im \tau}\,
\prod_{i=1}^r \dplus{a_i \\ b_i}(z_i{-}z,\tau)
\label{intdih}
\eeq
where the modular invariant measure with $\dd^2 z= -\frac{i}{2} \dd z \wedge \dd \bar z$ is normalized
to yield 
\beq
\int_{T^2} \frac{\dd^2 z}{\Im \tau} \, e^{2\pi i (nu-mv) }
= \int_{0}^1  \dd u \int_0^1 \dd v \, e^{2\pi i (nu-mv) }
= \delta_{n,0} \delta_{m,0}
\label{fourint}
\eeq
One can associate graphs with eMGFs by assigning
\begin{itemize}
 \item different types of vertices to integrated points $z$ and unintegrated ones $z_i \in T^2$;
 \item directed edges with decoration $(a_i, b_i) \in \mathbb N_0^2$ 
between vertices $z_i,z$ for each factor of $\dplus{a_i \\ b_i}(z_i{-}z,\tau)$
in the integrand of (\ref{intdih}).
\end{itemize} 
Generic eMGFs beyond the dihedral class in 
(\ref{defdiC}) and (\ref{intdih}) are obtained by integrating products of 
$\dplus{a_{ij} \\ b_{ij}}(z_i{-}z_j,\tau)$ with $a_{ij},b_{ij} \in \mathbb N_0$ over several points
$z_1,\ldots,z_r \in T^2$, i.e.\ through integrals of the form
\beq
\prod_{k=1}^r \int_{T^2} \frac{\dd^2 z_k}{\Im \tau} \,
 \prod_{i=1}^{r}   \prod_{j=i+1}^{r+s} 
\dplus{a_{ij} \\ b_{ij}}(z_i{-}z_j,\tau) 
\label{genint}
\eeq
which may depend on $s \in \mathbb N_0$ unintegrated points $z_{r+1},\ldots,z_{r+s} \in T^2$.
By repeated use of (\ref{fourint}), one can straightforwardly convert  the
integral representation (\ref{genint}) of general eMGFs into lattice sums. When interpreting
$m\tau{+}n$ and $m\bar\tau{+}n$ as the components of torus momenta
on the lattice (\ref{deflat}), then each integration over $z_k$ imposes momentum
conservation at the respective vertex of the associated graph.
For simple topologies of the graph associated with the integrands of (\ref{genint}), one
recovers the expression (\ref{defdiC}) for dihedral eMGFs, and the next more general
case of trihedral eMGFs is discussed in appendix A of \cite{Dhoker:2020gdz}.

Similar to our terminology for eMPLs (\ref{defempl}) with $a_i = 0$, we refer to eMGFs
(\ref{defdiC}) or (\ref{genint}) where all the unintegrated $z_i$ are equal to $0$ or the 
same point $z\in T^2$ as {\it eMGFs in one variable}. Integrals of the form (\ref{genint})
where all the unintegrated $z_i$ vanish are known as {\it modular graph forms} (MGFs) 
\cite{DHoker:2015wxz, DHoker:2016mwo}
which systematize the low-energy expansion of closed-string amplitudes at genus one 
\cite{Green:1999pv, Green:2008uj, DHoker:2015gmr}, see \cite{Gerken:review, DHoker:2024book} for an overview. 

%%%%%%%%%%%%%%%%%%%%%%%%%%%%%%%%%%%%%%%%%%%%%%%%%%
%%%%%%%%%%%%%%%%%%%%%%%%%%%%%%%%%%%%%%%%%%%%%%%%%%
\subsubsection{Differential equations of eMGFs}
\label{sec:2.2.3}

Both the integral and the lattice-sum representation of eMGFs
manifest their double periodicity in the unintegrated $z_i$ and their ${\rm SL}_2(\mathbb Z)$
transformations as modular forms of weight $(0, \sum_{i=1}^r \sum_{j=i+1}^{r+s} (b_{ij}{-}a_{ij}))$ 
as in (\ref{modeMGF}). However, these representations make it difficult to
foresee the wealth of algebraic relations among eMGFs
\cite{DHoker:2020tcq, Basu:2020pey, Basu:2020iok} or their expansion around the
cusp $\tau \rightarrow i \infty$.

Similar to the earlier investigation of MGFs \cite{DHoker:2015gmr, DHoker:2016mwo, DHoker:2016quv}\footnote{See \cite{Gerken:2020aju} for a {\sc Mathematica} package implementing relations among MGFs.},
relations among eMGFs can be conveniently found by studying their differential equations in the 
modular parameter \cite{Dhoker:2020gdz}. Their lattice-sum representations as in (\ref{defdiC}) are a convenient starting point to
evaluate their holomorphic $\tau$-derivatives at fixed co-moving coordinates $u_i,v_i$ of the points $z_i$,
\beq
2\pi i (\Im \tau)^2 \partial_\tau \cplus{a_1 &\ldots &a_r \\ b_1 &\ldots &b_r \\ z_1 &\ldots &z_r}(\tau) = 
\sum_{i=1}^r a_i \,
 \cplus{a_1 &\ldots &a_i {+} 1 &\ldots &a_r \\ b_1 &\ldots &b_i {-} 1 &\ldots &b_r \\ z_1 &\ldots &z_i &\ldots &z_r}(\tau) 
\label{nabdiC} 
\eeq
Similar formulae arise for derivatives in the marked points,
\beq
(\Im \tau)  \partial_z \cplus{a_1 &\ldots &a_r &a_{r+1} &\ldots &a_{r+s} \\ b_1 &\ldots &b_r &b_{r+1} &\ldots &b_{r+s} \\ z_1 &\ldots &z_r &z &\ldots &z}(\tau) =  \sum_{i=r+1}^{r+s} 
 \cplus{a_1 &\ldots &a_r  &a_{r+1} &\ldots &a_{i} &\ldots &a_{r+s} \\ b_1 &\ldots &b_r
 &b_{r+1} &\ldots &b_{i}-1 &\ldots &b_{r+s} \\ z_1 &\ldots &z_r &z &\ldots &z &\ldots &z}(\tau) 
\label{dzdiC} 
\eeq
assuming $z_1,\ldots,z_r$ to be unaffected by $\partial_z$. The antiholomorphic derivatives $\partial_{\bar \tau}$
and $\partial_{\bar z}$ can be derived from (\ref{nabdiC}) and (\ref{dzdiC}) through the complex-conjugation property
\beq
\overline{\cplus{a_1 &\ldots &a_r \\ b_1 &\ldots &b_r \\ z_1 &\ldots &z_r}(\tau) }
= \bigg( \prod_{j=1}^r (\pi \Im \tau)^{a_j-b_j} \bigg) \cplus{ b_1 &\ldots &b_r \\ a_1 &\ldots &a_r \\ -z_1 &\ldots &-z_r}(\tau) 
\label{ccemgf}
\eeq
The action of both $\partial_\tau$ and $\partial_z$
lowers the exponents $b_i$ of the antiholomorphic factors $(m_i\bar \tau{+}n_i)^{-1}$ in the lattice sum (\ref{defdiC})
and thereby gradually simplifies the eMGF: repeated derivatives in $\tau$ or $z$ inevitably lead to the situation
where two or more entries $b_i$ in the second line of $\cplus{a_1 &\ldots &a_r \\ b_1 &\ldots &b_r \\ z_1 &\ldots &z_r}(\tau)$ vanish. In these cases, the technique of {\it holomorphic subgraph reduction}, initially developed for MGFs in \cite{DHoker:2016mwo, Gerken:2018zcy} and extended to eMGFs in \cite{Dhoker:2020gdz}, 
leads to simplifications of the form
\beq
 \cplus{a_1 &a_2 &A \\ 0 &0 &B  \\ z_1 &z_2 &Z}(\tau) = 
 \sum^{a_1+a_2}_{k=0} (\Im \tau)^k f^{(k)}(z_1{-}z_2,\tau)  \cplus{A_k \\ B_k  \\ Z_k}(\tau) 
\label{HSRemgf} 
\eeq
where $A=[a_3,\ldots,a_r]$, $B=[b_3,\ldots,b_r]$ and $Z=[z_3,\ldots,z_r]$ gather an arbitrary number of entries $a_i,b_i \in \mathbb N_0$ and $z_i \in \mathbb C$. The shorthands $A_k,B_k, Z_k$ 
on the right side have a similar meaning, where $ \cplus{A_k \\ B_k  \\ Z_k}(\tau) $ is
a placeholder for a combination of dihedral eMGFs which can be determined 
from the techniques in section 3.3 of \cite{Dhoker:2020gdz}. The entries $B$ on the left
side of (\ref{HSRemgf}) are assumed to contain no additional zeros besides the two exposed ones,
and the $B_k$ on the right side have at most one vanishing entry. Taking further derivatives of these 
$ \cplus{A_k \\ B_k  \\ Z_k}(\tau) $ in $z$ and $\tau$ via (\ref{nabdiC}) and (\ref{dzdiC}) leads to
further instances of repeated zeros in the second line and new opportunities to perform
holomorphic subgraph reduction (\ref{HSRemgf}). Note that holomorphic subgraph reduction
can also be understood as applications of Fay identities among bilinears of $f^{(k)}$ in
the integrands of (\ref{intdih}) and (\ref{genint}) \cite{Gerken:2020aju}.

All of the steps in (\ref{nabdiC}), (\ref{dzdiC}) and (\ref{HSRemgf}) generalize beyond the dihedral topology,
i.e.\ to the most general eMGFs in (\ref{genint}). Iterations of these steps are referred
to as the {\it sieve algorithm} which was originally developed for MGFs \cite{DHoker:2016mwo} and is explained in section 3.4 of \cite{Dhoker:2020gdz}
to shed light on both the differential structure and algebraic relations of eMGFs.

%%%%%%%%%%%%%%%%%%%%%%%%%%%%%%%%%%%%%%%%%%%%%%%%%%
%%%%%%%%%%%%%%%%%%%%%%%%%%%%%%%%%%%%%%%%%%%%%%%%%%
\subsubsection{Towards iterated integrals for eMGFs}
\label{sec:2.2.4}

The differential equations and reality properties of eMGFs outlined in the previous section
implies that their dependence on $\tau$ at fixed co-moving coordinates $u_i,v_i \in \mathbb R$
can be captured via holomorphic iterated integrals of $\dd \tau\, f^{(k)}((u_i{-}u_j)\tau{+} (v_i{-}v_j),\tau)$ and their complex conjugates.
In the one-variable case where all the $z_i$ are either zero or given by the same $u\tau{+}v$,
this amounts to three types of integration kernels
\beq
\dd \tau \ \ \ \ {\rm and} \ \ \ \ 
\dd \tau \, f^{(k)}(u\tau{+} v,\tau) \, , \ \ k\geq 2 \ \ \ \ {\rm and} \ \ \ \ 
\dd \tau \, {\rm G}_k(\tau) \, , \ \ k\geq 4
\label{emgfkers}
\eeq
at fixed $u,v \in \mathbb R$, where the bounds on $k$ can for instance be seen from 
the generating-function methods in \cite{Dhoker:2020gdz}. As is familiar from iterated
integrals over holomorphic Eisenstein series \cite{brown2017multiple}, it is convenient to trade the
first kernel $\dd \tau$ in (\ref{emgfkers}) for insertions of non-negative powers of $\tau$ along
with the remaining kernels $\dd \tau \, f^{(k)}(u\tau{+} v,\tau)$ and $\dd \tau \, {\rm G}_k(\tau)$.\footnote{See
for instance \cite{Broedel:2018izr, Hidding:2022vjf, Dorigoni_2022} for conversion techniques.} Accordingly, 
we shall employ the iterated Kronecker-Eisenstein integrals with $k_i\geq 2$
\begin{align}
\eez{j_1 &\ldots &j_r}{k_1 &\ldots &k_r}{z_1 &\ldots &z_r}{\tau} &= (2\pi i)^{1+j_r - k_r} \int^\tau_{i\infty} \dd \tau_r \, \tau_r^{j_r} f^{(k_r)}(u_r \tau_r {+} v_r,\tau_r) \eez{j_1 &\ldots &j_{r-1}}{k_1 &\ldots &k_{r-1}}{z_1 &\ldots &z_{r-1}}{\tau_r}
\, , 
 \ \ \ \ \eez{\emptyset}{\emptyset}{\emptyset}{\tau}  = 1
\label{bsc.07} 
\end{align}
to represent eMGFs as initiated in \cite{Hidding:2022vjf}, and the
integration is performed at fixed $u_r,v_r \in \mathbb R$ to ensure homotopy invariance.
Here and throughout, the endpoint divergences of iterated $\tau$ integrals from the
region where $\tau_r\rightarrow i\infty$ are shuffle regularized through the tangential-base-point
method of \cite{brown2017multiple} whose net effect is $\int^\tau_{i\infty} \tau_r^n \dd \tau_r
= \frac{\tau^{n+1}}{n{+}1}$ for $n\in \mathbb N_0$.

We will only consider the iterated integrals (\ref{bsc.07}) with entries $0 \leq j_i \leq k_i{-}2$
in the upper row to ensure that the $(k{-}1)$-tuples
of differential forms $\dd \tau \, \tau^j f^{(k)}(u\tau{+} v,\tau)$ at fixed $k$ closes 
under ${\rm SL}_2(\mathbb Z)$ without introducing negative powers of $\tau$.

Finally, integration kernels $\dd \tau \, \tau^j {\rm G}_k(\tau)$ with $k\geq 4$ even and
$0\leq j \leq k{-}2$ are represented by means of empty slots in the last line of
\begin{align}
\eez{j_1 &\ldots &j_{r-1} &j_r}{k_1 &\ldots &k_{r-1}&k_r}{z_1 &\ldots &z_{r-1} &}{\tau} &= - (2\pi i)^{1+j_r - k_r} \int^\tau_{i\infty} \dd \tau_r \, \tau_r^{j_r} {\rm G}_{k_r}(\tau_r) \eez{j_1 &\ldots &j_{r-1}}{k_1 &\ldots &k_{r-1}}{z_1 &\ldots &z_{r-1}}{\tau_r}
\label{absc.07}
\end{align}
where the same recursive definition holds in case of empty slots in the place of some of the
$z_1,\ldots,z_{r-1}$. The extra minus sign in (\ref{absc.07}) relative to (\ref{bsc.07}) is due to the consequence
\beq
{\rm G}_k(\tau) = - f^{(k)}(0,\tau) \, , \ \ \ \ \ \ k\geq 3
\eeq
of the lattice-sum representations (\ref{latsumgf}). However, note that $f^{(2)}(z,\tau)$ does not have a well-defined
$z \rightarrow 0$ limit. Moreover, the limit $z_r \rightarrow 0$ limit in the integrand of (\ref{bsc.07})
generically does not commute with integration of $\dd \tau_r \, (\tau_r)^j f^{(k)}(u\tau_r{+} v,\tau_r)$,
see section 3.2 of \cite{Hidding:2022vjf} for simple counterexamples.
That is why we represent $\dd \tau_r \, (\tau_r)^j {\rm G}_j(\tau_r)$ kernels through a column
$\smallmatrix j_r \\ k_r \\ \endsmallmatrix$ rather than $\smallmatrix j_r \\ k_r \\ 0 \endsmallmatrix$ in the square bracket of (\ref{absc.07}). Note that section \ref{sec:itinemgf} organizes the iterated integrals ${\cal E}[\ldots;\tau]$ of (\ref{bsc.07}) and (\ref{absc.07}) into specific linear combinations $\beta_{\pm}[\ldots;\tau]$ of  tailored to the modularity properties of eMGFs.

Since eMGFs close under complex conjugation up to powers of $\Im \tau$, see for instance 
(\ref{ccemgf}), their $\bar \tau$-dependence can be accounted for through the complex conjugates
of the iterated integrals (\ref{bsc.07}) and (\ref{absc.07}). Hence, one can attain iterated
$\tau$-integral representations of eMGFs by iteratively integrating their total differential
$\dd_\tau = \dd \tau \, \partial_\tau + \dd \bar \tau \, \partial_{\bar \tau}$,
\beq
 \cplus{A \\ B  \\ Z}(\tau)   =  \cplus{A \\ B  \\ Z}(i \infty) + \int^\tau_{i\infty} \dd_{\tau_1}  \cplus{A \\ B  \\ Z}(\tau_1)  
 \label{totdtau}
\eeq
provided that their $u_i,v_i$-dependent integration constants at $\tau \rightarrow i \infty$ are known.
These integration constants turn out
to comprise multiple zeta values and Bernoulli polynomials in $u$
which compensate for the fact that the iterated integrals 
(\ref{bsc.07}) and (\ref{absc.07}) generically do not transform as modular forms of ${\rm SL}_2(\mathbb Z)$.
A systematic method to determine these integration constants is described in  \cite{Hidding:2022vjf} along
with applications to a variety of eMGFs. The generating series $\mathbb I^{\rm eqv}_{\ep,b}(u,v,\tau)$ 
to be presented in section \ref{sec:3.1} below determines the completions of
iterated integrals (\ref{bsc.07}) in the one-variable case $z_i=z$ with 
arbitrary $k_i,j_i$ satisfying $0 \leq j_i \leq k_i{-}2$ to attain the desired modular properties.

An alternative path towards iterated-integral representations of eMGFs is to
integrate their differential equations in $z$ instead of those in $\tau$ via (\ref{totdtau}).
By the $z$-derivatives (\ref{dzdiC}) of dihedral ${\cal C}^+$ together with holomorphic subgraph reduction (\ref{HSRemgf}), this alternative 
path expresses eMGFs in terms of iterated integrals of $\dd z \, f^{(k)}(z,\tau)$ 
and $\dd \bar z \,\overline{ f^{(k)}(z,\tau)}$,
i.e.\ Brown-Levin eMPLs and their complex conjugates. The integration constants
analogous to the ${\cal C}^+\![\ldots](i \infty)$ in (\ref{totdtau}) are then furnished
by MGFs obtained from the $z_i=0$ limit of eMGFs and eventually elliptic multiple
zeta values (\ref{defemzv}) as anticipated in talks by Panzer \cite{Panzertalk}. Representations
of eMGFs in terms of eMPLs are only known for the depth-one case of the
$\dplus{a \\ b}(z,\tau) $ in (\ref{defZag}) using the meromorphic formulation via $\tilde \Gamma$
in (\ref{defempl}) \cite{Broedel:2019tlz}.
A key result of this work is to spell out the conversion of iterated $\tau$ integrals
for general one-variable eMGFs into modular and single-valued combinations of
eMPLs at the level of generating series, see section \ref{sec:3.2}. 

In summary, the differential equations of eMGFs admit two types of iterated-integral 
representations -- either via $\tau$-integrals of Kronecker-Eisenstein kernels and 
holomorphic Eisenstein series (\ref{bsc.07}) and (\ref{absc.07}) or via eMPLs.
In both cases, eMGFs realize modular and single-valued combinations of the respective iterated
integrals and their complex conjugates. Moreover, both types of iterated-integral 
representations expose all algebraic relations among eMGFs
since the relations among the respective iterated integrals are completely known.
In the remainder of this work, we shall construct the modular and single-valued combinations of both types
of iterated integrals that enter one-variable eMGFs.

%%%%%%%%%%%%%%%%%%%%%%%%%%%%%%%%%%%%%%%%%%%%%%%%%%
%%%%%%%%%%%%%%%%%%%%%%%%%%%%%%%%%%%%%%%%%%%%%%%%%%
\subsection{Single-valued genus-zero polylogarithms}
\label{sec:2.3}

The problem of determining modular combinations of iterated integrals and their complex conjugates
for eMGFs can be viewed as the genus-one analogue of constructing single-valued MPLs at genus zero.
We shall here review the generating series for single-valued versions of
MPLs $G(a_1,\ldots,a_r;z)$ in one variable ($a_i \in \{0,1\}$) and two variables
($a_i \in \{0,1,y\}$ with $y\in \mathbb C \setminus\{0,1\}$). In order to set the stage for
the Lie-algebra methods employed in our genus-one results of later sections,
we shall focus on the recent reformulation \cite{Frost:2023stm, Frost:2025lre} of the earlier constructions
of single-valued MPLs in one variable \cite{svpolylog} or $\geq 2$ variables \cite{Broedel:2016kls, DelDuca:2016lad}.

%%%%%%%%%%%%%%%%%%%%%%%%%%%%%%%%%%%%%%%%%%%%%%%%%%
%%%%%%%%%%%%%%%%%%%%%%%%%%%%%%%%%%%%%%%%%%%%%%%%%%
\subsubsection{Generating series of MPLs and braid algebra}
\label{sec:2.3.1}

The meromorphic but multi-valued MPLs in one and two variables will be organized in the generating
series
\begin{align}
\mathbb G_{e_0,e_1}(z) &= {\rm Pexp}\bigg({-} \int^0_z \dd t \,\bigg[  \frac{e_0}{t} +  \frac{e_1}{t{-}1} \bigg] \bigg)\label{svmpl.00} \\
&=  \sum_{r=0}^\infty \sum_{a_1,\ldots,a_r \atop {\in \{ 0,1\}}} G(a_r,\ldots,a_2,a_1;z) e_{a_1}  e_{a_2 }\ldots e_{a_r}
\notag
\end{align}
and 
\begin{align}
\mathbb G_{e_0',e_1',e_y'}(y,z) &= {\rm Pexp}\bigg({-} \int^0_z \dd t \,\bigg[  \frac{e'_0}{t} +  \frac{e'_1}{t{-}1} +  \frac{e'_y}{t{-}y} \bigg] \bigg)\label{svmpl.01} \\
&=  \sum_{r=0}^\infty \sum_{a_1,\ldots,a_r \atop {\in \{ 0,1,y\}}} G(a_r,\ldots,a_2,a_1;z) e'_{a_1}  e'_{a_2 }\ldots e'_{a_r}
\notag
\end{align}
where both collections of non-commutative braid operators $e_0,e_1$ and 
$e_0', e_1', e_y'$ individually generate free Lie algebras
${\rm Lie}[e_0,e_1]$ and ${\rm Lie}[e_0', e_1', e_y']$. However, the commutators that
mix primed generators $e_i'$ with unprimed ones $e_j$ close on the primed generators
\begin{align}
[e_0' , e_0] &= 0 \, , &[e_0',e_1] &= 0 
\label{svmpl.02} \\
[e_1', e_0] &= [e_1', e_y'] \, , &[e_1', e_1] &= [e_y', e_1'] \notag \\
[e_y' , e_0] &= [e_0', e_y'] \, , &[e_y', e_1] &= [e_1',e_y'] \notag
\end{align}
i.e.\ both of ${\rm ad}_{e_0}$ and ${\rm ad}_{e_1}$ are derivations on ${\rm Lie}[e_0', e_1', e_y']$.
The bracket relations (\ref{svmpl.02}) follow from imposing flatness of the Knizhnik--Zamolodchikov (KZ) connection in two variables
\beq
\mathbbm{\Omega}^{\rm KZ}_{e,e'}(y,z) = \dd z\, \bigg( \frac{e_0'}{z} + \frac{e_1'}{z{-}1} + \frac{e_y'}{z{-}y} \bigg)
+ \dd y\, \bigg( \frac{e_0{-}e_y'}{y} + \frac{e_1}{y{-}1} + \frac{e_y'}{y{-}z} \bigg)
\label{svmpl.03}
\eeq
that features in the differential equation of the concatenation product of 
(\ref{svmpl.00}) and (\ref{svmpl.01}),
\beq
\dd \big(\mathbb G_{e_0,e_1}(y) \,\mathbb G_{e_0',e_1',e_y'}(y,z) \big) =
\mathbb G_{e_0,e_1}(y)\,  \mathbb G_{e_0',e_1',e_y'}(y,z) \, \mathbbm{\Omega}^{\rm KZ}_{e,e'}(y,z) 
\label{svmpl.04}
\eeq
The residue $e_0{-}e_y'$ of the pole of $\mathbbm{\Omega}^{\rm KZ}_{e,e'}(y,z)$ in $y$
is engineered such that the $\dd y$ component of $\mathbbm{\Omega}^{\rm KZ}_{e,e'}(y,z)$ reduces
 to the standard form $\dd y\, (\frac{e_0}{y}+\frac{e_1}{y-1})$ of a KZ connection in 
 one variable upon setting $z\rightarrow 0$ in (\ref{svmpl.03}).

%%%%%%%%%%%%%%%%%%%%%%%%%%%%%%%%%%%%%%%%%%%%%%%%%%
%%%%%%%%%%%%%%%%%%%%%%%%%%%%%%%%%%%%%%%%%%%%%%%%%%
\subsubsection{Generating series of MZVs and zeta generators}
\label{sec:2.3.2}

Already for single-valued MPLs in one variable, the coefficients of the
composing multi-valued MPLs and their complex conjugates are
$\mathbb Q$-linear combinations of single-valued MZVs \cite{svpolylog}.
We shall gather single-valued MZVs in the generating series
\begin{align}
\mathbb M^{\rm sv}_{0} &= \sum_{r=0}^{\infty} \sum_{i_1,\ldots,i_r \atop {\in 2\mathbb N+1} } \rho^{-1}\big({\rm sv}(f_{i_1} \ldots f_{i_r})\big) \, M_{i_1}\ldots  M_{i_r}
\label{g0.05} \\
&= 1 + 2 \sum_{i_1 \in 2\mathbb N+1} \zeta_{i_1} M_{i_1}
+ 2 \sum_{i_1,i_2 \in 2\mathbb N+1} \zeta_{i_1} \zeta_{i_2}M_{i_1} M_{i_2} + \ldots
\notag
\end{align}
with (conjecturally) indecomposable MZVs of depth $\geq 2$ in the ellipsis.
The isomorphism $\rho$ converts MZVs into the $f$-alphabet with non-commutative
generators $f_3,f_5,\ldots$, see appendix \ref{sec:A.2} for a brief review and in particular
(\ref{appA.00}) for the single-valued map. 

The series $\mathbb M^{\rm sv}_{0}$ in (\ref{g0.05})
combines all (conjecturally) $\mathbb Q$-linearly independent single-valued MZVs with
the genus-zero incarnation of zeta generators $M_3,M_5,\ldots$
\cite{DG:2005, Brown:depth3}.
By themselves, zeta generators form a free algebra 
${\rm Lie}[M_3,M_5,\ldots]$ \cite{Levine, BrownTate}.
However, they obey bracket relations with the above braid operators
such that the closure conditions $[M_w,e_i] \in {\rm Lie}[e_0,e_1]$ and $[M_w, e_j'] \in
{\rm Lie}[e_0', e_1', e_y']$ are satisfied. The explicit form of these
brackets involving zeta generators is most conveniently described
in terms of the Drinfeld associator given by the regularized $z\rightarrow 1$ limit 
\begin{align}
\Phinew(e_0,e_1) &= \mathbb G_{e_0,e_1}(1) = {\rm Pexp}\bigg({-}\int_1^0 \dd z\, \bigg[  \frac{e_0}{z} +  \frac{e_1}{z{-}1} \bigg] \bigg)
\notag \\
 &= \sum_{r=0}^\infty \sum_{a_1,\ldots,a_r \atop {\in \{0,1 \} } } G(a_1,\ldots,a_r;1) e_{a_r} \ldots e_{a_1}
\label{notsec.11}
\end{align}
and subject to $\Phinew^{-1}(e_0,e_1) = \Phinew(e_1,e_0)$.
More specifically, the expressions for $[M_w,e_i] \in {\rm Lie}[e_0,e_1]$ and $[M_w, e_j'] \in
{\rm Lie}[e_0', e_1', e_y']$ are given in terms of the following degree-$w$ Lie-polynomials $P_w$ obtained
by expressing the (motivic) associator in the $f$-alphabet of MZVs and extracting the 
coefficient of $f_w$
\beq
P_w(e_0,e_1)  =  \rho \big( \Phinew(e_0,e_1)  \big) \, \big|_{f_w}  = - P_w(e_1,e_0)\, , \ \ \ \ w\geq 3 \ {\rm odd}
\label{notsec.12}
\eeq
The simplest instances of these Lie polynomials are given by\footnote{Our conventions
$P_w$ for the Lie polynomials in the Drinfeld associator are related
to the $g_w$ in \cite{Dorigoni:2024iyt} via $P_w(e_0,e_1)=  - g_w(- e_0, e_1)$, for instance
$g_3(x,y) = [ x {-}y , [x,y]]$.}
\begin{align}
P_3(e_0,e_1) &= - [ e_0 + e_1 , [ e_0, e_1 ] ]
\label{notsec.13}\\
P_5(e_0,e_1) &=- [ e_0,[e_0,[e_0,[e_0,e_1 ]] ]]  - \frac{ 3}{2} [ e_1,[e_0,[e_0,[e_0,e_1]] ]]
-  \frac{1}{2} [ e_0,[e_1,[e_0,[e_0,e_1]] ]] \notag \\
&\quad  - \frac{1}{2} [ e_1,[e_0,[e_1,[e_0,e_1]] ]] 
-  \frac{3}{2} [ e_0,[e_1,[e_1,[e_0,e_1]] ]] - [e_1,[e_1,[e_1,[e_0,e_1]] ]]
\notag
\end{align}
With these definitions in place, the bracket relations of genus-zero zeta
generators relevant for MPLs in one and two variables are given by ($w\geq 3$ odd)
\begin{align}
\big[ e_0 , M_{w} \big] &= 0 \, , \ \ \ \ \ \
\big[ e_1 , M_{w} \big] = \big[ P_w(e_0,e_1) , e_1 \big]
\label{svmpl.07} 
\end{align}
and
\begin{align}
\big[ e_0' , M_{w} \big] &= 0 \, , \ \ \ \ \ \
\big[ e_y' , M_{w} \big] = \big[ P_w(e_0',e_y') , e_y' \big]
\label{svmpl.08} \\
\big[ e_1' , M_{w} \big] &= \big[ P_w(e_0'{+}e_y',e_1')+P_w(e_0',e_y')+P_w(e_y',e_0 - e_y') , e_1' \big]
\notag
\end{align}
respectively, see \cite{Frost:2023stm, Frost:2025lre} for the generalizations to an arbitrary number of variables.

%%%%%%%%%%%%%%%%%%%%%%%%%%%%%%%%%%%%%%%%%%%%%%%%%%
%%%%%%%%%%%%%%%%%%%%%%%%%%%%%%%%%%%%%%%%%%%%%%%%%%
\subsubsection{Constructing single-valued MPLs}
\label{sec:2.3.3}

The above series in MZVs and meromorphic MPLs will now be combined
to produce generating series of single-valued MPLs in one and two variables
that take values in the Lie groups generated by $e_0,e_1$ and $e_0', e_1', e_y'$
and define the individual single-valued MPLs according to
\begin{align}
\mathbb G^{\rm sv}_{e_0,e_1}(z) &= 
  \sum_{r=0}^\infty \sum_{a_1,\ldots,a_r \atop {\in \{ 0,1\}}} G^{\rm sv}(a_r,\ldots,a_2,a_1;z) e_{a_1}  e_{a_2 }\ldots e_{a_r}
\label{svmpl.11} \\
\mathbb G^{\rm sv}_{e_0',e_1',e_y'}(y,z) &= 
 \sum_{r=0}^\infty \sum_{a_1,\ldots,a_r \atop {\in \{ 0,1,y\}}} G^{\rm sv}(a_r,\ldots,a_2,a_1;z) e'_{a_1}  e'_{a_2 }\ldots e'_{a_r}
\notag
\end{align}

In the one-variable case, the single-valued series in the first line of 
(\ref{svmpl.11}) is constructed from the concatenation product \cite{Frost:2023stm, Frost:2025lre}
\beq
\mathbb G^{\rm sv}_{e_0,e_1}(z) = (\mathbb M_0^{\rm sv})^{-1} \, \overline{ \mathbb G_{e_0,e_1}(z)^T } \,
\mathbb M_0^{\rm sv} \,
\mathbb G_{e_0,e_1}(z)  
\label{svmpl.12}
\eeq
where the transposition symbol $(\ldots)^T$ reverses the concatenation order 
of the (primed or unprimed) braid operators
\beq
(\ldots e_i \ldots e_j \ldots )^T  = \ldots e_j \ldots e_i \ldots  \, , \ \ \ \ \ \ \
(\ldots e'_i \ldots e'_j \ldots )^T  = \ldots e'_j \ldots e'_i \ldots 
\label{svmpl.13}
\eeq
In order to see that the generators on the right side of (\ref{svmpl.12}) are expressible in terms of
$e_0, e_1$ and do not necessitate the larger alphabet of $e_i,M_w$, one expands the conjugation by 
$\mathbb M_0^{\rm sv}$ in terms of nested commutators acting on an arbitrary non-commutative symbol $X$,
\beq
(\mathbb M_0^{\rm sv})^{-1} \, X \,
\mathbb M_0^{\rm sv} = \sum_{r=0}^{\infty} \sum_{i_1,i_2,\ldots,i_r \atop {\in 2\mathbb N+1} } \rho^{-1}\big({\rm sv}(f_{i_1} f_{i_2}\ldots f_{i_r})\big) \, [[ \ldots [[ X,M_{i_1}] , M_{i_2}], \ldots ], M_{i_r}]
\label{svmpl.14}
\eeq
When applied to (\ref{svmpl.12}), the quantity $X$ is an arbitrary word in $e_0, e_1$.
In this case, iterative use of the commutation relations (\ref{svmpl.07}) converts
the nested brackets of $X$ with $M_{i_1} ,M_{i_2}\ldots$ in (\ref{svmpl.14}) to larger
words in $e_0,e_1$ with $i_1{+}i_2{+}\ldots$ additional letters $e_a$. Hence, the
expansion in the first line of (\ref{svmpl.11}) exists and leads to well-defined coefficients
$G^{\rm sv}(a_1,\ldots,a_r;z) $ for each word in $a_i \in \{0,1\}$ by freeness of ${\rm Lie}[e_0,e_1]$.

The same logic applies to single-valued MPLs in two or more variables: The
generating series in the second line of (\ref{svmpl.11}) is constructed from
the concatenation product \cite{Frost:2023stm, Frost:2025lre}
\beq
 \mathbb G_{e_0',e_1',e_y'}^{\rm sv}(y,z) =\big(  \mathbb G^{\rm sv}_{e_0,e_1}(y)  \big)^{-1}\, (\mathbb M_0^{\rm sv})^{-1} \, \overline{  \mathbb G_{e_0',e_1',e_y'}(y,z) ^T } \,
\mathbb M_0^{\rm sv} \,  \mathbb G^{\rm sv}_{e_0,e_1}(y) \, 
 \mathbb G_{e_0',e_1',e_y'}(y,z) 
 \label{svmpl.15}
\eeq
where the two conjugations of $\overline{  \mathbb G_{e_0',e_1',e_y'}(y,z) ^T }$ by
$\mathbb M_0^{\rm sv} $ and $  \mathbb G^{\rm sv}_{e_0,e_1}(y)$ both preserve
the alphabet of $e_i'$: The zeta generators enter through nested brackets by 
(\ref{svmpl.14}) which evaluate to longer words in $e_i'$ by iterative use of (\ref{svmpl.08}) (and also (\ref{svmpl.02}) to express $P_w(e_y',e_0 - e_y')$ in terms of Lie polynomials in $e_i'$).
By the shuffle relations of single-valued MPLs in one variable \cite{svpolylog}, the second conjugation in
(\ref{svmpl.15}) can be organized into
\beq
\big(  \mathbb G^{\rm sv}_{e_0,e_1}(y)  \big)^{-1}\, X\,  \mathbb G^{\rm sv}_{e_0,e_1}(y) 
= \sum_{r=0}^\infty \sum_{a_1,\ldots,a_r \atop {\in \{ 0,1\}}} G^{\rm sv}(a_r,\ldots,a_2,a_1;y) [[ \ldots [[ X,e_{a_1}],  e_{a_2 }], \ldots ],e_{a_r}]
 \label{svmpl.16}
\eeq
For words $X$ in $e_0',e_1', e_y'$, the nested brackets on the right side conspire to 
longer words in the same alphabet by virtue of the braid relations (\ref{svmpl.02}),
leading to well-defined $G^{\rm sv}(a_1,\ldots,a_r;z) $ for any word in $a_i \in \{0,1,y\}$
in the second line of (\ref{svmpl.11}) by freeness of ${\rm Lie}[e_0', e_1', e_y']$.

Note that the Lie-polynomials $P_w$ (defined in (\ref{notsec.12}) through the appearance of $f_w$
in the Drinfeld associator) and therefore the bracket relations (\ref{svmpl.07}), (\ref{svmpl.08}) of $M_w$ individually depend on the choice of the 
$f$-alphabet isomorphism $\rho$.
However, the conjugation in (\ref{svmpl.14}) by the entire series $\mathbb M_0^{\rm sv} $ in zeta 
generators is independent on $\rho$ since the conversion of words in $f_w$ into MZVs in (\ref{g0.05}) 
features a compensating dependence on the choice of $f$-alphabet isomorphism.

As detailed in section 7.4 of \cite{Frost:2025lre}, the series (\ref{svmpl.12}) and (\ref{svmpl.16})
in single-valued MPLs in one and two variables are equivalent to the earlier constructions
in \cite{svpolylog} and \cite{Broedel:2016kls, DelDuca:2016lad}
through a change of alphabet for $e_1$ and $e_y',e_1'$, respectively.

%%%%%%%%%%%%%%%%%%%%%%%%%%%%%%%%%%%%%%%%%%%%%%%%%%
%%%%%%%%%%%%%%%%%%%%%%%%%%%%%%%%%%%%%%%%%%%%%%%%%%
%%%%%%%%%%%%%%%%%%%%%%%%%%%%%%%%%%%%%%%%%%%%%%%%%%
%%%%%%%%%%%%%%%%%%%%%%%%%%%%%%%%%%%%%%%%%%%%%%%%%%
\section{Key ingredients of the main result}
\label{sec:key}

This section introduces some of the key ingredients to set the stage for the main results of this work: Lie-algebra structures of iterated integrals at genus one in section \ref{sec:2.4}, equivariant and single-valued iterated Eisenstein integrals in section \ref{sec:2.5}, building blocks for $z$-dependent equivariant series in section \ref{sec:3.1} and further discussions of the notion of equivariance in section~\ref{sec:3.eqv}.

%%%%%%%%%%%%%%%%%%%%%%%%%%%%%%%%%%%%%%%%%%%%%%%%%%
%%%%%%%%%%%%%%%%%%%%%%%%%%%%%%%%%%%%%%%%%%%%%%%%%%
\subsection{Lie-algebra structures of genus-one integrals}
\label{sec:2.4}

We have seen in section \ref{sec:2.3} that the construction of single-valued MPLs at genus zero
relies on the bracket relations (\ref{svmpl.02}), (\ref{svmpl.07}) and (\ref{svmpl.08}) among
the Lie-algebra generators $e_0,e_1$ as well as $e_0',e_1',e_y'$ and $M_w , \ w \in 2\mathbb N{+}1$.
The goal of this section is to introduce the genus-one generators and bracket relations relevant for the construction
of modular and single-valued combinations of eMPLs or iterated $\tau$-integrals (\ref{bsc.07}) and (\ref{absc.07}).
For this purpose, the free Lie algebra in generators $x,y$ that governs the generating series
(\ref{genKempl}) and (\ref{jbexp}) of eMPLs will be augmented by two kinds of derivations 
related to holomorphic Eisenstein series and odd Riemann zeta values, respectively.

%%%%%%%%%%%%%%%%%%%%%%%%%%%%%%%%%%%%%%%%%%%%%%%%%%
%%%%%%%%%%%%%%%%%%%%%%%%%%%%%%%%%%%%%%%%%%%%%%%%%%
\subsubsection{Tsunogai's derivation algebra}
\label{sec:2.4.1}

A first class of derivations on ${\rm Lie}[x,y]$ governs the $\tau$-dependence
of elliptic MZVs (\ref{defemzv}). More specifically, their generating series known
as the $A$-elliptic KZB associator and built from $\tilde{\mathbb K}_{x,y}(z,\tau)$ in (\ref{defKcon})
\begin{align}
\mathbb A_{x,y}(\tau) &= {\rm Pexp} \biggl(\int^0_1
 \tilde{\mathbb K}_{x,y}(z_1,\tau)\biggr) 
 \label{lieg1.01} \\
 &=   \sum_{r=0}^\infty (-1)^r  \! \! \! \! \!\sum_{n_1,\ldots,n_r=0}^\infty \! \! \! \! \!
 (2\pi i)^{ \sum_{i=1}^r (1-n_i) }
 \omega(n_1,\ldots,n_r;\tau)
 \ad_x^{n_1}(y) \ldots  \ad_x^{n_r}(y)
 \notag
\end{align}
obeys the differential equation \cite{KZB, EnriquezEllAss, Hain:KZB}\footnote{Our conventions for the elliptic KZB associator are related to those of \cite{Broedel:2015hia} by reversal of the integration path and rescaling the generators in the reference by $x \rightarrow \frac{x}{2\pi i}$ and $y \rightarrow 2\pi i y$ (which leads to the translation $\ep_0 \rightarrow (2\pi i)^{2} \ep_0$ and $\ep_k \rightarrow (2\pi i)^{2-k} \ep_k^{\rm TS}, \ k\geq 2$ of the derivations in the reference).}
\beq
 \partial_\tau \mathbb A_{x,y}(\tau) = 2\pi i\, \big[\ep_0 ,  \mathbb A_{x,y}(\tau) \big] + \sum_{k=2}^{\infty} (2\pi i)^{1-k} \, (k{-}1) \, {\rm G}_k(\tau)\,
 \big[ \epsilon_{k}^\text{TS} ,  \mathbb A_{x,y}(\tau) \big]
 \label{lieg1.02}
\eeq
The holomorphic Eisenstein series ${\rm G}_k(\tau)$ with
$k\geq 3$ are defined by absolutely convergent lattice sums (\ref{latsumgf}), furnish modular forms of weight $(k,0)$ and vanish for odd $k$. The
quasi-modular ${\rm G}_2(\tau)$ in (\ref{lieg1.02}) can be defined through the $k=2$
instance of the following $q$-series
\beq
{\rm G}_k(\tau) = \frac{(2\pi i)^k}{(k{-}1)!} \bigg\{ 
{-}\frac{B_k}{k}+ 2 \sum_{m,n=1}^\infty n^{k-1} q^{mn}
\bigg\} \, , \ \ \ \ q = e^{2\pi i \tau} \, , \ \ \ \ k\geq 2 \ {\rm even}
 \label{lieg1.03}
\eeq
The non-commuting $ \epsilon_{0}, \epsilon_{k}^\text{TS}$ that accompany the Eisenstein
series ${\rm G}_k(\tau)$ in the differential equation 
(\ref{lieg1.02}) of the KZB associator are known
as Tsunogai's derivations \cite{tsunogai4}. They 
act on the generators $x,y$ entering the associator via
\begin{align}
[\epsilon_0,x]=y \, , \ \ \ \ [\epsilon_0,y]=0
 \label{lieg1.04}
\end{align}
as well as ($\ell \in \mathbb N$)
\begin{align}
[\epsilon_{2\ell}^\text{TS},x]&=\ad^{2\ell}_x(y)  \label{lieg1.05} \\
[\epsilon_{2\ell}^{\rm TS},y]&=[y,\ad_x^{2\ell-1}(y)]+\sum_{j=1}^{\ell-1}(-1)^j [\ad_x^j(y),\ad_x^{2\ell-1-j}(y)]
\notag
\end{align}
which assigns degree 0 to $\epsilon_0$ and degree $k$ to $ \epsilon_{k}^\text{TS}$ in the
generators $x,y$. The second line of (\ref{lieg1.05}) can be derived from first line by demanding that the
commutator $[x,y]$ is annihilated by all Tsunogai derivations
\beq
\big[\ep_0 , [x,y] \big] = 0 \, , \ \ \ \ \ \
\big[ \epsilon_{k}^\text{TS} , [x,y] \big] = 0 \, , \ \ \ \ k\geq 2 \ {\rm even}
\label{killxy}
\eeq
Tsunogai's derivations have a rich history in the literature of both mathematics 
\cite{tsunogai1,tsunogai2, IharaTakao:1993, tsunogai5,tsunogai6,tsunogai7,Pollack,Brown:anatomy,
Hain:KZB,Brown:depth3, tsunogai12,hain_matsumoto_2020} and string theory 
\cite{Broedel:2015hia, Mafra:2019ddf, Mafra:2019xms, Gerken:2019cxz, Gerken:2020yii, 
Dorigoni:2021jfr, Dorigoni:2021ngn, Dorigoni_2022, Dorigoni:2024oft}. They obey a wealth 
of relations, starting from the facts that $\epsilon_{2}^\text{TS}$ is central
\beq
[\epsilon_{2}^\text{TS} , \epsilon_{k}^\text{TS} ] = 0 \, , \ \ \ \ k\geq 2
 \label{lieg1.06}
\eeq
and that all the $\epsilon_{k}^\text{TS}$ with $k\geq 2$ obey the nilpotency property with
respect to $\ad_{\ep_0}$
\beq
\ad_{ \epsilon_{0} }^{k-1} (\epsilon_{k}^\text{TS})  = 0 \, , \ \ \ \ k\geq 2
 \label{lieg1.07}
\eeq
Moreover, certain combinations of Tsunogai derivations starting from
degree 14 obey {\it Pollack relations} related to holomorphic cusp forms, for instance 
\cite{LNT, Pollack, Broedel:2015hia, WWWe}
\begin{align}
  0 &=   [\epsilon^\text{TS}_4,\epsilon^\text{TS}_{10}] - 3[\epsilon^\text{TS}_6,\epsilon^\text{TS}_8]
    \label{lieg1.08} \\
0 &= 80 [\ep_4^{(1) \text{TS}}, \ep^\text{TS}_{12} ] + 16 [\ep_{12}^{(1)\text{TS}},\ep_4^\text{TS} ]
- 250 [\ep_6^{(1) \text{TS}},\ep^\text{TS}_{10}] -125 [\ep_{10}^{(1) \text{TS}},\ep_6^\text{TS}] \notag \\
&\quad
+ 280 [\ep_8^{(1) \text{TS}},\ep_8^\text{TS}] 
- 462[\ep^\text{TS}_4, [\ep^\text{TS}_4,\ep^\text{TS}_8]] 
- 1725 [\ep^\text{TS}_6,[\ep^\text{TS}_6,\ep^\text{TS}_4]] 
\notag 
\end{align}
Here and throughout this work, we employ the shorthand for repeated adjoint actions of $\ep_0$
\beq
\epsilon_{k}^{(j) \text{TS}} = {\rm ad}_{\ep_0}^j( \epsilon_{k}^{\text{TS}})
 \label{lieg1.09}
\eeq
The non-vanishing $\epsilon_{k}^{(j) \text{TS}} $ at fixed $k$ and $0\leq j \leq k{-}2$
form $(k{-}1)$-dimensional multiplets with respect to the $\mathfrak{sl}_2$ algebra generated by
$\ep_0$ and the additional derivation $\epsilon_0^\vee$ of degree zero
\begin{align}
[\epsilon^\vee_0,x]=0 \, , \ \ \ \ [\epsilon^\vee_0,y]=x 
 \label{lieg1.11}
\end{align}
More specifically, $\ep_0$ and $\epsilon_0^\vee$ act as ladder operators of
the $\mathfrak{sl}_2$ with
\begin{align}
[\ep_0, \epsilon_{k}^{(j) \text{TS}} ] = \epsilon_{k}^{(j+1) \text{TS}} \, , \ \ \ \ \ \ 
[\epsilon^\vee_0,  \epsilon_{k}^{(j) \text{TS}} ]= j(k{-}j{-}1) \epsilon_{k}^{(j-1) \text{TS}} 
 \label{lieg1.13}
\end{align}
and their commutator yields the Cartan generator ${\rm h}$ of the $\mathfrak{sl}_2$
\beq
{\rm h} = [\ep_0, \ep_0^\vee]  \, , \ \ \ \ [ {\rm h},x] = - x \, , \ \ \ \ [ {\rm h}, y] = y \, , \ \ \ \
 \big[ {\rm h} , \epsilon_{k}^{(j) \text{TS}} \big] = (2j{-}k{+}2) \epsilon_{k}^{(j) \text{TS}}
 \label{lieg1.14}
 \eeq
In view of these Cartan eigenvalues, we refer to
$\epsilon_{k}^{(k-2) \text{TS}}$ and $\epsilon_{k}^{(0) \text{TS}}= \epsilon_{k}^{\text{TS}}$  as 
highest-weight vectors and lowest-weight
vectors, respectively,
\beq
[\ep_0, \epsilon_{k}^{(k-2) \text{TS}}] = 0
 \, , \ \ \ \ \ \
[\epsilon^\vee_0, \epsilon_{k}^\text{TS} ] = 0
 \, , \ \ \ \ \ \ k\geq 2
 \label{lieg1.15}
 \eeq
The first condition is equivalent to the nilpotency property (\ref{lieg1.07}), and
the second condition is (\ref{lieg1.13}) at $j=0$.

%%%%%%%%%%%%%%%%%%%%%%%%%%%%%%%%%%%%%%%%%%%%%%%%%%
%%%%%%%%%%%%%%%%%%%%%%%%%%%%%%%%%%%%%%%%%%%%%%%%%%
\subsubsection{Zeta generators at genus one}
\label{sec:2.4.2}

An additional class of derivations on ${\rm Lie}[x,y]$ is furnished by the genus-one uplift $\sigma_w$ of
the genus-zero zeta generators $M_w$ in section \ref{sec:2.3.2}. Their defining properties
are formulated in terms of the $\tau \rightarrow i \infty$ degeneration of the CEE connection (\ref{defKcon})
at fixed $z$
\begin{align}
\tilde{\mathbb K}_{x,y}(z, i\infty) &= \dd z \, \bigg\{  2 \pi i y + \pi \cot(\pi z) [x,y]
- 2 \sum_{\ell=1}^\infty (2\pi i)^{1-2\ell}  \zeta_{2\ell} \ad_x^{2\ell-1}(y)
\bigg\} \notag \\
&= {-}\dd \sigma \, \bigg( \frac{t_{12}}{\sigma{-}1} + \frac{t_{01}}{\sigma} \bigg)
 \label{lieg1.16}
\end{align}
where we use the shorthands
\beq
t_{12} = - [x,y]  \, , \ \ \ \ \ \  t_{01} =  -  \sum_{k=0}^\infty    \frac{B_k}{k!}  {\rm ad}_x^k(y)
= - \frac{{\rm ad}_x  }{ e^{{\rm ad}_x} - 1 }\, (y)   \label{lieg1.17}
\eeq
and $B_k$ denotes the $k^{\rm th}$ Bernoulli number.
In passing to the second line of (\ref{lieg1.16}), we have introduced the coordinate 
\beq
\sigma = e^{2\pi i z} \, , \ \ \ \ \ \
\dd z = \frac{\dd \sigma}{2\pi i \sigma} \, 
\label{zdsig}
\eeq
of the nodal sphere which is obtained from the non-separating 
degeneration of the torus. By comparison with the integration kernels of MPLs in (\ref{svmpl.00}), we recovered
a KZ connection from the CEE connection at $\tau \rightarrow i \infty$ where the
composite generators $t_{12},t_{01}$ in (\ref{lieg1.17}) take the role of the braid 
operators $e_0, e_1$ in section \ref{sec:2.3}. Given that $e_0, e_1$ govern the
defining bracket relations (\ref{svmpl.07}) of the zeta generators $M_w$ at genus zero,
we define the analogous genus-one zeta generators $\sigma_w$ \cite{EnriquezEllAss, Hain:KZB, Schneps:2015mzv, hain_matsumoto_2020} by transcribing
(\ref{svmpl.07}) into
\beq
[t_{12} , \sigma_w] = 0 \, , \ \ \ \ \ \ [t_{01}, \sigma_w] = \big[P_w(t_{12},t_{01}) , t_{01} \big] 
%=  - \big[P_w(t_{01},t_{12}) , t_{01} \big] 
\label{dfprpsig}
\eeq
where the Lie polynomials $P_w(t_{12},t_{01}) = - P_w(t_{01},t_{12})$ are defined by (\ref{notsec.12}) for a given $f$-alphabet isomorphism $\rho$ (see appendix \ref{sec:A.2}).
We are taking $t_{12}$ rather than $t_{01}$ to commute with $\sigma_w$
since the associated singularity of (\ref{lieg1.16}) at $\sigma=1$ and
therefore $z=0$ occurs at the starting point of the integration path
that defines eMPLs (\ref{defempl}).

Given the infinite series expansion of the generator
$t_{01} = -y+\frac{1}{2}[x,y]- \frac{1}{12} [x,[x,y]]+ {\cal O}(x^4)$
in the defining relations (\ref{dfprpsig}) of genus-one zeta generators $\sigma_w$, their action 
on $x,y$ mixes derivations of different degrees without any upper bound.\footnote{Inferring the action
$[\sigma_w,x]$, $[\sigma_w,y]$ on $x,y$ from the expressions (\ref{dfprpsig})
for $[\sigma_w,t_{01}]$, $[\sigma_w,t_{12}]$ requires the 
extension lemma 2.1.2 of \cite{Schneps:2015mzv}, also see section 5.3 of \cite{Dorigoni:2024iyt} for further details.}
It turns out that the contributions to $[\sigma_w,x]$ and $[\sigma_w,y]$ of all 
degrees except for $2w{+}1$ are expressible in terms of Tsunogai derivations \cite{hain_matsumoto_2020, Dorigoni:2024iyt}.
More precisely, (\ref{dfprpsig}) determines genus-one zeta generators
as an infinite series
\beq
\sigma_w  = z_w - \frac{ \epsilon_{w+1}^{(w-1) \text{TS}} }{(w{-}1)!} + \ldots
 \label{lieg1.21}
\eeq
with nested brackets of $\geq 2$ Tsunogai derivations with total degree $\geq w{+}3$ in the ellipsis.
The derivation $z_w$ of degree $2w$ which is not expressible
in terms of $\epsilon_{k}^{(j) \text{TS}}$ is referred to as the {\it arithmetic} 
part of the zeta generator $\sigma_w$ whereas the infinite series
of $\epsilon_{k}^{(j) \text{TS}}$ is referred to as its {\it geometric} part.
The arithmetic parts commute with all of $\ep_0,\ep_0^\vee$ 
and $[x,y]$ \cite{hain_matsumoto_2020}, 
\beq
[\ep_0 , z_w] = [ \ep_0^\vee, z_w] = 0 \, , \ \ \ \ \ \ \big[z_w, [x,y] \big] = 0
 \label{lieg1.22}
\eeq
i.e.\ the $z_w$ form $\mathfrak{sl}_2$ singlets. For a given zeta generator $\sigma_w$, the arithmetic derivations $z_{w}$
are canonical once the geometric remainder $\sigma_w{-} z_w$ is imposed to have
no $\mathfrak{sl}_2$ invariant terms~\cite{Dorigoni:2024iyt}.\footnote{Without this
extra condition, the arithmetic $z_w$ with $w\geq 7$ are only well defined up to
redefinitions by $\mathfrak{sl}_2$-invariant nested brackets of $\geq 3$ Tsunogai 
derivations \cite{Brown:2017qwo2}.}

A useful property of genus-one zeta generators is that they commute with
the following infinite series $N^{\rm TS} $ in Tsunogai derivations
\cite{hain_matsumoto_2020}\footnote{Up to the absence of the term $\epsilon_2^{\rm TS}$, the series $N^{\rm TS}$ in (\ref{notsec.22}) is obtained from $-1/(2\pi i)$ times the $\dd \tau$ component of the CEE connection $\mathbb K_{x,y,\ep}(z,\tau)$ in (\ref{notsec.02}) below in the limit $\tau \rightarrow i\infty$ and $z\rightarrow 0$.}
\begin{align}
[N^{\rm TS} , \sigma_w] &= 0 \, , \ \ \ \  w \geq 3 \, , \ \ \ \ \ \ 
N^{\rm TS} = -\ep_0 + \sum_{k=4}^\infty (k{-}1)   \frac{B_k}{k!} \, \ep^{\rm TS}_k 
 \label{notsec.22}
\end{align}
As detailed in section 7.3 of \cite{Dorigoni:2024iyt}, this property can be used
to determine the infinite series of Tsunogai derivations in the ellipsis
of (\ref{lieg1.21}) up to {\it finitely} many highest-weight vectors of $\mathfrak{sl}_2$,
i.e.\ nested brackets of $\ep^{(j){\rm TS}}_k $ in the kernel of ${\rm ad}_{\ep_0}$.
The arguments in the reference rely on the fact that
highest-weight vectors can only enter the expansion of $\sigma_w$ up to and including degree $2w$.\footnote{These bounds on the appearance of highest-weight vectors follow from the fact that $\sigma_w$ and $\epsilon_k^{(j){\rm TS}}$ have $y$-degrees $w$ and $j{+}1$, respectively, as one can deduce from (\ref{lieg1.04}), (\ref{lieg1.05}) and (\ref{dfprpsig}).}

Assigning modular depth one to each $\epsilon_{k}^{(j) \text{TS}}$, then the
contributions to $\sigma_w$ in (\ref{lieg1.21}) of modular depth two are reviewed
in appendix \ref{sec:C.1}. Partial results at modular depth three can be found in section 7.4 of 
\cite{Dorigoni:2024iyt}. Terms of low degrees in
$\sigma_w$ at $w=3,5,7,9$ can be found in appendix \ref{sec:C.2}
and the expressions for $[z_w,x] ,[z_w,y]$ at $w=3$ and $w=5$ are given in (\ref{z3onxy}) below and appendix \ref{sec:C.3}, respectively.

Finally, the vanishing brackets of $\sigma_w$ with $N^{\rm TS}$ in (\ref{notsec.22})
determine the commutators
\beq
[z_w,\ep^{\rm TS}_k] = \frac{B_{w+k-1}  \, k!}{B_k \, (w{+}k{-}1)!\, (w{+}k{-}3)!} \sum_{i=0}^{w-1} (-1)^i \frac{(k{+}i{-}2)!}{i!} [\ep_{w+1}^{(i) {\rm TS}} , \ep_{k+w-1}^{(w-i-1) {\rm TS}} ] + \ldots    
 \label{lieg1.23}
\eeq
as Lie polynomials in $ \ep_{m}^{(j) {\rm TS}} $ of degree $k{+}2w$. The ellipsis 
in (\ref{lieg1.23}) refers to terms of modular depth $\geq 3$, see
section 7.4 of \cite{Dorigoni:2024iyt} for partial results at modular depth three. Similar expressions for $[z_w,\ep^{(j){\rm TS}}_k]$ are simple corollaries under ${\rm ad}_{\ep_0}^j$ since $[\ep_0,z_w]=0$.

In summary, the genus-one incarnations of zeta generators $\sigma_w$ reviewed in this section are expressible
in terms of Tsunogai derivations $ \ep_{k}^{(j) {\rm TS}} $ up to a single arithmetic
derivation $z_w$ of degree $2w$ in $x,y$. The commutators $[\sigma_w , \ep_{k}^{(j) {\rm TS}}]$
yield Lie series in $\ep_{k'}^{(j') {\rm TS}}$ which can be determined from (\ref{dfprpsig})
and the methods of \cite{Schneps:2015mzv, Dorigoni:2024iyt}.

%%%%%%%%%%%%%%%%%%%%%%%%%%%%%%%%%%%%%%%%%%%%%%%%%%
%%%%%%%%%%%%%%%%%%%%%%%%%%%%%%%%%%%%%%%%%%%%%%%%%%
\subsubsection{The CEE connection and its Lie algebra}
\label{sec:2.4.3}

We shall finally introduce a variant $\ep_k^{(j)}$ of the Tsunogai 
derivations $\ep_k^{(j) {\rm TS}}$ by reinstating the $\dd \tau$ part of 
the flat and meromorphic CEE connection \cite{KZB} with $\dd z$ part 
in (\ref{defKcon}):
\begin{align}
 \mathbb K_{x,y,\ep}(z,\tau) &= 
 \dd z \, \ad_x F(z,{\rm ad}_{ \frac{ x}{2\pi i}} ,\tau)y
  +  2\pi i \text{d}\tau\bigg( \epsilon_0+\sum_{k=4}^\infty \frac{(k{-}1)}{(2\pi i)^k} 
 {\rm G}_k(\tau)\epsilon_k  \bigg)
\label{notsec.02} \\
 &\quad 
 +    \text{d}\tau \, {\rm ad}_x  \partial_{{\rm ad}_x} \bigg( F(z,{\rm ad}_{ \frac{ x}{2\pi i}} ,\tau)- \frac{2\pi i}{ \ad_x } \bigg)y
\notag \\
 &= 2\pi i \text{d}z\,  \sum_{k=0}^\infty \frac{ g^{(k)}(z,\tau)}{(2\pi i)^k } \, {\rm ad}_{x}^{k}(y) + 2\pi i \text{d}\tau\bigg( \epsilon_0+\sum_{k=2}^\infty \frac{(k{-}1)}{(2\pi i)^k} \big[ {\rm G}_k(\tau)\epsilon_k
 - g^{(k)}(z,\tau) b_k \big]  \bigg)
 \notag
\end{align}

Its flatness rests on the mixed heat equation (\ref{appA.31}), the action of the derivations $\epsilon_k^{(j)}$ on the free Lie algebra Lie$[x,y]$\footnote{In \cite{KZB} imposing the CEE connection to be flat fixes the action of $\epsilon_0$ and $\epsilon_k$ on $x$ and $y$ given in (\ref{epsxy}).
Using (\ref{lieg1.25}) they lead to the identities (\ref{lieg1.04}) and (\ref{lieg1.05}).}
and Fay identities of 
the Kronecker-Eisenstein series \cite{Broedel:2020tmd}.
It identifies the generators $b_k$ accompanying
$\dd \tau \, g^{(k)}(z,\tau)$ in the last line as Lie polynomials 
\beq
b_k = - {\rm ad}_x^{k-1}(y) \, , \ \ \ \ k \geq 2
 \label{lieg1.24}
\eeq
Moreover, the $z \rightarrow 0$ limit of the $\dd \tau$ component of (\ref{notsec.02}) 
where $g^{(k)}(0,\tau) = - {\rm G}_k(\tau)$ governs the 
differential equation (\ref{lieg1.02}) of the KZB associator
which relates the generators $\ep_k$ in (\ref{notsec.02}) to the Tsunogai derivations via
\beq
\ep^{\rm TS}_k = \ep_k + b_k \, , \ \ \ \ k \geq 2 \ {\rm even}
 \label{lieg1.25}
\eeq
In a shorthand notation analogous to $\epsilon_{k}^{(j) \text{TS}}$ in (\ref{lieg1.09})
\beq
b_{k}^{(j)} = {\rm ad}_{\ep_0}^j(b_{k}) \, , \ \ \ \ \ \
\epsilon_{k}^{(j)} = {\rm ad}_{\ep_0}^j( \epsilon_{k} )
 \label{lieg1.26}
\eeq
one can easily show via (\ref{lieg1.04}) and (\ref{lieg1.11}) that the lowest- and highest-weight vector
conditions (\ref{lieg1.15}) of Tsunogai derivations carry over to the Lie polynomials $b^{(j)}_k$
and thus to $\epsilon_{k}^{(j)}$:
\begin{align}
[ \ep_0, b_{k}^{(k-2)} ] &= 0  &[ \ep_0, \ep_{k}^{(k-2)} ] &= 0  \notag \\
[ \ep^\vee_0, b_{k}] &= 0  &[ \ep^\vee_0, \ep_{k}   ] &= 0 
 \label{lieg1.27}
\end{align}
Hence, both of $b_{k}^{(j)}$ and $\epsilon_{k}^{(j)}$ at fixed $k$ and integer $0\leq j \leq k{-}2$ form
$(k{-}1)$-dimensional representations of the $\mathfrak{sl}_2$ spanned by $\ep_0, \ep_0^\vee$
and ${\rm h}$, with the action of ladder operators and Cartan eigenvalues as
in (\ref{lieg1.13}) and (\ref{lieg1.14}), respectively.
\begin{align}
[ \ep_0, b_{k}^{(j)} ] &= b_{k}^{(j+1)}  &[ \ep_0, \ep_{k}^{(j)} ] &= \ep_{k}^{(j+1)}  \label{lieg1.28} \\
[ \ep^\vee_0, b_{k}^{(j)} ] &= j(k{-}j{-}1)b_{k}^{(j-1)}  &[ \ep^\vee_0, \ep_{k}^{(j)}   ] &= j(k{-}j{-}1)\ep_{k}^{(j-1)} 
\notag \\
[ {\rm h}, b_{k}^{(j)} ] &= (2j{-}k{+}2)b_{k}^{(j)}  &[ {\rm h}, \ep_{k}^{(j)}   ] &= (2j{-}k{+}2) \ep_{k}^{(j)} 
\notag
\end{align}
As detailed in appendix \ref{sec:D}, the above information can be used to express numerous classes of brackets among $\ep^{(j)}_k,b^{(j)}_k,x,y$ in terms of Lie polynomials in $b_{k'}^{(j')}$.

\begin{lemma}
\label{braklem}
Overview of brackets that evaluate to Lie polynomials in $b_{k}^{(j)}$:
\begin{itemize}
\item[(i)] all brackets $[ \ep^{(j)}_k , x ] $, $[ \ep^{(j)}_k , y ] $
and $[ b^{(j)}_k , x ] $, $[ b^{(j)}_k , y ] $ with $k\geq 2$ and $0\leq j\leq k{-}2$
are expressible as Lie polynomials in $b_{k'}^{(j')}$ of $(x,y)$-degree $k{+}1$;
\item[(ii)] all brackets $[ \ep^{(j_1)}_{k_1} , b^{(j_2)}_{k_2} ] $ (or equivalently 
$[ \ep^{(j_1) {\rm TS}}_{k_1} , b^{(j_2)}_{k_2} ] $) with $k_i\geq 2$ and $0\leq j_i\leq k_i{-}2$
for $i=1,2$ are expressible as Lie polynomials in $b_{k'}^{(j')}$ of $(x,y)$-degree $k_1{+}k_2$;
\item[(iii)] all Lie polynomials in $x,y$ of degree $\geq 2$ are 
Lie polynomials in $b_{k}^{(j)}$;
\item[(iv)] all brackets
$[\sigma_w,x]$, $[\sigma_w,y]$ with $w\geq 3$ odd are Lie series in $b_{k}^{(j)}$ of degrees $\geq w{+}2$;
\item[(v)] all brackets
$[z_w,x]$, $[z_w,y]$ with $w\geq 3$ odd are Lie polynomials in $b_{k}^{(j)}$ of degree $2w{+}1$.
\end{itemize}
\end{lemma}
The proof of the lemma can be found in appendix \ref{sec:D} and leads to an algorithm to determine the brackets of $(i)$ and $(ii)$ in terms of Lie polynomials in $b_k^{(j)}$ to any desired~degree.

For instance, in view of (\ref{lieg1.24}) and (\ref{lieg1.25}), the action (\ref{lieg1.05}) of $\ep^{\rm TS}_k$ on
$x,y$ translates into the following simpler expressions for the analogous $\ep_k$  action,
\beq
[ \ep_k , x ] = 0 \, , \ \ \ \ \ \  [ \ep_k , y ] = \sum_{\ell=1}^{\frac{k}{2}-1} (-1)^j [b_{j+1}, b_{k-j}]  \, , \ \ \ \ \ \  k \geq 2 \ {\rm even}
 \label{epsxy}
\eeq
Similarly, one can show from the flatness of the CEE connection (\ref{notsec.02})
or from iterative use of (\ref{epsxy}) on the Lie-polynomial representation (\ref{lieg1.24}) 
of $ b_{k_2} $ that \cite{Broedel:2020tmd}
\beq
[ \ep_{k_1}  , b_{k_2} ] = \sum_{\ell = 0}^{k_2-2} {k_2{-}2 \choose \ell} [ b_{\ell+2} , b_{k_1+k_2 - \ell - 2}]
\label{intrel}
\eeq
where the last term $\ell=k_2{-}2$ on the right side is absent in
the analogous expression for $[ \ep^{ {\rm TS}}_{k_1}  , b_{k_2} ]$.
Relations among $[ \ep^{(j_1)}_{k_1}  , b^{(j_2)}_{k_2} ]$ with non-zero $j_1,j_2$
are discussed in appendix \ref{sec:D.2}, also see section 4 and the ancillary file of \cite{Hidding:2022vjf}.

The simplest examples of the
Lie polynomials $[z_w,x]$ and $[z_w,y]$ are
\begin{align}
    [z_3,x]&= -\dfrac{1}{12}\big[b_2,[b_2,b_3] \big]-\dfrac{1}{8}[b_3,b_4^{(1)}]-\dfrac{1}{4}[b_4,b_3^{(1)}] \label{z3onxy}\\
    [z_3,y]&= -\dfrac{1}{12}\big[b_2,[b_2,b_3^{(1)}] \big]-\dfrac{1}{8}[b_3,b_4^{(2)}]+\dfrac{1}{8}[b_3^{(1)},b_4^{(1)}]
    \notag
\end{align}
and a considerably longer expression for $[z_5,x]$ can be found in appendix \ref{sec:C.3}.
\begin{prop}
\label{e2prop}
The derivation $\ep_2$ is central in the Lie algebra of $x,y,\ep_k$ and $b_\ell$ with $k\geq 0$ even and $\ell \geq 2$, i.e.\ it commutes with all other generators.
\end{prop}

\begin{proof} 
First, $\ep_2$ commutes with $x,y$ by the $k=2$ instance of (\ref{epsxy}). Second,
 $\ep_2$ commutes with $b_k^{(j)}$ for any $k\geq 2$ and $0\leq j\leq k{-}2$ 
in view of the previous point and the fact that 
(\ref{lieg1.24}) together with (\ref{lieg1.04}) determine $b_k^{(j)}$ as a Lie polynomial in $x,y$. Third,
 $\ep_2$ commutes with $\ep_k^{ {\rm TS}}$ and thus $\ep_k$ for any $k\geq 2$ by (\ref{killxy}),
(\ref{lieg1.06}) and $\ep_2 = \ep_2^{{\rm TS}} + [x,y]$. Finally, $\ep_2$ commutes with $\ep_k^{(j) {\rm TS}}$ and thus $\ep_k^{(j)}$ by
the previous point and $[\ep_0,\ep_2]=0$ in (\ref{lieg1.27}).
\end{proof}

We shall therefore set
\beq
\ep_2 = 0 \, , \ \ \ \ \ \ \ep^{\rm TS}_2 = b_2
 \label{lieg1.31}
\eeq
in the remainder of this work.

%%%%%%%%%%%%%%%%%%%%%%%%%%%%%%%%%%%%%%%%%%%%%%%%%%
%%%%%%%%%%%%%%%%%%%%%%%%%%%%%%%%%%%%%%%%%%%%%%%%%%
\subsubsection{The transposition operation}
\label{sec:2.4.4}

For later convenience, we shall here define the genus-one analogue of the 
transposition operation (\ref{svmpl.13}) among the braid operators $e_i, e_j'$ at genus zero.
We assign transposition parities
\begin{align}
x^T &= - x \, , &\ep_k^T &= \ep_k \, , \ \ \ \ \ \ k\geq 0 \ {\rm even}
 \label{lieg1.32} \\
y^T&= y \, , 
&z_w^T &= z_w \, , \ \ \ \ \ \ w\geq 3 \ {\rm odd}
\notag
\end{align}
and impose that concatenation products of the above generators are transposed according to
\beq
(l_1 l_2\ldots l_r)^T = (l_r)^T \ldots (l_2)^T (l_1)^T \, , \ \ \ \ \ \ 
l_i \in \{x,y,\ep_0,\ep_4,\ep_6,\ldots\} \, , \ \ \ \ r\in \mathbb N_0
 \label{lieg1.33}
\eeq
This implies the alternating transposition parities for
\beq
(\ep^{(j)}_{k})^T = (-1)^j \ep^{(j)}_{k} \, , \ \ \ \ \ \ 
(b^{(j)}_{k})^T = (-1)^j b_{k}^{(j)}
 \label{lieg1.34}
\eeq
and thus $\ep^{(j) {\rm TS}}_{k}$ and extends the composition rule (\ref{lieg1.33}) to the composite $l_i \in \{\ep^{(j)}_{k},b^{(j)}_{k}  \}$
with $k\geq 2$ and $0 \leq j \leq k{-}2$.

%%%%%%%%%%%%%%%%%%%%%%%%%%%%%%%%%%%%%%%%%%%%%%%%%%
%%%%%%%%%%%%%%%%%%%%%%%%%%%%%%%%%%%%%%%%%%%%%%%%%%
\subsection{Equivariant and single-valued iterated Eisenstein integrals}
\label{sec:2.5}

We shall here review Brown's construction of non-holomorphic modular forms 
\cite{brown2017multiple, Brown:2017qwo, Brown:2017qwo2}
from the restriction of the iterated $\tau$-integrals (\ref{bsc.07}), (\ref{absc.07}) to holomorphic Eisenstein series
\beq
\ee{j_1 &\ldots &j_r}{k_1 &\ldots &k_r}{\tau} = (2\pi i)^{1+j_r - k_r} \int_\tau^{i\infty} \dd \tau_r \, \tau_r^{j_r} {\rm G}_{k_r}(\tau_r) \ee{j_1 &\ldots &j_{r-1}}{k_1 &\ldots &k_{r-1}}{\tau_r}
\, , 
 \ \ \ \ \ee{\emptyset}{\emptyset}{\tau}  = 1
 \label{lieg1.35}
\eeq
and their complex conjugates with $r\geq 1$, $k_i\geq 2$ even and $0\leq j_i\leq k_i{-}2$ for $i=1,\ldots,r$.
The endpoint divergences from $\tau_r \rightarrow i\infty$ are again shuffle-regularized using Brown's 
tangential-base-point method \cite{brown2017multiple}. As will be reviewed below, genus-one
zeta generators guide the construction of modular complex combinations of (\ref{lieg1.35})
in the same way as genus-zero zeta generators gave rise to single-valued MPLs in one variable
in (\ref{svmpl.12}) \cite{Dorigoni:2024oft}.

%%%%%%%%%%%%%%%%%%%%%%%%%%%%%%%%%%%%%%%%%%%%%%%%%%
%%%%%%%%%%%%%%%%%%%%%%%%%%%%%%%%%%%%%%%%%%%%%%%%%%
\subsubsection{Generating series and modularity properties}
\label{sec:2.5.1}

Apart from the quasi-modular case of ${\rm G}_2(\tau)$, the integration kernels of
the iterated Eisenstein integrals (\ref{lieg1.35}) are gathered in the connection
 \begin{align}
 {\mathbb D}_{\ep^{\rm TS}}(\tau) &=    2\pi i \dd \tau
\sum_{k=4}^\infty \frac{ (k{-}1) }{(2\pi i)^{k}} \sum_{j=0}^{k-2}\dfrac{(-1)^j}{j!}(2\pi i \tau)^j  {\rm G}_k(\tau)\epsilon_k^{{ (j)\rm TS}}  
\label{tsconn}
 \end{align}
which takes values in Tsunogai's derivation algebra and is
closely related to the $\dd \tau$ part of the CEE connection (\ref{notsec.02})
at $z=0$.\footnote{More precisely, the connection (\ref{tsconn}) is obtained from
\[
  \mathbb K^{(\tau)}_{\ep,b}(z,\tau) = 2\pi i \bigg( \epsilon_0+\sum_{k=2}^\infty \frac{(k{-}1)}{(2\pi i)^k} \big[ {\rm G}_k(\tau)\epsilon_k - g^{(k)}(z,\tau) b_k \big]  \bigg)
  \]
by setting $z=0$, subtracting the $\epsilon_2^{{\rm TS}} $ part and performing following gauge transformation
\[
 {\mathbb D}_{\ep^{\rm TS}}(\tau) =
 \dd \tau\, e^{-2\pi i \ep_0 \tau} \mathbb K^{(\tau)}_{\ep,b}(0,\tau) 
e^{2\pi i \ep_0 \tau} + (\dd e^{-2\pi i \ep_0 \tau}) e^{2\pi i \ep_0 \tau}
- \frac{\dd \tau}{2\pi i}\, {\rm G}_2(\tau)\epsilon_2^{{\rm TS}}  
 \]} The iterated Eisenstein integrals (\ref{lieg1.35})
are then generated by the path-ordered exponential
\begin{align}
\mathbb I_{\ep^{\rm TS}}(\tau) &= {\rm Pexp} \bigg( \int^{i\infty}_\tau \mathbb D_{\ep^{\rm TS}}(\tau_1)   \bigg)
\label{notsec.18} \\
&= 1+\sum_{k_1=4}^\infty (k_1{-}1)  \sum_{j_1=0}^{k_1-2} \dfrac{ (-1)^{j_1} }{j_1!} \, \ee{j_1}{k_1}{\tau}\epsilon_{k_1}^{(j_1) {\rm TS}} \notag \\
&\quad
+\sum_{k_1,k_2=4}^\infty (k_1{-}1)(k_2{-}1) \sum_{j_1=0}^{k_1-2}  \sum_{j_2=0}^{k_2-2}\dfrac{ (-1)^{j_1+j_2} }{j_1!j_2!} \, \ee{j_1&j_2}{k_1&k_2}{\tau} \epsilon_{k_1}^{(j_1) {\rm TS}} \epsilon_{k_2}^{(j_2) {\rm TS}} +\ldots
\notag
\end{align}
By the modular weight $(k,0)$ of ${\rm G}_k(\tau)$, the differential forms $\dd \tau \, \tau^j {\rm G}_k(\tau)$
in the connection (\ref{tsconn}) close under ${\rm SL}_2(\mathbb Z)$ transformations.
In fact, the mixing of the differential forms $\dd \tau \, \tau^j {\rm G}_k(\tau)$
with fixed $k$ and $0\leq j \leq k{-}2$ under the modular group can be described
through the action (\ref{lieg1.13}) and (\ref{lieg1.15}) of the $\mathfrak{sl}_2$ generators $\ep_0,\ep_0^\vee$
on the accompanying $ \epsilon_{k}^{(j) {\rm TS}}$,
\beq
% {\mathbb D}_{\ep^{\rm TS}} \bigg(\frac{\alpha \tau {+} \beta}{\gamma \tau {+}\delta}\bigg) 
  {\mathbb D}_{\ep^{\rm TS}}  (\gamma\cdot \tau)   = U_\gamma^{-1} {\mathbb D}_{\ep^{\rm TS}}  ( \tau)  U_\gamma \, , \ \ \ \ \ \ \gamma \in {\rm SL}_2(\mathbb Z)
 \label{lieg1.38} 
 \eeq
where  $\gamma\cdot \tau = \frac{a\tau+b}{c\tau + d}$ for $\gamma = ( \smallmatrix a &b \\ c &d \endsmallmatrix ) \in {\rm SL}_2(\mathbb Z)$ and $U_\gamma$ is in the Lie group of $\mathfrak{sl}_2$. The action of
$U_T$ and $U_S$ for the modular $T: \tau \rightarrow \tau{+}1$
and $S: \tau \rightarrow -\frac{1}{\tau}$ transformations is given by\footnote{The expression for $U_S$ in (\ref{lieg1.39}) departs from that in section 4.2 of \cite{Dorigoni:2024oft},
\[
U^{[29]}_S = (2\pi i )^{[\ep_0,\ep_0^\vee]} e^{-\ep_0^\vee} e^{\ep_0} e^{-\ep_0^\vee}   = e^{- \ep_0^\vee/(2\pi i)} e^{ 2\pi i \ep_0} e^{-  \ep_0^\vee/ (2\pi i)} = U_S^{-1} 
\]
by an additional minus sign in the action on the generators $x,y$,
\[
  (U_S)^{-1} \bigg( \begin{array}{c} x \\   y \end{array} \bigg) U_S
  = - 
(U^{[29]}_S)^{-1} \bigg( \begin{array}{c} x \\   y \end{array} \bigg) U^{[29]}_S
= \bigg( \begin{array}{c} 2\pi i y \\  -x/(2\pi i) \end{array} \bigg)
\]
By the even degrees of $\ep_k^{(j){\rm TS}}$, both $U_S$ and $U^{[29]}_S$ have the same action (\ref{lieg1.40}) on Tsunogai derivations.}
\beq
U_T = e^{2\pi i \ep_0} \, , \ \ \ \ \ \ U_S 
%= (2\pi i )^{[\ep_0,\ep_0^\vee]} e^{-\ep_0^\vee} e^{\ep_0} e^{-\ep_0^\vee}   = e^{- \ep_0^\vee/(2\pi i)} e^{ 2\pi i \ep_0} e^{-  \ep_0^\vee/ (2\pi i)}  
= e^{\ep_0^\vee} e^{-\ep_0} e^{\ep_0^\vee} (2\pi i )^{-[\ep_0,\ep_0^\vee]}  = e^{ \ep_0^\vee/(2\pi i)} e^{ -2\pi i \ep_0} e^{\ep_0^\vee/ (2\pi i)}
 \label{lieg1.39} 
\eeq
which determines $U_\gamma$ by decomposing $\gamma \in {\rm SL}_2(\mathbb Z)$
into the generators $S,T$ and maps individual Tsunogai derivations to
\begin{align}
U_T^{-1}  \epsilon_{k}^{(j) {\rm TS}}  U_T &= 
\sum_{p=0}^{k-j-2} \frac{(-2\pi i)^p}{p!} \, \epsilon_{k}^{(j+p) {\rm TS}}   \label{lieg1.40} \\
U_S^{-1}  \epsilon_{k}^{(j) {\rm TS}}  U_S &=  \frac{ (-1)^j \, j!}{(k{-}2{-}j)!} (2\pi i )^{k-2-2j} \, \epsilon_{k}^{(k-j-2) {\rm TS}}
 \notag
\end{align}
However, the ${\rm SL}_2(\mathbb Z)$ transformation (\ref{lieg1.38}) of the connection
does not carry over to its path-ordered exponential $\mathbb I_{\ep^{\rm TS}}(\tau) $
in (\ref{notsec.18}): The obstruction is that the endpoint $i\infty$ of the integration path for
(\ref{notsec.18}) transforms to a distinct point $\gamma^{-1}(i\infty)$ 
as we perform the change of integration variable $\tau_1 = \gamma\cdot \rho_1$
in the first line of
\begin{align}
\mathbb I_{\ep^{\rm TS}}(\gamma\cdot \tau)  &= 
{\rm Pexp} \bigg( \int^{i\infty}_{\gamma\cdot \tau}    \mathbb  D_{\ep^{\rm TS}}(\tau_1)   \bigg) 
=
{\rm Pexp} \bigg( \int^{\gamma^{-1}(i\infty)}_{\tau}  U_\gamma^{-1} \, \mathbb  D_{\ep^{\rm TS}}(\rho_1) \, U_\gamma  \bigg)
\notag \\
&= {\rm Pexp} \bigg( \int^{\gamma^{-1}(i\infty)}_{i\infty}  U_\gamma^{-1}\,  \mathbb  D_{\ep^{\rm TS}}(\rho_1)  \,U_\gamma  \bigg)
\, U_\gamma^{-1}\,
\mathbb I_{\ep^{\rm TS}}(\tau) \,U_\gamma
\label{lieg1.41}
\end{align}
The path-ordered exponential in the second line connecting $\gamma^{-1}(i\infty)$ with $i\infty$ is
known as the {\it cocycle} associated with the modular transformation $\gamma \in {\rm SL}_2(\mathbb Z)$
and is the departure of $\mathbb I_{\ep^{\rm TS}}(\gamma\cdot \tau)$ from the
simple modular transformation law (\ref{lieg1.38}) of the connection.

%%%%%%%%%%%%%%%%%%%%%%%%%%%%%%%%%%%%%%%%%%%%%%%%%%
%%%%%%%%%%%%%%%%%%%%%%%%%%%%%%%%%%%%%%%%%%%%%%%%%%
\subsubsection{Equivariant iterated Eisenstein integrals}
\label{sec:2.5.2}

The previous section raises the question whether $\mathbb I_{\ep^{\rm TS}}(\tau)$ can be
composed with other group-like series in Tsunogai derivations to form a
series with the {\it equivariant} modular transformation
 \beq
\mathbb I^{\rm eqv}_{\ep^{\rm TS}}(\gamma\cdot \tau)  = 
U_\gamma^{-1} \mathbb I^{\rm eqv}_{\ep^{\rm TS}}( \tau) 
U_\gamma \ \ \forall \ \ \gamma \in {\rm SL}_2(\mathbb Z)
 \label{lieg1.42} 
\eeq
without any cocycles like the leftmost factor in the second line of (\ref{lieg1.41}). This question was answered in Brown's work \cite{Brown:2017qwo2} 
by composing $\mathbb I_{\ep^{\rm TS}}(\tau)$ with series in complex conjugate iterated Eisenstein integrals
and in single-valued MZVs. The fully explicit form of Brown's construction was proposed 
in \cite{Dorigoni:2024oft} and relies on series in genus-one zeta generators $\sigma_w$
\begin{align}
\mathbb M^{\rm sv}_{\sigma} &= \sum_{r=0}^{\infty} \sum_{i_1,\ldots,i_r \atop {\in 2\mathbb N+1} } \rho^{-1}\big({\rm sv}(f_{i_1} \ldots f_{i_r})\big) \, \sigma_{i_1}\ldots  \sigma_{i_r}
\label{not.05}
\end{align}
which mirrors the structure of its genus-zero counterpart $\mathbb M^{\rm sv}_0$ in (\ref{g0.05}). We shall similarly
write $\mathbb M^{\rm sv}_{z}$ and $\mathbb M^{\rm sv}_{\Sigma(u)}$ for the
variants of the series (\ref{not.05}) where each $\sigma_w$ is replaced by its
arithmetic, $\mathfrak{sl}_2$-invariant counterpart $z_w$, or by generators
$\Sigma_w(u)$ to be introduced in section \ref{sec:3.1.3} below.
With these series in single-valued MZVs in place, Brown's equivariant
series of iterated Eisenstein integrals of \cite{Brown:2017qwo2} translates into the statement of the following theorem:

\begin{theorem}
\label{2.thm:1}
 (conjectured in \cite{Dorigoni:2024oft} and equivalent to Theorem 8.2 of \cite{Brown:2017qwo2} based on results in section 15 of \cite{brown2017multiple})

The generating series $\mathbb I^{\rm eqv}_{\ep^{\rm TS}}(\tau)$ defined by 
 \beq
\mathbb I^{\rm eqv}_{\ep^{\rm TS}}(\tau)  = 
(\mathbb M^{\rm sv}_{z})^{-1}  \, \overline{ \mathbb I_{\ep^{\rm TS}}(\tau)^T} \, \mathbb M^{\rm sv}_{\sigma}\,
\mathbb I_{\ep^{\rm TS}}(\tau) 
 \label{lieg1.51} 
\eeq
\begin{itemize}
\item[(i)] takes values in the universal enveloping algebra of Tsunogai's derivations $\ep_k^{(j){\rm TS}}$ with $k\geq 4$ and $0\leq j \leq k{-}2$;
\item[(ii)] obeys the equivariance property (\ref{lieg1.42}) under ${\rm SL}_2(\mathbb Z)$.
\end{itemize}
\end{theorem}

Note that the transposition operation $(\ldots)^T$ reverses the Tsunogai derivations in the complex
conjugates of the series (\ref{notsec.18}) and introduces alternating signs
$ (\ep^{(j) {\rm TS} }_{k})^T = (-1)^j \ep^{(j) {\rm TS}}_{k} $ according to (\ref{lieg1.34}).
The expansion of $\overline{ \mathbb I_{\ep^{\rm TS}}(\tau)^T} $ in terms of complex conjugate iterated Eisenstein integrals is spelled out in (\ref{cc.18}).

\begin{proof}
\phantom{x}

$(i)$: We start by
decomposing all the zeta generators in $ \mathbb M^{\rm sv}_{\sigma}$ into $\sigma_w = z_w+\ldots$
with a Lie series of Tsunogai derivations in the ellipsis. The bracket relations (\ref{lieg1.23}) for
$[z_w, \ep_k^{(j){\rm TS}}]$ in terms of Lie polynomials in $\ep_m^{(\ell){\rm TS}}$ then allow to move all the $z_w$
to the leftmost position and to write $ \mathbb M^{\rm sv}_{\sigma} =  \mathbb M^{\rm sv}_{z} \mathbb B^{\rm sv}_{\ep^{\rm TS}}$ with a $\ep_k^{(j){\rm TS}}$-valued series $\mathbb B^{\rm sv}_{\ep^{\rm TS}}$. In a last step, one 
again uses the brackets in (\ref{lieg1.23}) to iteratively express the nested commutators
\beq
(\mathbb M_z^{\rm sv})^{-1} \, X \,
\mathbb M_z^{\rm sv} = \sum_{r=0}^{\infty} \sum_{i_1,i_2,\ldots,i_r \atop {\in 2\mathbb N+1} } \rho^{-1}\big({\rm sv}(f_{i_1} f_{i_2}\ldots f_{i_r})\big) \, [[ \ldots [[ X,z_{i_1}] , z_{i_2}], \ldots ], z_{i_r}]
\label{remzs}
\eeq
in terms of Tsunogai derivations, where $X$ represents an arbitrary word in 
$\ep_k^{(j){\rm TS}}$ from the expansion (\ref{cc.18}) of $\overline{ \mathbb I_{\ep^{\rm TS}}(\tau)^T}$.

$(ii)$: The key mechanism for the equivariance of (\ref{lieg1.51}) is the fact that 
the series $\mathbb M^{\rm sv}_{\sigma}$ absorbs the cocycle factors on its right and left due to the 
paths $\int^{\gamma^{-1}(i\infty)} _{i\infty}$ in the expression (\ref{lieg1.41})
for $\mathbb I_{\ep^{\rm TS}}(\gamma\cdot \tau) $ and its complex-conjugate
transpose $ \overline{ \mathbb I_{\ep^{\rm TS}}(\gamma \cdot \tau)^T}$, 
respectively \cite{Brown:2017qwo2, Dorigoni:2024oft}.
In case of the modular $T$ transformation, (\ref{lieg1.41}) specializes to
$\mathbb I_{\ep^{\rm TS}}(\tau{+}1)  = e^{2\pi i N^{\rm TS}} \mathbb I_{\ep^{\rm TS}}( \tau)  e^{2\pi i \ep_0}$
with $N^{\rm TS}$ given by (\ref{notsec.22}), and the the $T$ equivariance property (\ref{lieg1.42}) follows
from inserting
\beq
[N^{\rm TS} , \sigma_w ] = 0 \, , \ \ \ \ w\geq 3 \ {\rm odd} \ \ \ \ \Longrightarrow \ \ \ \
e^{-2\pi i N^{\rm TS}} \, \mathbb M^{\rm sv}_{\sigma} \, e^{2\pi i N^{\rm TS}} = \mathbb M^{\rm sv}_{\sigma}
 \label{lieg1.52} 
\eeq
into (see section 4.1 of \cite{Dorigoni:2024oft} and section 7.1 of \cite{Brown:2017qwo2})
\begin{align}
\mathbb I^{\rm eqv}_{\ep^{\rm TS}}(\tau{+}1) &= ( \mathbb M_z^{\rm sv} )^{-1} \, e^{-2\pi i \ep_0} \, \overline{ \mathbb I_{\ep^{\rm TS}}(\tau)^T} \, e^{-2\pi i N^{\rm TS}} \, \mathbb M^{\rm sv}_{\sigma}  \, e^{2\pi i N^{\rm TS}}\,
\mathbb I_{\ep^{\rm TS}}(\tau)\,
e^{2\pi i \ep_0} \notag \\
&= e^{-2\pi i \ep_0} \,  ( \mathbb M_z^{\rm sv} )^{-1} \, \overline{ \mathbb I_{\ep^{\rm TS}}(\tau)^T} \, \mathbb M^{\rm sv}_{\sigma} \,
\mathbb I_{\ep^{\rm TS}}(\tau)\,
e^{2\pi i \ep_0} \notag \\
&= e^{-2\pi i \ep_0} \, \mathbb I^{\rm eqv}_{\ep^{\rm TS}}(\tau) \,
e^{2\pi i \ep_0} 
 \label{extra.52} 
\end{align}
In case of the modular $S$ transformation, the absorption of cocycle factors via $\mathbb M^{\rm sv}_{\sigma}$ is based on a more refined argument which can be found in appendix \ref{app:thm21}.
In both cases, the left-multiplicative $U_\gamma^{-1} $ produced by $ \overline{ \mathbb I_{\ep^{\rm TS}}(\gamma \cdot \tau)^T}$
commutes with $(\mathbb M^{\rm sv}_{z})^{-1}$ for all $\gamma \in {\rm SL}_2(\mathbb Z)$ 
by $\mathfrak{sl}_2$ invariance of the arithmetic $z_w$ and becomes the leftmost factor of (\ref{lieg1.42}).
\end{proof}

Note that, after removing all the arithmetic zeta generators $z_w$ via (\ref{remzs}), the equivariant series (\ref{lieg1.51}) may be expanded
\begin{align}
\mathbb I^{\rm eqv}_{\ep^{\rm TS}}(\tau) &=1+\sum_{k_1=4}^\infty (k_1{-}1)  \sum_{j_1=0}^{k_1-2} \dfrac{ (-1)^{j_1} }{j_1!} \, \eeqqvv{j_1}{k_1}{\tau}\epsilon_{k_1}^{(j_1) {\rm TS}} \label{nwieqv} \\
&\quad
+\sum_{k_1,k_2=4}^\infty (k_1{-}1)(k_2{-}1) \sum_{j_1=0}^{k_1-2}  \sum_{j_2=0}^{k_2-2}\dfrac{ (-1)^{j_1+j_2} }{j_1!j_2!} \, \eeqqvv{j_1&j_2}{k_1&k_2}{\tau} \epsilon_{k_1}^{(j_1) {\rm TS}} \epsilon_{k_2}^{(j_2) {\rm TS}} +\ldots
\notag
\end{align}
The coefficients ${\cal E}^{\rm eqv}[\ldots;\tau]$ will be referred to as {\it equivariant iterated Eisenstein integrals}. Their combinations entering (\ref{nwieqv}) are expressible in terms of combinations of the meromorphic iterated Eisenstein integrals ${\cal E}[\ldots;\tau]$ in (\ref{notsec.18}) and their complex conjugates, with $\mathbb Q$-linear combinations of single-valued MZVs as coefficients \cite{Brown:2017qwo2, Dorigoni:2024oft}. 
%%%
%%%
However, the coefficients $\eeqqvv{j_1&j_2}{k_1&k_2}{\tau}$ at degree $k_1{+}k_2\geq14$
(and generic ${\cal E}^{\rm eqv}[\ldots;\tau]$ at modular depth $\geq 3$ and degree $\geq 16$) are not individually well-defined by (\ref{nwieqv})
since the accompanying derivations obey
Pollack relations such as (\ref{lieg1.08}). The reason is that Eisenstein series alone do not
suffice to find equivariant completions of holomorphic double integrals over ${\rm G}_{k_1}(\tau_1)
{\rm G}_{k_2}(\tau_2)$ with $k_1{+}k_2\geq14$, and holomorphic cusp forms are need
as additional integration kernels \cite{brown2017multiple, Brown:2017qwo}.
By providing the equivariant completion via cusp forms, standalone
definitions of all the $\eeqqvv{j_1&j_2}{k_1&k_2}{\tau}$ with arbitrary $k_1,k_2$
and $0\leq j_i \leq k_i{-}2$ has been given in \cite{Dorigoni:2021ngn, Dorigoni_2022}, also see 
\cite{Dorigoni:2024oft} for generalizations to triple Eisenstein integrals.

%%%%%%%%%%%%%%%%%%%%%%%%%%%%%%%%%%%%%%%%%%%%%%%%%%
%%%%%%%%%%%%%%%%%%%%%%%%%%%%%%%%%%%%%%%%%%%%%%%%%%
\subsubsection{Modular versus holomorphic frame}
\label{sec:2.5.3}

In spite of its equivariant transformation (\ref{lieg1.42}), the coefficients of
the words $\ep_k^{(j){\rm TS}}$ of the series $\mathbb I^{\rm eqv}_{\ep^{\rm TS}}(\tau) $ are
not yet modular forms. Still, one can attain a modular connection\footnote{The connection (\ref{tsconn}) transforms as follows under $\Umod(\tau)$ in (\ref{usl2}) 
\[
\Umod(\tau) \,  {\mathbb D}_{\ep^{\rm TS}}(\tau_1)  \,\Umod(\tau)^{-1}
= \frac{ \dd \tau_1 }{2\pi i}\sum_{k=4}^{\infty}  (k{-}1)  \sum_{j=0}^{k-2} \frac{(-1)^j}{j!} \,
\bigg( \frac{\tau{-}\tau_1}{4\pi \Im \tau} \bigg)^{k-j-2}\! (\bar\tau{-}\tau_1)^j \, {\rm G}_k(\tau_1)\,
\ep_k^{(j){\rm TS}}
\]
such that the coefficients of the series (\ref{defheqv}) are iterated integrals
of forms $\sim \dd \tau_1 (\tau{-}\tau_1)^{k-j-2} (\bar\tau{-}\tau_1)^j {\rm G}_k(\tau_1)$.} 
and series in iterated 
Eisenstein integrals by conjugation with the element \cite{Dorigoni:2024oft}\footnote{Note that the expression (\ref{usl2}) for $\Umod(\tau)$ is denoted by $U_{{\rm SL}_2}(\tau)$ in \cite{Dorigoni:2024oft}.} 
\begin{align}
\Umod(\tau) = \exp \bigg( \frac{-\epsilon_0^\vee}{4\pi \Im \tau} \bigg)  \exp(2\pi i \bar \tau \epsilon_0)
\label{usl2}
\end{align}
of the Lie group of $\mathfrak{sl}_2$. The transformed series\footnote{We depart from the conventions of \cite{Dorigoni:2024oft}, where 
the derivations in the expansion (\ref{defheqv}) of $\mathbb H^{\rm eqv}_{\ep^{\rm TS}}(\tau) $ 
are further reflected $\frac{1}{j!}\epsilon_{k}^{(j) {\rm TS}} \rightarrow \frac{1}{(k{-}j{-}2)!}\epsilon_{k}^{(k-j-2) {\rm TS}} $
and reversed in their concatenation order through an operation $R[\ldots]$. The $\tau$-derivative
of the resulting series $\mathbb J^{\rm eqv}_{\ep^{\rm TS}}  = R[ \mathbb H^{\rm eqv}_{\ep^{\rm TS}}]$ 
in the reference produces left-multiplicative Tsunogai derivations that line up with the
differential equations of closed-string genus-one integrals in \cite{Gerken:2019cxz, Gerken:2020yii} which involve
conjectural matrix representations of the $\epsilon_{k}^{(j) {\rm TS}}$.}
\begin{align}
\mathbb H^{\rm eqv}_{\ep^{\rm TS}}(\tau)  &= \Umod(\tau) \, \mathbb I^{\rm eqv}_{\ep^{\rm TS}}(\tau)  \, \Umod(\tau)^{-1} 
 \label{defheqv} \\
&= 1+\sum_{k_1=4}^\infty (k_1{-}1)  \sum_{j_1=0}^{k_1-2} \dfrac{ (-1)^{j_1} }{j_1!} \beqvtau{j_1}{k_1}{\tau}\epsilon_{k_1}^{(j_1) {\rm TS}} \notag \\
&\quad
+\sum_{k_1,k_2=4}^\infty (k_1{-}1)(k_2{-}1) \sum_{j_1=0}^{k_1-2}  \sum_{j_2=0}^{k_2-2}\dfrac{ (-1)^{j_1+j_2} }{j_1!j_2!} \beqvtau{j_1&j_2}{k_1&k_2}{\tau} \epsilon_{k_1}^{(j_1) {\rm TS}} \epsilon_{k_2}^{(j_2) {\rm TS}} +\ldots
\notag
\end{align}
then becomes invariant under $T$ and is diagonal under $S$ as one can see from the Cartan generator
${\rm h}=[\ep_0,\ep_0^\vee]$ in 
\beq
\mathbb H^{\rm eqv}_{\ep^{\rm TS}}(\tau{+}1)= \mathbb H^{\rm eqv}_{\ep^{\rm TS}}(\tau) \, , \ \ \ \ \ \
\mathbb H^{\rm eqv}_{\ep^{\rm TS}}\bigg({-}\frac{1}{\tau} \bigg) = 
\bar\tau^{-{\rm h}} \, \mathbb H^{\rm eqv}_{\ep^{\rm TS}}(\tau)\, \bar\tau^{{\rm h}}
\label{trfheqv}
\eeq
By the Cartan eigenvalues $2j{-}k{+}2$ of $\ep_k^{(j){\rm TS}}$ in (\ref{lieg1.14}),
the transformation (\ref{trfheqv}) translates into the following modular weights of
its coefficients in (\ref{defheqv})
\beq
\beqvtau{j_1 &\ldots &j_r}{k_1 &\ldots &k_r}{ \frac{a\tau{+}b}{ c\tau{+}d } }
= \bigg( \prod_{i=1}^r (c \bar \tau{+}d)^{k_i - 2j_i - 2} \bigg)\, \beqvtau{j_1 &\ldots &j_r}{k_1 &\ldots &k_r}{ \tau }
\label{wtbeqv}
\eeq
However, the coefficients $\beqvtau{j_1&j_2}{k_1&k_2}{\tau}$ at degree $k_1{+}k_2\geq14$
are not individually well-defined by (\ref{defheqv}) due to Pollack relations among the accompanying $\ep_k^{(j){\rm TS}}$, see the discussion below (\ref{nwieqv}). Hence, (\ref{wtbeqv}) is understood to apply to those combinations of $\beta^{\rm eqv}[\ldots;\tau]$ which arise in the expansion (\ref{defheqv}) after modding out by all Pollack relations among $\ep_k^{(j) {\rm TS}}$

%%%%%%%%%%%%%%%%%%%%%%%%%%%%%%%%%%%%%%%%%%%%%%%%%%
%%%%%%%%%%%%%%%%%%%%%%%%%%%%%%%%%%%%%%%%%%%%%%%%%%
\subsubsection{Single-valued iterated Eisenstein integrals}
\label{sec:2.5.4}

The following variant of the equivariant series (\ref{lieg1.51}),
with $(\mathbb M^{\rm sv}_{\sigma})^{-1} $ in the place of $(\mathbb M^{\rm sv}_{z})^{-1}$
as its leftmost factor, generates single-valued iterated Eisenstein integrals
\cite{Dorigoni:2024oft} 
 \beq
\mathbb I^{\rm sv}_{\ep^{\rm TS}}(\tau)  = 
(\mathbb M^{\rm sv}_{\sigma})^{-1}  \, \overline{ \mathbb I_{\ep^{\rm TS}}(\tau)^T} \, \mathbb M^{\rm sv}_{\sigma}\,
\mathbb I_{\ep^{\rm TS}}(\tau) 
 \label{lieg1.61} 
\eeq
In contrast to $\mathbb I^{\rm eqv}_{\ep^{\rm TS}}(\tau)$, the series $\mathbb I^{\rm sv}_{\ep^{\rm TS}}(\tau)$ is unaffected by the ambiguities in singling out the arithmetic generators
$z_w$ within $\sigma_w$ (though the criterion of \cite{Dorigoni:2024iyt}
to not admit any $\mathfrak{sl}_2$ invariant terms in $\sigma_w{-} z_w$ 
prescribes a canonical choice of $z_w$). Moreover, the conjugation of 
$ \overline{ \mathbb I_{\ep^{\rm TS}}(\tau)^T}$ by $\mathbb M^{\rm sv}_{\sigma}$ brings
the single-valued genus-one series (\ref{lieg1.61}) in close formal analogy with the
construction (\ref{svmpl.12}) of single-valued MPLs in one variable at genus zero
(with the non-equivariant $\mathbb I_{\ep^{\rm TS}}(\tau) $ taking the role of the
multi-valued MPL series $\mathbb G_{e_0,e_1}(z) $).
However, given that the leftmost factor of $(\mathbb M^{\rm sv}_{\sigma})^{-1} $ in
(\ref{lieg1.61}) does not commute with the transformation $U^{-1}_\gamma$, the single-valued series does not share the
 equivariant modular properties (\ref{lieg1.42}) of $\mathbb I^{\rm eqv}_{\ep^{\rm TS}}(\gamma\cdot \tau)$
 \beq
\mathbb I^{\rm sv}_{\ep^{\rm TS}}(\gamma\cdot \tau)  = 
(\mathbb M^{\rm sv}_{\sigma})^{-1}  \, U_\gamma^{-1} \,
\mathbb M^{\rm sv}_{\sigma} \, \mathbb I^{\rm sv}_{\ep^{\rm TS}}( \tau) 
U_\gamma
 \label{lieg1.62} 
\eeq
In Theorem \ref{3.cor:1} below, we present a series in single-valued eMPLs which by Remark \ref{3.rmk:1} generalizes the formal analogy between genus zero and genus one to one additional
marked point. In the same way as single-valued MPLs in two variables
are built in (\ref{svmpl.15}) via conjugation by a series in their one-variable counterpart,
the series of single-valued eMPLs features a conjugation by $\mathbb I^{\rm sv}_{\ep^{\rm TS}}(\tau)$
which takes the role of the conjugation by $ \mathbb G^{\rm sv}_{e_0,e_1}(y) $ at genus zero.

%%%%%%%%%%%%%%%%%%%%%%%%%%%%%%%%%%%%%%%%%%%%%%%%%%
%%%%%%%%%%%%%%%%%%%%%%%%%%%%%%%%%%%%%%%%%%%%%%%%%%
\subsection{Building blocks for $z$-dependent equivariant series}
\label{sec:3.1}

In preparation for the equivariant and single-valued series that depend on both $\tau$
and a point $z=u\tau {+}v$ on the torus, we shall here gather the composing series.

%%%%%%%%%%%%%%%%%%%%%%%%%%%%%%%%%%%%%%%%%%%%%%%%%%
%%%%%%%%%%%%%%%%%%%%%%%%%%%%%%%%%%%%%%%%%%%%%%%%%%
\subsubsection{Series in iterated $\tau$-integrals}
\label{sec:3.1.1}

Iterated integrals over combinations of $ f^{(k)}(u\tau{+}v,\tau)$ and ${\rm G}_k(\tau)$
will be generated from the connection\footnote{The $u,v$-dependent connection
(\ref{defdepb}) specializes to the solely $\tau$ dependent $ \mathbb D_{\ep^{\rm TS}}(\tau) $ in
(\ref{tsconn}) via 
 \[
  \mathbb D_{\ep^{\rm TS}}(\tau) = \lim_{u,v \rightarrow 0} \bigg\{
  \mathbb D_{\ep,b}(u,v,\tau) + \frac{\dd \tau}{2\pi i} \,  f^{(2)}(u\tau{+}v,\tau)\, b_2
  \bigg\}
 \]
 where the subtraction of $f^{(2)}(u\tau{+}v,\tau)$ is essential due to its ill-defined
 limit as $u,v\rightarrow 0$.}
\begin{align}
\mathbb D_{\ep,b}(u,v,\tau) &=  2\pi i \dd \tau
\sum_{k=2}^\infty \frac{ (k{-}1) }{(2\pi i)^{k}} \sum_{j=0}^{k-2}\dfrac{(-1)^j}{j!}\,(2\pi i \tau)^j\big[ {\rm G}_k(\tau)\epsilon_k^{(j)}
 - f^{(k)}(u\tau{+}v,\tau)b_k^{(j)} \big]  
 \label{defdepb}
 \end{align}
 which is meromorphic in $\tau$ at fixed co-moving coordinates $u,v \in \mathbb R$ 
 and thus gives rise to a homotopy invariant path-ordered exponential
 \begin{align}
\mathbb I_{\ep,b}(u,v,\tau) &= {\rm Pexp} \bigg( \int^{i\infty}_\tau \mathbb D_{\ep,b}(u,v,\tau_1)   \bigg)
\label{notsec.17}\\
&= 1+\sum_{k_1=2}^\infty (k_1{-}1) \sum_{j_1=0}^{k_1-2} \dfrac{(-1)^{j_1}}{j_1!}\biggl\{\ee{j_1}{k_1}{\tau}\epsilon_{k_1}^{(j_1)}+\eez{j_1}{k_1}{z}{\tau}b_{k_1}^{(j_1)}\biggr\}\notag\\
    &\quad+\sum_{k_1,k_2=2}^\infty (k_1{-}1)(k_2{-}1)  \sum_{j_1=0}^{k_1-2}  \sum_{j_2=0}^{k_2-2}\dfrac{ (-1)^{j_1+j_2}}{j_1!j_2!}\biggl\{\ee{j_1&j_2}{k_1&k_2}{\tau}\epsilon_{k_1}^{(j_1)}\epsilon_{k_2}^{(j_2)}  \notag\\
    &\qquad 
   +\eez{j_1&j_2}{k_1&k_2}{z&z}{\tau}b_{k_1}^{(j_1)}b_{k_2}^{(j_2)} 
    +\eez{j_1&j_2}{k_1&k_2}{&z}{\tau}\epsilon_{k_1}^{(j_1)}b_{k_2}^{(j_2)}
    +\eez{j_1&j_2}{k_1&k_2}{z&}{\tau}b_{k_1}^{(j_1)}\epsilon_{k_2}^{(j_2)}\biggr\}+\dots
\notag
\end{align}  
We set $\ep_2 = 0$ by Proposition \ref{e2prop}, and the words in $b_{k}^{(j)}, \ep_{k}^{(j)} $ are not independent in view of the bracket
relations (\ref{intrel}) and those in appendix \ref{sec:D.2}. The iterated integrals
${\cal E}[\ldots;\tau]$ are defined in (\ref{bsc.07}), (\ref{absc.07}), and the analogous
expansion of their (transposed) complex-conjugate series can be found in (\ref{cc.17}). Numerous properties of the iterated integrals ${\cal E}[\ldots;\tau]$ are discussed in section 3 of \cite{Hidding:2022vjf}.

%%%%%%%%%%%%%%%%%%%%%%%%%%%%%%%%%%%%%%%%%%%%%%%%%%
%%%%%%%%%%%%%%%%%%%%%%%%%%%%%%%%%%%%%%%%%%%%%%%%%%
\subsubsection{Series in eMPLs}
\label{sec:3.1.2}

The conversion of the above iterated $\tau$-integrals to eMPLs in later sections
will single out the following series in eMPLs
\begin{align}
\mathbbm{\Gamma}_{x,y}(z,\tau) &=
e^{-2\pi i \tau \ep_0}
\exp\bigg( \int^{i\infty}_\tau \frac{\dd \tau_1}{2\pi i}  \, {\rm G}_2(\tau_1)b_2\bigg)
{\rm Pexp} \biggl(\int_z^0
\tilde{\mathbb K}_{x,y}(z_1,\tau)
\biggr) 
e^{-ux} e^{2\pi i \tau \ep_0}
\notag \\
&= \exp\big( \ee{0}{2}{\tau} b_2\big)
{\rm Pexp} \biggl(\int_z^0
\tilde{\mathbb K}_{x-2\pi i \tau y,y}(z_1,\tau)
\biggr) 
e^{-u(x -2\pi i \tau y)} 
\notag \\
&= \exp\big( \ee{0}{2}{\tau} b_2\big)
{\rm Pexp} \biggl(\int_z^0
{\mathbb J}^{\rm BL}_{x-2\pi i \tau y,y}(z_1,\tau)
\biggr) 
\label{notsec.19}
\end{align}
The CEE connection $\tilde{\mathbb K}$ and the Brown-Levin connection 
${\mathbb J}^{\rm BL}$ are defined
in (\ref{defKcon}) and (\ref{notsec.06}), respectively, and the primitive
\beq
\ee{0}{2}{\tau}=\int^{i\infty}_\tau \frac{\dd \tau_1}{2\pi i}  \, {\rm G}_2(\tau_1)
\label{nmodg2}
\eeq
of the quasi-modular Eisenstein series ${\rm G}_2$ in (\ref{lieg1.03}) 
which was excluded from the earlier series $\mathbb I_{\ep^{\rm TS}}(\tau)$ in (\ref{notsec.18}) will be
seen in appendix \ref{app:gmod} to play an essential role for the modular 
properties of (\ref{notsec.19}). In passing to the last line of (\ref{notsec.19}), we have used
the gauge equivalence of the CEE and Brown-Levin connections, in particular
the corollary (\ref{gaugepexp}) for their path-ordered exponentials.

In view of the complex combinations in the letters of  $\tilde{\mathbb K}$,
${\mathbb J}^{\rm BL}$ in (\ref{notsec.19}), it might be worthwhile
to spell out the transposed complex-conjugate series
\begin{align}
\overline{\mathbbm{\Gamma}_{x,y}(z,\tau)^T} &=
e^{-2\pi i \bar \tau \ep_0} e^{ux}
{\rm Pexp} \biggl(\int^{\bar z}_0
\overline{ \tilde{\mathbb K}_{x,-y}(z_1,\tau)}\biggr) 
\exp\bigg( \int^{\bar \tau}_{-i\infty} \frac{\dd \bar \tau_1}{2\pi i}  \,\overline{ {\rm G}_2(\tau_1) } b_2\bigg)
 e^{2\pi i \bar \tau \ep_0} \notag \\
 &= e^{u(x-2\pi i \bar \tau y)}
{\rm Pexp} \biggl(\int^{\bar z}_0
\overline{ \tilde{\mathbb K}_{(x+2\pi i \tau y),-y}(z_1,\tau) } \biggr) 
\exp\big( \overline{\ee{0}{2}{\tau}} b_2\big)\notag \\
&= {\rm Pexp} \biggl(\int^{\bar z}_0
\overline{ {\mathbb J}^{\rm BL}_{(x+2\pi i \tau y),-y}(z_1,\tau) } \biggr) 
\exp\big( \overline{\ee{0}{2}{\tau}} b_2\big)
\label{ahologam}
\end{align}
see (\ref{cctilK}) for the expansion of ${\rm Pexp} \Big(\int^{\bar z}_0
\overline{ \tilde{\mathbb K}_{x,-y}(z_1,\tau)}\Big) $.

Note that, up to and including degree two in the letters $x,y$, the expansion of
(\ref{notsec.19}) reads
\begin{align}
\mathbbm{\Gamma}_{x,y}(z,\tau) 
&= 1 - ux - 2\pi i v y + \tfrac{1}{2} u^2 x^2 + \tfrac{1}{2} v^2 (2\pi i y)^2 + 2\pi i u v  yx  \notag \\
&\quad + [y,x] \bigg( i\pi u^2 \tau + \tilde \Gamma( \smallmatrix 1 \\ 0 \endsmallmatrix;z,\tau)
 +\ee{0}{2}{\tau} \bigg) + \ldots
\label{expga}
\end{align}
i.e.\ the simplest eMPL with a non-trivial $q$-series expansion occurs along with the degree-two Lie polynomial $[y,x] = b_2$.

%%%%%%%%%%%%%%%%%%%%%%%%%%%%%%%%%%%%%%%%%%%%%%%%%%
%%%%%%%%%%%%%%%%%%%%%%%%%%%%%%%%%%%%%%%%%%%%%%%%%%
\subsubsection{Series in single-valued MZVs}
\label{sec:3.1.3}

The genus-one zeta generators $\sigma_w$ of section \ref{sec:2.4.2}
will enter our main results on $u,v$-dependent equivariant and single-valued
series in the following combinations: we define {\it augmented zeta generators}  
   \beq
\Sigma_w(u) = e^{u x} \big( P_w( t_{12}, t_{01})+\sigma_w  \big) e^{-ux} \, , \ \ \ \ \ \ w\geq 3 \ {\rm odd}  
\label{not.06}
\eeq
where the Lie polynomials $P_w$ defined by (\ref{notsec.12}) involve the
letter $t_{12} = [y,x]$ and the Lie series $t_{01}$ in $x,y$ given by 
(\ref{lieg1.17}). The exponentials involving $e^{\pm u x}$ can be
rewritten in terms of the adjoint action $e^{u \ad_x}$ acting on
$P_w( t_{12}, t_{01})+\sigma_w$, and the resulting commutators
of the form $[x,[\ldots[x,[x,\sigma_w]]\ldots ]]$ are determined by the expansion of
$\sigma_w$ in terms of Tsunogai derivations and (\ref{dfprpsig}) for
the contribution of their arithmetic term $z_w$.
Following the discussion below (\ref{not.05}), the notation
$ \mathbb M^{\rm sv}_{\Sigma(u)}$ refers to the series in single-valued MZVs obtained
from $ \mathbb M^{\rm sv}_\sigma$ by replacing $\sigma_w \rightarrow \Sigma_w(u)$.

\begin{lemma}
\label{zwprop}
The augmented zeta generators $\Sigma_w(u)$ in (\ref{not.06}) are expressible
as Lie series in $\ep_k^{(j)},b_k^{(j)}$ up to a single arithmetic
term $z_w$ at degree $2w$ identical to the one
in (\ref{lieg1.21}) for $\sigma_w$.
\end{lemma}

\begin{proof}
The statement of the lemma will be demonstrated separately for $(i)$ the first term $ e^{u x} P_w(t_{12}, t_{01})  e^{-ux} $ in (\ref{not.06}); $(ii)$ the $u$-independent part $\sigma_w  $ of the second term in (\ref{not.06}) and $(iii)$ the $u$-dependent part $( e^{u {\rm ad}_x}-1) \sigma_w  $ of the second term in (\ref{not.06}).

$(i)$: Both $P_w(t_{12}, t_{01})$ and $ e^{u x} P_w(t_{12}, t_{01})  e^{-ux} $ are Lie series in $x,y$ with lowest degrees $w{+}1$. Since all Lie polynomials in $x,y$ of degree $\geq 2$ are Lie polynomials in $b_k^{(j)}$, see item $(iii)$ of Lemma \ref{braklem}, the same is true for each term in the series expansion of the first term $ e^{u x} P_w(t_{12}, t_{01})  e^{-ux} $.

$(ii)$: The statement of the lemma for this term follows from the decomposition $\sigma_w=z_w+\ldots$ with Lie series in Tsunogai derivations $\ep_k^{(j){\rm TS}} = \ep_k^{(j)} + b_k^{(j)}$ in the ellipsis (see section \ref{sec:2.4.2}) which exposes the only arithmetic term $z_w$ in $\Sigma_w(u)$.

$(iii)$:  Given that the non-arithmetic parts $\sigma_w{-}z_w$ are Lie series in $\ep_k^{(j){\rm TS}}$, the analogous $u$-dependent parts $( e^{u {\rm ad}_x}-1)(\sigma_w{-}z_w)   $ of the second term are Lie series in $b_k^{(j)}$ by item $(i)$ of Lemma \ref{braklem}. The remainder $( e^{u {\rm ad}_x}-1) z_w$ of the $u$-dependent part of the second term in turn is a Lie series in $b_k^{(j)}$ by items $(iv)$ and $(i)$ of Lemma \ref{braklem}.
\end{proof}

Note that examples of low-degree terms in the expansion of $\Sigma_3(u),\Sigma_5(u),\Sigma_7(u)$ can be found in appendix \ref{sec:C.4}.

%%%%%%%%%%%%%%%%%%%%%%%%%%%%%%%%%%%%%%%%%%%%%%%%%%
%%%%%%%%%%%%%%%%%%%%%%%%%%%%%%%%%%%%%%%%%%%%%%%%%%
\subsection{Equivariant transformation in the $z$-dependent case}
\label{sec:3.eqv}

Before stating the main theorems on our construction of equivariant and single-valued generating series, we shall specify the notion of ${\rm SL}_2(\mathbb Z)$ equivariance in the $z$-dependent case. In the $z$-independent situation of the connection $\mathbb D_{\ep^{\rm TS}}(\tau)$ in (\ref{lieg1.38}) and the series $\mathbb I^{\rm eqv}_{\ep^{\rm TS}}(\tau) $ in (\ref{lieg1.42}), (\ref{lieg1.51}), a modular transformation of $\tau$ by $\gamma \in {\rm SL}_2(\mathbb Z)$ translated into a conjugation by an element $U_\gamma$ in the Lie-group of the $\mathfrak{sl}_2$ generated by $\ep_0$ and $\ep_0^\vee$.

The Lie-group elements $U_T$ and $U_S$ associated with the modular $T$ and $S$ transformations are given by (\ref{lieg1.39}) and lead to the following $\mathfrak{sl}_2$ action on the Lie-algebra generators $x,y$:\footnote{This can be written in the unified form $\big( \smallmatrix x \\ -2\pi i y \endsmallmatrix\big) \rightarrow \big( \smallmatrix  a&b \\ c&d \endsmallmatrix\big)  \big( \smallmatrix x \\ -2\pi i y \endsmallmatrix\big) $ with $ ( \smallmatrix a &b \\ c &d \endsmallmatrix ) \in {\rm SL}_2(\mathbb Z)$ which identifies the vector $\big( \smallmatrix x \\ -2\pi i y \endsmallmatrix\big)$ as transforming in the defining representation of ${\rm SL}_2(\mathbb Z)$.}
\beq
 U_T^{-1} \bigg( \begin{array}{c} x \\   y \end{array} \bigg) U_T
= \bigg( \begin{array}{c} x{-}2\pi i y \\  y \end{array} \bigg)
\, , \ \ \ \ \ \
 U_S^{-1} \bigg( \begin{array}{c} x \\   y \end{array} \bigg) U_S
= \bigg( \begin{array}{c} 2\pi i y \\  -x/(2\pi i) \end{array} \bigg)
\label{uschoice.00}
\eeq
This implies the following $\mathfrak{sl}_2$ action on the Lie polynomials $b_k^{(j)}$
which mirrors that on the Tsunogai derivations (\ref{lieg1.40}) and by (\ref{lieg1.25}) extends to the $\epsilon_k^{(j)}$, 
\begin{align}
U_T^{-1} b_k^{(j)} U_T &= \sum_{p=0}^{k-j-2} \frac{(-2\pi i)^p}{p!}  \, b_k^{(j+p)}\, , &
U_S^{-1} b_k^{(j)} U_S &= \frac{ (-1)^j \, j!}{(k{-}2{-}j)!} (2\pi i )^{k-2-2j} b_k^{(k-j-2)}
\label{uschoice.01} \\
U_T^{-1} \ep_k^{(j)} U_T &= \sum_{p=0}^{k-j-2} \frac{(-2\pi i)^p}{p!}  \, \ep_k^{(j+p)}\, , &
U_S^{-1} \ep_k^{(j)} U_S &= \frac{ (-1)^j \, j!}{(k{-}2{-}j)!} (2\pi i )^{k-2-2j} \ep_k^{(k-j-2)}
\notag
\end{align}
with $k\geq 2$ and $0\leq j \leq k{-}2$ in all cases. These $\mathfrak{sl}_2$ actions capture the modular properties of the connection ${\mathbb D}_{\ep,b}  (u,v, \tau)$ at the heart of our main results, i.e.\ its transformation under
\beq
\gamma\cdot (z,\tau) = \bigg( \frac{z}{c\tau {+} d},\frac{a\tau{+}b}{c\tau {+} d}\bigg) \, , \ \ \ \ \gamma\cdot ( \smallmatrix v \\ -u \endsmallmatrix )
= ( \smallmatrix a &b \\ c &d \endsmallmatrix ) ( \smallmatrix v \\ -u \endsmallmatrix ) \, , \ \ \ \  \gamma = ( \smallmatrix a &b \\ c &d \endsmallmatrix ) \in {\rm SL}_2(\mathbb Z)
\label{uschoice.03}
\eeq

\begin{prop}
\label{dequiv}
The modular transformation (\ref{uschoice.03}) of the connection ${\mathbb D}_{\ep,b}  (u,v, \tau) $ defined by (\ref{defdepb}) translates into the following $\mathfrak{sl}_2$ action
\beq
{\mathbb D}_{\ep,b}  \big(\gamma\cdot (u,v,\tau)\big)   = U_\gamma^{-1} {\mathbb D}_{\ep,b}  (u,v, \tau)  U_\gamma \, , \ \ \ \ \ \ \gamma \in {\rm SL}_2(\mathbb Z)
\label{uschoice.02}
\eeq
\end{prop}

\begin{proof}
The proposition is a simple consequence of the modular transformation of $f^{(n)}(z,\tau)$ in (\ref{appA.13}), the fact that $\ep_2=0$ only leaves holomorphic Eisenstein series ${\rm G}_k(\tau)$ of weight $k\geq 4$ and the effect (\ref{uschoice.01}) of $U_\gamma$ on the non-commuting variables $b_k^{(j)},\ep_k^{(j)}$.
\end{proof}

\begin{definition}
\label{def:eqv}
Any combination $\mathbb F_{x,y,\ep}(z,\tau)$ of Lie algebra generators $x,y,\ep_k^{(j)}$ with the same modular transformation
\beq
\mathbb F_{x,y,\ep}\big(\gamma\cdot(z,\tau) \big)  = U_\gamma^{-1} \mathbb F_{x,y,\ep}(z,\tau) U_\gamma \, , \ \ \ \ \ \ \gamma \in {\rm SL}_2(\mathbb Z)
\label{uschoice.07}
\eeq
as the connection ${\mathbb D}_{\ep,b}  (u,v, \tau)$  in (\ref{uschoice.02}) will be referred to as equivariant.
\end{definition}

The equivariance (\ref{uschoice.02}) of ${\mathbb D}_{\ep,b}  (u,v,\tau)$
mirrors that of the $z$-independent $\mathbb D_{\ep^{\rm TS}}(\tau)$ in (\ref{lieg1.38}). Similar to the discussion of cocycles in (\ref{lieg1.41}), the $\gamma \in {\rm SL}_2(\mathbb Z)$ action on the path-ordered exponential (\ref{notsec.17}) will depart from $U_\gamma^{-1} \mathbb I_{\ep,b}(u,v,\tau)  U_\gamma$ by left-multiplicative cocycle factors which we do not need to compute for the purpose of this work.

Note that the expression (\ref{lieg1.39}) for $U_S$ and its action in (\ref{uschoice.00}), (\ref{uschoice.01}) single out the realization of the modular $S$ transformation via $(z,\tau) \rightarrow (\frac{z}{\tau}, - \frac{1}{\tau})$ with a positive sign of the term $\frac{z}{\tau}$, i.e.\ through the ${\rm SL}_2(\mathbb Z)$ matrix $( \smallmatrix 0 &-1 \\ 1 &0 \endsmallmatrix )$ as opposed to $( \smallmatrix 0 &1 \\ -1 &0 \endsmallmatrix )$.

%%%%%%%%%%%%%%%%%%%%%%%%%%%%%%%%%%%%%%%%%%%%%%%%%%
%%%%%%%%%%%%%%%%%%%%%%%%%%%%%%%%%%%%%%%%%%%%%%%%%%
%%%%%%%%%%%%%%%%%%%%%%%%%%%%%%%%%%%%%%%%%%%%%%%%%%
%%%%%%%%%%%%%%%%%%%%%%%%%%%%%%%%%%%%%%%%%%%%%%%%%%
\section{Main results}
\label{sec:3}

This section gathers the main results of this work while leaving several of the proofs
for later sections and appendices. More specifically, the series (\ref{lieg1.51}) of equivariant
iterated Eisenstein integrals will be generalized in section \ref{sec:3.2} to include the integration kernels
$\sim \dd \tau \, f^{(k)}(u\tau{+}v,\tau)$ of (\ref{bsc.07}) and related to 
series in single-valued eMPLs and their complex conjugates.
Sections \ref{sec:3.cc} and \ref{sec:3.3} are dedicated to the behavior of this series in single-valued eMPLs under complex conjugation and the degeneration limit $\tau \rightarrow i \infty$ of the torus at fixed co-moving coordinates $u,v$ of the point $z=u\tau{+}v$, respectively.
 
%%%%%%%%%%%%%%%%%%%%%%%%%%%%%%%%%%%%%%%%%%%%%%%%%%
%%%%%%%%%%%%%%%%%%%%%%%%%%%%%%%%%%%%%%%%%%%%%%%%%%
\subsection{Equivariant iterated integrals and single-valued eMPLs}
\label{sec:3.2}

A key step towards $z$- or $(u,v)$-dependent equivariant  and single-valued
generating series is stated in the following theorem:
\begin{theorem}
\label{3.thm:1}
With the definitions (\ref{notsec.17}), (\ref{notsec.19}), (\ref{lieg1.51}), (\ref{not.06}) and (\ref{not.05}) of $\mathbb I_{\ep,b}(u,v,\tau)$, $\mathbbm{\Gamma}_{x,y}(z,\tau)$, $\mathbb I_{\ep^{\rm TS}}^{\rm eqv}(\tau)$, $\Sigma_w(u)$ and $\mathbb M^{\rm sv}_{\Sigma(u)}$ as well as the transposition operation $(\ldots)^T$ in section \ref{sec:2.4.4}, 
the series $\mathbb I^{\rm eqv}_{\ep,b}(u,v,\tau)$
defined~by
\begin{align}
\mathbb I^{\rm eqv}_{\ep,b}(u,v,\tau)  &= 
(\mathbb M^{\rm sv}_{z})^{-1}\, \overline{ \mathbb I_{\ep,b}(u,v,\tau)^T} \, \mathbb M^{\rm sv}_{\Sigma(u)}\,
\mathbb I_{\ep,b}(u,v,\tau)  \label{grteq.01}
\end{align}
can alternatively be rewritten in terms of eMPLs via 
\begin{align}
\mathbb I^{\rm eqv}_{\ep,b}(u,v,\tau)  &=   (\mathbb M^{\rm sv}_z)^{-1}\,
\overline{\mathbbm{\Gamma}_{x,y}(z,\tau)^T} \,
  \mathbb M^{\rm sv}_z\,
\mathbb I^{\rm eqv}_{\ep^{\rm TS}}(\tau)  \, 
\mathbbm{\Gamma}_{x,y}(z,\tau)
 \label{altex.01} 
\end{align}
\end{theorem}

\noindent
The dependence on $z = u\tau{+}v$ is displayed via $\mathbbm{\Gamma}_{x,y}(z,\tau)$ to emphasize the link with the eMPLs in its series expansion and via $\mathbb I^{\rm eqv}_{\ep,b}(u,v,\tau)$ to emphasize the fact that the $\tau_1$ integral in (\ref{notsec.17}) is performed at fixed $u,v \in \mathbb R$. The same notation applies to the equivariant and single-valued series with the same parental letter.

The proof of Theorem \ref{3.thm:1} is one of the main tasks of this work and
will be carried out in section \ref{sec:4}. A major advantage of
having two representations of $\mathbb I^{\rm eqv}_{\ep,b}(u,v,\tau) $ in different fibration bases is that (\ref{grteq.01}) and (\ref{altex.01}) manifest different subsets of their
properties. While the periodicity under $v \rightarrow v{+}1$ is easier to see from (\ref{grteq.01}), the modular behavior of $\mathbb I^{\rm eqv}_{\ep,b}(u,v,\tau) $
can be studied on the basis of the following
lemma which establishes that $\mathbbm{\Gamma}_{x,y}(z,\tau)$ is equivariant
under ${\rm SL}_2(\mathbb Z)$ transformations of $z$ and $\tau$ up to simple
$b_2$ dependent phase factors:
\begin{lemma}
\label{3.lem:1}
The modular transformations of the eMPL series $\mathbbm{\Gamma}_{x,y}(z,\tau)$ in (\ref{notsec.19}) are
determined by 
\begin{align}
\mathbbm{\Gamma}_{x,y}(z,\tau{+}1) &= e^{i\pi b_2/6}\,
U_T^{-1} \, \mathbbm{\Gamma}_{x,y}(z,\tau)  \, U_T
\notag \\
\mathbbm{\Gamma}_{x,y}\bigg(\frac{z}{\tau},{-}\frac{1}{\tau} \bigg) &= e^{-i\pi b_2/2}  \, U_S^{-1} \mathbbm{\Gamma}_{x,y}(z,\tau) \, U_S
\label{modgser}
\end{align}
with $U_T$ and $U_S$ in the Lie group of $\mathfrak{sl}_2$ given by (\ref{lieg1.39}).
\end{lemma}
The lemma will be proven in appendix \ref{app:gmod}.

\begin{theorem}
\label{3.thm:2}
The series $\mathbb I^{\rm eqv}_{\ep,b}(u,v,\tau)$ in Theorem \ref{3.thm:1} 
\begin{itemize}
\item[(i)] has no monodromy when $z$ is moved around the origin,
\item[(ii)] has no monodromy $v\rightarrow v{+}1$ when $z$ is transported around the $A$-cycle of the torus,
\beq
 \mathbb I^{\rm eqv}_{\ep,b}(u,v{+}1,\tau) = \mathbb I^{\rm eqv}_{\ep,b}(u,v,\tau)
\label{clm.3.3a}
\eeq
\item[(iii)] is equivariant in the sense of Definition \ref{def:eqv}, with $\gamma\cdot(u,v,\tau)$ as in (\ref{uschoice.03})
\beq
\mathbb I^{\rm eqv}_{\ep,b}\big( \gamma\cdot(u,v,\tau)\big) = 
U_\gamma^{-1} \, \mathbb I^{\rm eqv}_{\ep,b}(u,v,\tau)\, U_\gamma
\label{clm.3.3b}
\eeq
\item[(iv)] has no monodromy $u\rightarrow u{+}1$ when $z$ is transported around the $B$-cycle of the torus,
\beq
\mathbb I^{\rm eqv}_{\ep,b}(u{+}1,v,\tau) = \mathbb I^{\rm eqv}_{\ep,b}(u,v,\tau)
\label{clm.3.3r}
\eeq
and thus single-valued in $z$ on the torus in view of (i) and (ii),
\item[(v)] is non-uniquely expressible as a series in words in $\ep_k^{(j)}$ and $b_k^{(j)}$.
\end{itemize}
\end{theorem}

\begin{proof}
\phantom{x}

$(i)$: The short-distance behavior ${\mathbb J}^{\rm BL}_{x-2\pi i \tau y,y}(z_1,\tau) = \frac{\dd z_1}{z_1} [x,y] + {\cal O}((z_1)^0)$ of the integrand in (\ref{notsec.19}) leads to monodromies
$\mathbbm{\Gamma}_{x,y}(z,\tau) \rightarrow e^{2\pi i [y,x]}\mathbbm{\Gamma}_{x,y}(z,\tau)$ as $z$ is moved around the origin. In the second representation of (\ref{altex.01}), these monodromies cancel from
\begin{align}
\mathbb I^{\rm eqv}_{\ep,b}(u,v,\tau)  \rightarrow \;   &(\mathbb M^{\rm sv}_z)^{-1}\,
\overline{\mathbbm{\Gamma}_{x,y}(z,\tau)^T} \, e^{-2\pi i [y,x]} \,
  \mathbb M^{\rm sv}_z\,
\mathbb I^{\rm eqv}_{\ep^{\rm TS}}(\tau)  \, e^{2\pi i [y,x]}  \,
\mathbbm{\Gamma}_{x,y}(z,\tau) \notag \\
=\; &(\mathbb M^{\rm sv}_z)^{-1}\,
\overline{\mathbbm{\Gamma}_{x,y}(z,\tau)^T} \,
  \mathbb M^{\rm sv}_z\,
\mathbb I^{\rm eqv}_{\ep^{\rm TS}}(\tau)  \, 
\mathbbm{\Gamma}_{x,y}(z,\tau) \notag \\
=\; &\mathbb I^{\rm eqv}_{\ep,b}(u,v,\tau) 
 \label{altex.61} 
\end{align}
using the fact that $[y,x]$ commutes with all the Tsunogai derivations in $\mathbb I^{\rm eqv}_{\ep^{\rm TS}}(\tau) $ and the arithmetic zeta generators in $\mathbb M^{\rm sv}_z$.

$(ii)$: Vanishing $A$-monodromies follow from the first representation (\ref{grteq.01}) of
$\mathbb I^{\rm eqv}_{\ep,b}(u,v,\tau)$ where all the dependence $u,v$ enters
through path-ordered exponentials in $\tau$ of the connection $\mathbb D_{\ep,b}(u,v,\tau)$ in (\ref{defdepb})
and its complex conjugate. More specifically, $\mathbb D_{\ep,b}(u,v,\tau_1)$ depends on $u,v$ through the
doubly-periodic $f^{(k)}(u\tau_1{+}v,\tau)$ and the $\tau_1$ integration in the expression
(\ref{notsec.17}) for $\mathbb I_{\ep,b}(u,v,\tau)$ is performed at fixed $u,v$. The tangential-base-point regularization in (\ref{notsec.17}) treats the divergent part of
\beq
(2\pi i)^{1-k} \, \dd \tau \, f^{(k)}(u\tau{+}v,\tau) = \frac{\dd q}{q}\, \frac{B_k(u)}{k!}+ {\cal O}(q^0)
\label{tanbsp}
\eeq
at the cusp $q\rightarrow 0$ on a different footing than the regular terms in $q$. Since the divergent terms in (\ref{tanbsp}) are by themselves $v\rightarrow v{+}1$ periodic, the tangential-base-point regularization in $\mathbb I_{\ep,b}(u,v,\tau)$ does not stop the $v\rightarrow v{+}1$ periodicity of the integrand from propagating to the path-ordered exponential, i.e.
\beq
\mathbb I_{\ep,b}(u,v{+}1,\tau)= \mathbb I_{\ep,b}(u,v,\tau)
\label{tanbsp.02}
\eeq
Note that this argument could not be used to extend the $u\rightarrow u{+}1$ periodicity of the integrands $f^{(k)}(u\tau_1{+}v,\tau)$ to the tangential-base-point regulated integral: The Bernoulli polynomials in the singular terms on the right side of (\ref{tanbsp}) fail to preserve the periodicity in $u$ by $B_k(u{+}1) =B_k(u) +k u^{k-1}$, so the regularization prevents $\mathbb I_{\ep,b}(u{+}1,v,\tau)$ from matching $\mathbb I_{\ep,b}(u,v,\tau)$ as one can easily check from examples at low degree.

$(iii)$: Equivariance will be shown through the second representation of $\mathbb I^{\rm eqv}_{\ep,b}(u,v,\tau)$ 
in (\ref{altex.01}). The modular properties (\ref{modgser}) of the eMPL series can be presented
in unified form as
\beq
\mathbbm{\Gamma}_{x,y} \big(\gamma\cdot (z,\tau)\big) = e^{i\pi r_\gamma b_2}\,
U_\gamma^{-1} \, \mathbbm{\Gamma}_{x,y}(z,\tau)  \, U_\gamma
\label{gaonGA}
\eeq
with $\gamma\cdot (z,\tau)$ in (\ref{uschoice.03}), and
where $r_\gamma \in \mathbb Q$ is determined by decomposing $\gamma \in {\rm SL}_2(\mathbb Z)$
into its generators $S$ and $T$. We are using the fact that $b_2$
commutes with the $\mathfrak{sl}_2$ generators of $U_T$ and $U_S$.
Inserting (\ref{gaonGA}) together with equivariance (\ref{lieg1.42}) of $\mathbb I^{\rm eqv}_{\ep^{\rm TS}}(\tau) $ in Theorem \ref{2.thm:1}
into (\ref{altex.01}) results in
\begin{align}
\mathbb I^{\rm eqv}_{\ep,b}\big(\gamma\cdot(u,v,\tau)\big)  &=   (\mathbb M^{\rm sv}_z)^{-1}\,
U_\gamma^{-1} \, \overline{\mathbbm{\Gamma}_{x,y}(z,\tau)^T} \,
U_\gamma\,  e^{-i\pi r_\gamma b_2} \,  \mathbb M^{\rm sv}_z\notag \\
&\quad \times
U_\gamma^{-1}\,\mathbb I^{\rm eqv}_{\ep^{\rm TS}}(\tau) \, U_\gamma  \, 
e^{i\pi r_\gamma b_2} \, U_\gamma^{-1}\, \mathbbm{\Gamma}_{x,y}(z,\tau) \, U_\gamma
\notag \\
&=  (\mathbb M^{\rm sv}_z)^{-1}\, U_\gamma^{-1} \,
 \overline{\mathbbm{\Gamma}_{x,y}(z,\tau)^T} \,
U_\gamma \,  \mathbb M^{\rm sv}_z \, 
U_\gamma^{-1}\,\mathbb I^{\rm eqv}_{\ep^{\rm TS}}(\tau) \,  \mathbbm{\Gamma}_{x,y}(z,\tau) \, U_\gamma
\notag \\
&=  U_\gamma^{-1} \, (\mathbb M^{\rm sv}_z)^{-1}\, 
 \overline{\mathbbm{\Gamma}_{x,y}(z,\tau)^T} \,
 \mathbb M^{\rm sv}_z \, 
\mathbb I^{\rm eqv}_{\ep^{\rm TS}}(\tau) \, \mathbbm{\Gamma}_{x,y}(z,\tau) \, U_\gamma
\label{clm.3.3c}
\end{align}
In the first step, we have used that $b_2 = [y,x]$ commutes with all of
$\mathbb I^{\rm eqv}_{\ep^{\rm TS}}(\tau) ,U_\gamma$ and 
$ \mathbb M^{\rm sv}_z$, see for instance (\ref{killxy}) and (\ref{lieg1.22}).
The second step is based on $U_\gamma  \mathbb M^{\rm sv}_z =   \mathbb M^{\rm sv}_z U_\gamma$
by $\mathfrak{sl}_2$ invariance of $z_w$. The last line of (\ref{clm.3.3c}) is equivalent to the 
right side of the claim (\ref{clm.3.3b}), finishing the proof of equivariance.

$(iv)$: Periodicity of $\mathbb I^{\rm eqv}_{\ep,b}(u,v,\tau)$ under $u \rightarrow u{+}1$ can be reduced to its $v \rightarrow v{+}1$ periodicity established in $(ii)$ by using equivariance demonstrated in $(iii)$ under modular $S$-transformations: applying $S\cdot (u,v,\tau) = ({-}v,u,-\tfrac{1}{\tau})$ to (\ref{clm.3.3b}) and conjugating by $U_S$ yields
\beq
\mathbb I^{\rm eqv}_{\ep,b}(u,v,\tau) = U_S \mathbb I^{\rm eqv}_{\ep,b}\big({-}v,u,-\tfrac{1}{\tau} \big)
U_S^{-1}
\label{prfuper}
\eeq
Hence, shifting $u \rightarrow u{+}1$ on the left side results in 
\begin{align}
\mathbb I^{\rm eqv}_{\ep,b}(u{+}1,v,\tau) &= U_S \mathbb I^{\rm eqv}_{\ep,b}\big({-}v,u{+}1,-\tfrac{1}{\tau} \big)
U_S^{-1}
= U_S \mathbb I^{\rm eqv}_{\ep,b}\big({-}v,u,-\tfrac{1}{\tau} \big)
U_S^{-1}
= \mathbb I^{\rm eqv}_{\ep,b}(u,v,\tau)
\label{prfuper.02}
\end{align}
where we used that the right side of (\ref{prfuper}) is periodic in its second argument by $(ii)$.
 
$(v)$: We shall use the representation (\ref{grteq.01}) of $\mathbb I^{\rm eqv}_{\ep,b}(u,v,\tau)$ to
show that all the Lie-algebra generators in its expansion are expressible in terms of $\ep_k^{(j)}$ and $b_k^{(j)}$.
The claim is obvious for the rightmost factor $\mathbb I_{\ep,b}(u,v,\tau) $, so it remains to
show that the same holds for the remaining factors
$(\mathbb M^{\rm sv}_{z})^{-1} \overline{ \mathbb I_{\ep,b}(u,v,\tau)^T} \mathbb M^{\rm sv}_{\Sigma(u)}$
with $\overline{ \mathbb I_{\ep,b}(u,v,\tau)^T}$ another series in $\ep_k^{(j)}$ and $b_k^{(j)}$.

Following the steps in the discussion around (\ref{remzs}), one uses the fact that 
$\mathbb M^{\rm sv}_{\Sigma(u)}= \mathbb M^{\rm sv}_{z} \mathbb B^{\rm sv}_{\ep,b}$
for some series $ \mathbb B^{\rm sv}_{\ep,b}$ in $\ep_k^{(j)}$ and $b_k^{(j)}$.
This is always possible since, by Lemma \ref{zwprop}, $\Sigma_w(u) = z_w + \ldots$ with a Lie series in 
$\ep_k^{(j)},b_k^{(j)}$ in the ellipsis. One can commute all the $z_w$ to
the left of $\ep_k^{(j)},b_k^{(j)}$ by using their bracket relations that only produce
Lie polynomials in $\ep_k^{(j)},b_k^{(j)}$.

Finally, given the above rewriting $\mathbb M^{\rm sv}_{\Sigma(u)}= \mathbb M^{\rm sv}_{z} \mathbb B^{\rm sv}_{\ep,b}$,
one expands the leftmost series $(\mathbb M^{\rm sv}_{z})^{-1} \overline{ \mathbb I_{\ep,b}(u,v,\tau)^T}  \mathbb M^{\rm sv}_{z}$
of $\mathbb I^{\rm eqv}_{\ep,b}(u,v,\tau)$ in terms of nested brackets with $z_w$ as done in (\ref{remzs})
and again uses the fact that the constituents $\ep_k^{(j)},b_k^{(j)}$ of $ \overline{ \mathbb I_{\ep,b}(u,v,\tau)^T}  $
are closed under $\ad_{z_w}$.

Non-uniqueness of the expansion of $\mathbb I^{\rm eqv}_{\ep,b}(u,v,\tau)$ in terms of 
$\ep_k^{(j)},b_k^{(j)}$ follows from the Pollack relations among $\ep_k^{(j) {\rm TS}}=\ep_k^{(j)}+ b_k^{(j)}$ at degrees $\geq 14$ combined with the bracket relations expressing 
$[ \ep_{k_1}^{(j_1)}, b_{k_2}^{(j_2)}]$ as Lie polynomials in $b_k^{(j)}$, see item $(ii)$ of Lemma \ref{braklem}.
\end{proof}
As will be detailed in section \ref{sec:itemgf.3}, the non-uniqueness of the coefficients of words in $\ep_k^{(j)},b_k^{(j)}$ can be relegated to degree $\geq14$ and modular depth two where Pollack relations kick in: When systematically moving all the $\ep_k^{(j)}$ to the left of $b_k^{(j)}$ (producing their ${\rm Lie}[b_{k'}^{(j')}]$-valued brackets), one arrives at unique coefficients of the resulting ordered words up to degree 13 where the $\ep_k^{(j)}$ do not obey any relations.

Theorem \ref{3.thm:2} sets the stage for accomplishing one of the main goals of this paper --
the construction of single-valued eMPLs in one variable -- in the following theorem:

\begin{theorem}
\label{3.cor:1} 
%\phantom{x}
The following modifications of the series $\mathbb I^{\rm eqv}_{\ep,b}(u,v,\tau)$ have unique expansion coefficients, preserve its modular equivariance in holomorphic or modular frame and produce two alternative formulations of single-valued elliptic polylogarithms in one variable:

\begin{itemize}
\item[(i)]
The series $\mathbbm{\Gamma}^{\rm sv}_{x,y}(z,\tau)$ which is defined by
\beq
\mathbbm{\Gamma}^{\rm sv}_{x,y}(z,\tau) = \mathbb I^{\rm eqv}_{\ep^{\rm TS}}(\tau)^{-1} \, \mathbb I^{\rm eqv}_{\ep,b}(u,v,\tau)
\label{cor.3.4a}
\eeq
and by (\ref{altex.01}) takes the alternative form
\begin{align}
\mathbbm{\Gamma}^{\rm sv}_{x,y}(z,\tau) &=   \mathbb I^{\rm eqv}_{\ep^{\rm TS}}(\tau)^{-1}
\, (\mathbb M^{\rm sv}_z)^{-1} \,
\overline{\mathbbm{\Gamma}_{x,y}(z,\tau)^T} \,
  \mathbb M^{\rm sv}_z \,
\mathbb I^{\rm eqv}_{\ep^{\rm TS}}(\tau)  \,
\mathbbm{\Gamma}_{x,y}(z,\tau) \label{grteq.02}
\end{align}
is single-valued in $z$, equivariant under ${\rm SL}_2(\mathbb Z)$ transformations
of $z,\tau$ and uniquely expressible as
a series in $b_k^{(j)}$ (with no separate appearance of the generators~$\ep_k^{(j)}$ or~$x,y$),
\begin{align}
\mathbbm{\Gamma}^{\rm sv}_{x,y}(z,\tau) &= 1 +\sum_{k_1=2}^\infty (k_1{-}1) \sum_{j_1=0}^{k_1-2} \dfrac{(-1)^{j_1}}{j_1!}\, \gab{j_1}{k_1}{z,\tau}b_{k_1}^{(j_1)} 
\label{cor.3.4b}\\
    &\quad+\sum_{k_1,k_2=2}^\infty (k_1{-}1)(k_2{-}1)  \sum_{j_1=0}^{k_1-2}  \sum_{j_2=0}^{k_2-2}\dfrac{ (-1)^{j_1+j_2}}{j_1!j_2!} \,  \gab{j_1&j_2}{k_1&k_2}{z,\tau}b_{k_1}^{(j_1)}b_{k_2}^{(j_2)} +\dots
    \notag
\end{align}
\item[(ii)] The series defined by
\beq
\mathbbm{\Lambda}^{\rm sv}_{x,y}(z,\tau) = 
U_{\rm mod}(\tau) \,\mathbbm{\Gamma}^{\rm sv}_{x,y}(z,\tau) \,U_{\rm mod}(\tau)^{-1}
\label{cor.3.4c}
\eeq
with the transformation $U_{\rm mod}(\tau)$ given by (\ref{usl2}) is uniquely expressible as a series
\begin{align}
\mathbbm{\Lambda}^{\rm sv}_{x,y}(z,\tau) &= 1 +\sum_{k_1=2}^\infty (k_1{-}1) \sum_{j_1=0}^{k_1-2} \dfrac{(-1)^{j_1}}{j_1!}\, \lab{j_1}{k_1}{z,\tau}b_{k_1}^{(j_1)} \label{cor.3.4d} \\
    &\quad+\sum_{k_1,k_2=2}^\infty (k_1{-}1)(k_2{-}1)  \sum_{j_1=0}^{k_1-2}  \sum_{j_2=0}^{k_2-2}\dfrac{ (-1)^{j_1+j_2}}{j_1!j_2!} \,  \lab{j_1&j_2}{k_1&k_2}{z,\tau}b_{k_1}^{(j_1)}b_{k_2}^{(j_2)} +\dots
    \notag
\end{align}
in non-holomorphic modular forms $\Lambda[\ldots;z,\tau]$ of purely antiholomorphic modular weights
\beq
\lab{j_1 &\ldots &j_r}{k_1 &\ldots &k_r}{ \frac{z}{c\tau{+}d}, \frac{a\tau{+}b}{ c\tau{+}d } }
= \bigg( \prod_{i=1}^r (c \bar \tau{+}d)^{k_i - 2j_i - 2} \bigg)\, \lab{j_1 &\ldots &j_r}{k_1 &\ldots &k_r}{z, \tau }
\label{cor.3.4e}
\eeq
\end{itemize}
\end{theorem}
As will become clear in the proof of item $(ii)$, the conjugation by 
$\Umod(\tau)$ is a universal way of converting equivariant series into generating series of modular forms. Following the terminology from the case of iterated Eisenstein integrals in section \ref{sec:2.5.3}, the equivariant series $\mathbbm{\Gamma}^{\rm sv}_{x,y}(z,\tau)$ is said to be in holomorphic frame (by the constituent $\mathbb I_{\ep,b}(u,v,\tau)$ holomorphic in $\tau$) while the transformed series $\mathbbm{\Lambda}^{\rm sv}_{x,y}(z,\tau)$ in (\ref{cor.3.4c}) is said to be in modular frame by (\ref{cor.3.4e}). The dictionary (\ref{gavsla}) below between the series coefficients in holomorphic and modular frame is universal to the case of equivariant iterated Eisenstein integrals \cite{Brown:2017qwo, Dorigoni_2022, Dorigoni:2024oft} and the present $z$-dependent setting. Its applications to iterated-integral representations of eMGFs will be further discussed in section \ref{sec:5}.

\begin{proof}
\phantom{x}

$(i)$: Single-valuedness and equivariance of $\mathbbm{\Gamma}^{\rm sv}_{x,y}(z,\tau)$ directly follows from the
analogous properties of $\mathbb I^{\rm eqv}_{\ep^{\rm TS}}(\tau)$ and $ \mathbb I^{\rm eqv}_{\ep,b}(u,v,\tau)$ 
in the ratio (\ref{cor.3.4a}). The expansion of $\mathbbm{\Gamma}^{\rm sv}_{x,y}(z,\tau)$ solely in terms of $b_k^{(j)}$
is based on the expansion of $\mathbb I^{\rm eqv}_{\ep,b}(u,v,\tau)$ in terms of $\ep_k^{(j)},b_k^{(j)}$. However, it additionally remains to show that, after elimination of $\ep_k^{(j)}= \ep_k^{(j) {\rm TS}} - b_k^{(j)}$,
all the appearances of $\ep_k^{(j) {\rm TS}}$ in $\mathbb I^{\rm eqv}_{\ep,b}(u,v,\tau)$ can be removed by
the inverse of the series $\mathbb I^{\rm eqv}_{\ep^{\rm TS}}(\tau)$ in (\ref{cor.3.4a}). 

For this purpose, we start by expanding (\ref{cor.3.4a}) using the expression (\ref{grteq.01}) for $ \mathbb I^{\rm eqv}_{\ep,b}(u,v,\tau)$
\ie
\mathbbm{\Gamma}^{\rm sv}_{x,y}(z,\tau)
= \big(\mathbb I_{\ep^{\rm TS}}(\tau) \big)^{-1} \, (\mathbb{M}_\sigma^\text{sv})^{-1} \, (\overline{ \mathbb I_{\epsilon^{\rm TS}}(\tau)^T})^{-1} \, \overline{ \mathbb I_{\ep,b}(u,v,\tau)^T} \, \mathbb{M}^\text{sv}_{\Sigma(u)}\, \mathbb I_{\ep,b}(u,v,\tau)
\label{bbprf.01}
\fe
and gradually identify series
$\mathbb B_r$ with $r=1,2,\ldots,6$ whose expansion is entirely expressible in terms of $b_k^{(j)}$, i.e.\ which take values in the Lie group of $b_k^{(j)}$. 
\begin{itemize}
    \item In a first step, we define a first series  $\mathbb B_1$ via
$\overline{ \mathbb I_{\ep,b}(u,v,\tau)^T}=\mathbb B_1 \overline{ \mathbb I_{\epsilon^{\rm TS}}(\tau)^T}$, i.e.\ by writing $\ep_k^{(j)}= \ep_k^{(j) {\rm TS}} - b_k^{(j)}$ and moving all Tsunogai derivations to the right at the cost of generating Lie polynomials in $b_{k'}^{(j')}$ from their brackets. This leads us to 
\ie
\mathbbm{\Gamma}^{\rm sv}_{x,y}(z,\tau) &= 
 \big(\mathbb I_{\ep^{\rm TS}}(\tau) \big)^{-1} \, (\mathbb{M}_\sigma^\text{sv})^{-1} \, (\overline{ \mathbb I_{\epsilon^{\rm TS}}(\tau)^T})^{-1} \, \mathbb B_1 \, \overline{ \mathbb I_{\epsilon^{\rm TS}}(\tau)^T} \, \mathbb{M}^\text{sv}_{\Sigma(u)} \, \mathbb I_{\ep,b}(u,v,\tau)\\
&= \big(\mathbb I_{\ep^{\rm TS}}(\tau) \big)^{-1} \, (\mathbb{M}_\sigma^\text{sv})^{-1} \, \mathbb B_2 \, \mathbb{M}^\text{sv}_{\Sigma(u)}\, \mathbb I_{\ep,b}(u,v,\tau)
\label{bbprf.02}
\fe
In passing to the second line, we have identified $\mathbb B_2 =  (\overline{ \mathbb I_{\epsilon^{\rm TS}}(\tau)^T})^{-1} \mathbb B_1  \overline{ \mathbb I_{\epsilon^{\rm TS}}(\tau)^T}$ which is expressible via $b_k^{(j)}$ since the brackets of $\ep_k^{(j){\rm TS}}$ with the expansion coefficients of $\mathbb B_1$ are expressible via $b_k^{(j)}$.

\item In a second step, we write $\mathbb{M}^\text{sv}_{\Sigma(u)} = \mathbb{B}_{3} \mathbb{M}^\text{sv}_{\sigma}$
where the discussion in Lemma \ref{zwprop} implies that $\mathbb{B}_{3}$ is again a series in $b_k^{(j)}$. Upon insertion into (\ref{bbprf.02}), we find
\ie
\mathbbm{\Gamma}^{\rm sv}_{x,y}(z,\tau)  &=  \big(\mathbb I_{\ep^{\rm TS}}(\tau) \big)^{-1} \,(\mathbb{M}_\sigma^\text{sv})^{-1} \, \mathbb B_2 \,\mathbb B_3 \, \mathbb{M}^\text{sv}_{\sigma} \, \mathbb I_{\ep,b}(u,v,\tau)\\
&= \big(\mathbb I_{\ep^{\rm TS}}(\tau) \big)^{-1} \, \mathbb B_4 \, \mathbb I_{\ep,b}(u,v,\tau)
\label{bbprf.05}
\fe
and identified a fourth series $ \mathbb B_4 = (\mathbb{M}_\sigma^\text{sv})^{-1} \mathbb B_2\mathbb B_3 \mathbb{M}^\text{sv}_{\sigma}$ in $b_k^{(j)}$ since commutators with $\sigma_w$ normalize the algebra of $b_{k'}^{(j')}$.
\item The third step is to rewrite $ \mathbb I_{\ep,b}(u,v,\tau) = \mathbb B_5 \mathbb I_{\ep^{\rm TS}}(\tau)$ as in the first step which casts (\ref{bbprf.05}) into the form
\ie
\mathbbm{\Gamma}^{\rm sv}_{x,y}(z,\tau) = \big(\mathbb I_{\ep^{\rm TS}}(\tau) \big)^{-1} \, \mathbb B_4\, \mathbb B_5 \, \mathbb I_{\ep^{\rm TS}}(\tau) = \mathbb B_6
\label{bbprf.06}
\fe
We have used once more that commutators with $\ep^{(j){\rm TS}}_k$ normalize the algebra of $b_{k'}^{(j')}$ and concluded that $\mathbbm{\Gamma}^{\rm sv}_{x,y}(z,\tau)$ is expressible in its Lie group.
\end{itemize}
Note that none of the steps in proving item $(i)$ required any information on the expansion coefficients of the series $\mathbb B_r$ in $b_k^{(j)}$ with $r=1,2,\ldots,6$.

$(ii)$: Conjugation by $\Umod(\tau)$ in (\ref{usl2}) universally converts
quantities $\mathbb F_{\rm holo}(\tau)$ with an equivariant ${\rm SL}_2(\mathbb Z)$
transformation $\mathbb F_{\rm holo}(\gamma\cdot \tau) = U_\gamma^{-1} \mathbb F_{\rm holo}(\tau) U_\gamma$
into generating series $\mathbb F_{\rm mod}(\tau)= \Umod(\tau) \mathbb F_{\rm holo}(\tau) \Umod(\tau)^{-1}$ of modular forms subject to $\mathbb F_{\rm mod}(\tau{+}1) = \mathbb F_{\rm mod}(\tau)$
and $\mathbb F_{\rm mod}(-\frac{1}{\tau}) = \bar \tau^{ -{\rm h} }\mathbb F_{\rm mod}(\tau) \bar \tau^{ {\rm h} }$
as in (\ref{trfheqv}). This can be shown solely from the  modular transformation of $\Umod(\tau)$ \cite{Brown:2017qwo2, Dorigoni:2024oft}, 
\beq
\Umod(\tau{+}1)
= \Umod(\tau) U_T \, , \ \ \ \
\Umod\big({-}\tfrac{1}{\tau})
= \bar \tau^{-{\rm h}} \Umod(\tau) U_S 
\eeq
which in turn follows from the adjoint action of both sides on $x,y$, see (\ref{uschoice.00}) for the contributions from $U_S, U_T$. Applying this logic to $\mathbb F_{\rm holo}(\tau) \rightarrow \mathbbm{\Gamma}^{\rm sv}_{x,y}(z,\tau)$ and $\mathbb F_{\rm mod}(\tau) \rightarrow \mathbbm{\Lambda}^{\rm sv}_{x,y}(z,\tau)$ results in the modular behavior
\beq
\mathbbm{\Lambda}^{\rm sv}_{x,y}(z,\tau{+}1) = \mathbbm{\Lambda}^{\rm sv}_{x,y}(z,\tau)
\, , \ \ \ \ \ \
\mathbbm{\Lambda}^{\rm sv}_{x,y}\bigg(\frac{z}{\tau},{-}\frac{1}{\tau}  \bigg) = \bar \tau^{ -{\rm h} }\mathbbm{\Lambda}^{\rm sv}_{x,y}(z,\tau) \bar \tau^{ {\rm h} }
\label{cor.3.4f}
\eeq
By the eigenvalues $2j{-}k{+}2$ of $b_k^{(j)}$ under the Cartan generator ${\rm ad}_{\rm h}$
in (\ref{lieg1.28}), the expansion coefficients $\Lambda^{\rm sv}[\ldots;z,\tau]$ defined by (\ref{cor.3.4d}) enjoy the ${\rm SL}_2(\mathbb Z)$ transformations (\ref{cor.3.4e}) as non-holomorphic modular forms.
\end{proof}

The following corollary of Theorem \ref{3.cor:1} elaborates on the properties and relations of the two alternative formulations of single-valued eMPLs $\Gamma^{\rm sv}[\ldots;z,\tau]$ and $\Lambda^{\rm sv}[\ldots;z,\tau]$ generated by (\ref{cor.3.4b}) and (\ref{cor.3.4d}), respectively.

\begin{corollary}
\label{corof44}
The expansion coefficients $\Gamma^{\rm sv}[\ldots;z,\tau]$ and $\Lambda^{\rm sv}[\ldots;z,\tau]$ in (\ref{cor.3.4b}) and (\ref{cor.3.4d})
%
%\begin{itemize}
\item[(i)] are single-valued combinations of eMPLs and their complex conjugates. 
In the conventions where eMPLs are normalized via $(2\pi i)^{ \sum_{i=1}^r (1-n_i) }
 \tilde \Gamma\big( \smallmatrix n_1  &\ldots &n_r \\ 0  &\ldots &0 \endsmallmatrix ;z,\tau\big)$ as in (\ref{genKempl}) or via $\Gamma_{\rm BL}(W;z,\tau)$ as in (\ref{jbexp}), the coefficients in their complex combinations $\Gamma^{\rm sv}[\ldots;z,\tau]$ and $\Lambda^{\rm sv}[\ldots;z,\tau]$ are $\mathbb Q$-linear combinations of products of the following quantities: 
\begin{itemize}
\item[(a)] polynomials in $2\pi i\tau$, $2\pi i\bar \tau$, $(\pi \Im \tau)^{-1}$; 
\item[(b)] equivariant iterated Eisenstein integrals ${\cal E}^{\rm eqv}[\ldots;\tau]$, $\beta^{\rm eqv}[\ldots;\tau]$ in (\ref{nwieqv}), (\ref{defheqv}); 
\item[(c)] the real part $\ee{0}{2}{\tau}+\overline{\ee{0}{2}{\tau}}$ of the primitives of ${\rm G}_2(\tau)$, $\overline{{\rm G}_2(\tau)}$ in (\ref{nmodg2});
\item[(d)] single-valued MZVs.
\end{itemize}

\item[(ii)] obey shuffle relations ($0\leq r\leq s$) 
\begin{align}
\Gamma^{\rm sv}\! \left[  \begin{smallmatrix}j_1&\ldots &j_r\\k_1&\ldots &k_r\end{smallmatrix}  ;z,\tau\right]
\Gamma^{\rm sv}\! \left[  \begin{smallmatrix}j_{r+1}&\ldots &j_s\\k_{r+1}&\ldots &k_s\end{smallmatrix}  ;z,\tau\right]
=
\Gamma^{\rm sv}\! \left[ \left(\begin{smallmatrix}j_1&\ldots &j_r\\k_1&\ldots &k_r\end{smallmatrix} \right) \shuffle \left(\begin{smallmatrix}j_{r+1}&\ldots &j_s\\k_{r+1}&\ldots &k_s\end{smallmatrix} \right) ;z,\tau\right] \label{shrls} \\
\Lambda^{\rm sv}\! \left[  \begin{smallmatrix}j_1&\ldots &j_r\\k_1&\ldots &k_r\end{smallmatrix}  ;z,\tau\right]
\Lambda^{\rm sv}\! \left[  \begin{smallmatrix}j_{r+1}&\ldots &j_s\\k_{r+1}&\ldots &k_s\end{smallmatrix}  ;z,\tau\right]
=
\Lambda^{\rm sv}\! \left[ \left(\begin{smallmatrix}j_1&\ldots &j_r\\k_1&\ldots &k_r\end{smallmatrix} \right) \shuffle \left(\begin{smallmatrix}j_{r+1}&\ldots &j_s\\k_{r+1}&\ldots &k_s\end{smallmatrix} \right) ;z,\tau\right] \notag
\end{align}
\item[(iii)] are related by
\begin{align}
 \! \! \! \! \! \! \!  \lab{j_1 & j_2 &\cdots & j_\ell}{k_1& k_2& \cdots & k_\ell}{z,\tau}=&\sum_{p_1=0}^{k_1-j_1-2}\sum_{p_2=0}^{k_2-j_2-2}\cdots\sum_{p_\ell=0}^{k_\ell-j_\ell-2}{k_1{-}j_1{-}2\choose p_1}{k_2{-}j_2{-}2\choose p_2}\cdots{k_\ell{-}j_\ell{-}2\choose p_\ell} \notag \\
&\times \bigg(\frac{1}{4\pi \Im \tau} \bigg)^{p_1+\cdots+p_\ell}\, \sum_{r_1=0}^{j_1+p_1}\sum_{r_2=0}^{j_2+p_2}\cdots\sum_{r_\ell=0}^{j_\ell+p_\ell}{j_1{+}p_1\choose r_1}{j_2{+}p_2\choose r_2}\cdots{j_\ell{+}p_\ell\choose r_\ell} \notag \\
&\times(-2\pi i\bar{\tau})^{r_1+\cdots+r_\ell}\, \gab{j_1+p_1-r_1 & j_2+p_2-r_2 &\cdots&j_\ell+p_\ell-r_\ell}{k_1&k_2&\cdots&k_\ell}{z,\tau}
\label{gavsla}
\end{align}
i.e.\ any $\Lambda^{\rm sv}[\ldots;z,\tau]$ is a finite $\mathbb Q[2\pi i \bar \tau,(4\pi \Im \tau)^{-1}]$-linear combination of $\Gamma^{\rm sv}[\ldots;z,\tau]$ and vice versa, with identical entries $k_i$ of their second lines.
%\end{itemize}
\end{corollary}

\begin{proof}
\phantom{x}

$(i)$: The appearance of eMPLs and their complex conjugates in the expansion coefficients $\Gamma^{\rm sv}[\ldots;z,\tau]$ and $\Lambda^{\rm sv}[\ldots;z,\tau]$ in (\ref{cor.3.4b}) and (\ref{cor.3.4d})  can be traced back to the series $\overline{\mathbbm{\Gamma}_{x,y}(z,\tau)^T}$ and $\mathbbm{\Gamma}_{x,y}(z,\tau)$ in the expression (\ref{grteq.02}) for $\mathbbm{\Gamma}^{\rm sv}_{x,y}(z,\tau)$. Tentative factors of $(2\pi i)^{\pm 1}$ are artifacts of the normalization conventions of the $\tilde \Gamma\big( \smallmatrix n_1  &\ldots &n_r \\ 0  &\ldots &0 \endsmallmatrix ;z,\tau\big)$ in (\ref{genKempl}) and do not appear when employing the Brown-Levin eMPLs $\Gamma_{\rm BL}(W;z,\tau)$ in the normalization of  (\ref{jbexp}).
The additional quantities that can multiply the eMPLs and their complex conjugates are $\mathbb Q$-linear combinations of products of
\begin{itemize}
\item[(a)] powers of $2\pi i\tau$, $2\pi i\bar \tau$ from the composite letters of $\mathbbm{\Gamma}_{x,y}(z,\tau)$, $\overline{\mathbbm{\Gamma}_{x,y}(z,\tau)^T}$, and the transformation $\Umod(\tau)$ to modular frame (\ref{usl2}) additionally introduces polynomials in $2\pi i\bar \tau$ and $(\pi\Im \tau)^{-1}$ at fixed degree of the $b_k^{(j)}$, $\epsilon_k^{(j)}$ in the generating series;
\item[(b)] equivariant iterated Eisenstein integrals $\eeqqvv{j_1&\ldots &j_r}{k_1&\ldots &k_r}{\tau}$ and $\beqvtau{j_1 &\ldots &j_r}{k_1&\ldots &k_r}{\tau}$ from the expansion of the series $(\mathbb I^{\rm eqv}_{\ep^{\rm TS}}(\tau))^{\pm 1}$ in holomorphic frame (\ref{nwieqv}) and modular frame (\ref{defheqv}), respectively;
\item[(c)] $\ee{0}{2}{\tau}+\overline{\ee{0}{2}{\tau}}$ from the first and last series of $\mathbbm{\Gamma}_{x,y}(z,\tau) \mathbb M^{\rm sv}_z\,
\mathbb I^{\rm eqv}_{\ep^{\rm TS}}(\tau)\overline{\mathbbm{\Gamma}_{x,y}(z,\tau)^T}$ in (\ref{altex.01}), using the fact that $b_2$ commutes with the series $\mathbb M^{\rm sv}_z\,
\mathbb I^{\rm eqv}_{\ep^{\rm TS}}(\tau)$ in between the eMPL series;
\item[(d)] single-valued MZVs from the series $(\mathbb M_z^{\rm sv})^{\pm 1}$ in the expression (\ref{altex.01}) for $\mathbb I^{\rm eqv}_{\ep,b}(u,v,\tau) $.
\end{itemize}

$(ii)$: Shuffle relations of the single-valued eMPLs $\Gamma^{\rm sv}[\ldots;z,\tau]$ and $\Lambda^{\rm sv}[\ldots;z,\tau]$ follow from the fact that all the series $\mathbb M^{\rm sv}_{z}$, $\mathbb M^{\rm sv}_{\Sigma(u)}$, $\mathbb I_{\ep,b}(u,v,\tau)$, $\mathbbm{\Gamma}_{x,y}(z,\tau)$, $\mathbb I^{\rm eqv}_{\ep^{\rm TS}}(\tau) $, etc.\ composing their generating series $\mathbbm{\Gamma}^{\rm sv}_{x,y}(z,\tau) $ and $\mathbbm{\Lambda}^{\rm sv}_{x,y}(z,\tau)$ are group-like, i.e.\ that their logarithms are Lie series in the respective generators \cite{Reutenauer, Brown:unpubl}.

$(iii)$: Follows solely from the action (\ref{lieg1.28}) of the $\mathfrak{sl}_2$ generators $\ep_0,\ep_0^\vee$ in the expansion (\ref{usl2}) of 
$\Umod(\tau)$ on the generators $b_k^{(j)}$ in the expansion (\ref{cor.3.4b}) of $\mathbbm{\Gamma}^{\rm sv}_{x,y}(z,\tau)$, see for instance section 3.1 of \cite{Dorigoni:2024oft}.
\end{proof}

\begin{remark}
\label{3.rmk:1}
Starting from the expression (\ref{grteq.02}) for $\mathbbm{\Gamma}^{\rm sv}_{x,y}(z,\tau)$, we can use the consequence $ \mathbb M^{\rm sv}_z 
\mathbb I^{\rm eqv}_{\ep^{\rm TS}}(\tau)  =  \mathbb M^{\rm sv}_\sigma 
\mathbb I^{\rm sv}_{\ep^{\rm TS}}(\tau) $ of (\ref{lieg1.61}) to trade arithmetic zeta
generators $z_w$ for the full ones $\sigma_w$ and at the same time {\it equivariant} iterated Eisenstein integrals for
{\it single-valued} ones:
\begin{align}
\mathbbm{\Gamma}^{\rm sv}_{x,y}(z,\tau) &=  \mathbb I^{\rm sv}_{\ep^{\rm TS}}(\tau)^{-1} \,
(\mathbb M^{\rm sv}_\sigma)^{-1} \,
\overline{\mathbbm{\Gamma}_{x,y}(z,\tau)^T} \,
  \mathbb M^{\rm sv}_\sigma \,
\mathbb I^{\rm sv}_{\ep^{\rm TS}}(\tau)  \,
\mathbbm{\Gamma}_{x,y}(z,\tau)
\label{agrteq.02}
\end{align}
This is in remarkable correspondence to the generating series (\ref{svmpl.15})
of single-valued MPLs at genus zero in two variables where
$\mathbb M^{\rm sv}_\sigma$ and $\mathbb I^{\rm sv}_{\ep^{\rm TS}}(\tau) $ can be viewed as the genus-one counterparts of $\mathbb M^{\rm sv}_0$ and $\mathbb G^{\rm sv}_{e_0,e_1}(y)$, respectively. The analogies between
genus-one and genus-zero constituents in (\ref{agrteq.02}) and (\ref{svmpl.15})
are summarized in table \ref{tab:1}. This extends and is consistent with 
the analogies between single-valued iterated Eisenstein integrals and single-valued 
MPLs in one variable in section \ref{sec:2.5.4} and \cite{Dorigoni:2024oft}.
\end{remark}

\begin{table}[h]
\centering
\begin{tabular}{c||c}
genus zero&genus one
\\\hline\hline
$ \displaystyle \Bigg. \mathbb M^{\rm sv}_0 \ {\rm and} \ M_w \ {\rm acting} \ {\rm on} \ e_i \ {\rm and} \ e_i'$ 
&$\ \ \ \mathbb M^{\rm sv}_\sigma \ {\rm and} \  \sigma_w \ {\rm acting} \ {\rm on} \ \ep_k^{(j){\rm TS}} \ {\rm and} \ x,y \ \ \ $
\\\hline
$\bigg.  \mathbb G^{\rm sv}_{e_0,e_1}(y)  \ {\rm and} \ e_i \ {\rm acting} \ {\rm on} \  e_i' \Big.$ 
&$\Big.  \mathbb I^{\rm sv}_{\ep^{\rm TS}}  \ {\rm and} \ \ep_k^{(j){\rm TS}} \ {\rm acting} \ {\rm on} \ x,y  \Big.$
\\\hline
$\bigg. \ \ \  \mathbb G_{e_0',e_1',e_y'}(y,z) \ {\rm and} \ % {\rm sv} \ {\rm completion} \ 
 \mathbb G^{\rm sv}_{e_0',e_1',e_y'}(y,z)
 \ \ \ $
  &$\mathbbm{\Gamma}_{x,y}(z,\tau)  \ {\rm and} \ \mathbbm{\Gamma}^{\rm sv}_{x,y}(z,\tau)
  $
\end{tabular}
\caption{\label{tab:1}\textit{Dictionary between the constructions of single-valued MPLs in two variables (genus zero) and single-valued eMPLs (genus one).}}
\end{table}

\begin{remark}
\label{rmk.F}
The series (\ref{agrteq.02}) in single-valued eMPLs can be equivalently presented as
\begin{align}
\mathbb F_{x,y,\ep,\sigma}^{\rm sv}(z,\tau) &= \overline{ \mathbb F_{x,y,\ep,\sigma}(z,\tau)^T } \, \mathbb F_{x,y,\ep,\sigma}(z,\tau)  \notag \\
%\mathbb F^{\rm sv}(z,\tau) = \overline{ \mathbb F(z,\tau)^T } \, \mathbb F(z,\tau) \, , \ \ \ \ \ \
\mathbb F_{x,y,\ep,\sigma}(z,\tau) &= \mathbb M_\sigma \,  \mathbb I_{\ep^{\rm TS}}(\tau)   \,
\mathbbm{\Gamma}_{x,y}(z,\tau)
\label{prevs.01}
\end{align}
by unpacking $\mathbb I^{\rm sv}_{\ep^{\rm TS}}(\tau)$ via (\ref{lieg1.61}) and introducing the series in MZVs
\beq
\mathbb M_{\sigma} = \sum_{r=0}^{\infty} \sum_{i_1,\ldots,i_r \atop {\in 2\mathbb N+1} } \rho^{-1}(f_{i_1} \ldots f_{i_r}) \, \sigma_{i_1}\ldots  \sigma_{i_r}
\label{nosv.05}
\eeq
without the single-valued projection of the $f_{i_1} \ldots f_{i_r}$. The factorized form of (\ref{prevs.01}) comes at the cost of taking values in the larger Lie group involving the three types of generators $\sigma_w,\epsilon_k^{(j)},b_k^{(j)}$. The extraction of coefficients then requires a choice of ordering prescription, say using the bracket relations among the generators to move all the $\sigma_w$ to the left of $\epsilon_k^{(j)},b_k^{(j)}$ and all the $\epsilon_k^{(j)}$ to the left of $b_k^{(j)}$. Moreover, cases with $\geq 2$ generators $\ep^{(j_1)}_{k_1} \ep^{(j_2)}_{k_2}\ldots$ at degrees $\geq 14$ again require us to mod out by Pollack relations before investigating coefficients.

Given that factorized formulae similar to (\ref{prevs.01}) were observed for genus-zero MPLs in any number of variables \cite{Frost:2023stm} and related to period matrices of Koba-Nielsen type integrals over points on a disk or sphere \cite{Brown:2019jng, Britto:2021prf}, it would be interesting if (\ref{prevs.01}) signals similar links to concrete configuration-space integrals at genus one. Moreover, it remains to explore connections of (\ref{prevs.01}) to the general formalism for single-valued period matrices in \cite{Brown:2018omk}.
\end{remark}

\begin{remark}
\label{3.rmk:2}
The non-holomorophic modular forms $\lab{j_1 &\ldots &j_r}{k_1 &\ldots &k_r}{z, \tau }$ 
obtained as the coefficients (\ref{cor.3.4d}) of the series 
$\mathbbm{\Lambda}^{\rm sv}_{x,y}(z,\tau)$ in (\ref{cor.3.4c})
are the right class of iterated integrals for eMGFs in one variable, see section \ref{sec:2.2.4}. As detailed in section \ref{sec:6.4}, the counting of shuffle-independent $\lab{j_1 &\ldots &j_r}{k_1 &\ldots &k_r}{z, \tau }$ governs that of indecomposable eMGFs in one variable that are independent under algebraic relations with $\mathbb Q$-combinations of MGFs and/or single-valued MZVs as coefficients. More precisely, the total degree $k_1{+}\ldots{+}k_r$ of the $\lab{j_1 &\ldots &j_r}{k_1 &\ldots &k_r}{z, \tau }$ will match the sum of the exponents $a_i, b_i$ of the eMGFs in (\ref{defdiC}) and their generalizations to trihedral or more involved topologies.
\end{remark}

%%%%%%%%%%%%%%%%%%%%%%%%%%%%%%%%%%%%%%%%%%%%%%%%%%
%%%%%%%%%%%%%%%%%%%%%%%%%%%%%%%%%%%%%%%%%%%%%%%%%%
\subsection{Behavior under complex conjugation}
\label{sec:3.cc}

The placement of the factors $\mathbb M^{\rm sv}$ in equivariant generating series (\ref{lieg1.51}) or (\ref{grteq.01}) creates an asymmetry between the building blocks with meromorphic and antimeromorphic dependence on $\tau$. Already at genus zero, this is a well-known feature of the single-valued map \cite{svpolylog, Schnetz:2013hqa, Brown:2013gia, Broedel:2016kls, DelDuca:2016lad, Brown:2018omk}, so it is not surprising to see the genus-one echo in the series (\ref{cor.3.4a}) in single-valued eMPLs. The asymmetry becomes even more pronounced in modular frame since the transformation $\Umod(\tau)$ in (\ref{usl2}) is engineered to produce modular forms of purely antiholomorphic modular weight by distributing suitable powers of $\pi \Im \tau$ among the series coefficients.

The goal of this section is to quantify the asymmetry between meromorphic and antimeromorphic quantities in equivariant series:  we will express the complex conjugates of $\mathbb I^{\rm eqv}_{\ep^{\rm TS}}(\tau)$, $\mathbb I^{\rm eqv}_{\ep,b}(u,v,\tau)$, $\mathbbm{\Gamma}^{\rm sv}_{x,y}(z,\tau)$ and their counterparts in modular frame in terms of the original series conjugated by additional series that depend on fewer variables $z,\tau$. In this way, the single-valued eMPLs produced by (\ref{cor.3.4b}) or (\ref{cor.3.4d}) are shown to close under complex conjugation, with single-valued MZVs, powers of $\pi \Im \tau$ and equivariant iterated Eisenstein integrals among their coefficients.

After spelling out the differential equations of the equivariant series in $z$, $\bar z$, we will see that, by the Brown-Levin connection in (\ref{notsec.06}), the holomorphic $z$-derivative of $\mathbbm{\Gamma}^{\rm sv}_{x,y}(z,\tau)$ does not take the same form as that of $\mathbbm{\Gamma}_{x,y}(z,\tau)$. Our eMPL series therefore do not share the property of MPL series at genus zero that the single-valued map $\mathbb G \rightarrow \mathbb G^{\rm sv}$ preserves the holomorphic differential equations in all variables.

%%%%%%%%%%%%%%%%%%%%%%%%%%%%%%%%%%%%%%%%%%%%%%%%%%
\subsubsection{Complex conjugation in holomorphic frame}
\label{s:3.c.a}

Throughout this section, it will be convenient to act on the complex conjugates of the generating series of interest with the transposition operation $(\ldots)^T$ in section \ref{sec:2.4.4}. In holomorphic frame, the resulting links of three equivariant series to their complex conjugates are gathered in the following proposition: 

\begin{prop}
\label{ccholo}
We find the following transposes of the complex-conjugate series of

\begin{itemize}
\item[(i)] equivariant iterated Eisenstein integrals (also see section 8.2 of \cite{Brown:2017qwo2})
\beq
\overline{ \mathbb I^{\rm eqv}_{\ep^{\rm TS}}( \tau)^T } = \mathbb M^{\rm sv}_z \, \mathbb I^{\rm eqv}_{\ep^{\rm TS}}(\tau)\, (\mathbb M^{\rm sv}_z)^{-1}
\label{ccsc.01}
\eeq
\item[(ii)] their generalizations to integration kernels ${\rm d} \tau \,f^{(k)}(u\tau{+}v,\tau)$ at fixed $u,v \in \mathbb R$
\beq
\overline{ \mathbb I^{\rm eqv}_{\ep,b}(u,v,\tau)^T } = \mathbb M^{\rm sv}_z \, \mathbb I^{\rm eqv}_{\ep,b}(u,v,\tau)\, (\mathbb M^{\rm sv}_z)^{-1}
\label{ccsc.02}
\eeq
\item[(iii)] single-valued eMPLs
\begin{align}
\overline{ \mathbbm{\Gamma}^{\rm sv}_{x,y}(z,\tau)^T } &= \mathbb M^{\rm sv}_z \, \mathbb I^{\rm eqv}_{\ep^{\rm TS}}(\tau) \, \mathbbm{\Gamma}^{\rm sv}_{x,y}(z,\tau)\, \big( \mathbb I^{\rm eqv}_{\ep^{\rm TS}}(\tau) \big)^{-1} \, (\mathbb M^{\rm sv}_z )^{-1}
\label{ccsc.03} \\
&= \mathbb M^{\rm sv}_\sigma \, \mathbb I^{\rm sv}_{\ep^{\rm TS}}(\tau) \, \mathbbm{\Gamma}^{\rm sv}_{x,y}(z,\tau)\, \big( \mathbb I^{\rm sv}_{\ep^{\rm TS}}(\tau) \big)^{-1} \, (\mathbb M^{\rm sv}_\sigma )^{-1}
\notag
\end{align}
\end{itemize}
\end{prop}

\begin{proof}
\phantom{x}

$(i)$: readily follows from the complex conjugate transpose of (\ref{lieg1.51})
\beq
\overline{ \mathbb I^{\rm eqv}_{\ep^{\rm TS}}( \tau)^T } = \overline{ \mathbb I_{\ep^{\rm TS}}(\tau)^T} \, \mathbb M^{\rm sv}_{\sigma}\,
\mathbb I_{\ep^{\rm TS}}(\tau) \, (\mathbb M^{\rm sv}_{z})^{-1}
\label{ccsc.04}
\eeq
by identifying the first three terms on the right side as $\mathbb M^{\rm sv}_{z} \mathbb I^{\rm eqv}_{\ep^{\rm TS}}( \tau)$. We have used that $(\mathbb M^{\rm sv}_{\sigma})^T = \mathbb M^{\rm sv}_{\sigma}$ since the single-valued map (\ref{appA.00}) in the $f$-alphabet organizes the zeta generators in the expansion of (\ref{not.05}) in palindromic combinations $\sigma_{i_1}\ldots \sigma_{i_r} + \sigma_{i_r}\ldots \sigma_{i_1}$.
This argument applies to any variant $\sigma_w \rightarrow z_w$ or $\sigma_w \rightarrow \Sigma_w(u)$ of genus-one zeta generators which are all unaffected by complex conjugation.

$(ii)$: similarly follows from the complex conjugate transpose of (\ref{grteq.01})
\beq
\overline{ \mathbb I^{\rm eqv}_{\ep,b}(u,v, \tau)^T } = \overline{ \mathbb I_{\ep,b}(u,v,\tau)^T} \, \mathbb M^{\rm sv}_{\Sigma(u)}\,
\mathbb I_{\ep,b}(u,v,\tau) \, (\mathbb M^{\rm sv}_{z})^{-1}
\label{ccsc.05}
\eeq
by identifying the first three terms on the right side as $\mathbb M^{\rm sv}_{z} \mathbb I^{\rm eqv}_{\ep,b}(u,v, \tau)$.

$(iii)$: As a result of  (\ref{cor.3.4a}),  (\ref{ccsc.01}) and (\ref{ccsc.02}), we have
\beq
\overline{ \mathbbm{\Gamma}^{\rm sv}_{x,y}(z,\tau)^T } = \mathbb M^{\rm sv}_z \, \mathbb I^{\rm eqv}_{\ep,b}(u,v, \tau) \, \big( \mathbb I^{\rm eqv}_{\ep^{\rm TS}}(\tau) \big)^{-1} \, (\mathbb M^{\rm sv}_z )^{-1}
\label{ccsc.06}
\eeq
The first line of the claim (\ref{ccsc.03}) then follows by identifying $\mathbb I^{\rm eqv}_{\ep,b}(u,v, \tau) = \mathbb I^{\rm eqv}_{\ep^{\rm TS}}(\tau) \mathbbm{\Gamma}^{\rm sv}_{x,y}(z,\tau)$, and the second line is due to $\mathbb M^{\rm sv}_z  \mathbb I^{\rm eqv}_{\ep^{\rm TS}}(\tau) = \mathbb M^{\rm sv}_\sigma  \mathbb I^{\rm sv}_{\ep^{\rm TS}}(\tau)$, see (\ref{lieg1.61}).
\end{proof}

\begin{corollary}
\label{closegam}
The single-valued eMPLs $\Gamma^{\rm sv}[\ldots;z,\tau]$ defined by the expansion (\ref{cor.3.4b}) of the series $\mathbbm{\Gamma}^{\rm sv}_{x,y}(z,\tau)$ close under complex conjugation. The coefficients in relating $\overline{\Gamma^{\rm sv}[\ldots;z,\tau]}$ to $\Gamma^{\rm sv}[\ldots;z,\tau]$ are determined by (\ref{ccsc.03}) as $\mathbb Q$-linear combinations of products of the equivariant iterated Eisenstein integrals in (\ref{nwieqv})  and single-valued MZVs.
\end{corollary}

%%%%%%%%%%%%%%%%%%%%%%%%%%%%%%%%%%%%%%%%%%%%%%%%%%
\subsubsection{Towards modular frame}
\label{s:3.c.b}

The following definition introduces an interpolation between the series $\mathbb H^{\rm eqv}_{\ep^{\rm TS}}(\tau)$ and $\mathbbm{\Lambda}^{\rm sv}_{x,y}(z,\tau)$ in modular frame given by (\ref{defheqv}) and (\ref{cor.3.4c}):

\begin{definition}
\label{newdefheqv}
The modular-frame analogue of $\mathbb I^{\rm eqv}_{\ep,b}(u,v,\tau)$ in (\ref{grteq.01}) is given by the series $\mathbb H^{\rm eqv}_{\ep,b}(u,v,\tau)$ defined by
\beq
\mathbb H^{\rm eqv}_{\ep,b}(u,v,\tau) = 
U_{\rm mod}(\tau) \, \mathbb I^{\rm eqv}_{\ep,b}(u,v,\tau) \, U_{\rm mod}(\tau)^{-1}
\label{thm.3.4a}
\eeq
\end{definition}
By the arguments
in the proof of item $(ii)$ of Theorem \ref{3.cor:1}, we have
\beq
\mathbb H^{\rm eqv}_{\ep,b}(u,v{-}u,\tau{+}1) =  \mathbb H^{\rm eqv}_{\ep,b}(u,v,\tau)
\, , \ \ \ \ \ \
\mathbb H^{\rm eqv}_{\ep,b}\bigg({-}v,u,{-}\frac{1}{\tau} \bigg) = \bar \tau^{ -{\rm h} }\,
\mathbb H^{\rm eqv}_{\ep,b}(u,v,\tau) \,
\bar \tau^{ {\rm h} }
\label{thm.3.4b}
\eeq
such that each term in the expansion in $\ep_k^{(j)}, b_k^{(j)}$ is a non-holomorphic
modular form after modding out by the bracket relations among the generators. We note in passing that the images of the generators $x,y$ in modular frame are given by
\begin{align}\label{Umodxy}
    \Umod(\tau) \begin{pmatrix}
        x\\y
    \end{pmatrix}
    \Umod^{{-}1}(\tau)=\begin{pmatrix}\frac{\tau}{\tau{-}\bar\tau} x+2\pi i\bar\tau y\\
    y+\frac{1}{2\pi i(\tau{-}\bar\tau)}x\end{pmatrix}
\end{align}
The following lemma sets the stage for investigating the complex-conjugation properties of equivariant series in modular frame and pinpointing
the effect of the transformation $\Umod(\tau)$:

\begin{lemma}
\label{lemusl2}
\phantom{x}

\begin{itemize}
\item[(i)] The generalized reflection operators ${\cal Q}_{\alpha}$ in the Lie group of $\mathfrak{sl}_2$ defined by
\beq
{\cal Q}_{\alpha} = \exp\big({-}\ep_0^\vee/\alpha \big) \, \exp(\alpha \ep_0) \, \exp\big({-}\ep_0^\vee/\alpha \big)
\label{lemq.01}
\eeq
with arbitrary commuting $\alpha \in \mathbb C \setminus \{ 0 \}$ act on the generators $x,y$ via
\beq
{\cal Q}_{\alpha}^{-1} \bigg( \begin{array}{c} x \\   y \end{array} \bigg) {\cal Q}_{\alpha} = \bigg( \begin{array}{c} {-}\alpha y \\   x/\alpha \end{array} \bigg) 
\label{lemq.02}
\eeq
which leads to the following action on $b_k^{(j)}$ and $\ep_k^{(j)}$
\begin{align}
{\cal Q}_{\alpha}^{-1} 
\bigg( \begin{array}{c} b_k^{(j)}  \\  \ep_k^{(j)}  \end{array} \bigg) 
{\cal Q}_{\alpha}  &= \frac{(-)^j j! (-\alpha)^{k-2j-2}}{(k{-}j{-}2)!} \, \bigg( \begin{array}{c} b_k^{(k-j-2)}  \\  \ep_k^{(k-j-2)}  \end{array} \bigg) 
\label{lemq.03}
\end{align}
\item[(ii)] The transformation $U_{\rm mod}(\tau)$ in (\ref{usl2}) to modular frame and its complex conjugate transpose combine to the $\alpha = 4\pi \Im \tau$ instance of the generalized reflection operator (\ref{lemq.01}),
\beq
U_{\rm mod}(\tau) \, \overline{ U_{\rm mod}(\tau)^T} = {\cal Q}_{4\pi \Im \tau} 
\label{lemq.04}
\eeq
\end{itemize}
\end{lemma}
Note that the $\alpha=-2\pi i$ instance of the generalized reflection operator (\ref{lemq.01}) realizes the $U_S$ in (\ref{lieg1.39}) for modular $S$-transformations,
\beq
U_S = {\cal Q}_{-2\pi i}
\label{lemq.05}
\eeq

\begin{proof}
\phantom{x}

$(i)$: Both components of (\ref{lemq.02}) are straightforwardly derived from (\ref{lemq.01}) by expanding the exponentials to subleading order. By the Lie-polynomial representation (\ref{lieg1.24}) of $b_k^{(j)}$, the upper component of (\ref{lemq.03}) follows from (\ref{lemq.02}), and the lower component can be derived by conjugating (\ref{lieg1.05}) with ${\cal Q}_\alpha$ and using $\ep_k^{(j)} = \ep_k^{(j){\rm TS}}- b_k^{(j)}$.

$(ii)$: follows from the definitions (\ref{usl2}) and (\ref{lemq.01}) of $\Umod(\tau)$ and ${\cal Q}_\alpha$, respectively.
\end{proof}

%%%%%%%%%%%%%%%%%%%%%%%%%%%%%%%%%%%%%%%%%%%%%%%%%%
\subsubsection{Complex conjugation in modular frame}
\label{s:3.c.c}

Lemma \ref{lemusl2} of the previous section yields convenient representations for the complex conjugates of equivariant series in modular frame, gathered in the following proposition:

\begin{prop}
\label{moreccholo}
We find the following transposes of the complex-conjugate series of

\begin{itemize}
\item[(i)] equivariant iterated Eisenstein integrals or modular graph forms 
\beq
\overline{ \mathbb H^{\rm eqv}_{\ep^{\rm TS}}(\tau)^T } = \mathbb M^{\rm sv}_z \, {\cal Q}_{4\pi \Im \tau}^{-1} \, \mathbb H^{\rm eqv}_{\ep^{\rm TS}}(\tau) \, {\cal Q}_{4\pi \Im \tau} \, (\mathbb M^{\rm sv}_z )^{-1}
\label{ccsc.21}
\eeq
\item[(ii)] their generalizations to integration kernels ${\rm d} \tau \,f^{(k)}(u\tau{+}v,\tau)$ or eMGFs
\beq
\overline{ \mathbb H^{\rm eqv}_{\ep,b}(u,v,\tau)^T } = \mathbb M^{\rm sv}_z \, {\cal Q}_{4\pi \Im \tau}^{-1} \, \mathbb H^{\rm eqv}_{\ep,b}(u,v,\tau) \, {\cal Q}_{4\pi \Im \tau} \, (\mathbb M^{\rm sv}_z )^{-1}
\label{ccsc.22}
\eeq
\item[(iii)] single-valued eMPLs adapted to eMGFs
\beq
\overline{ \mathbbm{\Lambda}^{\rm sv}_{x,y}(z,\tau)^T } = \mathbb M^{\rm sv}_z \, {\cal Q}_{4\pi \Im \tau}^{-1} \, \mathbb H^{\rm eqv}_{\ep^{\rm TS}}(\tau) \, \mathbbm{\Lambda}^{\rm sv}_{x,y}(z,\tau) \, \big(\mathbb H^{\rm eqv}_{\ep^{\rm TS}}(\tau) \big)^{-1} \, {\cal Q}_{4\pi \Im \tau} \, (\mathbb M^{\rm sv}_z )^{-1}
\label{ccsc.23}
\eeq
\end{itemize}
\end{prop}

By Lemma \ref{lemusl2}, the conjugation by ${\cal Q}_{4\pi \Im \tau} = \Umod(\tau)  \overline{ \Umod(\tau)^T}$ merely switches the generators $b_k^{(j)} \leftrightarrow b_k^{(k-j-2)}$ and $\ep_k^{(j)} \leftrightarrow \ep_k^{(k-j-2)}$ of the series expansions up to powers of $(\pi \Im \tau)^{\pm 1}$ with rational prefactors as in (\ref{lemq.03}). Note that $\mathbb M^{\rm sv}_z$ commutes with ${\cal Q}_{4\pi \Im \tau}$ by the composing $\mathfrak{sl}_2$ singlets $z_w$.

\begin{proof}
For any triplet $(\mathbb X,\mathbb Y,\mathbb W)$ of group-like series related by
\begin{align}
\mathbb Y &= \Umod(\tau) \, \mathbb X \, \Umod(\tau)^{-1}  \, , \ \ \ \ \ \ 
\overline{\mathbb X ^T} = \mathbb W  \, \mathbb X \, (\mathbb W)^{-1}
\label{ccsc.11}
\end{align}
we have
\beq
\overline{\mathbb Y^T} = (\overline{ \Umod(\tau)^T })^{-1} \, \mathbb W \, \Umod(\tau)^{-1} \, \mathbb Y \, \Umod(\tau) \, (\mathbb W)^{-1} \, \overline{ \Umod(\tau)^T }
\label{ccsc.12}
\eeq

$(i)$: follows from (\ref{ccsc.12}) specialized to $(\mathbb X,\mathbb Y,\mathbb W) = ( \mathbb I^{\rm eqv}_{\ep^{\rm TS}}(\tau) , \mathbb H^{\rm eqv}_{\ep^{\rm TS}}(\tau) , \mathbb M^{\rm sv}_{z} )$, where (\ref{ccsc.01}) is the required relation between $\mathbb W $ and $\mathbb X$,
\beq
\overline{ \mathbb H^{\rm eqv}_{\ep^{\rm TS}}(\tau)^T }
=
(\overline{ \Umod(\tau)^T })^{-1} \, \mathbb M^{\rm sv}_{z} \, \Umod(\tau)^{-1} \, \mathbb H^{\rm eqv}_{\ep^{\rm TS}}(\tau) \, \Umod(\tau) \, (\mathbb M^{\rm sv}_{z})^{-1} \, \overline{ \Umod(\tau)^T }
\label{ccsc.13}
\eeq
which readily implies the claim (\ref{ccsc.21}) since $ \Umod(\tau)$ commutes with $\mathbb M^{\rm sv}_{z}$ and ${\cal Q}_{4\pi \Im \tau}$ can be identified via (\ref{lemq.04}).

$(ii)$: follows in the same way  by specializing $(\mathbb X,\mathbb Y,\mathbb W) = ( \mathbb I^{\rm eqv}_{\ep,b}(u,v,\tau),\mathbb H^{\rm eqv}_{\ep,b}(u,v,\tau),\mathbb M^{\rm sv}_{z} )$.

$(iii)$: follows from (\ref{ccsc.12}) specialized to $(\mathbb X,\mathbb Y,\mathbb W) = ( \mathbbm{\Gamma}^{\rm sv}_{x,y}(z,\tau) , \mathbbm{\Lambda}^{\rm sv}_{x,y}(z,\tau)  , \mathbb M^{\rm sv}_{z} \mathbb I^{\rm eqv}_{\ep^{\rm TS}}(\tau)  )$, where (\ref{ccsc.03}) is the required relation between $\mathbb W $ and $\mathbb X$
\begin{align}
\overline{ \mathbbm{\Lambda}^{\rm sv}_{x,y}(z,\tau) ^T }
&=
(\overline{ \Umod(\tau)^T })^{-1} \, \mathbb M^{\rm sv}_{z} \, \mathbb I^{\rm eqv}_{\ep^{\rm TS}}(\tau) \, \Umod(\tau)^{-1} \, \mathbbm{\Lambda}^{\rm sv}_{x,y}(z,\tau)  \notag \\
&\quad \times \Umod(\tau) \, \big( \mathbb I^{\rm eqv}_{\ep^{\rm TS}}(\tau) \big)^{-1} \, (\mathbb M^{\rm sv}_{z})^{-1} \, \overline{ \Umod(\tau)^T }
\label{ccsc.17}
\end{align}
The claim (\ref{ccsc.23}) then follows by rewriting $\mathbb I^{\rm eqv}_{\ep^{\rm TS}}(\tau)  \Umod(\tau)^{-1} = \Umod(\tau)^{-1} \mathbb H^{\rm eqv}_{\ep^{\rm TS}}(\tau) $ as well as $\Umod(\tau)  ( \mathbb I^{\rm eqv}_{\ep^{\rm TS}}(\tau) )^{-1} = ( \mathbb H^{\rm eqv}_{\ep^{\rm TS}}(\tau) )^{-1} \Umod(\tau)$, by commuting $ \Umod(\tau)$ with $\mathbb M^{\rm sv}_{z}$ and by identifying ${\cal Q}_{4\pi \Im \tau}$ via (\ref{lemq.04}).
\end{proof}

\begin{corollary}
\label{closelamb}
The single-valued eMPLs $\Lambda^{\rm sv}[\ldots;z,\tau]$ defined by the expansion (\ref{cor.3.4d}) of the series $\mathbbm{\Lambda}^{\rm sv}_{x,y}(z,\tau)$ close under complex conjugation. The coefficients in relating $\overline{\Lambda^{\rm sv}[\ldots;z,\tau]}$ to $\Lambda^{\rm sv}[\ldots;z,\tau]$ are determined by (\ref{ccsc.23}) as $\mathbb Q$-linear combinations of products of equivariant iterated Eisenstein integrals $\beta^{\rm eqv}$ in modular frame, see (\ref{defheqv}), single-valued MZVs and, as an additional feature of modular frame, powers of $(\pi \Im \tau)^{\pm 1}$. 
\end{corollary}

%%%%%%%%%%%%%%%%%%%%%%%%%%%%%%%%%%%%%%%%%%%%%%%%%%
\subsubsection{Differential equations of equivariant and single-valued series}
\label{sec.4.diff}

This section gathers the total differentials $\dd_z = \dd z \,\partial_z+\dd \bar z \,\partial_{\bar z}$ at constant $\tau$ of the equivariant generating series $\mathbb I^{\rm eqv}_{\ep,b}(u,v,\tau)$, $\mathbbm{\Gamma}^{\rm sv}_{x,y}(z,\tau)$, $\mathbb H^{\rm eqv}_{\ep,b}(u,v,\tau)$, $\mathbbm{\Lambda}^{\rm sv}_{x,y}(z,\tau)$ and pinpoints why $\mathbbm{\Gamma}_{x,y}(z,\tau)$ and $\mathbbm{\Gamma}^{\rm sv}_{x,y}(z,\tau)$ do not obey the same holomorphic differential equations with respect to $z$.

\begin{prop}\label{thm_deriv}
The $z$- and $\bar z$-derivatives of the equivariant generating series $\mathbb I^{\rm eqv}_{\ep,b}(u,v,\tau)$, $\mathbbm{\Gamma}^{\rm sv}_{x,y}(z,\tau)$ and $\mathbb H^{\rm eqv}_{\ep,b}(u,v,\tau)$, $\mathbbm{\Lambda}^{\rm sv}_{x,y}(z,\tau)$ in holomorphic and modular frame are given as follows in terms of the Brown-Levin connection $\mathbb J^{\rm BL}$ in (\ref{notsec.06}) and its modular analogue $\mathbb J^{\rm mod}$ in (\ref{notsec.07}):
\begin{itemize}
    \item[(i)] the series $\mathbb I^{\rm eqv}_{\ep,b}(u,v,\tau)$ in (\ref{altex.01}) obeys
\begin{align}
{\rm d}_z \mathbb I^{\rm eqv}_{\ep,b}(u,v,\tau) &=
(\mathbb{M}_z^{\rm sv})^{-1}  \, \overline{ \mathbb J^{\rm BL}_{x+2\pi i \tau y,-y}(z,\tau) } \, \mathbb{M}_z^{\rm sv} \,  \mathbb I^{\rm eqv}_{\ep,b}(u,v,\tau) \notag \\
&\quad
-   \mathbb I^{\rm eqv}_{\ep,b}(u,v,\tau)\, \mathbb J^{\rm BL}_{x-2\pi i \tau y,y}(z,\tau)
\label{nwdzi}
\end{align}
or in component form
\begin{align}
\partial_z \mathbb I^{\rm eqv}_{\ep,b}(u,v,\tau) &=  
(\mathbb{M}_z^{\rm sv})^{-1} \, \frac{x {-}2\pi i \bar \tau y}{\tau {-} \bar \tau} \, \mathbb{M}_z^{\rm sv} \, \mathbb I^{\rm eqv}_{\ep,b}(u,v,\tau)
- \mathbb I^{\rm eqv}_{\ep,b}(u,v,\tau)\, \frac{x {-} 2\pi i  \tau y}{\tau {-} \bar \tau} \label{dzbarzbis} \\
&\quad
-  \mathbb I^{\rm eqv}_{\ep,b}(u,v,\tau)\,  e^{-2\pi i \tau \epsilon_0}\,  {\rm ad}_x \Omega(z,{\rm ad}_{\frac{x}{2\pi i}},\tau )\,y\,  e^{2\pi i \tau \epsilon_0}
\notag \\
\partial_{\bar z} \mathbb I^{\rm eqv}_{\ep,b}(u,v,\tau) &=  {-}(\mathbb{M}_z^{\rm sv})^{-1} \,\frac{x {-}2\pi i \bar \tau y}{\tau {-} \bar \tau} \, \mathbb{M}_z^{\rm sv} \, \mathbb I^{\rm eqv}_{\ep,b}(u,v,\tau)
+ \mathbb I^{\rm eqv}_{\ep,b}(u,v,\tau) \,\frac{x {-} 2\pi i  \tau y}{\tau {-} \bar \tau} \notag \\
&\quad - (\mathbb{M}_z^{\rm sv})^{-1}\, e^{-2\pi i \bar \tau  \epsilon_0} \, \ad_x \overline{ \Omega( z,{\rm ad}_{\frac{x}{2\pi i}}, \tau )}\,y\, e^{2\pi i \bar \tau  \epsilon_0}  \,
\mathbb{M}_z^{\rm sv} \, \mathbb I^{\rm eqv}_{\ep,b}(u,v,\tau)  
\notag
\end{align}
\item[(ii)] the series $\mathbb H^{\rm eqv}_{\ep,b}(u,v,\tau)$ in (\ref{thm.3.4a}) obeys
\begin{align}
{\rm d}_z \mathbb H^{\rm eqv}_{\ep,b}(u,v,\tau) &=
(\mathbb{M}_z^{\rm sv})^{-1}  \, \overline{ \mathbb J^{\rm mod}_{x,-y}(z,\tau) } \, \mathbb{M}_z^{\rm sv} \,  \mathbb H^{\rm eqv}_{\ep,b}(u,v,\tau) \notag \\
&\quad
-   \mathbb H^{\rm eqv}_{\ep,b}(u,v,\tau)\, \mathbb J^{\rm mod}_{4\pi\Im\tau y,\frac{-x}{4\pi \Im \tau}}(z,\tau)
\label{nwdzh}
\end{align}
or in component form
 \begin{align}
\partial_z\mathbb H^{\rm eqv}_{\ep,b}(u,v,\tau) &= \frac{1}{\tau {-} \bar \tau}\,  \Big(
(\mathbb{M}_z^{\rm sv})^{-1} \, x \,\mathbb{M}_z^{\rm sv} 
\mathbb H^{\rm eqv}_{\ep,b}(u,v,\tau) - \mathbb H^{\rm eqv}_{\ep,b}(u,v,\tau)\, x \Big) \label{altdz.4bis} \\
&\quad
{+} \sum_{k=2}^\infty \frac{ (\bar \tau {-} \tau)^{k{-}2} }{(k{-}2)!} \, f^{(k{-}1)}(u\tau{+}v,\tau)\, \mathbb{H}^{\rm eqv}_{\ep,b}(u,v,\tau) \, b_k^{(k{-}2)} 
 \notag \\
\partial_{\bar z}\mathbb H^{\rm eqv}_{\ep,b}(u,v,\tau) &=   -2\pi i\, \Big( \mathbb H^{\rm eqv}_{\ep,b}(u,v,\tau)\, y
- (\mathbb{M}_z^{\rm sv})^{-1} \, y\, \mathbb{M}_z^{\rm sv} \,\mathbb H^{\rm eqv}_{\ep,b}(u,v,\tau) \Big)\notag \\
&\quad
+ \sum_{k=2}^{\infty} \frac{1}{({-}2\pi i)^{k{-}2}} \,\overline{ f^{(k{-}1)}(u\tau{+}v,\tau) } \,
(\mathbb{M}_z^{\rm sv})^{-1} \, b_k \,\mathbb{M}_z^{\rm sv}\,  \mathbb H^{\rm eqv}_{\ep,b}(u,v,\tau) 
\notag
 \end{align} 
\item[(iii)] the series $\mathbbm{\Gamma}^{\rm sv}_{x,y}(z,\tau)$ in (\ref{cor.3.4a}) obeys
\begin{align}
{\rm d}_z \mathbbm{\Gamma}^{\rm sv}_{x,y}(z,\tau)&= \big( \mathbb I^{\rm eqv}_{\ep^{\rm TS}}(\tau) \big)^{-1}\,
(\mathbb{M}_z^{\rm sv})^{-1}  \, \overline{ \mathbb J^{\rm BL}_{x+2\pi i \tau y,-y}(z,\tau) } \, \mathbb{M}_z^{\rm sv} \, \mathbb I^{\rm eqv}_{\ep^{\rm TS}}(\tau) \,  \mathbbm{\Gamma}^{\rm sv}_{x,y}(z,\tau)\notag \\
&\quad
-  \mathbbm{\Gamma}^{\rm sv}_{x,y}(z,\tau)\, \mathbb J^{\rm BL}_{x-2\pi i \tau y,y}(z,\tau)
\label{nwdzga}
\end{align}
\item[(iv)] the series $\mathbbm{\Lambda}^{\rm sv}_{x,y}(z,\tau)$ in (\ref{cor.3.4c}) obeys
\begin{align}
{\rm d}_z \mathbbm{\Lambda}^{\rm sv}_{x,y}(z,\tau)&= \big(  \mathbb H^{\rm eqv}_{\ep^{\rm TS}}(\tau) \big)^{-1}\,
(\mathbb{M}_z^{\rm sv})^{-1}  \, \overline{ \mathbb J^{\rm mod}_{x,-y}(z,\tau) } \, \mathbb{M}_z^{\rm sv} \, \mathbb H^{\rm eqv}_{\ep^{\rm TS}}(\tau) \,  \mathbbm{\Lambda}^{\rm sv}_{x,y}(z,\tau) \notag \\
&\quad
-  \mathbbm{\Lambda}^{\rm sv}_{x,y}(z,\tau)\, \mathbb J^{\rm mod}_{4\pi\Im\tau y,\frac{-x}{4\pi \Im \tau}}(z,\tau)
\label{nwdzla}
\end{align}
\end{itemize}
\end{prop}

\begin{proof}
\phantom{x}

$(i)$: The definitions (\ref{notsec.19}) and (\ref{ahologam}) of $\mathbbm{\Gamma}_{x,y}(z,\tau)$ and $\overline{\mathbbm{\Gamma}_{x,y}(z,\tau)^T}$ as path-ordered exponentials readily imply that
\begin{align}
\dd_z\mathbbm{\Gamma}_{x,y}(z,\tau) &= - \mathbbm{\Gamma}_{x,y}(z,\tau) \, \mathbb J^{\rm BL}_{x-2\pi i \tau y,y}(z,\tau) \label{notsame} \\
\dd_z\overline{\mathbbm{\Gamma}_{x,y}(z,\tau)^T} &= \overline{ \mathbb J^{\rm BL}_{x+2\pi i \tau y,-y}(z,\tau) }\, \overline{\mathbbm{\Gamma}_{x,y}(z,\tau)^T}
\notag
\end{align}
When applied to the second representation (\ref{altex.01}) of $ \mathbb I^{\rm eqv}_{\ep,b}(u,v,\tau)$, we arrive at the $\dd_z$ differential in (\ref{nwdzi}) which, after inserting (\ref{notsec.06}) for the Brown-Levin connection, unpacks to the component form (\ref{dzbarzbis}).

$(ii)$: Given that $ \mathbb H^{\rm eqv}_{\ep,b}(u,v,\tau)$ in (\ref{thm.3.4a}) is the conjugate of $ \mathbb I^{\rm eqv}_{\ep,b}(u,v,\tau)$ by the $z$-independent transformation $\Umod(\tau)$ towards modular frame, its $\dd_z$-derivatives follow from the analogous conjugations of the one-forms in the expression (\ref{nwdzi}) for $\dd_z \mathbb I^{\rm eqv}_{\ep,b}(u,v,\tau)$. Since the letters $x,y$ in the Brown-Levin connection are related to those in its modular version $\mathbb J^{\rm mod}$ via (\ref{notsec.07}) and transform under $\Umod(\tau)$ according to (\ref{Umodxy}), we arrive at
\begin{align}
\Umod(\tau)\, \mathbb J^{\rm BL}_{x- 2\pi i \tau y,y}(z,\tau) 
\, \Umod(\tau)^{-1}  &= 
\Umod(\tau)\, \mathbb J^{\rm mod}_{x- 2\pi i \tau y, \frac{x - 2\pi i \bar \tau y}{2\pi i (\tau - \bar \tau)}}(z,\tau)  \, \Umod(\tau)^{-1}  \notag \\
&= \mathbb J^{\rm mod}_{ 4\pi (\Im \tau) y , \frac{-x }{4 \pi \Im \tau}}(z,\tau) 
\end{align}
and
\begin{align}
\Umod(\tau)\, \overline{ \mathbb J^{\rm BL}_{x+ 2\pi i \tau y,-y}(z,\tau)  } \, \Umod(\tau)^{-1} &=  \Umod(\tau)\,
\overline{  
\mathbb J^{\rm mod}_{x+ 2\pi i \tau y,\frac{x+2\pi i \bar \tau y}{2\pi i (\tau - \bar \tau)}}(z,\tau)  } \, \Umod(\tau)^{-1}
 \notag \\
&=  \overline{    \mathbb J^{\rm mod}_{x,-y}(z,\tau)   } 
\end{align}
This determines the letters of $\mathbb J^{\rm mod}$ in the target expression (\ref{nwdzh}) for $ \dd_z\mathbb H^{\rm eqv}_{\ep,b}(u,v,\tau)$ from those of $\mathbb J^{\rm BL}$ in item $(i)$. We have also used that the $\mathfrak{sl}_2$ invariant $(\mathbb M^{\rm sv}_z)^{\pm 1}$ commutes with $\Umod(\tau)$. The component form (\ref{altdz.4bis}) follows from (\ref{nwdzh}) through the definition (\ref{notsec.07}) of~$\mathbb J^{\rm mod}$.

$(iii)$ and $(iv)$: follow by relating $\mathbbm{\Gamma}^{\rm sv}_{x,y}(z,\tau)$ and $\mathbbm{\Lambda}^{\rm sv}_{x,y}(z,\tau)$ to $\mathbb I^{\rm eqv}_{\ep,b}(u,v,\tau)$ and $\mathbb H^{\rm eqv}_{\ep,b}(u,v,\tau)$ via left-multiplication with $z$-independent series $(\mathbb I^{\rm eqv}_{\ep^{\rm TS}}(\tau))^{-1}$ and $(\mathbb H^{\rm eqv}_{\ep^{\rm TS}}(\tau))^{-1}$ together with the respective expressions (\ref{nwdzi}) and (\ref{nwdzh}) for ${\rm d}_z \mathbb I^{\rm eqv}_{\ep,b}(u,v,\tau)$ and ${\rm d}_z \mathbb H^{\rm eqv}_{\ep,b}(u,v,\tau)$.
\end{proof}

With the $z$-derivatives of Proposition \ref{thm_deriv} in place, we can compare $\partial_z\mathbbm{\Gamma}_{x,y}(z,\tau)$ with $\partial_z\mathbbm{\Gamma}^{\rm sv}_{x,y}(z,\tau)$ and settle if the parallels between holomorphic differential equations at genus~zero
\beq
\partial_z \mathbb G_{e_0,e_1}(z) =   \mathbb G_{e_0,e_1}(z) \, \bigg( \frac{e_0}{z} + \frac{e_1}{z{-}1} \bigg)
\ \ \ 
\leftrightarrow \ \ \ 
\partial_z \mathbb G^{\rm sv}_{e_0,e_1}(z) =   \mathbb G^{\rm sv}_{e_0,e_1}(z) \, \bigg( \frac{e_0}{z} + \frac{e_1}{z{-}1} \bigg)
\label{g0par.01}
\eeq
following from (\ref{svmpl.12})
have an echo in our construction of single-valued eMPLs at genus one.

\begin{corollary}
\label{cor.cav}
The holomorphic $z$-derivatives of multi-valued and single-valued eMPLs are generated by
\begin{align}
\partial_z \mathbbm{\Gamma}_{x,y}(z,\tau) &=  - \mathbbm{\Gamma}_{x,y}(z,\tau)\,\bigg\{ \frac{x {-} 2\pi i  \tau y}{\tau {-} \bar \tau}
+  e^{-2\pi i \tau \epsilon_0}\,  {\rm ad}_x \Omega(z,{\rm ad}_{\frac{x}{2\pi i}},\tau )\,y\,  e^{2\pi i \tau \epsilon_0} \bigg\}
 \label{diffdzs}
\end{align}
and
\begin{align}
\partial_z \mathbbm{\Gamma}^{\rm sv}_{x,y}(z,\tau) &=  
\big(\mathbb I^{\rm eqv}_{\ep^{\rm TS}}(\tau) \big)^{-1}\,
(\mathbb{M}_z^{\rm sv})^{-1} \, \frac{x {-}2\pi i \bar \tau y}{\tau {-} \bar \tau} \, \mathbb{M}_z^{\rm sv} \, 
\mathbb I^{\rm eqv}_{\ep^{\rm TS}}(\tau)
\, \mathbbm{\Gamma}^{\rm sv}_{x,y}(z,\tau)
 \label{svdiffzs} \\
&\quad
- \mathbbm{\Gamma}^{\rm sv}_{x,y}(z,\tau)\, \bigg\{ \frac{x {-} 2\pi i  \tau y}{\tau {-} \bar \tau}
+  e^{-2\pi i \tau \epsilon_0}\,  {\rm ad}_x \Omega(z,{\rm ad}_{\frac{x}{2\pi i}},\tau )\,y\,  e^{2\pi i \tau \epsilon_0} \bigg\}
\notag
\end{align}
The terms in the curly brackets in (\ref{diffdzs}) and the second line of (\ref{svdiffzs}) are identical and come from the coefficient of ${\rm d} z$ in $\mathbb J^{\rm BL}_{x- 2\pi i \tau y,y}$. However, the first line of (\ref{svdiffzs}) prevents the parallels in the holomorphic derivatives of multi-valued and single-valued genus-zero MPLs in (\ref{g0par.01}) from extending to our genus-one construction. 
\end{corollary}
While (\ref{diffdzs}) follows from the component form of the differential equation (\ref{notsame}) of the series $\mathbbm{\Gamma}_{x,y}(z,\tau)$, we obtained (\ref{svdiffzs}) by isolating the $\dd z$ (as opposed to $\dd \bar z$) components of $\dd_z \mathbbm{\Gamma}^{\rm sv}_{x,y}(z,\tau)$ in (\ref{nwdzga}). Given that the left-multiplicative $\overline{ \mathbb J^{\rm BL}_{x+2\pi i \tau y,-y}(z,\tau) } $ in (\ref{nwdzga}) also contributes terms $\sim \dd z$, see the complex conjugate of (\ref{notsec.06}), we obtain the first line of (\ref{svdiffzs}) which prevents $\partial_z$ from commuting with a formal replacement $\mathbbm{\Gamma}_{x,y}(z,\tau) \rightarrow \mathbbm{\Gamma}^{\rm sv}_{x,y}(z,\tau)$. 

The analogous replacement $\mathbb{G} \rightarrow \mathbb{G}^{\rm sv}$ for genus-zero MPLs in any number of variables implements the single-valued map which commutes with holomorphic derivatives in any variable \cite{Schnetz:2013hqa, Brown:2013gia, Brown:2018omk}. By the statement of Corollary \ref{cor.cav}, the present form of our construction of single-valued eMPLs via $\mathbbm{\Gamma}^{\rm sv}_{x,y}(z,\tau)$ cannot be identified as the image of $\mathbbm{\Gamma}_{x,y}(z,\tau)$ under the single-valued map of the references, at least when the latter is required to commute with all holomorphic derivatives also beyond genus zero.

%%%%%%%%%%%%%%%%%%%%%%%%%%%%%%%%%%%%%%%%%%%%%%%%%%
%%%%%%%%%%%%%%%%%%%%%%%%%%%%%%%%%%%%%%%%%%%%%%%%%%
\subsection{Expansion around the cusp $\tau \rightarrow i\infty$}
\label{sec:3.3}

This section is dedicated to the asymptotic behavior of the equivariant
series $\mathbb H^{\rm eqv}_{\ep,b}(u,v,\tau)$, $\mathbbm{\Lambda}^{\rm sv}_{x,y}(z,\tau)$ in modular frame and single-valued eMPLs as $\tau$ approaches
the cusp $i\infty$ while keeping $u$ and $v$ fixed. 
The analogous expansion
of equivariant iterated Eisenstein integrals in the modular frame (\ref{defheqv})
produces Laurent polynomials in %
\beq
Y = 4\pi \Im \tau
\label{Yvar}
\eeq
with $\mathbb Q$-linear combinations of single-valued MZVs as coefficients
at each degree in the $\ep_k^{(j){\rm TS}}$ (which is preserved by Pollack relations).
More specifically, this asymptotic behavior of equivariant iterated Eisenstein integrals in the series 
$\mathbb H^{\rm eqv}_{\ep^{\rm TS}}(\tau)$ of (\ref{defheqv}) given in (\ref{thm.3.4k}) below
\begin{itemize}
\item is contained in Theorem 10.6 of Brown's work \cite{Brown:2017qwo2} upon
restriction to the zero mode of the Fourier expansion in $\Re \tau$;
\item can be verified by combining (\ref{lieg1.51}), (\ref{defheqv}) and the fact that
the expansion of non-holomorphic modular forms around the cusp
takes the form of $\sum_{m,n=0}^\infty c_{m,n}(\Im \tau) q^{m} \bar q^n$
where $c_{m,n}(\Im \tau)$ cannot depend on $\Re \tau$;
\item is conjectured to carry over to MGFs \cite{Zerbini:2015rss, DHoker:2015wxz},
tested in a wealth of explicitly known cases and proven for infinite families of
examples \cite{DHoker:2019xef, Zagier:2019eus}. Still, it is an open problem to prove in generality that the
asymptotics of arbitrary MGFs at $\tau \rightarrow i \infty$ described by
Laurent polynomials in $Y$ solely features
$\mathbb Q$-linear combinations of single-valued MZVs as their coefficients.\footnote{It is for instance not proven (though commonly expected) that all MGFs are
linear combinations of equivariant iterated Eisenstein integrals in the
expansion (\ref{defheqv}) with rational combinations of single-valued MZVs as
coefficients, e.g.\ it remains to exclude coefficients of the form $\mathbb Q \pi^{2k}$ with $k \in \mathbb N$.}
 \end{itemize}
The goal of this section is to extend the number-theoretic analysis
of the asymptotics at the cusp to the $u,v$-dependent equivariant
series of the previous section and to the single-valued eMPLs obtained from their expansion in $b_k^{(j)}$ in Theorem \ref{3.cor:1}.

We start by degenerating the equivariant series $\mathbb H^{\rm eqv}_{\ep,b}(u,v,\tau)$ in Definition \ref{newdefheqv} with the modular properties (\ref{thm.3.4b}) which produces non-uniquely defined non-holomorphic
modular forms through its expansion in $\ep_k^{(j)}, b_k^{(j)}$:

\begin{theorem}
\label{3.thm:5}
The $\tau \rightarrow i\infty$ asymptotics of the series $\mathbb H^{\rm eqv}_{\ep,b}(u,v,\tau) $ in (\ref{thm.3.4a})
at fixed values of $0\leq u,v< 1$ is given by
\beq
\mathbb H^{\rm eqv}_{\ep,b}(u,v,\tau) = \exp\bigg( {-} \frac{\ep_0^\vee}{Y} \bigg)
\, (\mathbb M^{\rm sv}_z)^{-1} \, \mathbb M^{\rm sv}_{\Sigma(u)} \, \mathbb G^{\rm sv}_{E_0(u),E_1(u)}(e^{2\pi i z})
\, e^{-Y \ep_0} \,  \exp\bigg( \frac{\ep_0^\vee}{Y} \bigg)+ {\cal O}(q^{1-u},\bar q^{1-u})
\label{thm.3.4c}
\eeq
where $Y=4\pi \Im \tau$, and the $u$-dependent letters of the series
(\ref{svmpl.12}) in single-valued MPLs in one variable are given by
\beq
E_0(u) = \dfrac{1}{u}\biggl({-}\epsilon_0+\sum_{k=4}^\infty (k{-}1) \dfrac{B_k}{k!}\epsilon_k +\sum_{k=2}^\infty (k{-}1)\dfrac{B_k(u)}{k!}b_k\biggr) \, , \ \ \ \ E_1(u) = \sum_{k=2}^\infty \dfrac{u^{k-2}}{(k{-}2)! } \, b_k
\label{thm.3.4d}
\eeq
Moreover, ${\cal O}(q^{1-u},\bar q^{1-u})$ indicates that at each
degree of the expansion in $\ep_k^{(j)}, b_k^{(j)}$, terms of
the form $(\Im \tau)^k q^{1-u}$ and $(\Im \tau)^k \bar q^{1-u}$
with $k \in \mathbb Z$ or with higher powers of $q,\bar q$ are discarded.
\end{theorem}
The proof will be split into four steps to be presented in sections \ref{sec:3.3.1} to \ref{sec:3.3.4} below: 
\begin{itemize}
\item[(i)] show that the analogous asymptotics of the series $\mathbb I_{\ep,b}(u,v,\tau)$ in (\ref{notsec.17}) is given by
\beq
\mathbb I_{\ep,b}(u,v,\tau) = e^{-2\pi i v E_0(u)} \, \mathbb G_{E_0(u),E_1(u)}(e^{2\pi i z})  \, e^{2\pi i \ep_0 \tau}
+ {\cal O}(q^{1-u})
\label{thm.3.4e}
\eeq
\item[(ii)] establish numerous bracket relations among the generators
including
\beq
[E_0(u),\Sigma_w(u)]=0
\label{goaltwo}
\eeq
\item[(iii)] show that the series (\ref{svmpl.12}) in single-valued MPLs arises by proving the analogue
\beq
\big[ E_1(u) , \Sigma_w(u) \big] = \big[ P_w\big(E_0(u), E_1(u) \big) , E_1(u) \big]
\label{thm.3.4f}
\eeq
of the genus-zero bracket relations (\ref{svmpl.07}) among zeta generators and braid operators;
\item[(iv)] assemble the asymptotics of the series $\mathbb H^{\rm eqv}_{\ep,b}(u,v,\tau) $ from (\ref{thm.3.4e}).
\end{itemize}

With Theorem \ref{3.thm:5} on the degeneration of the $z$-dependent equivariant series $\mathbb H^{\rm eqv}_{\ep,b}(u,v,\tau)$ in place, we can now proceed to investigating the generating series
\beq
\mathbbm{\Lambda}^{\rm sv}_{x,y}(z,\tau) = \big( \mathbb H^{\rm eqv}_{\ep^{\rm TS}}(\tau) \big)^{-1}\mathbb H^{\rm eqv}_{\ep,b}(u,v,\tau)
\label{thm.3.4g}
 \eeq
of single-valued eMPLs defined by (\ref{cor.3.4c}).\footnote{Its alternative representation as a ratio of two equivariant series in modular frame follows from (\ref{cor.3.4c}) together with (\ref{defheqv}), (\ref{cor.3.4a}) and (\ref{thm.3.4a}).}

\begin{corollary}
\label{3.cor:9} \phantom{x}

\begin{itemize}
\item[(i)] The series $\mathbbm{\Lambda}^{\rm sv}_{x,y}(z,\tau) $ in single-valued eMPLs in modular frame exhibits the following $\tau \rightarrow i \infty$ asymptotics 
at fixed $u,v$ analogous to (\ref{thm.3.4c})
\begin{align}
\mathbbm{\Lambda}^{\rm sv}_{x,y}(z,\tau) &= \exp\bigg( {-} \frac{\ep_0^\vee}{Y} \bigg) \, e^{Y \ep_0} \, 
e^{N^{\rm TS}Y} \, (\mathbb M^{\rm sv}_\sigma)^{-1} \, \mathbb M^{\rm sv}_{\Sigma(u)} 
\label{thm.3.4h} \\
&\quad \times
 \mathbb G^{\rm sv}_{E_0(u),E_1(u)}(e^{2\pi i z})
\, e^{-Y \ep_0} \,  \exp\bigg( \frac{\ep_0^\vee}{Y} \bigg)+ {\cal O}(q^{1-u},\bar q^{1-u})
\notag
\end{align}
with the series $N^{\rm TS}$ in Tsunogai derivations
given by (\ref{notsec.22}) and letters $E_0, E_1$ in (\ref{thm.3.4d}).
\item[(ii)] The $\tau \rightarrow i\infty$ asymptotics of the 
single-valued eMPLs $\lab{j_1 &\ldots &j_r}{k_1 &\ldots &k_r}{z, \tau }$ 
generated by the $b_k^{(j)}$-expansion (\ref{cor.3.4d}) of $\mathbbm{\Lambda}^{\rm sv}_{x,y}(z,\tau)$ consists of Laurent polynomials in $Y$ whose coefficients are
$\mathbb Q$-linear combinations of products of single-valued MZVs and single-valued
MPLs $G^{\rm sv}(a_1,\ldots,a_r;e^{2\pi i z})$ in one variable, i.e.\ with $a_i \in \{0,1\}$. 
\end{itemize} 
\end{corollary}
Note that the single-valued
MPLs $G^{\rm sv}(a_1,\ldots,a_r;e^{2\pi i z})$ in one variable already cover polynomials in $u$ when admitting variable powers of $Y$
since
\beq
G^{\rm sv}(0;e^{2\pi i z}) = 2\pi i (z{-}\bar z) = -4\pi u \Im \tau = - u\,Y
\label{getbacku}
\eeq

\begin{proof}
\phantom{x}

$(i)$: The techniques of section \ref{sec:3.3.1} below straightforwardly carry over
to determine the asymptotics
\beq
\mathbb I_{\ep^{\rm TS}}(\tau)  = e^{2\pi i \tau N^{\rm TS}} e^{2\pi i \tau \ep_0} + {\cal O}(q,\bar q)
\label{thm.3.4i}
\eeq
Upon combination with its transposed complex conjugate in the
expression (\ref{lieg1.51}) for $\mathbb I^{\rm eqv}_{\ep^{\rm TS}}(\tau)$
and conjugating with $\Umod(\tau) $ in (\ref{usl2}), we find
\beq
\mathbb H^{\rm eqv}_{\ep^{\rm TS}}(\tau)  =
\exp\bigg( {-} \frac{\ep_0^\vee}{Y} \bigg)  \, ( \mathbb M_{z}^{\rm sv})^{-1} \,   \mathbb M_{\sigma}^{\rm sv} \,  e^{-Y N^{\rm TS}}
\, e^{-Y \ep_0} \,  \exp\bigg( \frac{\ep_0^\vee}{Y} \bigg)
 + {\cal O}(q,\bar q)
\label{thm.3.4k}
\eeq
By concatenating its inverse with the asymptotics of $\mathbb H^{\rm eqv}_{\ep,b}(u,v,\tau)$ 
in (\ref{thm.3.4c}), one arrives at the first claim (\ref{thm.3.4h}).

$(ii)$: Concentrating on the inner factors $e^{N^{\rm TS}Y}  (\mathbb M^{\rm sv}_\sigma)^{-1}  \mathbb M^{\rm sv}_{\Sigma(u)}  \mathbb G^{\rm sv}_{E_0(u),E_1(u)}(e^{2\pi i z})$ of the asymptotics (\ref{thm.3.4h})
of $\mathbbm{\Lambda}^{\rm sv}_{x,y}(z,\tau)$, one arrives at a series
in $b_k^{(j)}$ where the coefficients of each word comprise $\mathbb Q$ combinations of
$G^{\rm sv}(a_1,\ldots,a_r;e^{2\pi i z})$, single-valued MZVs and polynomials in $u$,
the latter coming from $\Sigma_w(u), E_0(u), E_1(u)$.
While $\Sigma_w(u)$ in (\ref{not.06}) and $E_1(u)$ in (\ref{thm.3.4d}) are evident to only feature non-negative powers of $u$, the expression for $E_0(u)$ in (\ref{thm.3.4d}) has a simple pole in $u$. These poles cannot persist in $\mathbbm{\Lambda}^{\rm sv}_{x,y}(z,\tau)$ or at any order of its expansion in $b_k^{(j)}$ since its $z$-dependence is entirely carried by $\mathbb I^{\rm eqv}_{\ep,b}(u,v,\tau)$ and ultimately by the path-ordered exponentials $\overline{\mathbbm{\Gamma}_{x,y}(z,\tau)^T} $ and $
\mathbbm{\Gamma}_{x,y}(z,\tau)$. The latter are known to be non-singular as $z=u\tau{+}v$ approaches generic real values $v\in \mathbb R$ (at each order of their expansions in $x,y$) from the properties of the connections (\ref{defKcon}) or (\ref{notsec.06}).
The resulting polynomial dependence of the inner factors $e^{N^{\rm TS}Y}  (\mathbb M^{\rm sv}_\sigma)^{-1}  \mathbb M^{\rm sv}_{\Sigma(u)}  \mathbb G^{\rm sv}_{E_0(u),E_1(u)}(e^{2\pi i z})$ on $u$ can then be absorbed into single-valued MPLs via (\ref{getbacku}).

Finally, the conjugation by $\exp( {-} \frac{\ep_0^\vee}{Y} )  e^{Y \ep_0} $ in (\ref{thm.3.4h}) can only give a finite range of 
integer powers of $Y$ along with a fixed word in $b_k^{(j)}$: the lowest- and highest-weight
conditons (\ref{lieg1.27}) imply that both exponentials truncate on a given word
$\sim b_{k_1}^{(j_1)} \ldots b_{k_r}^{(j_r)}$ whose coefficient determines the asymptotics of  
$\lab{j_1 &\ldots &j_r}{k_1 &\ldots &k_r}{z, \tau }$.
\end{proof}

%%%%%%%%%%%%%%%%%%%%%%%%%%%%%%%%%%%%%%%%%%%%%%%%%%
%%%%%%%%%%%%%%%%%%%%%%%%%%%%%%%%%%%%%%%%%%%%%%%%%%
\subsubsection{Step (i) of proving Theorem \ref{3.thm:5}}
\label{sec:3.3.1}

We recall that the target series $\mathbb H^{\rm eqv}_{\ep,b}(u,v,\tau)$ is given by the $\Umod(\tau)$ conjugate (\ref{thm.3.4a}) of the equivariant series $\mathbb I^{\rm eqv}_{\ep,b}(u,v,\tau) $ built from
$\mathbb I_{\ep,b}(u,v,\tau)$ and its complex conjugate according to (\ref{grteq.01}).
The asymptotics of the meromorphic series $\mathbb I_{\ep,b}(u,v,\tau)$ 
is most conveniently determined
by rewriting the defining path-ordered exponential in (\ref{notsec.17})  as
\beq
\mathbb I_{\ep,b}(u,v,\tau) = {\rm Pexp} \bigg( \int^{i\infty}_\tau \, \dd \tau_1 \, \mathbb J^{(\tau)}_{x,y,\ep}(u,v,\tau_1)    \bigg) \, e^{2\pi i \tau \ep_0}
\label{thm.3.4l}
\eeq
with
\beq
\mathbb J^{(\tau)}_{x,y,\ep}(u,v,\tau) = 2\pi i   \bigg( \epsilon_0+\sum_{k=2}^\infty \frac{(k{-}1)}{(2\pi i)^k} \big[ {\rm G}_k(\tau)\epsilon_k
 - f^{(k)}(u\tau {+} v,\tau) b_k \big]  \bigg)
 \label{thm.3.4m}
 \eeq
The rewriting (\ref{thm.3.4l}) can for instance be verified by checking that it reproduces the differential
equation and initial value of (\ref{notsec.17}). The $\tau \rightarrow i \infty$ behavior at fixed $u,v$ is
then studied at the level of the connection, namely using
\begin{align}
{\rm G}_k(\tau) &= - (2\pi i)^k\, \frac{B_k}{k!}+\mathcal{O}(q)
 \label{thm.3.4o} \\
f^{(k)}(u\tau{+}v,  \tau) &= \dfrac{(2\pi i)^k}{(k{-}1)!} \biggl( \frac{ B_k(u) }{k} + \dfrac{ u^{k-1} \, e^{2\pi i z} }{ e^{2\pi i z} {-} 1 }\biggr)+\mathcal{O}(q^{1-u})
\notag 
\end{align}
together with the change of variables
\beq
\sigma = e^{2\pi i z} = e^{2\pi i (u \tau_1 + v)}  \ \ \ \Rightarrow \ \ \ 2\pi i \, \dd\tau_1 = \frac{\dd \sigma}{u \, \sigma}
 \label{thm.3.4p}
\eeq
Inserting both of (\ref{thm.3.4o}) and (\ref{thm.3.4p}) into $\dd \tau \, \mathbb J^{(\tau)}_{x,y,\ep}(u,v,\tau)$ then results in
\begin{align}
\dd \tau \, \mathbb J^{(\tau)}_{x,y,\ep}(u,v,\tau) &= \dd \sigma\, \bigg\{ \frac{1}{u \, \sigma}
\, \bigg( \ep_0 - \sum_{k=4}^\infty (k{-}1) \, \frac{B_k}{k!} \, \ep_k
- \sum_{k=2}^\infty (k{-}1) \, \frac{B_k(u)}{k!} \, b_k \bigg) \notag \\
&\quad\quad\quad\quad\quad - \frac{1}{\sigma{-}1} \sum_{k=2}^\infty \frac{u^{k-2}}{(k{-}2)!} \, b_k \bigg\} + {\cal O}(q^{1-u}) \notag \\
&= - \dd \sigma \, \bigg(  \frac{E_0(u)}{\sigma} + \frac{E_1(u)}{\sigma{-}1}\bigg) + {\cal O}(q^{1-u})
\label{thm.3.4q}
\end{align}
where we have identified the letters $E_0(u), E_1(u)$ in (\ref{thm.3.4d}) in the last step.

The integration limits $\int^{ i\infty}_\tau \dd \tau_1$ of (\ref{thm.3.4l}) naively translate into
$\int^0_\sigma \dd \sigma_1$, but it is worth noting that the upper limit $\sigma_1 \rightarrow 0$ is attained
from the $\rho \rightarrow i\infty$ limit of $e^{2\pi i (u\rho + v)}$ with real $u,v$.
The endpoint singularity of the KZ connection in
\begin{align}
\mathbb I_{\ep,b}(u,v,\tau) &= {\rm Pexp} \bigg({-}  \lim_{\rho \rightarrow i\infty}\int^{e^{2\pi i (u\rho + v)}}_{e^{2\pi i z}} \dd \sigma_1 \, \bigg(  \frac{E_0(u)}{\sigma_1} + \frac{E_1(u)}{\sigma_1{-}1}\bigg)    \bigg)  e^{2\pi i \tau \ep_0} + {\cal O}(q^{1-u}) \notag \\
&= e^{-2\pi i v E_0(u)}\, {\rm Pexp} \bigg({-}  \int^{0}_{e^{2\pi i z}} \dd \sigma_1 \, \bigg(  \frac{E_0(u)}{\sigma_1} + \frac{E_1(u)}{\sigma_1{-}1}\bigg)    \bigg)  e^{2\pi i \tau \ep_0} + {\cal O}(q^{1-u})  \notag \\
&= e^{-2\pi i v E_0(u)}\, \mathbb G_{E_0(u),E_1(u)}(e^{2\pi i z})\, e^{2\pi i \tau \ep_0} + {\cal O}(q^{1-u})
\label{thm.3.4r}
\end{align}
makes the path-ordered exponential sensitive to the direction
in which the integration path leaves the origin. The mismatch between the
tangential vector with complex phase $\sim e^{2\pi i v}$ in the first line of (\ref{thm.3.4r})
and the tangent vector in positive real direction in the definition (\ref{svmpl.00})
of the MPL series $\mathbb G_{E_0(u),E_1(u)}(e^{2\pi i z})$ leads to the phase factor
$e^{-2\pi i v E_0(u)}$. This concludes our derivation of the
expression (\ref{thm.3.4e}) for the asymptotics of $\mathbb I_{\ep,b}(u,v,\tau) $.

%%%%%%%%%%%%%%%%%%%%%%%%%%%%%%%%%%%%%%%%%%%%%%%%%%
%%%%%%%%%%%%%%%%%%%%%%%%%%%%%%%%%%%%%%%%%%%%%%%%%%
\subsubsection{Step (ii) of proving Theorem \ref{3.thm:5}}
\label{sec:3.3.2}

We will as a next step show that the augmented zeta generators $\Sigma_w(u)$ in
(\ref{not.06}) and the following series  $T_{01}(u) , N(u)$ all commute with each other (see (\ref{lieg1.17}) for $t_{01}$):
\begin{align}
T_{01}(u) &= e^{ux} t_{01} e^{-ux}
= - \frac{{\rm ad}_x e^{u {\rm ad}_x }}{ e^{{\rm ad}_x} - 1 }\, (y)  \label{notsec.21} \\
&=  -  \sum_{k=0}^\infty    \frac{B_k(u)}{k!}  {\rm ad}_x^k(y)
=  - y +  \sum_{k=1}^\infty    \frac{B_k(u)}{k!}  b_{k+1}
\notag
\\
N(u) &=  -\ep_0 + \sum_{k=2}^\infty (k{-}1) \bigg( \frac{B_k}{k!} \, \ep_k + \frac{B_k(u)}{k!}  b_k \bigg) 
\notag
\end{align}
namely, for any odd $w\geq3$
\begin{align}
 \big[ T_{01}(u) , \Sigma_w(u) \big] &= 0   \label{thm.3.4s} \\
  \big[ T_{01}(u) , N(u) \big]  &= 0 \notag \\
 \big[ N(u)  ,  \Sigma_w(u)  \big] &= 0 \notag 
\end{align}
In section \ref{sec:4} below, these identities will be useful for the proof of Theorem \ref{3.thm:1},
and $T_{01}(u) ,N(u)$ will be identified as components of a gauge transformed version
of the CEE connection (\ref{notsec.02}) at $\tau \rightarrow i\infty$.

In order to prove (\ref{thm.3.4s}), we first demonstrate that the commutators on the left side vanish at $u=0$.
The first one involving $t_{01} = T_{01}(0)$ follows from the defining property (\ref{dfprpsig}) of genus-one zeta generator 
\beq
 [t_{01} , \Sigma_w(0)] = 0 \,, \ \ \ \ \ \ 
\Sigma_w(0) = P_w(t_{12},t_{01}) +\sigma_w 
 \label{thm.3.4t}
\eeq
The second line of (\ref{thm.3.4s}) at $u=0$ follows from flatness
of the KZB connection (\ref{notsec.02}) at $\tau \rightarrow i \infty$:
We compute the limit
\beq
 \mathbb K_{x,y,\ep}(z,i\infty) = - \dd \sigma \, \bigg( \frac{t_{01}}{\sigma} + \frac{t_{12}}{\sigma{-}1} \bigg)
 -2\pi i \dd \tau \, N(0)
  \label{thm.3.4u}
\eeq
in the coordinate $\sigma = e^{2\pi i z}$ and use the fact that the $\dd \sigma$, $\dd \tau$ 
components already commute at generic $\tau$ \cite{KZB}. Since $t_{12} $ commutes
with itself and all the Tsunogai derivations in $N^{\rm TS}$ of (\ref{notsec.22}), we conclude
 \beq
 [t_{01} , N(0)] = 0 \,, \ \ \ \ \ \ 
  [t_{12} , N(0)] = 0 \,, \ \ \ \ \ \ 
 N(0) = N^{\rm TS}+ \frac{B_2}{2} t_{12}
   \label{thm.3.4v}
 \eeq
This implies that $N(0)$ commutes with arbitrary
words in $t_{01},t_{12}$, in particular that
\beq
\big[ N(0) , P_w(t_{12},t_{01})\big] =0
  \label{thm.3.4x}
\eeq
As a consequence, we also find that the third line of 
(\ref{thm.3.4s}) vanishes at $u=0$
\beq
\big[ N(0)  ,  \Sigma_w(0)  \big] = [ N^{\rm TS}  ,  \sigma_w  ]  + \frac{B_2}{2} \, [ t_{12}  ,  \sigma_w  ]
+  \big[ N(0)  ,  P_w(t_{12},t_{01})  \big]  = 0
  \label{thm.3.4y}
\eeq
using the fact that $  \sigma_w$ commutes with $N^{\rm TS}$ and $t_{12}$.

It remains to prove (\ref{thm.3.4s}) at non-zero values of $u$.
This will be done by means of the following differential equations
which are all elementary consequences of the definitions (\ref{not.06})
and (\ref{notsec.21})
\begin{align}
 \partial_u T_{01}(u) &= \ad_x T_{01}(u)   \label{thm.3.4z} \\
\partial_u \Sigma_w(u) &= \ad_x \Sigma_w(u)  \notag\\
\partial_u N(u) &= \ad_x N(u) + T_{01}(u) \notag
\end{align}
Since $\ad_x$ obeys the same Leibniz rule as $\partial_u$, these
differential equations imply
\begin{align}
\partial_u \big[ T_{01}(u) , \Sigma_w(u) \big] &=  \ad_x \big[ T_{01}(u) , \Sigma_w(u) \big] 
 \label{thm.3.4a2}  \\
\partial_u \big[ T_{01}(u) , N(u) \big] &=  \ad_x \big[ T_{01}(u) , N(u) \big]  + \big[ T_{01}(u) , T_{01}(u) \big]  
%=   \ad_x \big[ T_{01}(u) , N(u) \big] 
\notag \\
\partial_u \big[ N(u)  ,  \Sigma_w(u)  \big]  &= \ad_x \big[ N(u)  ,  \Sigma_w(u)  \big]  + \big[ T_{01}(u)  ,  \Sigma_w(u)  \big]
\notag
\end{align}
for the commutators under investigation. Since all the commutators vanish  at $u=0$
by the discussion above, the first two lines of (\ref{thm.3.4a2}) imply 
$ [ T_{01}(u) , \Sigma_w(u) ] = 0$ and $[ T_{01}(u) , N(u) ] =0$ as the unique solutions 
of the differential equations with vanishing initial values. The third line of (\ref{thm.3.4a2}) then reduces
to $(\partial_u - \ad_x) [ N(u)  ,  \Sigma_w(u)  ]=0$ which again implies 
$[ N(u)  ,  \Sigma_w(u)  ] = 0$ by the vanishing initial value at $u=0$.
Hence, we have established the vanishing of all the three commutators in (\ref{thm.3.4s}).

The expression (\ref{thm.3.4d}) then reveals that $E_0(u) = N(u)/u$, 
so the third line of (\ref{thm.3.4s}) implies the
vanishing of $ [ E_0(u)  ,  \Sigma_w(u)  ]$ in (\ref{goaltwo}).

%%%%%%%%%%%%%%%%%%%%%%%%%%%%%%%%%%%%%%%%%%%%%%%%%%
%%%%%%%%%%%%%%%%%%%%%%%%%%%%%%%%%%%%%%%%%%%%%%%%%%
\subsubsection{Step (iii) of proving Theorem \ref{3.thm:5}}
\label{sec:3.3.3}

The Lie-algebra techniques of the previous section will now be extended to prove (\ref{thm.3.4f}), i.e.\ to show that $[ E_1(u) , \Sigma_w(u)  ] =[P_w(E_0(u),E_1(u)), E_1(u)]$.
For this purpose, we relate the letters $E_0(u), E_1(u)$ in (\ref{thm.3.4d})
to the letters $t_{01}, t_{12}$ in (\ref{lieg1.17}) seen in the definition
(\ref{not.06}) of $\Sigma_w(u)$:
\beq
e^{-ux} E_1(u) e^{ux} = t_{12} \, , \ \ \ \ \ \
e^{-ux} E_0(u) e^{ux} = t_{01} + \frac{N(0)}{u} 
\label{thm.3.31}
\eeq
see (\ref{notsec.21}) for the series of generators $N(u)$. The first identity of (\ref{thm.3.31})
is an immediate consequence of the expression (\ref{thm.3.4d}) for $E_1(u)$
together with $b_k = \ad_x^{k-2} b_2$ and $b_2 = t_{12}$.
The second identity of (\ref{thm.3.31}) is equivalent to
\beq
e^{-ux} N(u) e^{ux} = u \, t_{01}+ N(0)
\label{thm.3.32}
\eeq
since $N(u) = u E_0(u)$ which follows from the fact that both sides match at $u=0$ and
have the same $u$-derivative: the right side clearly reduces to $t_{01}$ under $\partial_u$
while the left side yields
\begin{align}
\partial_u \big( e^{-ux} N(u) e^{ux}  \big) &= e^{-ux} \, \Big({-}x  N(u)+  N(u) x +  \partial_u N(u) \Big) \, e^{ux}
\notag \\
&= e^{-ux} \, \Big({-}x  N(u)+  N(u) x + \ad_x N(u) + T_{01}(u) \Big) \, e^{ux} 
\notag \\
&= e^{-ux}\, T_{01}(u)  \,e^{ux} = t_{01}
\label{thm.3.33}
\end{align}
using the differential equation (\ref{thm.3.4z}) of $N(u)$ in passing to the second line.
This establishes (\ref{thm.3.32}) and therefore the second identity of (\ref{thm.3.31}).

With both identities of (\ref{thm.3.31}) in place, we can prove the equivalent
\beq
e^{-ux}\, \big[ E_1(u) , \Sigma_w(u) \big] \, e^{ux} =  e^{-ux} \, \big[ P_w\big(E_0(u), E_1(u) \big) , E_1(u) \big]\, e^{ux}
\label{thm.3.34}
\eeq
of the target identity (\ref{thm.3.4f}). The left side of (\ref{thm.3.34}) simplifies to
\beq
\big[ t_{12} , \Sigma_w(0) \big]  = \big[ t_{12} , P_w(t_{12},t_{01})+\sigma_w  \big] 
=  \big[ t_{12} ,P_w(t_{12},t_{01}) \big] 
\label{thm.3.35}
\eeq
using (\ref{thm.3.31}), the definition (\ref{not.06}) of $ \Sigma_w(u)$ 
and the defining property of $\sigma_w$ to commute with $t_{12}$.
The right side of (\ref{thm.3.34}) in turn can be rewritten as
\begin{align}
 \big[ P_w\big(e^{-ux}  E_0(u) e^{ux} ,\, e^{-ux}  E_1(u) e^{ux}  \big) , e^{-ux}  E_1(u)  e^{ux} \big]
 &=  \big[ P_w\big( t_{01} + \tfrac{1}{u} N(0),\, t_{12}  \big) , t_{12} \big] \notag \\
& =  \big[ P_w( t_{01},t_{12}) , t_{12} \big] 
\label{thm.3.36}
\end{align}
where we have first pulled the 
conjugation by $e^{-ux}$ into the letters of the Lie polynomials $P_w$,
then rewrote both instances of $e^{-ux}  E_i(u) e^{ux}$ via (\ref{thm.3.31}).
The second step is based on the fact that $N(0)$ commutes with both $t_{01}$
and $t_{12}$, see (\ref{thm.3.4v}), which guarantees that it drops out of the
Lie polynomials in $P_w\big( t_{01} {+} \tfrac{1}{u} N(0), t_{12}  \big)  = P_w( t_{01},t_{12})$.
Hence, the two sides of (\ref{thm.3.34}) in (\ref{thm.3.35}) and (\ref{thm.3.36}) match by $P_w( t_{01},t_{12})= - P_w( t_{12},t_{01})$
which concludes the proof of (\ref{thm.3.4f}).

%%%%%%%%%%%%%%%%%%%%%%%%%%%%%%%%%%%%%%%%%%%%%%%%%%
%%%%%%%%%%%%%%%%%%%%%%%%%%%%%%%%%%%%%%%%%%%%%%%%%%
\subsubsection{Step (iv) of proving Theorem \ref{3.thm:5}}
\label{sec:3.3.4}

The asymptotics (\ref{thm.3.4e})
of $\mathbb I_{\ep,b}(u,v,\tau) $ and the commutation relations 
(\ref{goaltwo}) and (\ref{thm.3.4f}) for $ [ E_i(u)  ,  \Sigma_w(u)  ]$ derived in the previous steps will now be combined to bring the asymptotics of $\mathbb I^{\rm eqv}_{\ep,b}(u,v,\tau) $ into the desired form. In the first place, (\ref{grteq.01})
gives rise to
\begin{align}
\mathbb I^{\rm eqv}_{\ep,b}(u,v,\tau) &=
(\mathbb M^{\rm sv}_{z})^{-1}\, 
e^{-2\pi i \bar \tau \ep_0} \, \overline{  \mathbb G_{E_0(u),E_1(u)}(e^{2\pi i z})^T } \, e^{2\pi i v E_0(u)}
 \, \mathbb M^{\rm sv}_{\Sigma(u)}
\, e^{-2\pi i v E_0(u)}\notag \\
&\quad \quad \quad \times \mathbb G_{E_0(u),E_1(u)}(e^{2\pi i z})\, e^{2\pi i \tau \ep_0}  + {\cal O}(q^{1-u},\bar q^{1-u})
\notag\\
&= (\mathbb M^{\rm sv}_{z})^{-1}\, 
e^{-2\pi i \bar \tau \ep_0} \, \overline{  \mathbb G_{E_0(u),E_1(u)}(e^{2\pi i z})^T } 
 \, \mathbb M^{\rm sv}_{\Sigma(u)}
 \, \mathbb G_{E_0(u),E_1(u)}(e^{2\pi i z})\, e^{2\pi i \tau \ep_0} \notag \\
&\quad + {\cal O}(q^{1-u},\bar q^{1-u})
\label{thm.3.4a3} 
\end{align}
where we have used that $ E_0(u) $ commutes with $\Sigma_w(u) $ to cancel the factors of $e^{\pm 2\pi i v E_0(u)}$. 

In the next step, it would be desirable to convert the following combination of series 
$ \overline{  \mathbb G_{E_0(u),E_1(u)}(e^{2\pi i z})^T }  \mathbb M^{\rm sv}_{\Sigma(u)}
 \mathbb G_{E_0(u),E_1(u)}(e^{2\pi i z})$ to the series $\mathbb M^{\rm sv}_{\Sigma(u)}
 \mathbb G^{\rm sv}_{E_0(u),E_1(u)}(e^{2\pi i z})$ in single-valued MPLs via (\ref{svmpl.12}).
 However, this is tied to the bracket relations between the
 augmented zeta generators $\Sigma_w(u)$ and $E_0(u), E_1(u)$ which need to mirror those among $M_w$ and $e_0,e_1$ in (\ref{svmpl.07}).
This is indeed the case by the
commutation relations 
(\ref{goaltwo}) and (\ref{thm.3.4f}) for $ [ E_i(u)  ,  \Sigma_w(u)  ]$ derived in the previous steps (ii) and (iii). Hence, we can identify the series in single-valued MPLs in (\ref{thm.3.4a3})
and bring the asymptotics of $\mathbb I^{\rm eqv}_{\ep,b}(u,v,\tau)$ into the form
\begin{align}
\mathbb I^{\rm eqv}_{\ep,b}(u,v,\tau) &= (\mathbb M^{\rm sv}_{z})^{-1}\, 
e^{-2\pi i \bar \tau \ep_0} \, \mathbb M^{\rm sv}_{\Sigma(u)}
 \, \mathbb G^{\rm sv}_{E_0(u),E_1(u)}(e^{2\pi i z})\, e^{2\pi i \tau \ep_0} + {\cal O}(q^{1-u},\bar q^{1-u})
\label{thm.3.4a4}
\end{align}
Finally, the $\Umod(\tau)$ conjugation (\ref{thm.3.4a}) leads to
\begin{align}
\mathbb H^{\rm eqv}_{\ep,b}(u,v,\tau) &= 
\exp\bigg( {-} \frac{\ep_0^\vee}{Y} \bigg)\,
e^{2\pi i \bar \tau \ep_0}\,
(\mathbb M^{\rm sv}_{z})^{-1}\, 
e^{-2\pi i \bar \tau \ep_0} \, \mathbb M^{\rm sv}_{\Sigma(u)}
 \label{thm.3.4a5} \\
 &\quad \times  \mathbb G^{\rm sv}_{E_0(u),E_1(u)}(e^{2\pi i z})\, e^{2\pi i \tau \ep_0} \,
e^{-2\pi i \bar \tau \ep_0}\,
\exp\bigg(  \frac{\ep_0^\vee}{Y} \bigg)
+ {\cal O}(q^{1-u},\bar q^{1-u})
\notag
\end{align}
where $Y= 4\pi \Im \tau$ allows us to merge $e^{2\pi i \tau \ep_0} 
e^{-2\pi i \bar \tau \ep_0} = e^{-Y \ep_0}$ in the second line.
By $\mathfrak{sl}_2$ invariance of the arithmetic zeta
generators, we have $e^{2\pi i \bar \tau \ep_0}
(\mathbb M^{\rm sv}_{z})^{-1}
e^{-2\pi i \bar \tau \ep_0} = (\mathbb M^{\rm sv}_{z})^{-1}$
such that (\ref{thm.3.4a5}) reproduces the desired expression (\ref{thm.3.4c})
and concludes the proof of Theorem \ref{3.thm:5}.

%%%%%%%%%%%%%%%%%%%%%%%%%%%%%%%%%%%%%%%%%%%%%%%%%%
%%%%%%%%%%%%%%%%%%%%%%%%%%%%%%%%%%%%%%%%%%%%%%%%%%
%%%%%%%%%%%%%%%%%%%%%%%%%%%%%%%%%%%%%%%%%%%%%%%%%%
%%%%%%%%%%%%%%%%%%%%%%%%%%%%%%%%%%%%%%%%%%%%%%%%%%
\section{Proof of Theorem \ref{3.thm:1}}
\label{sec:4}

This section is dedicated to the proof of Theorem \ref{3.thm:1}, stating the equivalence of the two expression (\ref{grteq.01}) and (\ref{altex.01}) for $\mathbb I^{\rm eqv}_{\ep,b}(u,v,\tau)$ which was needed to establish its equivariance. The claim to be proven can be restated as
\beq
\overline{ \mathbb I_{\ep,b}(u,v,\tau)^T} \, \mathbb M^{\rm sv}_{\Sigma(u)}\,
\mathbb I_{\ep,b}(u,v,\tau) 
= 
\overline{\mathbbm{\Gamma}_{x,y}(z,\tau)^T} \,
  \mathbb M^{\rm sv}_z\,
\mathbb I^{\rm eqv}_{\ep^{\rm TS}}(\tau)  \, 
\mathbbm{\Gamma}_{x,y}(z,\tau)
\label{restatcl}
\eeq
and amounts to matching iterated integrals over $\tau$ on the left side with iterated $z$-integrals on the right side. This requires a change of fibration basis which we shall perform by unifying the connections $\mathbb D_{\ep,b}(u,v,\tau_1) \sim \dd \tau_1$ in (\ref{defdepb}) governing $\mathbb I_{\ep,b}(u,v,\tau) $ on the left side and  $\tilde{\mathbb K}_{x,y}(z_1,\tau) \sim \dd z_1$ in (\ref{defKcon}) governing $\mathbbm{\Gamma}_{x,y}(z,\tau)$ on the right side of (\ref{restatcl}). More specifically, the proof of (\ref{restatcl}) and thus Theorem \ref{3.thm:1} is organized as follows:
\begin{itemize}
\item using a flat connection unifying $\dd \tau$ and $\dd z$ components to compute  differential equations of $\mathbb I_{\ep,b}(u,v,\tau) $ in $u$, $v$ or $z$, $\bar z$ in section \ref{sec:4.1};
\item integrating the differential equations of the previous point in terms of  the series $\mathbb I_{\ep^{\rm TS}}(\tau)\mathbbm{\Gamma}_{x,y}(z,\tau)$ in iterated Eisenstein integrals and eMPLs in section \ref{sec:4.2};
\item determining the $z$- and $\tau$-independent initial values in the solution $\sim \mathbb I_{\ep^{\rm TS}}(\tau)\mathbbm{\Gamma}_{x,y}(z,\tau)$ of the earlier differential equation in section \ref{sec:4.3};
\item assembling the second expression (\ref{altex.01}) for $\mathbb I^{\rm eqv}_{\ep,b}(u,v,\tau)$ including $\mathbb I^{\rm eqv}_{\ep^{\rm TS}}(\tau)$ in section \ref{sec:4.5}.
\end{itemize}

%%%%%%%%%%%%%%%%%%%%%%%%%%%%%%%%%%%%%%%%%%%%%%%%%%
%%%%%%%%%%%%%%%%%%%%%%%%%%%%%%%%%%%%%%%%%%%%%%%%%%
\subsection{Deriving differential equations of $\mathbb I_{\ep,b}(u,v,\tau) $ in $z$}
\label{sec:4.1}

The definition (\ref{notsec.17}) of $\mathbb I_{\ep,b}(u,v,\tau) $ as a path-ordered exponential  in $\tau$ exposes its differential equations in the modular parameter. In this section, we will derive the differential equations in the variables $u,v$ that are kept constant in the $\tau$-integrals.

%%%%%%%%%%%%%%%%%%%%%%%%%%%%%%%%%%%%%%%%%%%%%%%%%%
%%%%%%%%%%%%%%%%%%%%%%%%%%%%%%%%%%%%%%%%%%%%%%%%%%
\subsubsection{A doubly-periodic flat connection}
\label{sec:4.1.1}

The $u$- and $v$-derivatives of $\mathbb I_{\ep,b}(u,v,\tau) $ will be determined by means of a flat connection $ \mathbb J_{x,y,\ep}(u,v,\tau)$ involving
doubly-periodic $f^{(k)}(u\tau{+}v,\tau)$ kernels in all of its $\dd u$, $\dd v$ and $\dd \tau$ components. A connection with these properties can be constructed from the meromorphic but multi-valued CEE connection $\mathbb K_{x,y,\ep}(z,\tau)$ in (\ref{notsec.02}) by performing the gauge transformation by $e^{ux}$ that produced the Brown-Levin connection in  (\ref{notsecCEE}) from its $\dd z$ component. By eliminating $\dd z$
from the CEE connection via
\beq
\dd z = \dd(u\tau{+}v)  =u\, \dd\tau + \tau \, \dd u + \dd v\,, \ \ \ \ \ \
\dd \bar z = \dd(u\bar \tau{+}v)  =u\, \dd \bar \tau + \bar \tau \, \dd u + \dd v  
\label{sc5.01}
\eeq
 and considering $u,v,\tau,\bar \tau$ instead of $z,\bar z,\tau,\bar \tau$ as independent variables, the desired gauge transformed connection takes the form
\begin{align}
  \mathbb J_{x,y,\ep}(u,v,\tau) &=  e^{ux}  \mathbb K_{x,y,\ep}(z,\tau) e^{-ux} + (\dd e^{ux} ) e^{-ux}
 \label{notsec.03} \\
 &= x\, \dd u + (\tau \,\dd u+ \dd v)  \ad_x \Omega(u\tau{+}v,\ad_{\frac{x}{2\pi i}},\tau)y \notag \\
 &\quad 
+ \dd \tau \, {\rm ad}_x  \partial_{{\rm ad}_x} \Omega_0(u\tau{+}v,{\rm ad}_{ \frac{ x}{2\pi i}} ,\tau)y 
 + 2\pi i \dd \tau  \bigg( \epsilon_0+\sum_{k=2}^\infty \frac{(k{-}1)}{(2\pi i)^k} \,  {\rm G}_k(\tau)\epsilon_k
    \bigg)  \notag \\
&=  x\, \dd u + 2\pi i (\tau \,\dd u+ \dd v)  \sum_{k=0}^\infty \frac{ f^{(k)}(u\tau {+}v,\tau)}{(2\pi i)^k } \, {\rm ad}_{x}^{k}(y) \notag \\
&\quad + 2\pi i \dd \tau  \bigg( \epsilon_0+\sum_{k=2}^\infty \frac{(k{-}1)}{(2\pi i)^k} \, \big[ {\rm G}_k(\tau)\epsilon_k
 - f^{(k)}(u\tau {+} v,\tau) b_k \big]  \bigg) \notag
\end{align}
where $\Omega(z,\alpha,\tau)$ is the doubly-periodic Kronecker-Eisenstein series defined in (\ref{appA.11}) and $\Omega_0(z,\alpha,\tau)=\Omega(z,\alpha,\tau)-\frac{1}{\alpha}$.
In passing to the last line, we have used 
\begin{align}
\ad_x\partial_{\ad_x}\Omega_0(u\tau{+}v,\ad_{\frac{x}{2\pi i}},\tau)y&=-\sum_{k=2}^\infty (k{-}1) (2\pi i)^{1-k}  f^{(k)}(u\tau{+}v,\tau)b_k\notag\\
\ad_x \Omega(u\tau{+}v,\ad_{\frac{x}{2\pi i}},\tau)y&=\sum_{k=0}^\infty (2\pi i)^{1-k} f^{(k)}(u\tau{+}v,\tau)\ad_x^k(y)
\label{sc5.02}
\end{align}
As a gauge transform of the flat CEE connection $\mathbb K_{x,y,\epsilon}(z,\tau)$, the connection $\mathbb J_{x,y,\epsilon}(u,v,\tau)$ in (\ref{notsec.03}) is guaranteed to be flat (away from the singular points at $z \notin \mathbb Z {+} \tau \mathbb Z$): 
\beq
\dd \mathbb J_{x,y,\ep}(u,v,\tau) = \mathbb J_{x,y,\ep}(u,v,\tau) \wedge \mathbb J_{x,y,\ep}(u,v,\tau) 
\label{sc5.03}
\eeq
We decompose $\mathbb J_{x,y,\ep}(u,v,\tau)$ into components $\sim \dd u , \, \dd v, \, \dd \tau$,
\beq
\mathbb J_{x,y,\ep}(u,v,\tau)  = \dd u\,\mathbb J^{(u)}_{x,y,\ep}(u,v,\tau)  + \dd v \,
\mathbb J^{(v)}_{x,y,\ep}(u,v,\tau) +  \dd \tau\, 
\mathbb J^{(\tau)}_{x,y,\ep}(u,v,\tau) 
\label{sc5.04}
\eeq
which can be read off from (\ref{notsec.03}):
\begin{align}
\mathbb J^{(u)}_{x,y,\ep}(u,v,\tau) &= 
x + \tau  \, \ad_x \Omega(u\tau{+}v,\ad_{\frac{x}{2\pi i}},\tau)y
\label{sc5.05}
\\
\mathbb J^{(v)}_{x,y,\ep}(u,v,\tau) &=  \ad_x \Omega(u\tau{+}v,\ad_{\frac{x}{2\pi i}},\tau)y
\notag \\
\mathbb J^{(\tau)}_{x,y,\ep}(u,v,\tau) &= 2\pi i   \bigg( \epsilon_0+\sum_{k=2}^\infty \frac{(k{-}1)}{(2\pi i)^k} \, \big[ {\rm G}_k(\tau)\epsilon_k
 - f^{(k)}(u\tau {+} v,\tau) b_k \big]  \bigg)
\notag
 \end{align}
The flatness condition (\ref{sc5.03}) then implies three separate identities 
\begin{align}
\partial_\tau \mathbb J^{(v)}_{x,y,\ep}(u,v,\tau)-\partial_v \mathbb J^{(\tau)}_{x,y,\ep}(u,v,\tau)&=\big[\mathbb J^{(\tau)}_{x,y,\ep}(u,v,\tau),\mathbb J^{(v)}_{x,y,\ep}(u,v,\tau) \big]\label{flatness_1}\\
\partial_\tau \mathbb J^{(u)}_{x,y,\ep}(u,v,\tau)-\partial_u \mathbb J^{(\tau)}_{x,y,\ep}(u,v,\tau)&=\big[\mathbb J^{(\tau)}_{x,y,\ep}(u,v,\tau),\mathbb J^{(u)}_{x,y,\ep}(u,v,\tau) \big]\notag\\
\partial_u \mathbb J^{(v)}_{x,y,\ep}(u,v,\tau)-\partial_v \mathbb J^{(u)}_{x,y,\ep}(u,v,\tau)&=\big[\mathbb J^{(u)}_{x,y,\ep}(u,v,\tau),\mathbb J^{(v)}_{x,y,\ep}(u,v,\tau)\big]\notag
\end{align}
In contrast to the flatness condition of the CEE connection \cite{KZB}, the curls on the left side and the commutators on the right side typically do not vanish individually. However, the first line of (\ref{flatness_1}) is an exception, where the mixed heat equation (\ref{appA.32}) immediately implies the vanishing of the left side.\footnote{A direct proof of $\big[\mathbb J^{(\tau)}_{x,y,\ep}(u,v,\tau),\mathbb J^{(v)}_{x,y,\ep}(u,v,\tau) \big]=0$ on the right side of (\ref{flatness_1}) is more challenging and requires Fay identities as, for instance, in section 4.4 of \cite{Broedel:2020tmd} as well as the action of $\epsilon_k$ on $x,y$ as in section~\ref{sec:2.4}.} 

%%%%%%%%%%%%%%%%%%%%%%%%%%%%%%%%%%%%%%%%%%%%%%%%%%
%%%%%%%%%%%%%%%%%%%%%%%%%%%%%%%%%%%%%%%%%%%%%%%%%%
\subsubsection{Differentiating path-ordered exponentials}
\label{sec:4.1.2}

The $\dd \tau$-component of the flat connection  $\mathbb J_{x,y,\ep}(u,v,\tau)$ in the third line of (\ref{sc5.05}) was already encountered in (\ref{thm.3.4m}) and used in the alternative representation (\ref{thm.3.4l}) of $\mathbb I_{\ep,b}(u,v,\tau)$. We shall generalize this alternative path-ordered exponential further and introduce the series 
\beq 
\mathbb L_{x,y,\epsilon}(u,v;a,b)=\text{Pexp}\biggl(\int_b^a \dd \tau \, \mathbb J_{x,y,\ep}^{(\tau)}(u,v,\tau)\biggr)
\label{sc5.11}
\eeq
depending on two generic endpoints $a,b$ in the upper half plane and retrieve $\mathbb I_{\ep,b}(u,v,\tau)$ via (\ref{thm.3.4l}) as
\beq
\mathbb I_{\ep,b}(u,v,\tau) =
\mathbb L_{x,y,\epsilon}(u,v;i\infty,\tau) e^{2\pi i \tau \ep_0}
\label{sc5.12}
\eeq
We can now differentiate the path-ordered exponential (\ref{sc5.11}) with respect to $v$ or $u$ via
\begin{align}
\partial_u \mathbb L_{x,y,\epsilon}(u,v;i\infty, \tau)&=\int_\tau^{i\infty}\text{d}s\,\mathbb L_{x,y,\epsilon}(u,v;i\infty,s)\bigl(\partial_u \mathbb J_{x,y,\epsilon}^{(s)}(u,v,s)\bigr)\mathbb L_{x,y,\ep}(u,v;s,\tau)
\label{der_u} \\
&=\int_\tau^{i\infty}\text{d}s\,\mathbb L_{x,y,\epsilon}(u,v;i\infty,s)\bigl(\partial_s \mathbb J_{x,y,\epsilon}^{(u)}(u,v,s)\bigr)\mathbb L_{x,y,\ep}(u,v;s,\tau)\notag\\
&\quad+\int_\tau^{i\infty}\text{d}s\,\mathbb L_{x,y,\epsilon}(u,v;i\infty,s)\bigl[ \mathbb J_{x,y,\epsilon}^{(u)}(u,v,s),\mathbb J_{x,y,\epsilon}^{(s)}(u,v,s)\bigr]\mathbb L_{x,y,\ep}(u,v;s,\tau)\notag\\
&=\int_\tau^{i\infty}\text{d}s \, \partial_s\biggl(\mathbb L_{x,y,\epsilon}(u,v;i\infty,s)\mathbb J_{x,y,\epsilon}^{(u)}(u,v,s)\mathbb L_{x,y,\ep}(u,v;s,\tau)\biggr)\notag\\
&=\mathbb J^{(u)}_{x,y,\epsilon}(u,v,i\infty)\mathbb L_{x,y,\epsilon}(u,v;i\infty,\tau)-\mathbb L_{x,y,\epsilon}(u,v;i\infty,\tau)\mathbb J^{(u)}_{x,y,\epsilon}(u,v,\tau)
\notag
\end{align}
where the second and third line are derived from the flatness condition in  (\ref{flatness_1}). The total derivative in the fourth line arises from $\partial_s \mathbb L_{x,y,\epsilon}(u,v;i\infty,s) = - \mathbb L_{x,y,\epsilon}(u,v;i\infty,s) \mathbb J_{x,y,\epsilon}^{(s)}(u,v,s)$ and $\partial_s \mathbb L_{x,y,\ep}(u,v;s,\tau) = \mathbb J_{x,y,\epsilon}^{(s)}(u,v,s) \mathbb L_{x,y,\ep}(u,v;s,\tau)$ which produce the $\mathbb J_{x,y,\epsilon}^{(s)}$ terms of the commutator in the third line.
The derivative in $v$ is obtained from the same method:
\begin{align}\label{der_v}
\partial_v \mathbb L_{x,y,\epsilon}(u,v;i\infty,\tau)=\mathbb J^{(v)}_{x,y,\epsilon}(u,v,i\infty)\mathbb L_{x,y,\epsilon}(u,v;i\infty,\tau)-\mathbb L_{x,y,\epsilon}(u,v;i\infty,\tau)\mathbb J^{(v)}_{x,y,\epsilon}(u,v,\tau)
\end{align}
Both derivatives in
(\ref{der_u}), (\ref{der_v})
involve degeneration limits $\tau \rightarrow i\infty$ of the connection $\mathbb J_{x,y,\epsilon}$ at fixed $u,v$ which follows from inserting the leading terms
$f^{(k)}(u\tau{+}v,  \tau) \rightarrow  (2\pi i)^k \frac{ B_k(u) }{k!}$ and ${\rm G}_k(\tau) \rightarrow - (2\pi i)^k \frac{B_k}{k!}$ of their expansion around the cusp into (\ref{notsec.03}),
\begin{align}
\mathbb J^{(u)}_{x,y,\epsilon}(u,v,i\infty)&=x\label{bound_u}\\
\mathbb J^{(v)}_{x,y,\epsilon}(u,v,i\infty)&=-2\pi i T_{01}(u) \notag \\
\mathbb J^{(\tau)}_{x,y,\epsilon}(u,v,i\infty)&=-2\pi i N(u)
\notag
\end{align}
The first line is understood as a regularized limit where the second term $\sim \tau$ in the first line of (\ref{sc5.05}) is suppressed. The last two lines of (\ref{bound_u}) reproduce the $u$-dependent series $T_{01}(u)$ and $N(u)$ of generators defined in (\ref{notsec.21}), where the $\tau \rightarrow i\infty$ limit is understood to be taken at non-zero $u$.\footnote{Among the leading terms $f^{(k)}(u\tau{+}v,  \tau) = \dfrac{(2\pi i)^k}{(k{-}1)!} \Bigl( \frac{ B_k(u) }{k} + \dfrac{ u^{k-1} \, e^{2\pi i (u\tau+v)} }{ e^{2\pi i (u\tau+v)} {-} 1 }\Bigr)+\mathcal{O}(q^{1-u})$ noted in (\ref{thm.3.4o}), the second one is exponentially suppressed as $\tau \rightarrow i\infty$ by virtue of the numerator factor $e^{2\pi i (u\tau+v)}$.}

By combining the derivatives (\ref{der_u}), (\ref{der_v}) with the degeneration limits (\ref{bound_u}), the total differential of the path-ordered exponential (\ref{sc5.11}) becomes 
\begin{align}\label{djplus} 
\dd \mathbb L_{x,y,\ep}(u,v;i\infty,\tau) &= \big(x \, \dd u - 2\pi i T_{01}(u) \, \dd v \big)\, \mathbb L_{x,y,\ep}(u,v;i\infty,\tau)\notag\\
&\quad - \mathbb L_{x,y,\ep}(u,v;i\infty,\tau) \, \mathbb J_{x,y,\ep}(u,v,\tau)
\end{align}
where the right-multiplicative $\mathbb J^{(\tau)}_{x,y,\epsilon}$ from the $\tau$-derivative completes the flat connection~(\ref{notsec.03}).

%%%%%%%%%%%%%%%%%%%%%%%%%%%%%%%%%%%%%%%%%%%%%%%%%%
%%%%%%%%%%%%%%%%%%%%%%%%%%%%%%%%%%%%%%%%%%%%%%%%%%
\subsubsection{Total differential of the series $\mathbb I_{\ep,b}(u,v,\tau)$}
\label{sec:4.1.3}

The path-ordered exponentials $\mathbb L_{x,y,\ep}(u,v;i\infty,\tau)$ of the previous section and $\mathbb I_{\ep,b}(u,v,\tau)$ relevant to Theorem \ref{3.thm:1} are related by the right-multiplicative $e^{2\pi i \tau \ep_0}$ in (\ref{sc5.12}). Accordingly, the connections in their total differentials are related by the gauge transformation
\begin{align}
\mathbb D_{\ep,b}(u,v,\tau) &= \dd \tau\, e^{-2\pi i \ep_0 \tau} \mathbb J^{(\tau)}_{x,y,\ep}(u,v,\tau) 
e^{2\pi i \ep_0 \tau} + (\dd e^{-2\pi i \ep_0 \tau}) \,e^{2\pi i \ep_0 \tau}
\label{sc5.15}
 \end{align}
and similar conjugations of the $\dd u$ and $\dd v$ components of $\mathbb J_{x,y,\ep}(u,v,\tau)$ in (\ref{notsec.03}), see (\ref{defdepb}) for $\mathbb D_{\ep,b}(u,v,\tau)$. This translates the total differential (\ref{djplus})
into
\begin{align}
\dd \mathbb I_{\ep,b}(u,v,\tau) &= \big(x \, \dd u - 2\pi i T_{01}(u) \, \dd v \big)\, \mathbb I_{\ep,b}(u,v,\tau) \label{diplus} \\
&\quad  - \mathbb I_{\ep,b}(u,v,\tau)\, \mathbb D_{\ep,b}(u,v,\tau)
- \mathbb I_{\ep,b}(u,v,\tau) \, (x - 2\pi i \tau y) \, \dd u
\notag \\
&\quad - \mathbb I_{\ep,b}(u,v,\tau) \, e^{-2\pi i \tau\ep_0} \,
\ad_x \Omega(x,\ad_{\frac{x}{2\pi i} },\tau) \,  y \, 
e^{2\pi i \tau\ep_0} \,(\dd v + \tau \, \dd u)
\notag 
\end{align}
which also encodes the $z$- and $\bar z$- derivatives of $\mathbb I_{\ep,b}(u,v,\tau)$ by inverting the relation (\ref{sc5.01}) between  differentials $\dd u$, $\dd v$ and $\dd z$, $\dd \bar z$.
Note that the transposed complex conjugates of (\ref{djplus}) and (\ref{diplus}) can be found in appendix \ref{secdifc}.

%%%%%%%%%%%%%%%%%%%%%%%%%%%%%%%%%%%%%%%%%%%%%%%%%%
%%%%%%%%%%%%%%%%%%%%%%%%%%%%%%%%%%%%%%%%%%%%%%%%%%
\subsection{Solving differential equations of $\mathbb I_{\ep,b}(u,v,\tau) $ in $z$}
\label{sec:4.2}

The differential equations of both $\mathbb L_{x,y,\ep}(u,v;i\infty,\tau)$ in (\ref{djplus}) and $\mathbb I_{\ep,b}(u,v,\tau)$ in (\ref{diplus}) mix right-multiplicative connections $-\mathbb J_{x,y,\ep}(u,v,\tau)$ and its gauge transform by $e^{-2\pi i \tau\ep_0}$ with a left-multiplicative contribution $x \, \dd u - 2\pi i T_{01}(u) \, \dd v $. The latter can be absorbed into the total differential of a left-multiplicative 
\beq
\dd \Big( e^{ux} \, e^{-2\pi i v  t_{01}}  \Big) = 
\big(  x  \, \dd u - 2\pi i T_{01}(u) \, \dd v \big) \,\Big( e^{ux} \, e^{-2\pi i v  t_{01}}  \Big)
\label{sc5.16}
\eeq
where the relation $T_{01}(u)= e^{ux} t_{01} e^{-ux}$ of (\ref{notsec.21}) leads to the alternative form $e^{-2\pi i v T_{01}(u)}e^{ux} = e^{ux} e^{-2\pi i v t_{01}}$ of the exponentials. By passing to the ``reduced'' series $\mathbb I^{\rm red}_{\ep,b}(u,v,\tau)$ defined by
\begin{align}
 \mathbb I_{\ep,b}(u,v,\tau) &= 
e^{ux} \, e^{-2\pi i v  t_{01}} \,
 \mathbb I^{\rm red}_{\ep,b}(u,v,\tau)\, e^{2\pi i \tau \ep_0}
  \label{defired}
\end{align}
we take advantage of (\ref{sc5.16}) to eliminate the left-multiplicative terms from (\ref{djplus}) and to find the particularly simple differential equation
\beq
\text{d}
    \mathbb I^{\rm red}_{\ep,b}(u,v,\tau) =  -  \mathbb I^{\rm red}_{\ep,b}(u,v,\tau) \, \mathbb J_{x,y,\ep}(u,v,\tau)
\label{sc5.18}
\eeq

\subsubsection{Towards Brown-Levin eMPLs}
\label{sc:5.2.1}

Our next goal is to integrate the differential equation (\ref{sc5.18}) on a path in the moduli space of $(z,\tau)$ that retrieves eMPLs. For this purpose, it is advantageous to rewrite $\dd u$, $\dd v$ in terms of $\dd z$, $\dd \tau$ and their complex conjugates by inverting the relations of (\ref{sc5.01}). Indeed, this rewriting organizes the $\dd z$, $\dd \bar z$ components into the Brown-Levin connection (\ref{notsec.06})
\beq
\mathbb J_{x,y,\ep}(u,v,\tau) =
\mathbb J^{\rm BL}_{x,y}(z,\tau) + \mathbb T_{x,y,\ep}(z,\tau)
\label{sc5.17}
\eeq
at the cost of more lengthy $\dd \tau$, $\dd\bar \tau$ components
\begin{align}
 \mathbb T_{x,y,\ep}(z,\tau)  &= 2\pi i\text{d}\tau\, \biggl(\epsilon_0+\sum_{k=4}^\infty \dfrac{(k{-}1)}{(2\pi i)^k}\, {\rm G}_k(\tau)\epsilon_k-\sum_{k=2}^\infty \dfrac{(k{-}1)}{(2\pi i)^k} \, f^{(k)}(u\tau{+}v,\tau)b_k\biggr)\notag\\
    &\quad-\text{d}\tau \, u \, \ad_x\Omega(z,\ad_{\frac{x}{2\pi i}},\tau)y - u\, x\, \dfrac{\text{d}\tau -\text{d}\bar\tau}{\tau-\bar\tau}
    \label{sc5.19}
\end{align}
In the desired path-ordered exponential $\mathbbm{\Gamma}_{x,y}(z,\tau)$ of the Brown-Levin connection in (\ref{notsec.19}) and (\ref{restatcl}), the integration over $z$ is performed at generic but fixed $\tau$. This is obtained by
integrating the equivalent
\begin{align}
    \text{d}
    \mathbb I^{\rm red}_{\ep,b}(u,v,\tau)&= 
 -  \mathbb I^{\rm red}_{\ep,b}(u,v,\tau) \,\Big( \mathbb J^{\rm BL}_{x,y}(z,\tau) + \mathbb T_{x,y,\ep}(z,\tau) \Big)  
 \label{sc5.20}
\end{align}
of (\ref{sc5.18}) from $(0,i\infty)$ via $(0,\tau)$ to $(z,\tau)$.
This in turn necessitates the $z\rightarrow 0$ limit of the $\dd \tau$, $\dd \bar \tau$ components $\mathbb T_{x,y,\ep}(z,\tau)$ in (\ref{sc5.19}) which enters the left
path-ordered exponential in
\begin{align}
\mathbb I^{\rm red}_{\ep,b}(u,v,\tau) &=
\mathbb I^{\rm in}_{x,y,\ep} \,
{\rm Pexp}\bigg( \int^{i\infty}_\tau 
\lim_{z \rightarrow 0} \mathbb T_{x,y,\ep}(z,\tau_1) \bigg)\,
{\rm Pexp}\bigg( \int^{0}_z
\mathbb J^{\rm BL}_{x,y}(z_1,\tau) \bigg)
 \label{sc5.22}
\end{align}
for the path segment $(0,i\infty)$ via $(0,\tau)$ for the integration variables $(z_1,\tau_1)$. The notation $\mathbb I^{\rm in}_{x,y,\ep}$
refers to an initial value to be determined in section \ref{sec:4.3} below that does not depend on $z$ or $\tau$ but may be a series in arbitrary Lie-algebra generators.

\subsubsection{Restoring the $\tau$-dependence}
\label{sc:5.2.2}

The first path-ordered exponential on the right side of (\ref{sc5.22}) ensures a $\tau$-dependence of $\mathbb I^{\rm red}_{\ep,b}(u,v,\tau)$ as mandated by the differential equation (\ref{sc5.20}). The limit $z\rightarrow 0$ of the connection (\ref{sc5.19}) in the integrand is delicate and would be ill-defined for some of its contributions in isolation, namely $-\frac{\dd \tau}{2\pi i} f^{(2)}(z,\tau)b_2$ in the first line and $\dd \tau \, u f^{(1)}(z,\tau) b_2$ in the second line. However, their sum is found to be free of problematic terms $\sim \frac{u}{z}$ by converting the single-valued $f^{(k)}(z,\tau)$ to meromorphic
Kronecker-Eisenstein kernels $g^{(k)}(z,\tau)$ via (\ref{appA.12})
\begin{align}
 \lim_{z\rightarrow 0} \bigg( u f^{(1)}(z,\tau) 
-\frac{1}{2\pi i} \, f^{(2)}(z,\tau) \bigg) = 
 \lim_{z\rightarrow 0} \bigg( i \pi u^2
 -\frac{1}{2\pi i} \, g^{(2)}(z,\tau)\bigg) = \frac{1}{2\pi i}\, {\rm G}_2(\tau)
 \label{sc5.21}
\end{align}
where we have stripped off $\dd \tau \, b_2$ and identified the quasi-modular holomorphic Eisenstein series via ${\rm G}_2(\tau) = - g^{(2)}(0,\tau)$. With the resolution (\ref{sc5.21}) of the delicate parts of the $z\rightarrow 0$ limit of (\ref{sc5.19}) as well as $f^{(k)}(0,\tau) = - {\rm G}_k(\tau)$ for all $k\geq 3$, we find 
\begin{align}
\lim_{z\rightarrow 0}\mathbb T_{x,y,\ep}(z,\tau)
&= 2\pi i \dd \tau \, \bigg(\ep_0 + \frac{1}{(2\pi i)^2} \, {\rm G}_2(\tau) b_2 + \sum_{k=4}^\infty
\frac{(k{-}1)}{(2\pi i)^k} \,{\rm G}_k(\tau) \ep_k^{\rm TS} \bigg)  \label{gtrfD}\\
&=  e^{2\pi i \tau \ep_0} \,\bigg( \frac{\dd \tau}{2\pi i} \, {\rm G}_2(\tau) b_2 + \mathbb D_{\ep^{\rm TS}}(\tau) \bigg) \,e^{-2\pi i \tau \ep_0} + (\dd e^{2\pi i \tau \ep_0}) e^{-2\pi i \tau \ep_0} \notag
\end{align}
and recover a gauge transformed version of $\frac{\dd \tau}{2\pi i} {\rm G}_2(\tau) b_2$ plus the solely $\tau$-dependent connection $ \mathbb D_{\ep^{\rm TS}}(\tau)$ in (\ref{tsconn}) which was used for the construction of iterated Eisenstein integrals.
Indeed, by translating the gauge transformation in the second line of (\ref{gtrfD}) into a right-multiplicative $e^{-2\pi i \tau \ep_0}$ at the level of the path-ordered exponential, we arrive at
\begin{align}
{\rm Pexp}\bigg( \int^{i\infty}_\tau 
\lim_{z \rightarrow 0} \mathbb T_{x,y,\ep}(z,\tau_1) \bigg)
&= \mathbb I_{\ep^{\rm TS}}(\tau) \,
\exp \bigg( \int^{i\infty}_\tau \frac{\dd \tau_1}{2\pi i} \, {\rm G}_2(\tau_1) b_2 \bigg) \, e^{-2\pi i \tau \ep_0}
 \label{sc5.23}
\end{align}
using the vanishing of all $[b_2 ,\ep_k^{\rm TS}]$ to move the primitive of ${\rm G}_2(\tau)$ to the~right of $\mathbb I_{\ep^{\rm TS}}(\tau)$. 

\subsubsection{Assembling the $z$- and $\tau$-dependence}
\label{sc:5.2.3}

We are now in a position to assemble the path-ordered exponentials in (\ref{sc5.22}),
\begin{align}
\mathbb I^{\rm red}_{\ep,b}(u,v,\tau) &=
\mathbb I^{\rm in}_{x,y,\ep}\,
\mathbb I_{\ep^{\rm TS}}(\tau)
\,\exp \big( \ee{0}{2}{\tau} b_2 \big) \, e^{-2\pi i \tau \ep_0}
\, {\rm Pexp}\bigg( \int^{0}_z
\mathbb J^{\rm BL}_{x,y}(z_1,\tau) \bigg)
 \label{sc5.24} \\
&=
\mathbb I^{\rm in}_{x,y,\ep}
\, \mathbb I_{\ep^{\rm TS}}(\tau)  \, 
\exp \big( \ee{0}{2}{\tau} b_2 \big)\, 
{\rm Pexp}\bigg( \int^{0}_z
\mathbb J^{\rm BL}_{x-2\pi i\tau  y,y}(z_1,\tau) \bigg)\, e^{-2\pi i \tau \ep_0}
 \notag
\end{align}
with the primitive $\ee{0}{2}{\tau}$ of ${\rm G}_2(\tau)$ defined in (\ref{nmodg2}). In passing to the second line, the factor of $e^{-2\pi i \tau \ep_0}$ has been moved to the rightmost position which shifts the letters $(x,y)$ of the Brown-Levin connection to $(x-2\pi i \tau y , y)$. When furthermore inserting the reduced series into the original one in (\ref{defired}), we arrive at the close-to-final form
\begin{align}
\mathbb I_{\ep,b}(u,v,\tau) &=  e^{ux}\, e^{-2\pi i v  t_{01}}\,
\mathbb I^{\rm in}_{x,y,\ep}\,
\mathbb I_{\ep^{\rm TS}}(\tau)\,  
\exp \big( \ee{0}{2}{\tau} b_2 \big)\, 
{\rm Pexp}\bigg( \int^{0}_z
\mathbb J^{\rm BL}_{x-2\pi i\tau y,y}(z_1,\tau) \bigg)  \notag \\
&=  e^{ux}\, e^{-2\pi i v  t_{01}} \,
\mathbb I^{\rm in}_{x,y,\ep}\,
\mathbb I_{\ep^{\rm TS}}(\tau)\,
\mathbbm{\Gamma}_{x,y}(z,\tau)
 \label{sc5.25}
\end{align}
by identifying the series $\mathbbm{\Gamma}_{x,y}(z,\tau)$ in (\ref{notsec.19}) where it only remains to determine the initial value $\mathbb I^{\rm in}_{x,y,\ep}$. 

%%%%%%%%%%%%%%%%%%%%%%%%%%%%%%%%%%%%%%%%%%%%%%%%%%
%%%%%%%%%%%%%%%%%%%%%%%%%%%%%%%%%%%%%%%%%%%%%%%%%%
\subsection{Fixing the initial value}
\label{sec:4.3}

The expression (\ref{sc5.25}) captures the complete $z$- and $\tau$-dependence of $\mathbb I_{\ep,b}(u,v,\tau)$ but leaves an undetermined initial value $\mathbb I^{\rm in}_{x,y,\ep}$ in the middle of the concatenation product. We shall here demonstrate that the initial value is given in terms of the Drinfeld associator (\ref{notsec.11}) by
\beq
\mathbb I^{\rm in}_{x,y,\ep} = \Phinew^{-1}(t_{01},t_{12}) = \Phinew(t_{12},t_{01})
\label{inisc.01}
\eeq

%%%%%%%%%%%%%%%%%%%%%%%%%%%%%%%%%%%%%%%%%%%%%%%%%%
%%%%%%%%%%%%%%%%%%%%%%%%%%%%%%%%%%%%%%%%%%%%%%%%%%
\subsubsection{The consistency condition}
\label{sec:4.3.1}

The expression (\ref{sc5.22}) for $\mathbb I^{\rm red}_{\ep,b}(u,v,\tau)$ defining the initial value has been obtained by integrating its differential equation (\ref{sc5.18}) from $(0,i\infty)$ via $(0,\tau)$ to $(z,\tau)$. By flatness of the gauge-transformed CEE connection $\mathbb J_{x,y,\ep}(z,\tau)$ in (\ref{sc5.18}), the series $\mathbb I^{\rm red}_{\ep,b}(u,v,\tau)$ must equivalently arise from alternative integration paths in the same homotopy class. 
In real co-moving coordinates $u,v$ with $0<u,v\leq 1$, we now integrate along the alternative path from $(0,i\infty)$ via $(u,v,i\infty)$ to $(u,v,\tau)$ with the same (as of yet unknown) initial conditions where
\begin{align}
\mathbb I^{\rm red}_{\ep,b}(u,v,\tau) &= \mathbb I^{\rm in}_{x,y,\ep} \,
  {\rm Pexp}\bigg( \int^{(0,0,i\infty)}_{(u,v,i\infty)}\mathbb J_{x,y,\ep}(u_1,v_1,\tau_1) \ \bigg)  
\,
{\rm Pexp}\bigg( \int^{i\infty}_\tau \dd \tau_1 \,
\mathbb J^{(\tau)}_{x,y,\ep}(u,v,\tau_1) \bigg)
\notag \\
&= \mathbb I^{\rm in}_{x,y,\ep} \,
 \lim_{\tau \rightarrow i\infty} {\rm Pexp}\bigg( \int^{0}_{u\tau+v} 
\tilde{\mathbb K}_{x,y}(z_1,\tau) \bigg)\, e^{-ux}
\, \mathbb I_{\ep,b}(u,v,\tau) \,e^{-2\pi i \tau \ep_0}
\label{inisc.02}
\end{align}
using (\ref{thm.3.4l}) for the rewriting of the
rightmost path-ordered exponential in the first line. The left path-ordered exponential is performed at $\dd \tau_1= \dd \bar \tau_1 = 0$ where $\mathbb J_{x,y}(u_1,v_1,\tau_1)$ reduces to the Brown-Levin connection by (\ref{sc5.17}). The latter is gauge equivalent to the $\dd z_1$ component $\tilde{\mathbb K}_{x,y}(z_1,\tau)$ of the meromorphic CEE connection in (\ref{defKcon}), so the respective path-ordered exponentials are related by a right-multiplicative factor of $e^{-ux}$ as in (\ref{gaugepexp}) that leads to the second line of (\ref{inisc.02}).

Comparing the alternative representation (\ref{inisc.02}) with the defining equation (\ref{defired}) for $\mathbb I^{\rm red}_{\ep,b}(u,v,\tau)$, we find the consistency condition
\beq
 \mathbb I^{\rm in}_{x,y,\ep} \,
%{\rm Pexp}\bigg( \int^{(0,0)}_{(u,v)}\mathbb J^{i\infty}_{x,y}(u_1,v_1) \bigg)
 \lim_{\tau \rightarrow i\infty} {\rm Pexp}\bigg( \int^{0}_{u\tau+v} 
\tilde{\mathbb K}_{x,y}(z_1,\tau) \bigg) \, 
 e^{-2\pi i v  t_{01}}  = 1
\label{inisc.03}
\eeq
which fixes the initial value $\mathbb I^{\rm in}_{x,y,\ep}$ in terms of the path-ordered exponential of the connection $\tilde{\mathbb K}_{x,y}$ at the cusp.

%%%%%%%%%%%%%%%%%%%%%%%%%%%%%%%%%%%%%%%%%%%%%%%%%%
%%%%%%%%%%%%%%%%%%%%%%%%%%%%%%%%%%%%%%%%%%%%%%%%%%
\subsubsection{Evaluating the path-ordered exponential at the cusp}
\label{sec:4.3.2}

We shall now evaluate the $\tau \rightarrow i\infty$ limit of the following path-ordered exponential that determines the initial value $\mathbb I^{\rm in}_{x,y,\ep}$ through the consistency condition (\ref{inisc.03}):
\begin{align} 
{\rm Pexp}\bigg( \int^{0}_{z} \lim_{\tau \rightarrow i\infty} 
\tilde{\mathbb K}_{x,y}(z_1,\tau) \bigg) &= 
{\rm Pexp}\bigg( - \int^{1}_{\sigma} \dd \sigma_1 \, \bigg(\frac{t_{12}}{\sigma_1{-}1} + \frac{t_{01}}{\sigma_1}  \bigg) \bigg)\notag \\
&= \Phinew^{-1}(t_{01},t_{12}) \,  \mathbb G_{t_{01},t_{12}}(e^{2\pi i z}) \label{tkpexp.01}
\end{align}
The first step is based on the degeneration (\ref{lieg1.16}) of the $\dd z$ component of the meromorphic CEE connection with coordinate $\sigma_1 = e^{2\pi i z_1}$ and generators $t_{12},t_{01}$ given by (\ref{lieg1.17}). In the second step of (\ref{tkpexp.01}), we have used the composition-of-paths formula to decompose ${\rm Pexp}(\int^1_\sigma  \tilde{\mathbb K}_{x,y} )
={\rm Pexp}(\int^1_0  \tilde{\mathbb K}_{x,y} ) {\rm Pexp}(\int^0_\sigma \tilde{\mathbb K}_{x,y} )$ and then identified the Drinfeld associator and the MPL series $\mathbb G_{t_{01},t_{12}}$ via (\ref{notsec.11}) and (\ref{svmpl.00}), respectively. 

The limit $\tau \rightarrow i\infty$ in (\ref{inisc.03}) is in fact taken at fixed co-moving coordinates $u,v \in \mathbb R$ of the endpoint $z=u\tau{+}v$ of the integration path in (\ref{tkpexp.01}). Hence, it remains to evaluate
\beq
\lim_{\tau \rightarrow i\infty} \mathbb G_{t_{01},t_{12}}(e^{2\pi i (u\tau{+}v)}) 
= \lim_{\tau \rightarrow i\infty}  e^{2\pi i (u\tau{+}v) t_{01}} = e^{2\pi i  v t_{01}}
\label{tkpexp.02}
\eeq
where the first step exploits the fact that shuffle-regularized MPLs $G(a_1,\ldots,a_r;\sigma)$ with $G(0;\sigma) = \log(\sigma)$ are suppressed by at least one power of $\sigma = e^{2\pi i (u\tau{+}v)} = q^u e^{2\pi i v}$ as $q\rightarrow 0$ if one or more of the $a_1,\ldots, a_r$ are non-zero. In the second step of (\ref{tkpexp.02}), we have taken the regularized $\tau \rightarrow i\infty$ limit suppressing any positive power of $\tau$ along with any word in generators $t_{01}$ or~$x,y$.

Inserting the combination of (\ref{tkpexp.01}) and (\ref{tkpexp.02}) into (\ref{inisc.03}) simplifies the consistency condition to
\beq
 \mathbb I^{\rm in}_{x,y,\ep} \, \Phinew^{-1}(t_{01},t_{12}) = 1
\label{altkpexp.03}
\eeq
and thus determines the initial value to be the Drinfeld associator as previewed in (\ref{inisc.01}).

%%%%%%%%%%%%%%%%%%%%%%%%%%%%%%%%%%%%%%%%%%%%%%%%%%
%%%%%%%%%%%%%%%%%%%%%%%%%%%%%%%%%%%%%%%%%%%%%%%%%%
\subsubsection{Assembling the complete change of fibration basis}
\label{sec:4.3.3}

With the $z,\tau$ dependence of $\mathbb I_{\ep,b}(u,v,\tau) $ in (\ref{sc5.25}) and the result (\ref{inisc.01}) for the initial value, our final expression is
\begin{align}
\mathbb I_{\ep,b}(u,v,\tau) &=  e^{ux}  \,  e^{-2\pi i v  t_{01}} \, \Phinew^{-1}(  t_{12}, t_{01} )\,  \mathbb I_{\ep^{\rm TS}}(\tau) \,
\mathbbm{\Gamma}_{x,y}(z,\tau)
\label{soliplus} \\
&=
e^{-2\pi i v T_{01}(u)} \, e^{ux} \, \Phinew^{-1}(  t_{12}, t_{01} )\,  \mathbb I_{\ep^{\rm TS}}(\tau) \,
\mathbbm{\Gamma}_{x,y}(z,\tau)
\notag
\end{align}
This is the advertised change of fibration basis converting the iterated $\tau$ integrals in the definition (\ref{notsec.17}) of $\mathbb I_{\ep,b}(u,v,\tau)$ to eMPLs generated by $\mathbbm{\Gamma}_{x,y}(z,\tau)$. The second line is obtained from the aforementioned identity $e^{-2\pi i v T_{01}(u)}e^{ux} = e^{ux} e^{-2\pi i v t_{01}}$ due to (\ref{notsec.21}) and will be more convenient for the next section.

%%%%%%%%%%%%%%%%%%%%%%%%%%%%%%%%%%%%%%%%%%%%%%%%%%
%%%%%%%%%%%%%%%%%%%%%%%%%%%%%%%%%%%%%%%%%%%%%%%%%%
\subsection{Producing the $\mathbb I^{\rm eqv}_{\ep^{\rm TS}}(\tau)$ constituent of 
$\mathbb I^{\rm eqv}_{\ep,b}(u,v,\tau)$}
\label{sec:4.5}

This section finishes the proof of Theorem \ref{3.thm:1} by inserting the change of fibration basis (\ref{soliplus}) into (\ref{grteq.01}). 
Its right-hand side then takes the form
\begin{align}
\mathbb I^{\rm eqv}_{\ep,b}(u,v,\tau)  &= 
(\mathbb M^{\rm sv}_{z})^{-1}\, 
\overline{\mathbbm{\Gamma}_{x,y}(z,\tau)^T} \, \overline{ \mathbb I_{\ep^{\rm TS}}(\tau)^T} \, \overline{ \Phinew^{-1}(  t_{12}, t_{01} )^T}\, e^{-ux} \,
e^{2\pi i v T_{01}(u)}  \notag \\
&\quad \times \mathbb M^{\rm sv}_{\Sigma(u)}\,
e^{-2\pi i v T_{01}(u)} \, e^{ux} \, \Phinew^{-1}(  t_{12}, t_{01} )\,  \mathbb I_{\ep^{\rm TS}}(\tau) \,
\mathbbm{\Gamma}_{x,y}(z,\tau)
\label{tkpexp.03}
\end{align}
which we shall here match with the right side of (\ref{altex.01}).

%%%%%%%%%%%%%%%%%%%%%%%%%%%%%%%%%%%%%%%%%%%%%%%%%%
%%%%%%%%%%%%%%%%%%%%%%%%%%%%%%%%%%%%%%%%%%%%%%%%%%
\subsubsection{Simplifying the MZV sector}
\label{sec:4.5.1}

A first step in simplifying the series between the Drinfeld associators in (\ref{tkpexp.03}) is to note that
\beq
e^{-ux}\,
e^{2\pi i v  T_{01}(u)} \, \mathbb M^{\rm sv}_{\Sigma(u)} \, e^{-2\pi i v  T_{01}(u)} \, e^{ux} = \mathbb M^{\rm sv}_{\Sigma(0)} 
\label{tkpexp.04}
\eeq
Since $T_{01}(u)$ commutes with the augmented zeta generators $\Sigma_w(u)$ in $\mathbb M^{\rm sv}_{\Sigma(u)}$, see section \ref{sec:3.3.2}, the exponentials of $\pm 2\pi i v T_{01}(u)$ on the left side of (\ref{tkpexp.04}) cancel. The simple $u$-dependence in the definition (\ref{not.06}) of augmented zeta generators propagates to the entire series $\mathbb M^{\rm sv}_{\Sigma(u)} = e^{ux} \mathbb M^{\rm sv}_{\Sigma(0)} e^{-ux}$ and explains the cancellation of $e^{\pm ux}$ from (\ref{tkpexp.04}).

The transposition properties of $x,y$ in (\ref{lieg1.32}) imply that both of $t_{12},t_{01}$ in (\ref{lieg1.17}) are even in the sense that $t_{12}^T=t_{12}$ and $t_{01}^T=t_{01}$. As a result, the transposition of the Drinfeld associator $\overline{ \Phinew^{-1}(  t_{12}, t_{01} )^T}$ compensates for the inversion up to an alternating minus sign with the number of letters,
\beq
\overline{ \Phinew^{-1}(  t_{12}, t_{01} )^T} = \Phinew( {-} t_{12},{-} t_{01} )
\label{tkpexp.05}
\eeq
using its reality to remove the complex-conjugation bar. The simplifications of (\ref{tkpexp.04}) and (\ref{tkpexp.05}) cast (\ref{tkpexp.03}) into the form of
\begin{align}
\mathbb I^{\rm eqv}_{\ep,b}(u,v,\tau)  &= 
(\mathbb M^{\rm sv}_{z})^{-1}\, 
\overline{\mathbbm{\Gamma}_{x,y}(z,\tau)^T} \, \overline{ \mathbb I_{\ep^{\rm TS}}(\tau)^T} \,  \Phinew( {-} t_{12}, {-} t_{01} )  \, \mathbb M^{\rm sv}_{\Sigma(0)}  \, \Phinew^{-1}(  t_{12}, t_{01} )\,  \mathbb I_{\ep^{\rm TS}}(\tau) \,
\mathbbm{\Gamma}_{x,y}(z,\tau)
\label{tkpexp.06}
\end{align}
The next simplification step is based on the following lemma:

\begin{lemma}
\label{sigphilem}
The series $\mathbb M^{\rm sv}_{\bullet}$ in (\ref{not.05}) with genus-one zeta generators $\sigma_w$ and their augmented versions $\Sigma_w(0) = P_w(t_{12},t_{01}){+}\sigma_w $ in (\ref{not.06}) are related by
\beq
 \Phinew( {-} t_{12}, {-} t_{01} )  \, \mathbb M^{\rm sv}_{\Sigma(0)}  \, \Phinew^{-1}(  t_{12}, t_{01} )
= \mathbb M^{\rm sv}_{\sigma}
\label{tkpexp.07s}
\eeq
\end{lemma}
The lemma is proven in appendix \ref{app:phi} and simplifies (\ref{tkpexp.06}) to
\begin{align}
\mathbb I^{\rm eqv}_{\ep,b}(u,v,\tau)
&= (\mathbb M^{\rm sv}_z)^{-1}
\overline{\mathbbm{\Gamma}_{x,y}(z,\tau)^T} \,
 \overline{\mathbb I_{\ep^{\rm TS}}(\tau)^T} \,
\mathbb M^{\rm sv}_{\sigma}\,
\mathbb I_{\ep^{\rm TS}}(\tau)\, 
\mathbbm{\Gamma}_{x,y}(z,\tau)
\label{tkpexp.00}
\end{align}

%%%%%%%%%%%%%%%%%%%%%%%%%%%%%%%%%%%%%%%%%%%%%%%%%%
%%%%%%%%%%%%%%%%%%%%%%%%%%%%%%%%%%%%%%%%%%%%%%%%%%
\subsubsection{Finishing the proof}
\label{sec:4.5.2}

Now that the first expression (\ref{grteq.01}) for $\mathbb I^{\rm eqv}_{\ep,b}(u,v,\tau)$ has been brought into the form of (\ref{tkpexp.00}), it is straightforward to match it with the target expression in (\ref{altex.01}): Rewriting three of the middle terms of (\ref{tkpexp.00}) as
\beq
\overline{ \mathbb I_{\ep^{\rm TS}}(\tau)^T}\,  
\mathbb M^{\rm sv}_\sigma \, \mathbb I_{\ep^{\rm TS}}(\tau)
=
\mathbb M^{\rm sv}_z \, \mathbb I^{\rm eqv}_{\ep^{\rm TS}}(\tau) 
\eeq
using (\ref{lieg1.51}) clearly reproduces (\ref{altex.01}) which completes the  proof of Theorem \ref{3.thm:1} and the objective of this section.

%%%%%%%%%%%%%%%%%%%%%%%%%%%%%%%%%%%%%%%%%%%%%%%%%%
%%%%%%%%%%%%%%%%%%%%%%%%%%%%%%%%%%%%%%%%%%%%%%%%%%
%%%%%%%%%%%%%%%%%%%%%%%%%%%%%%%%%%%%%%%%%%%%%%%%%%
%%%%%%%%%%%%%%%%%%%%%%%%%%%%%%%%%%%%%%%%%%%%%%%%%%
\section{Connection with eMGFs}
\label{sec:5}

In this section, we connect the equivariant iterated integrals and single-valued eMPLs constructed in section \ref{sec:3.2} with eMGFs which were reviewed in section \ref{sec:2.2}. A particularly convenient flavor of iterated integrals in modular frame is identified in section \ref{sec:itinemgf}. We then display the resulting expressions for several eMGFs in one variable in section \ref{sec:6.big2} and comment on connections with the literature. Finally, section \ref{sec:6.4} describes a group-theoretic method to infer the counting of independent eMGFs in one variable from shuffle-independent single-valued eMPLs.

\subsection{Iterated integrals for eMGFs}
\label{sec:itinemgf}

The equivariant and single-valued generating series of section \ref{sec:3} were analyzed in both holomorphic frame -- $\mathbb I^{\rm eqv}_{\ep,b}(u,v,\tau)$ in (\ref{grteq.01}) or $\mathbbm \Gamma_{x,y}^{\rm sv}(z,\tau)$ in (\ref{cor.3.4a}) -- and in modular frame -- $\mathbb H^{\rm eqv}_{\ep,b}(u,v,\tau)$ in (\ref{thm.3.4a}) or $\mathbbm \Lambda_{x,y}^{\rm sv}(z,\tau)$ in (\ref{cor.3.4c}). 
While the composing iterated integrals over modular parameters in holomorphic frame are indeed (anti)holomorphic in $\tau$, the series coefficients of $\mathbb H^{\rm eqv}_{\ep,b}(u,v,\tau)$ in (\ref{thm.3.4a}) or $\mathbbm \Lambda_{x,y}^{\rm sv}(z,\tau)$ in modular frame transform as modular forms of ${\rm SL}_2(\mathbb Z)$, see for instance (\ref{cor.3.4e}) and (\ref{cor.3.4f}).

Particularly natural iterated-integral representations of eMGFs arise from the above series in modular frame, since
\begin{itemize}
\item eMGFs individually transform as modular forms, as spelled out in (\ref{modeMGF}) for their dihedral cases and manifested by their surface-integral representations (\ref{genint}) for arbitrary graph topologies;
\item the $\tau$ derivatives of eMGFs \cite{DHoker:2016mwo, Dhoker:2020gdz} only produce kernels $\dd \tau \, {\rm G}_k$ and $\dd \tau \, f^{(k)}(u\tau{+}v,\tau)$ in (\ref{emgfkers}) and their complex conjugates, without any admixtures of independent powers $\tau^{j=1,\ldots,k-2}$ that one would get from the iterated integrals ${\cal E}[\ldots;\tau]$ in (\ref{bsc.07}) of the series in holomorphic frame in~(\ref{notsec.17}).
\end{itemize}
This section \ref{sec:itinemgf} offers a more detailed description of the iterated-integral constituents of eMGFs in modular frame.

\subsubsection{Integration kernels in modular frame}
\label{sec:itemgf.1}

We start by spelling out the integration kernels $\sim \dd \tau,\, \dd \bar \tau$ that arise from conjugating the connection $ {\mathbb D}_{\ep,b}(u,v,\tau_1) $ in (\ref{defdepb}) by the $\exp(\mathfrak{sl}_2)$ transformation $\Umod(\tau)$ towards modular frame in (\ref{usl2}). By distinguishing the integration variable $\tau_1$ of the path-ordered exponential (\ref{notsec.17}) from the endpoint $\tau$ of its path entering $\Umod(\tau)$, the Lie-algebra generators $b_k^{(j)}$ and $\ep_k^{(j)}$ in (recall that $\ep_2=0$)
\beq
\Umod(\tau) \,  {\mathbb D}_{\ep,b}(u,v,\tau_1)  \,\Umod(\tau)^{-1}
= \sum_{k=2}^\infty (k{-}1)\sum_{j=0}^{k-2} \frac{(-1)^j}{j!} \, \bigg\{
\wpz{j \\ k \\ z}{\tau_1}\, b_k^{(j)}
+ \wpz{j \\ k }{\tau_1} \ep_k^{(j)}
\bigg\}
\label{intemgf.01}
\eeq
are now accompanied by $z$-dependent one-forms in $\tau_1$ scalar in $\tau$,
\begin{align}
\wpz{j \\ k \\ z}{\tau_1}&=- \frac{\text{d}\tau_1}{2\pi i} \,(-1)^{k} \,\bigg( \frac{\tau{-}\tau_1}{ 4\pi \Im \tau } \bigg)^{k-2-j} (\bar\tau{-}\tau_1)^{j}\, f^{(k)} (u\tau_1{+}v,\tau_1) \notag \\
\wmz{j \\ k \\ z}{\tau_1}&=\frac{ \text{d}\bar{\tau}_1}{2\pi i } \,\bigg( \frac{\tau{-}\bar\tau_1}{4\pi \Im \tau}  \bigg)^{k-2-j} (\bar\tau{-}\bar\tau_1)^{j} \,\overline{f^{(k)}(u \tau_1{+}v,\tau_1)}
\label{intemgf.02}
\end{align}
The modular Eisenstein kernels $\omega_{\pm}[  \smallmatrix j \\ k\endsmallmatrix;\tau,\tau_1]$ with the same form degrees known from~\cite{Dorigoni_2022}
\begin{align}
\wpz{j \\ k }{\tau_1}&= \frac{\text{d}\tau_1}{2\pi i} \, \bigg( \frac{\tau{-}\tau_1}{ 4\pi \Im \tau } \bigg)^{k-2-j} (\bar\tau{-}\tau_1)^{j} \, {\rm G}_k(\tau_1) \notag \\
\wmz{j \\ k}{\tau_1}&=-\frac{ \text{d}\bar{\tau}_1}{2\pi i } \,\bigg( \frac{\tau{-}\bar\tau_1}{4\pi \Im \tau}  \bigg)^{k-2-j} (\bar\tau{-}\bar\tau_1)^{j} \, \overline{{\rm G}_k(\tau_1)}
\label{intemgf.03}
\end{align}
For both cases, we have also defined $(0,1)$-forms $\omega_-[\ldots;\tau,\tau_1]$ in $\tau_1$ that will be relevant to the image of the complex conjugate series $\overline{\mathbb I_{\ep,b}(u,v,\tau)^T}$ in modular frame. 
Note that the replacement of integration kernels $f^{(k)} (u\tau_1{+}v,\tau_1) \rightarrow -{\rm G}_k(\tau_1)$ with a minus sign in passing from $\wpmz{j \\ k \\ z}{\tau_\ell} $ to the analogue $\wpmz{j \\ k }{\tau_\ell}$ with an empty third line is chosen in view of $f^{(k)} (0,\tau_1) = -{\rm G}_k(\tau_1)$ for $k\geq 3$.

\subsubsection{Iterated integrals in modular frame}
\label{sec:itemgf.2}

Upon transformation to modular frame, the series $\mathbb I_{\ep,b}(u,v,\tau)$ and $\overline{\mathbb I_{\ep,b}(u,v,\tau)^T}$ generate iterated integrals of the one-forms (\ref{intemgf.02}) and (\ref{intemgf.03}) which we shall denote by
\begin{align}
\bplusz{j_1 &j_2&\cdots&j_\ell}{k_1 &k_2&\cdots&k_\ell}{z &z&\cdots&z}{\tau}=&\int_{\tau}^{i \infty}\wpz{j_\ell \\ k_\ell \\ z}{\tau_\ell}\cdots\int_{\tau_3}^{i \infty}\wpz{j_2 \\ k_2 \\ z}{\tau_2}\int_{\tau_2}^{i \infty}\wpz{j_1 \\ k_1 \\ z}{\tau_1} \label{intemgf.04}\\
%%%%%
%%%%%
\bminusz{j_1 &j_2&\cdots&j_\ell}{k_1 &k_2&\cdots&k_\ell}{z &z&\cdots&z}{\tau}=&\int_{\bar\tau}^{-i\infty}\wmz{j_\ell \\ k_\ell \\ z}{\tau_\ell}\cdots\int_{\bar\tau_3}^{-i\infty}\wmz{j_2 \\ k_2 \\ z}{\tau_2}\int_{\bar\tau_2}^{-i\infty}\wmz{j_1 \\ k_1 \\ z}{\tau_1} \notag
\end{align}
and symbols like $\bpmz{j_1 &\cdots&j_{\ell-1}&j_\ell}{k_1 &\cdots&k_{\ell-1}&k_\ell}{z &\cdots&z&}{\tau}$ with empty slots in the third line if some of the integration kernels are replaced by their $z$-independent counterparts, $\wpmz{j_i \\ k_i \\ z}{\tau_i} \rightarrow \wpmz{j_i \\ k_i }{\tau_i}$. By the summation range in (\ref{intemgf.01}) and $\ep_2=0$, the entries in (\ref{intemgf.04}) are understood to take values $k_i\geq 2$ and $0\leq j_i\leq k_i{-}2$, further restricted to $k_i\geq 4$ even in $z$-independent columns.

The advantage of the one-forms (\ref{intemgf.02}), (\ref{intemgf.03}) and their iterated integrals (\ref{intemgf.04}) for applications to eMGFs is that their $\tau$-derivatives directly produce the kernels (\ref{emgfkers}) in their differential equations (avoiding separate appearance of $\tau^j f^{(k)}(u\tau{+}v,\tau)$ for different $0\leq j\leq k{-}2$),
\begin{align}
2\pi i(\tau{-}\bar \tau)^2 \partial_\tau  \bplusz{j_1 &j_2& \ldots &j_\ell}{k_1 &k_2 &\ldots &k_\ell}{z&z&\ldots&z}{\tau} &= \sum_{i=1}^\ell (k_i{-}j_i{-}2) \bplusz{j_1 & \ldots &j_i+1 &\ldots &j_\ell}{k_1 &\ldots &k_i &\ldots &k_\ell}{z&\ldots&z&\ldots&z}{\tau}\label{dtauzbeta} \\
&\quad  + \delta_{j_\ell,k_\ell-2} (\tau{-}\bar \tau)^{k_\ell} f^{(k_\ell)}(u\tau{+}v,\tau) 
 \bplusz{j_1 &j_2& \ldots &j_{\ell-1}}{k_1 &k_2 &\ldots &k_{\ell-1}}{z&z&\ldots&z}{\tau} 
 \notag \\
2\pi i(\tau{-}\bar \tau)^2 \partial_\tau  \bminusz{j_1 &j_2& \ldots &j_\ell}{k_1 &k_2 &\ldots &k_\ell}{z&z&\ldots&z}{\tau}
&= \sum_{i=1}^\ell (k_i{-}j_i{-}2) \bminusz{j_1 & \ldots &j_i+1 &\ldots &j_\ell}{k_1 &\ldots &k_i &\ldots &k_\ell}{z&\ldots&z&\ldots&z}{\tau} \notag
\end{align} 
with the following variant in case of an Eisenstein kernel associated with the rightmost column
\begin{align}
2\pi i(\tau{-}\bar \tau)^2 \partial_\tau  \bplusz{j_1 & \ldots &j_{\ell-1} &j_\ell}{k_1  &\ldots &k_{\ell-1} &k_\ell}{z&\ldots&z&}{\tau} &= \sum_{i=1}^\ell (k_i{-}j_i{-}2) \bplusz{j_1 & \ldots &j_i+1 &\ldots &j_{\ell-1} &j_\ell}{k_1 &\ldots &k_i &\ldots &k_{\ell-1} &k_\ell}{z &\ldots &z &\ldots &z &}{\tau}\label{dtaubeta} \\
&\quad  - \delta_{j_\ell,k_\ell-2} (\tau{-}\bar \tau)^{k_\ell} {\rm G}_{k_\ell}(\tau) 
 \bplusz{j_1   \ldots &j_{\ell-1}}{k_1  &\ldots &k_{\ell-1}}{z&\ldots&z}{\tau} 
\notag
\end{align} 
The antiholomorphic derivatives have a similar structure but involve an extra term which can be absorbed into the Maa\ss{} operator $(\bar \tau{-}\tau)\partial_{\bar \tau}+w$ tailored to modular forms of antiholomorphic weight $w=\sum_{i=1}^{\ell}(k_i{-}2{-}2j_i)$  (see section 2.1.4 of \cite{Dorigoni:2024oft}) and can be found in appendix \ref{secdifc}.

The iterated integrals $\beta_{\pm}[\ldots;\tau]$ of modular frame in (\ref{intemgf.04}) can be straightforwardly translated into finite combinations of the ${\cal E}[\ldots;\tau]$ of holomorphic frame in (\ref{bsc.07}), 
\begin{align}
\bplusz{j_1 & j_2 &\cdots & j_\ell}{k_1& k_2& \cdots & k_\ell}{z & z&\cdots & z}{\tau}&=\sum_{p_1=0}^{k_1-j_1-2}\sum_{p_2=0}^{k_2-j_2-2}\cdots\sum_{p_\ell=0}^{k_\ell-j_\ell-2}{k_1{-}j_1{-}2\choose p_1}{k_2{-}j_2{-}2\choose p_2}\cdots{k_\ell{-}j_\ell{-}2\choose p_\ell}  \label{intemgf.07}\\
&\quad \times \bigg(\frac{1}{4\pi \Im \tau} \bigg)^{p_1+\cdots+p_\ell}\sum_{r_1=0}^{j_1+p_1}\sum_{r_2=0}^{j_2+p_2}\cdots\sum_{r_\ell=0}^{j_\ell+p_\ell}{j_1{+}p_1\choose r_1}{j_2{+}p_2\choose r_2}\cdots{j_\ell{+}p_\ell\choose r_\ell} \notag \\
&\quad\times(-2\pi i\bar{\tau})^{r_1+\cdots+r_\ell}\eezno{j_1+p_1-r_1 & j_2+p_2-r_2 &\cdots&j_\ell+p_\ell-r_\ell}{k_1&k_2&\cdots&k_\ell}{z&z&\cdots&z} \notag \\
\bminusz{j_1 & j_2 &\cdots & j_\ell}{k_1& k_2& \cdots & k_\ell}{z & z&\cdots & z}{\tau}&=\sum_{p_1=0}^{k_1-j_1-2}\sum_{p_2=0}^{k_2-j_2-2}\cdots\sum_{p_\ell=0}^{k_\ell-j_\ell-2}{k_1{-}j_1{-}2\choose p_1}{k_2{-}j_2{-}2\choose p_2}\cdots{k_\ell{-}j_\ell{-}2\choose p_\ell} \notag\\
&\quad\times \bigg(\frac{1}{4\pi \Im \tau} \bigg)^{p_1+\cdots+p_\ell}\sum_{r_1=0}^{j_1+p_1}\sum_{r_2=0}^{j_2+p_2}\cdots\sum_{r_\ell=0}^{j_\ell+p_\ell}{j_1{+}p_1\choose r_1}{j_2{+}p_2\choose r_2}\cdots{j_\ell{+}p_\ell\choose r_\ell} \notag\\
&\quad \times(-1)^{\sum_{i=1}^\ell (j_i+p_i-r_i)}(-2\pi i\bar{\tau})^{r_1+\cdots+r_\ell}\eezbno{j_1+p_1-r_1 & j_2+p_2-r_2 &\cdots&j_\ell+p_\ell-r_\ell}{k_1&k_2&\cdots&k_\ell}{z&z&\cdots&z}
\notag
\end{align}
as can be verified by binomial expansion of factors such as $(\tau{-}\tau_1)^{k-2-j}$ or $(\bar \tau{-}\tau_1)^j$ in (\ref{intemgf.02}), (\ref{intemgf.03}). The  relations (\ref{intemgf.07}) still hold if we replace any subset of the entries $z$ on the third line by empty slots, provided that this is done consistently in the same columns of the $\beta_{\pm}\!\left[\begin{smallmatrix}
    \dots\\\dots\\\dots
\end{smallmatrix}\right]$ and $\cal{E}\!\left[\begin{smallmatrix}
    \dots\\\dots\\\dots
\end{smallmatrix}\right]$ (or $\overline{\cal{E}\!\left[\begin{smallmatrix}
    \dots\\\dots\\\dots
\end{smallmatrix}\right]}$) on both sides.

\subsubsection{Generating series in modular frame}
\label{sec:itemgf.3}

The conjugation by $U_{\rm mod}(\tau)$ that produces the equivariant and single-valued series $\mathbb H^{\rm eqv}_{\ep,b}(u,v,\tau)$, $\mathbbm \Lambda_{x,y}^{\rm sv}(z,\tau)$ in modular frame is most conveniently performed at the level of the individual  series factors in (\ref{grteq.01}). The factors $\mathbb I_{\ep,b}(u,v,\tau)$, $\overline{\mathbb I_{\ep,b}(u,v,\tau)^T}$ (anti-)holomorphic in $\tau$ then transform into generating series of the iterated integrals $\beta_{\pm}[\ldots;\tau]$ in (\ref{intemgf.04}), 
\begin{align}
\mathbb{H}^{+}_{\ep,b}(u,v,\tau) &=U_{\rm mod}(\tau)\, \mathbb{I}_{\ep,b}(u,v,\tau) \, U_{\rm mod}(\tau)^{-1} \label{intemgf.08}\\
\mathbb{H}^{-}_{\ep,b}(u,v,\tau)&=U_{\rm mod}(\tau) \, \overline{\mathbb{I}_{\ep,b}(u,v,\tau)^T} \, U_{\rm mod}(\tau)^{-1}
\notag
\end{align}
namely
\begin{align}
\mathbb H^{+}_{\ep,b}(u,v,\tau) 
&= 1+\sum_{k_1=2}^\infty (k_1{-}1) \sum_{j_1=0}^{k_1-2} \dfrac{(-1)^{j_1}}{j_1!}\biggl\{\bplus{j_1}{k_1}{\tau}\epsilon_{k_1}^{(j_1)}+\bplusz{j_1}{k_1}{z}{\tau}b_{k_1}^{(j_1)}\biggr\}
\label{intemgf.10} \\
    &\quad+\sum_{k_1,k_2=2}^\infty (k_1{-}1)(k_2{-}1)  \sum_{j_1=0}^{k_1-2}  \sum_{j_2=0}^{k_2-2}\dfrac{ (-1)^{j_1+j_2}}{j_1!j_2!}\biggl\{\bplus{j_1&j_2}{k_1&k_2}{\tau}\epsilon_{k_1}^{(j_1)}\epsilon_{k_2}^{(j_2)}  \notag\\
    &\quad \ \
   +\bplusz{j_1&j_2}{k_1&k_2}{z&z}{\tau}b_{k_1}^{(j_1)}b_{k_2}^{(j_2)} 
    {+}\bplusz{j_1&j_2}{k_1&k_2}{&z}{\tau}\epsilon_{k_1}^{(j_1)}b_{k_2}^{(j_2)}
    {+}\bplusz{j_1&j_2}{k_1&k_2}{z&}{\tau}b_{k_1}^{(j_1)}\epsilon_{k_2}^{(j_2)}\biggr\}+\dots \notag
\end{align}
and the series $\mathbb{H}^{-}_{\ep,b}(u,v,\tau)$ in $\beta_-[\ldots,\tau]$ is given by (\ref{ccH}). Their equivariant combination then takes the form
\beq
\mathbb H^{\rm eqv}_{\ep,b}(u,v,\tau) = 
(\mathbb M^{\rm sv}_z)^{-1} \,
\mathbb{H}_{\ep,b}^-(u,v,\tau)
\, U_{\rm mod}(\tau) \, \mathbb M^{\rm sv}_{\Sigma(u)} \, U_{\rm mod}(\tau)^{-1}\, 
\mathbb{H}_{\ep,b}^+(u,v,\tau)
\label{defhplmi}
\eeq
which reveals a minor downside of modular frame: The $\tau$-independent series $\mathbb M^{\rm sv}_{\Sigma(u)}$ in single-valued MZVs and augmented zeta generators of section \ref{sec:3.1.3} now acquires polynomial dependence on $2\pi i \bar \tau$ and $(4\pi \Im \tau)^{-1}$ at each order in the expansion of $U_{\rm mod}(\tau)  \mathbb M^{\rm sv}_{\Sigma(u)}  U_{\rm mod}(\tau)^{-1}$ with respect to $\ep_k^{(j)}$, $b_k^{(j)}$ and $z_w$. Still, the $\mathfrak{sl}_2$-invariant factors of $z_w$ can eventually be made to cancel those of the leftmost factor $(\mathbb M^{\rm sv}_z)^{-1}$ in (\ref{defhplmi}) by the reasoning around (\ref{remzs}).

The resulting expansion of (\ref{defhplmi}) in words in $\ep_k^{(j)}$, $b_k^{(j)}$ does not have well-defined coefficients unless we prescribe a scheme of modding out by the bracket relations among the generators.\footnote{See item $(v)$ of Theorem \ref{3.thm:2} for the analogous ambiguities in holomorphic frame.} A first type of bracket relations expressing any $[\ep_{k_1}^{(j_1)}, b_{k_2}^{(j_2)}]$ as a Lie polynomials in $b_k^{(j)}$ is discussed in (\ref{intrel}) and appendix \ref{sec:D.2}. These relations are systematically accounted for by moving all the $\ep_k^{(j)}$ in (\ref{defhplmi}) to the left of the free-Lie-algebra generators~$b_{k_i}^{(j_i)}$, e.g.
\begin{align}
\mathbb H^{\rm eqv}_{\ep,b}(u,v,\tau) &= 1+\sum_{k_1=2}^\infty (k_1{-}1) \sum_{j_1=0}^{k_1-2} \dfrac{(-1)^{j_1}}{j_1!}\biggl\{ 
\beqvtau{j_1}{k_1}{ \tau }
\epsilon_{k_1}^{(j_1)}+
%%%
\beqvztau{j_1 }{k_1 }{z }{\tau}
b_{k_1}^{(j_1)}\biggr\}\notag\\
%%% 
&\quad+\sum_{k_1,k_2=2}^\infty (k_1{-}1)(k_2{-}1)  \sum_{j_1=0}^{k_1-2}  \sum_{j_2=0}^{k_2-2}\dfrac{ (-1)^{j_1+j_2}}{j_1!j_2!}\biggl\{
\beqvtau{j_1 &j_2}{k_1 &k_2}{ \tau }
\epsilon_{k_1}^{(j_1)}\epsilon_{k_2}^{(j_2)}  \notag\\
    &\qquad \qquad 
    +\beqveqvtau{j_1}{k_1}{\phantom{z}}{j_2}{k_2}{z}{\tau}
\epsilon_{k_1}^{(j_1)}b_{k_2}^{(j_2)}
   +
\beqvztau{j_1 &j_2}{k_1 &k_2}{z &z}{\tau}  
 b_{k_1}^{(j_1)}b_{k_2}^{(j_2)} \biggr\}+\dots\label{sec5.1_1}
\end{align}
after eliminating $ b_{k_2}^{(j_2)}\epsilon_{k_1}^{(j_1)} = [b_{k_2}^{(j_2)},\epsilon_{k_1}^{(j_1)}]+ \epsilon_{k_1}^{(j_1)}b_{k_2}^{(j_2)}$.
However, terms involving two or more $\ep_k^{(j)} = \epsilon^{(j){\rm TS}}_{k}{-}b^{(j)}_{k}$ obey the Pollack relations among Tsunogai derivations such that the coefficients $\beta^{\rm eqv}[\ldots;\tau]$ are only well-defined for degrees $k_1{+}k_2{+}\ldots \leq 13$, and for arbitrary degrees at modular depth one. In these cases below degree 14, the series factors $\mathbb{H}^{-}_{\ep,b}(u,v,\tau) \ldots
\mathbb{H}^{+}_{\ep,b}(u,v,\tau)$ of $\mathbb H^{\rm eqv}_{\ep,b}(u,v,\tau)$ in (\ref{defhplmi}) contribute the simple deconcatenation~sum
\beq
\beqvztau{j_1 &j_2 &\ldots &j_\ell}{k_1 &k_2 &\ldots &k_\ell}{z &z&\ldots &z}{\tau} = \sum_{i=0}^\ell  
\bminusz{j_i &\ldots &j_2 &j_1}{k_i &\ldots &k_2  &k_1}{z &\ldots &z &z}{\tau}
\bplusz{j_{i+1} &j_{i+2}  &\ldots &j_\ell}{k_{i+1} &k_{i+2}  &\ldots &k_\ell}{z &z &\ldots &z}{\tau} 
+ \ldots
\label{intemgf.13}
\eeq
to the coefficient of $b_{k_1}^{(j_1)} b_{k_2}^{(j_2)} \ldots b_{k_\ell}^{(j_\ell)}$. However, even at degrees $k_1{+}k_2{+}\ldots \leq 13$, there are two types of extra contributions in the ellipsis of (\ref{intemgf.13}) needed for the modular completion of $\beta^{\rm eqv}[\ldots;\tau]$: 
\begin{itemize}
\item[(i)] terms involving single-valued MZVs from the series $(\mathbb M^{\rm sv}_z)^{-1}$ \& $U_{\rm mod}(\tau) \mathbb M^{\rm sv}_{\Sigma(u)} U_{\rm mod}(\tau)^{-1}$ in (\ref{defhplmi});
\item[(ii)] mixed integrals like
$\bpmz{j_1 &j_2}{k_1 &k_2}{z &}{\tau}$ from situations where $b_{k_1}^{(j_1)} \ep_{k_2}^{(j_2)}$ needs to be reordered to $ \ep_{k_2}^{(j_2)} b_{k_1}^{(j_1)}$ at the cost of producing a Lie polynomial in $b_k^{(j)}$ through the bracket $[\ep_{k_2}^{(j_2)} ,b_{k_1}^{(j_1)}]$.
\end{itemize}
The only way of obtaining well-defined coefficients to all orders is to further left-concatenate (\ref{defhplmi}) with the inverse of $\mathbb H^{\rm eqv}_{\ep^{\rm TS}}(\tau)$ in (\ref{defheqv}). This leads to the generating series $\mathbbm \Lambda_{x,y}^{\rm sv}(z,\tau)$ in (\ref{thm.3.4g}) of single-valued eMPLs $\Lambda^{\rm sv}[\ldots;z,\tau]$ defined by the expansion (\ref{cor.3.4d}) solely in $b_k^{(j)}$. The deconcatenation term in (\ref{intemgf.13}) will then persist in the single-valued eMPLs
\beq
\lab{j_1 &\ldots &j_\ell}{k_1 &\ldots &k_\ell}{z, \tau }
= \sum_{i=0}^\ell  
\bminusz{j_i &\ldots &j_2 &j_1}{k_i &\ldots &k_2  &k_1}{z &\ldots &z &z}{\tau}
\bplusz{j_{i+1} &j_{i+2}  &\ldots &j_\ell}{k_{i+1} &k_{i+2}  &\ldots &k_\ell}{z &z &\ldots &z}{\tau} 
+ \ldots
\label{intemgf.14}
\eeq
at arbitrary modular depth $\ell$ and degree $k_1{+}\ldots{+} k_\ell$. However, the extra terms in the ellipsis are not only composed of the above (i) and (ii) but also of the equivariant iterated Eisenstein integrals $\beqvtau{j_1&\ldots &j_r}{k_1 &\ldots &k_r}{\tau} $ from the expansion (\ref{defheqv}) of $\mathbb H^{\rm eqv}_{\ep^{\rm TS}}(\tau)^{-1}$. 

On the one hand, closed all-degree formulae for the decompositions of $\lab{j_1 &\ldots &j_\ell}{k_1 &\ldots &k_\ell}{z, \tau }$ at modular depth $\ell \geq 2$ into $\beta_{\pm}[\ldots;\tau]$ appear out of reach. On the other hand, the explicitly known form of the series in (\ref{defhplmi}), $\mathbb H^{\rm eqv}_{\ep^{\rm TS}}(\tau)$ in (\ref{defheqv}) and the bracket relations needed to expand $\mathbbm \Lambda_{x,y}^{\rm sv}(z,\tau)$ in terms of $b_k^{(j)}$ give access to $\lab{j_1 &\ldots &j_\ell}{k_1 &\ldots &k_\ell}{z, \tau }$ at any desired degree in a finite number of steps. The iterated integrals $\beta_{\pm}[\ldots;\tau]$ in modular frame defined by (\ref{intemgf.04}) give rise to particularly economic decompositions of single-valued eMPLs (with $\mathbb Q[u,2\pi i \bar \tau,(4\pi \Im \tau)^{-1}]$-linear combinations of single-valued MZVs as coefficients), and we will discuss their value for eMGFs in section \ref{sec:6.big2} below.

\subsubsection{Examples of equivariant integrals in modular frame}
\label{sec:itemgf.4}

We conclude this section with a discussion of simple instances of the above $\beta^{\rm eqv}[\ldots;\tau]$ and $\Lambda^{\rm sv}[\ldots;z,\tau]$ at modular depth $\ell \leq 2$. First of all, the depth-one cases of equivariant iterated Eisenstein integrals are well-known in different formulations \cite{Brown:2017qwo, Dorigoni_2022}
\beq
\beqvtau{j}{k}{\tau} = \bplus{j}{k}{\tau}+ \bminus{j}{k}{\tau} -\frac{2\zeta_k}{(k{-}1)\, (4\pi\Im \tau)^{k-2-j}} 
\label{intemgf.17}
\eeq
where the odd zeta values stem from the term $\sigma_w \rightarrow -\frac{\ep_{w+1}^{(w-1){\rm TS}}}{(w{-}1)!}$ in (\ref{lieg1.51}). The analogous depth-one expressions for eMGFs do not involve any analogue of the zeta value \cite{Hidding:2022vjf},
\beq
\beqvztau{j}{k}{z}{\tau} = \bplusz{j}{k}{z}{\tau}+ \bminusz{j}{k}{z}{\tau} 
\label{intemgf.18}
\eeq
which can here be seen from (\ref{defhplmi}) and (\ref{sec5.1_1}) together with the fact that the only geometric depth-one term of augmented zeta generators is given by $\Sigma_w(u) \rightarrow -\frac{\ep_{w+1}^{(w-1)}}{(w{-}1)!}$.

Single-valued eMPLs at modular depth one resulting from the construction of $\mathbbm \Lambda_{x,y}^{\rm sv}(z,\tau)$ as the ratio (\ref{thm.3.4g}) are given by
\beq
\lab{j}{k}{z,\tau}=\beqvztau{j}{k}{z}{\tau}-\left\{\begin{array}{cl}\beqvtau{j}{k}{\tau}   &: \ k\geq 4\,\,\text{even}\\ 0&: \  k\geq 3\,\,\text{odd} \ \text{or} \, \, k=2\end{array}\right.
\label{intemgf.19}
\eeq
At modular depth two, the coefficient of the mixed word $\ep_{k_1}^{(j_1)} b_{k_2}^{(j_2)}$ in (\ref{sec5.1_1}) factorizes into 
\beq
\beqveqvtau{j_1}{k_1}{\phantom{z}}{j_2}{k_2}{z}{\tau}=\beqvtau{j_1}{k_1}{\tau}\beqvztau{j_2}{k_2}{z}{\tau}
\label{intemgf.20}
\eeq
It is well-defined even if  $k_1{+}k_2\geq 14$ since all the $\ep_{k}^{(j)}$ are understood to be moved to the left of $b_{k}^{(j)}$ and Tsunogai relations only relate alike brackets $[\ep_{k_1}^{(j_1)}, \ep_{k_2}^{(j_2)}]$ and $[b_{k_1}^{(j_1)}, b_{k_2}^{(j_2)}]$ after eliminating any $[\ep_{k_1}^{(j_1)}, b_{k_2}^{(j_2)}]$ via (\ref{intrel}) and the results of appendix \ref{sec:D.2}. The right side of (\ref{intemgf.20}) follows from the fact that any $\ep_{k}^{(j)}$ in the ratio (\ref{thm.3.4g}) can be eliminated in favor of $b_k^{(j)}$, see the proof of item $(i)$ of Theorem \ref{3.cor:1} for details.

Among the $\beqvY{j_1 &j_2\\ k_1 &k_2\\ z&z}$ which are defined at degrees $k_1{+}k_2\leq 13$ in the ordering conventions of $\ep_{k_1}^{(j_1)}, b_{k_2}^{(j_2)}$ in  (\ref{sec5.1_1}), we content ourselves to providing the simple examples
\begin{align}
\beqvY{0 &0\\ 2 &3\\ z&z}&=\bplusY{0&0\\2&3\\z&z}+\bminusY{0&0\\3&2\\z&z}+\bminusY{0\\2\\z}\bplusY{0\\3\\z}-\dfrac{1}{4\pi {\rm Im}\,\tau} \zeta_3 B_1(u)
 \notag\\
   \beqvY{1 &0\\ 3 &2\\ z&z}&= \bplusY{1&0\\3&2\\z&z}+\bminusY{0&1\\2&3\\z&z}+\bminusY{1\\3\\z}\bplusY{0\\2\\z}+\zeta_3B_1(u)
\notag \\
\beqvY{1 &0\\ 3 &3\\ z&z}&= \bplusY{1&0\\3&3\\z&z}+\bminusY{0&1\\3&3\\z&z}+\bminusY{1\\3\\z}\bplusY{0\\3\\z}+\dfrac{1}{4}\zeta_3B_2(u)-\frac{1}{8}\zeta_3
\notag \\
%%%%
\beqvY{1&2\\3&4\\z&z}&=\bplusY{1&2\\3&4\\z&z}+\bminusY{2&1\\4&3\\z&z}+\bminusY{1\\3\\z}\bplusY{2\\4\\z}\notag\\
&\quad -\bplusY{1&2\\3&4\\z&}-\bminusY{2&1\\4&3\\&z}-\bminusY{1\\3\\z}\bplusY{2\\4\\}\notag\\
&\quad+\dfrac{2}{3}\zeta_3\biggl(\bminusY{1\\3\\z}-\dfrac{B_3(u)}{12}(-2\pi i \bar{\tau})^2\biggr)+\dfrac{\zeta_5}{3}B_1(u)
\label{exbeqvd2}
\end{align}
and a few other cases at degrees $k_1{+}k_2=7$ and $8$ in appendix \ref{apbeqv}.
The first three terms in each case reproduce the deconcatenation sums in (\ref{intemgf.13}). The  example at degree seven exhibits another deconcatenation sum of the same type in the second line from below which arises from the process of reordering $b_3^{(1)}\ep_4^{(2)}$ through the bracket,
\begin{align}
[\ep_4^{(2)}, b_3^{(1)}] &=
[b_3^{(1)},b_4^{(2)}] + \frac{1}{3} \, [b_2, b_5^{(3)}]
\label{bepbrks}
\end{align}
Finally, the odd zeta values in (\ref{exbeqvd2}) along with the Bernoulli polynomials $B_k(u)$ are due to the expansion of $\Sigma_3(u)$ and $\Sigma_5(u)$ in (\ref{explSIG}) up to degree seven. The accompanying factors of $2\pi i\bar \tau$ or $(4\pi \Im \tau)^{-1}$ can be traced back to the transformation to modular frame, and the product $\zeta_3\bminusY{1\\3\\z}$ arises from the reordering of $b_3^{(1)}\Sigma_3(u) \rightarrow - \frac{1}{2} b_3^{(1)} \ep_4^{(2)}$, again using (\ref{bepbrks}).

Finally, the single-valued eMPLs $\lab{j_1&j_2}{k_1&k_2}{z,\tau}$ associated with the $j_1,j_2,k_1,k_2$ in (\ref{exbeqvd2}) and (\ref{wt8exbs}) of even degree $k_1{+}k_2$ are given by
\begin{align}
\lab{1&0}{3&3}{z,\tau}&=\beqvztau{1&0}{3&3}{z&z}{\tau}
\label{intemgf.27} \\
\lab{2&1}{4&4}{z,\tau} &=\beqvY{2&1\\4&4\\z&z}+\beqvY{2&1\\4&4}-\beqvY{2\\4\\}\beqvY{1\\4\\z}
 \notag \\
\lab{2&1}{5&3}{z,\tau}&=\beqvY{2&1\\5&3\\z&z}-\frac{3}{2}\beqvY{2&1\\4&4}-\frac{3}{4}\beqvY{1&2\\4&4}
\notag
\end{align}
The remaining cases in
(\ref{exbeqvd2}) have odd degree and line up with the general formulae (valid for $k_1{+}k_2\leq 13$ where the $\beta^{\rm eqv}[\ldots;\tau]$ on the right side are well-defined)
\begin{align}
\lab{j_1&j_2}{k_1&k_2}{z,\tau}&= \left\{ \begin{array}{cl}
\bigg. \beqvztau{j_1&j_2}{k_1&k_2}{z&z}{\tau} \bigg.&: \ k_1 \ \textrm{odd and} \ k_2 \ \textrm{even} \\
\beqvztau{j_1&j_2}{k_1&k_2}{z&z}{\tau}  - \beqvtau{j_1}{k_1}{\tau} \beqvztau{j_2}{k_2}{z}{\tau} &: \ k_1\geq 4 \ \textrm{even and} \ k_2 \ \textrm{odd} \\
%%%
\bigg.\beqvztau{j_1&j_2}{k_1&k_2}{z&z}{\tau}  \bigg. &: \ k_1=2 \ \textrm{and} \ k_2 \ \textrm{odd}
\end{array} \right.
\label{intemgf.21}
\end{align}
They follow from left-multiplication of (\ref{sec5.1_1}) by $\mathbb H^{\rm eqv}_{\ep^{\rm TS}}(\tau)^{-1} = 1- \sum_{k=4}^\infty(k{-}1)\sum_{j=0}^{k-2}\frac{(-1)^j}{j!}$ $\times \beqvtau{j}{k}{\tau} (\ep_k^{(j)} {+}b_k^{(j)})+\ldots$ and noticing that this does not introduce any terms $b_{k_1}^{(j_1)} \ep_{k_2}^{(j_2)}$ at odd $k_1{+}k_2$ into  $\mathbb H^{\rm eqv}_{\ep^{\rm TS}}(\tau)^{-1}\mathbb H^{\rm eqv}_{\ep,b}(u,v,\tau)$. On these grounds, there is no need to use bracket relations for $[b_{k_1}^{(j_1)} ,\ep_{k_2}^{(j_2)}]$ at
odd $k_1{+}k_2$ to maintain the ordering of 
$\ep_{k_1}^{(j_1)} b_{k_2}^{(j_2)}$ in (\ref{sec5.1_1}) which would have otherwise modified the coefficient of $b_{k_1}^{(j_1)} b_{k_2}^{(j_2)}$ beyond the terms in (\ref{intemgf.21}).

%%%%%%%%%%%%%%%%%%%%%%%%%%%%%%%%%%%%
%%%%%%%%%%%%%%%%%%%%%%%%%%%%%%%%%%%%
%%%%%%%%%%%%%%%%%%%%%%%%%%%%%%%%%%%%

\subsection{eMGFs in terms of equivariant iterated integrals and single-valued eMPLs}
\label{sec:6.big2}

We shall here translate several showcases of eMGFs into iterated integrals and their complex conjugates. By the results of the previous section, the representations of eMGFs in terms of $\beta^{\rm eqv}[\ldots;\tau]$ in section \ref{sec:6.1} and $\Lambda^{\rm sv}[\ldots;z,\tau]$ in section \ref{sec:6.3} boil down to the iterated integrals ${\cal E}[\ldots;\tau]$ meromorphic in $\tau$ through their counterparts $\beta_{\pm}[\ldots;\tau]$ in modular frame, with known $\mathbb Q[u,2\pi i \bar \tau,(4\pi \Im \tau)^{-1}]$ combinations of single-valued MZVs as coefficients.

%%%%%%%%%%%%%%%%%%%%%%%%%%%%%%%%%%%%
%%%%%%%%%%%%%%%%%%%%%%%%%%%%%%%%%%%%
%%%%%%%%%%%%%%%%%%%%%%%%%%%%%%%%%%%%

\subsubsection{$\mathbb H^{\rm eqv}_{\ep,b}(u,v,\tau)$ as a generating series of eMGFs}
\label{sec:6.1}

A first representation of eMGFs employs the modular forms $\beqvztau{j}{k}{z}{\tau}$ and $\beqvztau{j_1 &j_2}{k_1 &k_2}{z &z}{\tau}$ in the expansion (\ref{sec5.1_1}) of the series $\mathbb H^{\rm eqv}_{\ep,b}(u,v,\tau)$ which are well-defined at the degrees $k_1{+}k_2\leq 13$ under consideration in this section.

The simplest family of $z$-dependent eMGFs is furnished by Zagier's \cite{Ramakrish} single-valued elliptic polylogartithms ${\cal D}^+[\smallmatrix a \\ b \endsmallmatrix](z,\tau)$ defined in (\ref{defZag}). Their instances with $a,b\in \mathbb N$ are given by  equivariant integrals $\beqvtau{j}{k}{\tau}$ and $\beqvztau{j}{k}{z}{\tau}$ at modular depth one \cite{Dorigoni_2022, Hidding:2022vjf}
\begin{align}
\beqvtau{j}{k}{\tau}&=-(2i)^{2j-k+2} \, \frac{j!(k{-}2{-}j)!}{(k{-}1)!} \, {\cal D}^+[\smallmatrix j+1 \\ k-1-j \endsmallmatrix](0,\tau) \, , \ \ \ \ \ \ k\geq 4 \ {\rm even}
\notag\\
\beqvztau{j}{k}{z}{\tau}&=-(2i)^{2j-k+2}\, \frac{j!(k{-}2{-}j)!}{(k{-}1)!} \, {\cal D}^+[\smallmatrix j+1 \\ k-1-j \endsmallmatrix](z,\tau) \, , \ \ \ \ \ \ k\geq 2
\label{intemgf.31}
\end{align}
Equivariant iterated integrals at modular depth $\geq 2$ arise from eMGFs (\ref{defdiC}) involving three or more columns.\footnote{See \cite{Dorigoni:2021jfr} for presentations of three-column MGFs in terms of equivariant double Eisenstein~integrals and \cite{DHoker:2015gmr, DHoker:2017zhq} for earlier discussions of their Laplace equations and Fourier expansion.} The simplest non-trivial instances that are independent under shuffles $\beqvztau{j_1 &j_2}{k_1 &k_2}{z &z}{\tau}+ \beqvztau{j_2 &j_1}{k_2 &k_1}{z &z}{\tau} = \beqvztau{j_1}{k_1}{z }{\tau} \beqvztau{j_2}{k_2}{z }{\tau}$ read as follows 
\begin{align}
\beqvztau{0&0}{2&3}{z&z}{\tau}&=-\frac{i}{8}{\cal C}^+\!\left[\smallmatrix 0&1&1 \\ 1&1&1 \\ z&0&0 \endsmallmatrix\right](\tau)-\frac{i}{12}{\cal D}^+[\smallmatrix 2 \\ 3 \endsmallmatrix](z,\tau)
\label{intemgf.32}
\\
\beqvztau{1&0}{3&3}{z&z}{\tau}&=
\frac{1}{8}
%C_{1|1,1}
\mathcal C^+\!\left[\begin{smallmatrix}1&1&1\\1&1&1\\z&0&0\end{smallmatrix}\right](\tau)
-\frac{1}{12}{\cal D}^+[\smallmatrix 3 \\ 3 \endsmallmatrix](z,\tau)-\frac{1}{24}{\cal D}^+[\smallmatrix 3 \\ 3 \endsmallmatrix](0,\tau)-\frac{1}{8}\zeta_3 \notag \\
\beqvztau{2&0}{4&3}{z&z}{\tau}&=\frac{1}{12}\nabla_z
%C_{2|1,1}
\mathcal C^+\!\left[\begin{smallmatrix}
    2&1&1\\2&1&1\\z&0&0
\end{smallmatrix}\right](\tau)+\frac{i}{3}{\cal D}^+[\smallmatrix 1 \\ 2 \endsmallmatrix](z,\tau){\cal D}^+[\smallmatrix 3 \\ 1 \endsmallmatrix](0,\tau)-\frac{i}{6}{\cal D}^+[\smallmatrix 4 \\ 3 \endsmallmatrix](z,\tau) \notag
\end{align}
We follow the conventions of \cite{DHoker:2016mwo, Dhoker:2020gdz, Hidding:2022vjf} for differential operators that preserve the vanishing holomorphic modular weight of ${\cal C}^+[\ldots](\tau)$,
\beq
\nabla_\tau = 2i (\Im \tau)^2 \, \partial_\tau \, , \ \ \ \ \ \
\nabla_z = 2i \Im \tau \, \partial_z
\label{intemgf.33}
\eeq
see (\ref{nabdiC}) and (\ref{dzdiC}) for their
action on general dihedral eMGFs.

Conversely, we can also express the eMGFs in (\ref{intemgf.32}) as simple combinations of the equivariant integrals $\beta^\text{eqv}[\ldots;\tau]$.  
We give all such relations for the $\mathbb Q$-basis of eMGFs in \cite{Hidding:2022vjf} up to and including degree seven in the ancillary file {\tt eMGFdata.nb}. Some of their examples~read 
\begin{align}
\nabla_z 
%C_{1|1,1}
\mathcal{C}^+\!\left[\begin{smallmatrix}1&1&1\\1&1&1\\z&0&0\end{smallmatrix}\right](\tau)&= 4\beqvztau{1&0}{3&2}{z&z}{\tau}  -8\beqvztau{2}{5}{z}{\tau} \label{intemgf.34}\\
%C_{1|1,1}(z)
\mathcal{C}^+\!\left[\begin{smallmatrix}1&1&1\\1&1&1\\z&0&0\end{smallmatrix}\right](\tau)&=8\beqvztau{1&0}{3&3}{z&z}{\tau}-10\beqvtau{2}{6}{\tau}-20\beqvztau{2}{6}{z}{\tau}+\zeta_3 \notag \\
%%%%%%
\nabla_z
%C_{2|1,1}(z)
\mathcal{C}^+\!\left[\begin{smallmatrix}2&1&1\\2&1&1\\z&0&0\end{smallmatrix}\right](\tau)&=12\beqvztau{2&0}{4&3}{z&z}{\tau}-60\beqvztau{3}{7}{z}{\tau}-12\beqvtau{2}{4}{\tau}\beqvztau{0}{3}{z}{\tau} \notag \\
%%%%%%
\pi \nabla_\tau\nabla_z 
%C_{2|1,1}(z)
\mathcal{C}^+\!\left[\begin{smallmatrix}2&1&1\\2&1&1\\z&0&0\end{smallmatrix}\right](\tau)&=-3\beqvztau{2&1}{4&3}{z&z}{\tau}+30\beqvztau{4}{7}{z}{\tau}+3\beqvtau{2}{4}{\tau}\beqvztau{1}{3}{z}{\tau}
\notag
\end{align}
and illustrate the appearance of single-valued MZVs and equivariant iterated Eisenstein integrals $\beqvtau{j_1 &\ldots &j_\ell}{k_1 &\ldots &k_\ell}{\tau}$ besides their $z$-dependent counterparts.
The double integrals on the right side incorporate contributions of $\zeta_{w}B_k(u)$ with $w\geq 3$ odd that are known from the expansion of eMGFs around the cusp in \cite{DHoker:2018mys} and their iterated-integral representations in \cite{Hidding:2022vjf}. The terms $4\beqvztau{1&0}{3&2}{z&z}{\tau}$ and $8\beqvztau{1&0}{3&3}{z&z}{\tau}$ in the first two rows for instance lead to the terms $4\zeta_3 B_1(u)$ in (6.20) and $2 \zeta_3 B_2(u)$ in (6.11) of \cite{Hidding:2022vjf} through their decomposition (\ref{exbeqvd2}).

\subsubsection{$ \mathbbm{\Lambda}^{\rm sv}(z,\tau)$ as a generating series of eMGFs}
\label{sec:6.3}

A second and closely related representation of eMGFs is based on the modular formulation of single-valued eMPLs $\Lambda^\text{sv}\!\left[\begin{smallmatrix}j_1&\dots&j_\ell\\k_1&\dots&k_\ell\end{smallmatrix};z,\tau\right]$ which are generated by the series $\mathbbm{\Lambda}^{\rm sv}_{x,y}(z,\tau)$ according to (\ref{cor.3.4d}) and well-defined for arbitrary degrees.

Cases of modular depth $\ell=1$ recover Zagier's single-valued eMPLs \cite{Ramakrish}
\beq
\lab{j}{k}{z,\tau}=
-(2i)^{2j-k+2}\,\frac{j!(k{-}2{-}j)!}{(k{-}1)!} \times \left\{\begin{array}{cl} {\cal D}^+[\smallmatrix j+1 \\ k-1-j \endsmallmatrix](z,\tau)-{\cal D}^+[\smallmatrix j+1 \\ k-1-j \endsmallmatrix](0,\tau)
&: \ k \geq 3 \\
{\cal D}^+[\smallmatrix j+1 \\ k-1-j \endsmallmatrix](z,\tau) &: \ k=2
\end{array}
\right.
\label{intemgf.38}
\eeq
where the values ${\cal D}^+[\smallmatrix j+1 \\ k-1-j \endsmallmatrix](0,\tau)$ at the origin vanish at odd $k$ and are absent in case of $k=2$ since ${\cal D}^+[\smallmatrix 1 \\1 \endsmallmatrix](z,\tau)$ has a logarithmic singularity $\sim -\log|z|^2$ as $z\rightarrow 0$.

At modular depth two, the expressions (\ref{intemgf.32}) of $\beta^{\rm eqv}[\ldots;\tau]$ in terms of eMGFs together with the translations (\ref{intemgf.27}) and (\ref{intemgf.21}) into single-valued eMPLs lead to the examples
\begin{align}
\lab{0&0}{2&3}{z,\tau}&=-\frac{i}{8}{\cal C}^+\!\left[\smallmatrix 0&1&1 \\ 1&1&1 \\ z&0&0 \endsmallmatrix\right](\tau)-\frac{i}{12}{\cal D}^+[\smallmatrix 2 \\ 3 \endsmallmatrix](z,\tau) \label{intemgf.39}\\
\lab{1&0}{3&3}{z,\tau}&=\frac{1}{8}
%C_{1|1,1}
\mathcal C^+\!\left[\begin{smallmatrix}1&1&1\\1&1&1\\z&0&0\end{smallmatrix}\right](\tau)-\frac{1}{12}{\cal D}^+[\smallmatrix 3 \\ 3 \endsmallmatrix](z,\tau)-\frac{1}{24}{\cal D}^+[\smallmatrix 3 \\ 3 \endsmallmatrix](0,\tau)-\frac{1}{8}\zeta_3 \notag \\
\lab{2&0}{4&3}{z,\tau}&=\frac{1}{12}\nabla_z
%C_{2|1,1}
\mathcal{C}^+\!\left[\begin{smallmatrix}
    2&1&1\\2&1&1\\z&0&0
\end{smallmatrix}\right](\tau)-\frac{i}{6}{\cal D}^+[\smallmatrix 4 \\ 3 \endsmallmatrix](z,\tau)
\notag
\end{align}
Conversely we can also express eMGFs in terms of $\lab{\dots}{\dots}{z,\tau}$ where the examples in (\ref{intemgf.34}) translate into
\begin{align}
\nabla_z %C_{1|1,1}(z)
\mathcal C^+\!\left[\begin{smallmatrix}
    1&1&1\\1&1&1\\z&0&0
\end{smallmatrix}\right](\tau)&= 4\lab{1&0}{3&2}{z,\tau} -8\lab{2}{5}{z,\tau} \label{cplexs} \\
%C_{1|1,1}(z)
\mathcal C^+\!\left[\begin{smallmatrix}
    1&1&1\\1&1&1\\z&0&0
\end{smallmatrix}\right](\tau)&=8\lab{1&0}{3&3}{z,\tau}-30\beqvtau{2}{6}{\tau}-20\lab{2}{6}{z,\tau}+\zeta_3 \notag \\
\nabla_z
%C_{2|1,1}(z)
\mathcal C^+\!\left[\begin{smallmatrix}
    2&1&1\\2&1&1\\z&0&0
\end{smallmatrix}\right](\tau)&=12\lab{2&0}{4&3}{z,\tau}-60\lab{3}{7}{z,\tau}
\notag \\
\pi \nabla_\tau\nabla_z 
%C_{2|1,1}(z)
\mathcal C^+\!\left[\begin{smallmatrix}
    2&1&1\\2&1&1\\z&0&0
\end{smallmatrix}\right](\tau)&=-3\lab{2&1}{4&3}{z,\tau}+30\lab{4}{7}{z,\tau}
\notag
\end{align}
The second line illustrates that generic eMGFs are linear combinations of single-valued eMPLs $\Lambda^{\rm sv}[\ldots;\tau]$ with rational polynomials in MGFs (represented by $z$-independent $\beta^{\rm eqv}[\ldots;\tau]$) and single-valued MZVs as coefficients. This generalizes the decomposition of MGFs into equivariant iterated Eisenstein integrals accompanied by $\mathbb Q$-linear combinations of single-valued  MZVs to the case of an additional modulus $z$. In particular, the contribution $\beqvtau{2}{6}{\tau}$ to the second line of (\ref{cplexs}) realizes the simple MGF $\sim {\cal D}^+[\smallmatrix 3 \\ 3 \endsmallmatrix](0,\tau)$ which by its independence on $z$ is not expressible in terms of $\Lambda^{\rm sv}[\ldots;\tau]$ -- the special value $\lab{2}{6}{0,\tau}$ vanishes by (\ref{intemgf.38}).
In the ancillary file {\tt eMGFdata.nb}, we give all relations between $\Lambda^{\rm sv}[\ldots;z,\tau]$ and eMGFs in the $\mathbb Q$-bases of \cite{Hidding:2022vjf} up to and including degree 7.

As detailed in Corollary \ref{corof44} the single-valued eMPLs $\Lambda^{\rm sv}[\ldots;z,\tau]$ are determined as combinations of eMPLs and their complex conjugates through their generating series (\ref{grteq.02}) and (\ref{cor.3.4c}). At modular depth one, this relates Zagier's single-valued eMPLs to complex combinations of eMPLs via (\ref{intemgf.38}) and should reproduce the representations of ${\cal D}^+[\smallmatrix a \\b \endsmallmatrix](z,\tau)$ in section 4 of \cite{Broedel:2019tlz} when the contributing elliptic MZVs are expressed in terms of $\beqvtau{j}{k}{\tau}$ and $\ee{0}{2}{\tau}{+}\overline{\ee{0}{2}{\tau}}$. It would be interesting to extract the explicit form of the decomposition of $\Lambda^{\rm sv}[\ldots;z,\tau]$ at modular depth $\geq 2$ in terms of eMPLs and their complex conjugates from the generating series of this work.

\subsubsection{Comments on $\beta^{\rm sv}$ of \cite{Hidding:2022vjf}}
\label{comsec623}

Reference \cite{Hidding:2022vjf} describes a method to decompose arbitrary eMGFs into iterated $\tau$ integrals through a generalization of the sieve algorithm \cite{DHoker:2016mwo, DHoker:2016quv} for MGFs to the $z$-dependent case. These iterated integrals are organized into complex and $T$-invariant combinations denoted by $\bsvtau{j_1 &\ldots &j_\ell}{k_1 &\ldots &k_\ell}{\tau}$ and singled out by integrating the holomorphic $\tau$ derivatives of the generating series of genus-one Koba-Nielsen integrals in \cite{Dhoker:2020gdz}. 

More specifically, the $\bsvtau{j_1 &\ldots &j_\ell}{k_1 &\ldots &k_\ell}{\tau}$ in section 3 of \cite{Hidding:2022vjf} supplement the deconcatenation sums of (\ref{intemgf.13}) or (\ref{intemgf.14}) over products of $\beta_{\pm}[\ldots;\tau]$ in (\ref{intemgf.04}) by antiholomorphic iterated integrals of ${\rm G}_k$, $f^{(k)}(u\tau{+}v,\tau)$ multiplied by MZVs and lower-depth instances of $\beta^{\rm sv}[\ldots;\tau]$. Following the strategy of \cite{Gerken:2020yii, Dorigoni:2021jfr}, the antiholomorphic admixtures are determined by imposing the complex-conjugation properties of the Koba-Nielsen integrals in \cite{Dhoker:2020gdz} on each eMGF in their low-energy expansion.

In the approach of this work, the additional terms that complete the deconcatenation sums of (\ref{intemgf.13}) or (\ref{intemgf.14}) to modular forms are entirely determined by the group-like series in section \ref{sec:3.2} and the bracket relations among zeta generators, Tsunogai derivations, and Lie-polynomials in $x,y$. Hence, there is no need for extra information from Koba-Nielsen integrals, and the complex-conjugation properties of eMGFs and single-valued eMPLs are the outcome of the compositions of the group-like series, see section \ref{sec:3.cc}.

\subsubsection{Comments on $q$-expansions and numerical methods}
\label{comsec624}

In early stages of this work, the iterated integrals entering the $\beta^{\rm sv}[\ldots;\tau]$ of \cite{Hidding:2022vjf} were translated into the $\beta_{\pm}[\ldots;\tau]$ in (\ref{intemgf.04}). Hence, the $\beta^{\rm sv}[\ldots;\tau]$-representations of eMGFs up to degree ten in the reference identified numerous modular combinations of $\beta_{\pm}[\ldots;\tau]$ and thereby low-degree instances of the $\beta^{\rm eqv}[\ldots;\tau]$ in (\ref{sec5.1_1}). This step of data-mining was helpful to anticipate the ingredients of the equivariant and single-valued generating series of this work before their rigorous derivations were found.

Moreover, the constituents of $\beta^{\rm sv}[\ldots;\tau]$ have a known all-order expansion around the cusp at fixed $u,v$, see section 3.5 of \cite{Hidding:2022vjf}. Hidding's {\sc Mathematica} implementation of these expansion methods \cite{Hidding:2022zzz} leads to numerical evaluations of the $\beta_{\pm}[\ldots;\tau]$ at high precision. 
In intermediate stages of developing the generating series of this work, numerical checks of modular transformations via \cite{Hidding:2022zzz} were a valuable tool to test candidate expressions or to determine unknown rational coefficients of MZVs via PSLQ. 
This gave access to higher orders of the expansion (\ref{explSIG}) of augmented zeta generators $\Sigma_w(u)$ before their closed formula (\ref{not.06}) was derived.

\subsection{Implications for counting of eMGFs}
\label{sec:6.4}

This section applies the iterated-integral representations of eMGFs to count their independent representatives under algebraic relations with $\mathbb Q$-linear combinations of (conjecturally single-valued) MZVs and MGFs as coefficients \cite{DHoker:2020tcq, Basu:2020pey, Basu:2020iok, Dhoker:2020gdz}. 
For dihedral cases $\cplus{a_1 &\ldots &a_r \\ b_1 &\ldots &b_r \\ z_1 &\ldots &z_r}(\tau)$ in (\ref{defdiC}), the sums
\beq
|A|=\sum_{i=1}^r a_i \, , \ \ \ \  |B|=\sum_{i=1}^r b_i \, , \ \ \ \ {\rm in} \ (\ref{defdiC})
\label{ctemgf.01}
\eeq
of the exponents $a_i$ and $b_i$ of the holomorphic and antiholomorphic factors $(m_i\tau{+}n_i)^{-1}$ and $(m_i\bar \tau{+}n_i)^{-1}$ in the summand are conserved quantities in these types of relations among eMGFs when working modulo products involving MZVs. 
Also for eMGFs of trihedral or more general graph topologies \cite{Dhoker:2020gdz}, the total exponents of all factors $(m_i\tau{+}n_i)^{-1}$ and $(m_i\bar \tau{+}n_i)^{-1}$ in their lattice-sum representation are conserved modulo MZVs and given as follows for the representation (\ref{genint}) of general eMGFs as integrals over products of $\dplus{a_{ij} \\ b_{ij}}(z_i{-}z_j,\tau) $:
\beq
|A|=\sum_{i=1}^r \sum_{j=i+1}^{r+s} a_{ij} \, , \ \ \ \  |B|= \sum_{i=1}^r \sum_{j=i+1}^{r+s}  b_{ij} \, , \ \ \ \ {\rm in} \ (\ref{genint})
\label{ctemgf.00}
\eeq
Accordingly, we refer to these characteristic integers $|A|$, $|B|$ of eMGFs as their {\it lattice weights} which are understood to be additive in products among eMGFs and MGFs.

\subsubsection{Bounding the counting of indecomposable eMGFs}
\label{sec631bd}

For eMGFs in one variable with $z_i \in \{0,z\}$ in (\ref{defdiC}) and (\ref{genint}), the following lemma correlates their lattice weights $|A|$ and $|B|$ with the entries $k_i,j_i$ in their representation via single-valued eMPLs $\Lambda^\text{sv}\!\left[\begin{smallmatrix}j_1&\dots&j_r\\k_1&\dots&k_r\end{smallmatrix};z,\tau\right]$.

\begin{lemma}
\label{ABKJ}
When expressing eMGFs in one variable of lattice weights $|A|$ and $|B|$ as a linear combination of single-valued eMPLs $\Lambda^{\rm sv}\!\left[\begin{smallmatrix}j_1&\dots&j_\ell\\k_1&\dots&k_\ell\end{smallmatrix};z,\tau\right]$, then the entries $j_i$, $k_i$ in each term without factors of MZVs in their coefficients satisfy
\beq
|A|=\sum_{i=1}^\ell (j_i{+}1) \, , \ \ \ \ \ \ |B|=\sum_{i=1}^\ell (k_i{-}j_i{-}1)
\label{ctemgf.02}
\eeq
\end{lemma}
As exemplified by (\ref{cplexs}), the modular depth $\ell$ of the single-valued eMPLs for a given eMGF can vary as long as (\ref{ctemgf.02}) at fixed $|A|$, $|B|$ has solutions within the admissible range of $k_i\geq 2$ and $0\leq j_i \leq k_i{-}2$. For dihedral eMGFs $\cplus{a_1 &\ldots &a_r \\ b_1 &\ldots &b_r \\ z_1 &\ldots &z_r}(\tau)$, the highest modular depth $\ell$ in its representation via single-valued eMPLs is expected to be bounded by the number $r$ of columns as $\ell \leq r{+}1$ (not necessarily saturating the bound). The correlation (\ref{ctemgf.02}) was also proposed in section 4.1 of \cite{Hidding:2022vjf} for representations of one-variable eMGFs in terms of the quantities $\bsvtau{j_1 &\ldots &j_\ell}{k_1 &\ldots &k_\ell}{\tau}$ in the reference (see section \ref{comsec623}), again for those terms in an iterated-integral representation of eMGFs without any accompanying factors of MZVs.

\begin{proof}
The proof of Lemma \ref{ABKJ} proceeds in two steps, namely by separately showing the two identities
\beq
(i) \ \ |B|{-}|A| = \sum_{i=1}^\ell (k_i{-}2j_i{-}2) \, , \ \ \ \ \ \ 
(ii) \ \ |A|{+}|B| = \sum_{i=1}^\ell k_i
\label{ctemgf.03}
\eeq
which are equivalent to the claim (\ref{ctemgf.02}).

$(i)$: follows by matching the modular weight $(0,\sum_{i=1}^\ell (k_i{-}2j_i{-}2))$ of the single-valued eMPL $\Lambda^\text{sv}\!\left[\begin{smallmatrix}j_1&\dots&j_\ell\\k_1&\dots&k_\ell\end{smallmatrix};z,\tau\right]$ in (\ref{cor.3.4e}) with that of eMGFs: the lattice weights (\ref{ctemgf.01}) and (\ref{ctemgf.00}) of dihedral and general eMGFs give rise to modular weight $(0,|B|{-}|A|)$ by (\ref{modeMGF}) and by~(\ref{genint}).\footnote{In case of the integral representations (\ref{genint}) of general eMGFs, we have used the modular weights $(0,0)$ and $(0,b_{ij}{-}a_{ij})$ of the integration measures and the factors $\dplus{a_{ij} \\ b_{ij}}(z_i{-}z_j,\tau) $ in the integrand, respectively.}

$(ii)$: follows from the structure of $\tau$-derivatives of eMGFs computed from the method of holomorphic subgraph reduction in \cite{Dhoker:2020gdz}: modulo products involving MGFs and/or MZVs, each factor of $f^{(k)}(u\tau{+}v,\tau)$ or ${\rm G}_k(\tau)$ in their $\tau$-derivatives is accompanied by eMGFs whose initial lattice weights $|A|$ and $|B|$ have dropped to $|A|{-}k$ and $|B|$. By iteratively integrating back these $\tau$-derivatives, we identify terms of the form $\bplusz{j_{1} &\ldots &j_\ell}{k_{1} &\ldots &k_\ell}{z &\ldots &z}{\tau} $ (or their analogues with empty slots in the last line) in the iterated-integral representation of that eMGF which by (\ref{intemgf.14}) signal a single-valued eMPL with the same entries $j_i$, $k_i$.
\end{proof}
Based on this lemma, we infer an upper bound on the number of indecomposable eMGFs in one variable at given lattice weights $|A|$ and $|B|$. An eMGF is referred to as {\it indecomposable} if it cannot be expressed in terms of eMGFs of lower lattice weights, MZVs and/or MGFs.

\begin{corollary}
\label{corcount}
The number of indecomposable eMGFs in one variable of lattice weights $|A|$ and $|B|$ is bounded from above by the number of shuffle-independent single-valued eMPLs $\Lambda^{\rm sv}\!\left[\begin{smallmatrix}j_1&\dots&j_\ell\\k_1&\dots&k_\ell\end{smallmatrix};z,\tau\right]$ of arbitrary modular depth $\ell\geq 1$ whose entries in the range $k_i\geq 2$ and $0\leq j_i\leq k_i{-}2$ obey (\ref{ctemgf.02}).
\end{corollary}
The corollary follows from the fact that products of eMGFs of lower lattice weights are not counted into their indecomposable representatives which rules out any shuffle product (\ref{shrls}) of single-valued eMPLs.

If the bound in Corollary \ref{corcount} is saturated, then eMGFs in one variables can be used to express arbitrary $\Lambda^{\rm sv}\!\left[\begin{smallmatrix}j_1&\dots&j_\ell\\k_1&\dots&k_\ell\end{smallmatrix};z,\tau\right]$ with entries in the above range when combined with MGFs and single-valued MZVs. The counting of indecomposable eMGFs at $|A|{+}|B| \leq 10$ in section~4 of \cite{Hidding:2022vjf} and the KZB equations obeyed by the generating series of genus-one Koba-Nielsen integrals in sections 4 and 5 of \cite{Dhoker:2020gdz} with matrix representations of the generators $b_k^{(j)},\ep_k^{(j)}$ suggest that this is the case. Hence, the counting of shuffle-independent single-valued eMPLs $\Lambda^{\rm sv}\!\left[\begin{smallmatrix}j_1&\dots&j_\ell\\k_1&\dots&k_\ell\end{smallmatrix};z,\tau\right]$ is actually expected to {\it equal} the number of indecomposable eMGFs in one variable with lattice weights in (\ref{ctemgf.02}).

\begin{table}[h]
\centering
\begin{tabular}{c||c |c|c|c|c|c|c|c|c|c|c|c|c|c|c|c|c|c}
\backslashbox{$k$}{$w$}&{-}8&{-}7&{-}6&{-}5&{-}4&{-}3&{-}2&{-}1&0&1&2&3&4&5&6&7&8
\\\hline\hline
2&&&&&&&&&1&&&&&&&&
\\\hline
3&&&&&&&&1&&1&&&&&&&
\\\hline
4&&&&&&&1&&1&&1&&&&&&
\\\hline
5&&&&&&1&&2&&2&&1&&&&&
\\\hline
6&&&&&1&&2&&3&&2&&1&&&&
\\\hline
7&&&&1&&3&&5&&5&&3&&1&&&
\\\hline
8&&&1&&3&&7&&8&&7&&3&&1&&
\\\hline
9&&1&&4&&9&&14&&14&&9&&4&&1&
\\\hline
10&1&&4&&12&&20&&25&&20&&12&&4&&1
\end{tabular}
\caption{\label{tab:count} Counting of independent eMGFs, sorting the columns by modular weight $(0,w)$ with $w=|B|{-}|A|$ and the rows by degree $k= |A|{+}|B|$.}
\end{table}

Table \ref{tab:count} displays the numbers of indecomposable eMGFs in one variable up to and including $|A|{+}|B|=10$, organized by the degree $k= \sum_{i}k_i = |A|+|B|$ of the associated single-valued eMPL and its antiholomorphic modular weight $(0,w)$ with $w = \sum_i(k_i{-}2j_i{-}2) =|B|{-}|A|$. All entries are consistent with table 2 of \cite{Hidding:2022vjf} in early section 4 of the reference, which actually counts the union of indecomposable MGFs and one-variable eMGFs (excluding MZVs). The counting of indecomposable MGFs one needs to subtract from the entries of table 2 of \cite{Hidding:2022vjf} can be found in table 2 in section 6.2 \cite{Gerken:2020yii}. Combining the input from the references leads to the numbers in table \ref{tab:count} obtained from the counting of shuffle-independent single-valued eMPLs.

\subsubsection{Counting shuffle-independent single-valued eMPLs}
\label{sec632sh}

The shuffle-independent coefficients $\Lambda^{\rm sv}\!\left[\begin{smallmatrix}j_1&\dots&j_\ell\\k_1&\dots&k_\ell\end{smallmatrix};z,\tau\right]$ of the series $\mathbbm{\Lambda}^{\rm sv}_{x,y}(z,\tau)$ in (\ref{cor.3.4d}) are in one-to-one correspondence with Lie polynomials in the accompanying generators $b_k^{(j)}$ \cite{Reutenauer}. 
Hence, the problem of bounding the number of indecomposable eMGFs in one variable via Corollary \ref{corcount} reduces to the counting of Lie polynomials in the free-Lie-algebra generators $b_k^{(j)}$ with $k\geq 2$ and $0\leq j \leq k{-}2$. For given lattice weights $|A|$, $|B|$, the admissible Lie polynomials are built from $b_{k_1}^{(j_1)}\ldots b_{k_\ell}^{(j_\ell)}$ at different modular depth $\ell$ whose entries $j_i$, $k_i$ satisfy (\ref{ctemgf.02}).

A convenient way of simultaneously addressing all lattice weights compatible with a given degree $k_1{+}\ldots{+}k_\ell = |A|{+}|B|$ is to exploit the action (\ref{lieg1.28}) of the $\mathfrak{sl}_2$ generators $\ep_0$ and $\ep^\vee_0$ on $b_k^{(j)}$. Given that the set $\{ b_k^{(j)} , \, j=0,1,\ldots,k{-}2 \}$ at fixed $k$ forms a $(k{-}1)$-dimensional irreducible representation (irrep) of the $\mathfrak{sl}_2$ algebra to be denoted by $(k{-}1)$, both concatenation products and Lie polynomials in $b_{k_1}^{(j_1)}\ldots b_{k_\ell}^{(j_\ell)}$ fall into the tensor products $(k_1{-}1)\otimes \ldots\otimes (k_\ell{-}1)$. The latter can be straightforwardly decomposed into irreps of $\mathfrak{sl}_2$ by iteratively using
\beq
(i)\otimes (j)= \big(|i{-}j|{+}1 \big)\oplus \big(|i{-}j|{+}3\big)\oplus\dots\oplus (i{+}j{-}3)\oplus (i{+}j{-}1)
\label{tensssl2}
\eeq
Whenever an irrep $(d)$ in this tensor product corresponds to Lie polynomials in $b_k^{(j)}$ as opposed to shuffles, then the middle $d$ entries in the relevant row of table \ref{tab:count} are incremented by one. In other words, irreps $(d)$ of $\mathfrak{sl}_2$ corresponding to Lie polynomials in $b_{k_1}^{(j_1)}\ldots b_{k_\ell}^{(j_\ell)}$ predict indecomposable eMGFs with $|A|{+}|B|=k_1{+}\ldots{+}k_\ell$, namely one representative per antiholomorphic modular weight ${-}d{+}1,\,{-}d{+}3,\,\ldots,d{-}3,\, d{-}1$.

It remains to identify the irreps of $\mathfrak{sl}_2$ in the tensor product $(k_1{-}1)\otimes \ldots\otimes (k_\ell{-}1)$ that correspond to Lie polynomials. The criterion depends on the modular depth $\ell$ of the $b_{k_1}^{(j_1)}\ldots b_{k_\ell}^{(j_\ell)}$ under consideration and whether some of the $k_r=k_s$ coincide for $1\leq r<s\leq \ell$. The analogous identification of ``shuffle-irreducible'' combinations of $\ep_k^{(j)}$ up to $\ell=3$ has been performed in section 3.4 of \cite{Dorigoni:2024oft} and readily translates into the setting of~$b_k^{(j)}$:
\begin{itemize}
\item $b_{k_1}^{(j_1)}$ at $\ell=1$ is a particularly simple case of a Lie polynomial, and the associated $(k_1{-}1)$ contributes to the entries of table \ref{tab:count};
\item for $b_{k_1}^{(j_1)} b_{k_2}^{(j_2)}$ at $\ell=2$ with $k_1=k_2$, only the antisymmetric part of the tensor product $(k_1{-}1)^{\otimes 2}$ corresponds to Lie polynomials which is given by $\big(|i{-}j|{+}3 \big)\oplus\big(|i{-}j|{+}7 \big)\oplus \dots\oplus (i{+}j{-}7)\oplus (i{+}j{-}3) 
= (3)\oplus (7)\oplus \dots \oplus (2k_1{-}9)\oplus (2k_1{-}5)$, skipping every other irrep in (\ref{tensssl2}) including the outermost~ones;
\item for $b_{k_1}^{(j_1)} b_{k_2}^{(j_2)}$ at $\ell=2$ with $k_1\neq k_2$, independent commutators $[b_{k_1}^{(j_1)}, b_{k_2}^{(j_2)}]$ are characterized by ordered pairs $k_1<k_2$ such that each irrep in the tensor product $(k_1{-}1) \otimes (k_2{-}1)$ with $k_1<k_2$ contributes to table \ref{tab:count};\footnote{At $k_1=3$ and $k_2=4$, for instance, the brackets $[b^{(j_1)}_3,b^{(j_2)}_4]$ built from a doublet $(2)$ and triplet $(3)$ of $\mathfrak{sl}_2$ fall into the tensor product $(2)\otimes (3)= (2)\oplus (4)$. The irreps on the right side signal indecomposable eMGFs with $|A|{+}|B|=7$ and modular weights
$(0,{-}1)$ and $(0,1)$ from $(2)$ as well as $(0,{-}3)$, $(0,{-}1)$, $(0,1)$ and $(0,3)$ from~$(4)$. The additional irreps contributing to the relevant line at $k=7$ in table \ref{tab:count} are one $(6)$ from $b_7^{(j)}$, one $(4)$ from $[b_2,b^{(j)}_5]$ and one $(2)$ from $[b_2,[b_2,b^{(j)}_3]]$.}
\item for $b_{k_1}^{(j_1)} b_{k_2}^{(j_2)} b_{k_3}^{(j_3)}$ at $\ell=3$ with all of $k_1= k_2=k_3$ equal, Lie-polynomials fall into half of the mixed-symmetry projection of the triple tensor product; in other words, contributions to table \ref{tab:count} are obtained by removing the totally symmetric and totally antisymmetric parts of $(k_1{-}1)^{\otimes 3}$ and diving the multiplicities of the leftover irreps by~2, see (3.89) of \cite{Dorigoni:2024oft};
\item for $b_{k_1}^{(j_1)} b_{k_2}^{(j_2)} b_{k_3}^{(j_3)}$ at $\ell=3$ with $k_1= k_2$ but one distinct $k_3\neq k_1,k_2$, one can take $[b_{k_1}^{(j_1)}, [b_{k_2}^{(j_2)},b_{k_3}^{(j_3)} ]]$ as representatives of the independent Lie polynomials; each irrep within $(k_1{-}1) \otimes (k_2{-}1) \otimes (k_3{-}1)$ contributes to table \ref{tab:count};
\item for $b_{k_1}^{(j_1)} b_{k_2}^{(j_2)} b_{k_3}^{(j_3)}$ at $\ell=3$ with $k_1, k_2,k_3$ pairwise distinct, independent Lie polynomials can be spanned by $[b_{k_1}^{(j_1)}, [b_{k_2}^{(j_2)},b_{k_3}^{(j_3)} ]]$ and $[b_{k_2}^{(j_2)}, [b_{k_1}^{(j_1)},b_{k_3}^{(j_3)} ]]$ after imposing an ordering $k_1<k_2<k_3$. By the two inequivalent bracketings under Jacobi identities, each irrep within $(k_1{-}1) \otimes (k_2{-}1) \otimes (k_3{-}1)$ contributes twice to table \ref{tab:count}.
\end{itemize}
Similar group-theoretic methods can be applied to $b_{k_1}^{(j_1)} \ldots b_{k_\ell}^{(j_\ell)}$ at higher modular depth $\ell\geq 4$ to organize their independent Lie polynomials into irreps of $\mathfrak{sl}_2$.

\begin{table}[h]
\centering
\begin{tabular}{c||c }
\text{Lie polynomials in}&\text{irreps of $\mathfrak{sl}_2$}
\\\hline\hline
$b_9^{(j)}$ & $$(8)$$\\\hline
$b_2\, b_7^{(j_2)}$& $(1)\otimes (6)= (6)$
\\\hline
$b_3^{(j_1)}\, b_6^{(j_2)}$& $(2)\otimes (5)=(4)\oplus (6)$
\\\hline
$b_4^{(j_1)}\, b_5^{(j_2)}$& $(3)\otimes (4)=(2)\oplus (4)\oplus (6)$
\\\hline
$b_2\, b_2\, b_5^{(j_3)}$& $(1)\otimes (1)\otimes (4)=(4)$
\\\hline
$b_2\, b^{(j_2)}_3\, b^{(j_3)}_4$& $(1)\otimes (2)\otimes (3)=(2)\oplus (4)$
\\\hline
$b^{(j_2)}_3\, b_2 \,b^{(j_3)}_4$& $(1)\otimes (2)\otimes (3)=(2)\oplus (4)$
\\\hline
$b^{(j_1)}_3 \, b^{(j_2)}_3\, b^{(j_3)}_3$& 
%$(2)\otimes (1)=(2)$
$\frac{1}{2}\big[(2)^{\otimes 3} \ominus (4) \big] =(2)$
\\\hline
$b_2\,b_2\,b_2\,b^{(j_4)}_3$ & $(1)\otimes (1)\otimes (1)\otimes (2)=(2)$
\end{tabular}
\caption{\label{tab:count2} Irreps of $\mathfrak{sl}_2$ for Lie polynomials in $b^{(j_1)}_{k_1}\ldots b^{(j_\ell)}_{k_\ell}$ of degree $k_1{+}\ldots{+}k_\ell = 9$ which give rise to the entries in the second line from below of table \ref{tab:count}. By footnote \ref{ftaltsym}, the $\ominus$ symbol in the second line from below instructs to remove the totally symmetric irrep $(4)$ from $(2)^{\otimes 3}$.}
\end{table}

As an example at degree $k_1{+}\ldots{+}k_\ell = 9$, the irreps of $\mathfrak{sl}_2$ contributing to indecomposable eMGFs with lattice weights $ |A|{+}|B|=9$ are listed in the right column of table \ref{tab:count2}. This reproduces the second line from below of table \ref{tab:count} by assembling its entries from the irreps $(8)\oplus 3(6)\oplus 5(4)\oplus 5(2)$ of the Lie polynomials (recalling that $(d)$ contributes 1 to the middle $d$ entries). Note that for the case of $b^{(j_1)}_3 b^{(j_2)}_3 b^{(j_3)}_3$, the  relevant tensor product is $(2)\otimes (2)\otimes (2) = (4)\oplus(2)\oplus(2)$ but only one of the $(2)$ on the right side corresponds to Lie polynomials by the above  criteria.\footnote{\label{ftaltsym}The triple tensor product $(2)\otimes (2)\otimes (2) = (4)\oplus(2)\oplus(2)$ has $(4)$ as its totally symmetric part and no totally antisymmetric part. Hence, the mixed-symmetry projection is $(2)\oplus(2)$, resulting in a single $(2)$ contribution to the Lie polynomials $[b^{(j_1)}_3 ,[b^{(j_2)}_3, b^{(j_3)}_3]]$ by taking half of the multiplicities according to the above criterion.} The case of $b_2 b_2 b_2 b^{(j)}_3$ at modular depth four gives rise to one $(2)$ of Lie polynomials since $[b_2,[ b_2,[ b_2 ,b^{(j)}_3]]]$ is the unique bracketing of the generators.

%%%%%%%%%%%%%%%%%%%%%%%%%%%%%%%%%%%%%%%%%%%%%%%%%%
%%%%%%%%%%%%%%%%%%%%%%%%%%%%%%%%%%%%%%%%%%%%%%%%%%
%%%%%%%%%%%%%%%%%%%%%%%%%%%%%%%%%%%%%%%%%%%%%%%%%%
%%%%%%%%%%%%%%%%%%%%%%%%%%%%%%%%%%%%%%%%%%%%%%%%%%
\section{Conclusion and outlook}
\label{sec:outopen}

In this work, the construction of non-holomorphic modular forms from equivariant iterated Eisenstein integrals \cite{Brown:2017qwo, Brown:2017qwo2, Dorigoni_2022, Dorigoni:2024oft} is generalized to the case of additional Kronecker-Eisenstein integration kernels $\dd \tau \, f^{(k)}(u\tau{+}v,\tau)$ that depend on a point $z=u\tau{+}v$ on the torus. We present a generating series that combines iterated $\tau$ integrals over (Kronecker-)Eisenstein kernels with their complex conjugates and single-valued multiple zeta values to attain the desired modular properties. The differential equations and double periodicity of our equivariant series in~$z$ identifies its coefficients as single-valued elliptic polylogarithms. Our construction relies on Lie-algebraic tools including Tsunogai derivations and zeta generators, as well as their action on the fundamental group of the torus. The counting of single-valued elliptic polylogarithms at fixed modular weight and degrees in the accompanying Lie-algebra generators translates into that of independent elliptic modular graph forms.

The techniques and results of this work raise a variety of follow-up questions at the interface of string perturbation theory, algebraic geometry, and number theory, with potential applications to Feynman integrals in particle physics and gravity:
\begin{itemize}
\item An immediate task is to relate the single-valued elliptic polylogarithms in this work to those of Baune, Broedel and M\"ockli obtained from complementary methods~\cite{Baune:2025svempl}.
\item Another direction of
 follow-up research is the construction of single-valued elliptic polylogarithms in multiple variables $z_1,\ldots,z_n$. Their generating series is expected to mirror the structure of the genus-zero formulae for single-valued polylogarithms in $n{+}1$ variables in terms of zeta generators \cite{Frost:2023stm}. It will be interesting to compare different fibration bases, including formulations of single-valued $z_1,\ldots,z_n$-dependent elliptic polylogarithms as iterated $\tau$~integrals.
\item The generating series of elliptic modular graph forms in \cite{Dhoker:2020gdz} obey the same type of differential equations as that of equivariant iterated integrals in this work, but with matrix representations of some of the Lie-algebra generators. It would be interesting to investigate echos of zeta generators in these matrix representations and to understand the general link between the Lie-algebra structure of this work and twisted cohomology, for instance, via double-copy formulae for Riemann-Wirtinger integrals \cite{Bhardwaj:2023vvm, Mazloumi:2024wys, Pokraka:2025zlh}. 
\item Upon specialization to rational points $z=u\tau{+}v$ with $u,v \in \mathbb{Q}$, the Kronecker-Eisenstein kernels $\dd \tau \, f^{(k)}(u\tau{+}v,\tau)$ become Eisenstein series of congruence subgroups of ${\rm SL}_2(\mathbb Z)$ \cite{Broedel:2018iwv}. It would be interesting to extend the construction of their equivariant iterated integrals in \cite{Drewitt:2023con, Duhr:2025lvr} to higher depth and to compare with special values of the single-valued elliptic polylogarithms constructed in this work.
\item The recent literature offers several approaches to polylogarithms on Riemann surfaces of arbitrary genus \cite{BEFZ:2110, Ichikawa:2022qfx, BEFZ:2212, DHoker:2023vax, Baune:2024biq, DHoker:2024ozn, Baune:2024ber, DHoker:2025szl, Baune:2025sfy, Ichikawa:2025kbi}. It would be interesting to explore the analogues of the Tsunogai derivations or zeta generators beyond genus one, to construct single-valued versions of higher-genus polylogarithms and to pinpoint their connection with higher-genus modular graph functions \cite{DHoker:2013fcx, DHoker:2014oxd, Pioline:2015qha, DHoker:2017pvk, DHoker:2018mys, DHoker:2020tcq, Basu:2020goe, Basu:2021xdt} and tensors \cite{DHoker:2020uid}.
\item Zeta generators give rise to a formal similarity between the motivic coaction and the single-valued map of genus-zero polylogarithms \cite{Frost:2023stm, Frost:2025lre}. In the same way as the description of single-valued iterated Eisenstein integrals \cite{Dorigoni:2024oft} via zeta generators translated into a recent proposal for the motivic coaction of their holomorphic counterparts \cite{Kleinschmidt:2025dtk}, we plan to study the implications of the single-valued elliptic polylogarithms of this work for explicit coaction formulae for meromorphic or Brown-Levin elliptic polylogarithms.
\end{itemize}

%%%%%%%%%%%%%%%%%%%%%%%%%%%%%%%%%%%%%%%%%%%%%%%%%%
%%%%%%%%%%%%%%%%%%%%%%%%%%%%%%%%%%%%%%%%%%%%%%%%%%
%%%%%%%%%%%%%%%%%%%%%%%%%%%%%%%%%%%%%%%%%%%%%%%%%%
%%%%%%%%%%%%%%%%%%%%%%%%%%%%%%%%%%%%%%%%%%%%%%%%%%
\appendix

%%%%%%%%%%%%%%%%%%%%%%%%%%%%%%%%%%%%%%%%%%%%%%%%%%
%%%%%%%%%%%%%%%%%%%%%%%%%%%%%%%%%%%%%%%%%%%%%%%%%%
%%%%%%%%%%%%%%%%%%%%%%%%%%%%%%%%%%%%%%%%%%%%%%%%%%
%%%%%%%%%%%%%%%%%%%%%%%%%%%%%%%%%%%%%%%%%%%%%%%%%%
\section{Standard definitions and background material}
\label{sec:A}

This appendix gathers standard definitions and techniques relevant for different types of iterated integrals at genus zero and genus one.

%%%%%%%%%%%%%%%%%%%%%%%%%%%%%%%%%%%%%%%%%%%%%%%%%%
%%%%%%%%%%%%%%%%%%%%%%%%%%%%%%%%%%%%%%%%%%%%%%%%%%
\subsection{Multiple polylogarithms and multiple zeta values}
\label{sec:A.1}

Our conventions for multiple polylogarithms (MPLs) are fixed by \cite{Goncharov:2001iea}
\beq
G(a_1,\ldots,a_r;z) = \int^z_0 \frac{\dd t}{t{-}a_1} \, G(a_2,\ldots,a_r;t)
\, , \ \ \ \ \ \  G(\emptyset;z) = 1
\label{appA.01}
\eeq
where $r\in \mathbb N$ and $z,a_1,\ldots,a_r \in \mathbb C$. Endpoint divergences for $a_1=z$
and $a_r= 0$ are shuffle-regularized by assigning regularized values (see for instance \cite{Panzer:2015ida, Abreu:2022mfk})
\beq
G(0;z) = \log(z) \, , \ \ \ \ \ \ G(z;z) = - \log(z)
\label{appA.02}
\eeq
Multiple zeta values (MZVs) are defined by nested sums over integers $k_i$ and can alternatively be obtained from special values of MPLs at $z=1$
\begin{align}
    \zeta_{n_1,n_2,\dots, n_r}&=\sum_{0<k_1<\dots <k_r}k_1^{-n_1}k_2^{-n_2}\dots k_r^{-n_r}
\label{appA.03}\\
    &=(-1)^r \, G( \underbrace{0,\dots, 0}_{n_r-1},1,\dots,\underbrace{0,\dots, 0}_{n_2-1},1,\underbrace{0,\dots, 0}_{n_1-1},1;1)
    \notag
\end{align}
where $n_1,\ldots,n_r \in \mathbb N$ with $n_r\geq 2$.

%%%%%%%%%%%%%%%%%%%%%%%%%%%%%%%%%%%%%%%%%%%%%%%%%%
%%%%%%%%%%%%%%%%%%%%%%%%%%%%%%%%%%%%%%%%%%%%%%%%%%
\subsection{$f$-alphabet description of multiple zeta values}
\label{sec:A.2}

MZVs obey infinite families of relations over $\mathbb Q$ that follow from their representations (\ref{appA.03}) as nested sums or iterated integrals \cite{Jianqiang, BurgosFresan}. In several situations of this work, we shall rely on a description of conjecturally $\mathbb Q$-independent MZVs through the $f$-alphabet \cite{Brown:2011ik, BrownTate}: Its key ingredients are non-commuting generators $f_3,f_5,\ldots$ for odd Riemann zeta values and a commuting generator $f_2$ for even Riemann zeta values obtained from an isomorphism $\rho$ via
\beq
\rho( \zeta_{2k+1} ) = f_{2k+1}
\, , \ \ \ \ \ \
\rho( \zeta_{2k} ) = f_{2k} = \frac{\zeta_{2k}}{(\zeta_2)^k}\, (f_2)^k
\label{appA.04}
\eeq
The extension of the $f$-alphabet isomorphism $\rho$ to MZVs of depth $\geq 2$ and to products of MZVs is largely determined by the Hopf-algebra structure of motivic MZVs \cite{Goncharov:2005sla, BrownTate, Brown2014MotivicPA, brown2017notes}\footnote{We refer to the references in the algebraic-geometry literature for the definition of motivic MZVs $\zeta^\mot_{n_1,\ldots,n_r}$. Even though the $f$-alphabet isomorphism $\rho$ and the single-valued map are only well-defined in a motivic setting, we shall write $\zeta_{n_1,\ldots,n_r}$ instead of $\zeta^\mot_{n_1,\ldots,n_r}$ in (\ref{appA.06}), (\ref{appA.08}), (\ref{appA.sv}) and (\ref{sv335}) in abuse of notation.}: multiplication of (motivic) MZVs translates into the shuffle product of the $f$-alphabet
\beq
\rho( \zeta_{n_1,\ldots,n_r} \, \zeta_{m_1,\ldots ,m_s} ) = \rho(  \zeta_{n_1,\ldots,n_r} ) \shuffle  \rho(\zeta_{m_1,\ldots ,m_s} )
\label{appA.05}
\eeq
where $f_2^m f_{i_1}\cdots f_{i_r} \shuffle f_2^n f_{j_1}\cdots f_{j_s} = f_2^{m+n} (f_{i_1}\cdots f_{i_r} \shuffle  f_{j_1}\cdots f_{j_s})$ for any $m,n \in \mathbb N_0$ and $i_k,j_k\in 2\mathbb N{+}1$, with the standard shuffle product for words in $f_w$ with $w \in 2\mathbb N{+}1$.\footnote{The shuffle product among words in $f_w$ with $w\in 2\mathbb N{+}1$ is defined by the role of the empty word $\emptyset$ as a neutral element, $\emptyset \shuffle f_{i_1}\cdots f_{i_r} = f_{i_1}\cdots f_{i_r} \shuffle  \emptyset = f_{i_1}\cdots f_{i_r} $ and the following recursion for $r,s\geq 1$:
\[
f_{i_1} f_{i_2}\cdots f_{i_r} \shuffle  f_{j_1} f_{j_2}\cdots f_{j_s}
= f_{i_1}(f_{i_2}\cdots f_{i_r} \shuffle  f_{j_1}\cdots f_{j_s})
+ f_{j_1} (f_{i_1}\cdots f_{i_r} \shuffle  f_{j_2}\cdots f_{j_s})
\]}
The Goncharov-Brown coaction $\Delta$ of (motivic) MZVs \cite{Goncharov:2001iea, Goncharov:2005sla, BrownTate, Brown:2011ik} is imposed to translate into the deconcatenation coaction $\Delta_{\rm dec}$ of words in $f_w$ with $w \in 2\mathbb N{+}1$ according to $\Delta_{\rm dec}\circ \rho = \rho \circ \Delta$. This fixes the action of the $f$-alphabet isomorphism $\rho$ on each independent indecomposable higher-depth MZV $\zeta_{n_1,\ldots,n_r}$ up to a rational number $q_i \in \mathbb Q$ multiplying $f_{n_1+\ldots+n_r}$, for instance
\begin{align}
\rho( \zeta_{3,5}) &= - 5 f_3 f_5 + q_8 f_8  \label{appA.06} \\
\rho( \zeta_{3,7}) &= - 14 f_3 f_7 - 6 f_5 f_5 + q_{10} f_{10} \notag \\
\rho( \zeta_{3,3,5}) &= - 5 f_3 f_3 f_5 - 45 f_2 f_9 - \frac{6}{5} f_2^2 f_7  + \frac{4}{7} f_2^3 f_5 +q_{11} f_{11} 
\notag
\end{align}
By choosing $q_i=0$ for all depth $\geq 2$ MZVs in the reference $\mathbb Q$-bases of the datamine \cite{Blumlein:2009cf}, one arrives at the 
$f$-alphabet images up to weight 16 given in \cite{Schlotterer:2012ny}. The canonical choice of $f$-alphabet isomorphism described in \cite{Dorigoni:2024iyt}, by contrast, yields $(q_8,q_{10},q_{11}) = ( \frac{100471 }{ 35568 } ,  \frac{408872741707 }{40214998720} , \frac{1119631493 }{14735232 } )$ for the rational parameters in (\ref{appA.06}), also see \cite{Keilthy} for an earlier discussion of canonical zeta generators and the ancillary file of \cite{Dorigoni:2024oft} for the $f$-alphabet images of the indecomposable MZVs up to and including weight 17.

The definition (\ref{not.05}) of the series $\mathbb M^{\rm sv}_\sigma$ in zeta generators involves the inverse of the isomorphism $\rho$ which is given as follows up to weight 11 as a consequence of (\ref{appA.05}) and (\ref{appA.06}):
\begin{align}
\rho^{-1}(f_3 f_5) &= - \frac{1}{5} \zeta_{3, 5} +\frac{q_8}{5}\zeta_8   \label{appA.08}\\
\rho^{-1}(f_3 f_7) &=  - \frac{1}{14} \zeta_{3,7}   - \frac{ 3}{14} (\zeta_5)^2 
+\frac{q_{10}}{14}\zeta_{10}  
\notag \\
\rho^{-1}(f_3 f_3 f_5) &= -\frac{1}{5} \zeta_{3, 3, 5} +
\frac{1}{2}\zeta_5 \zeta_6 - 
  \frac{3}{5}  \zeta_7 \zeta_4 -
  9 \zeta_9 \zeta_2+
  \frac{q_{11}}{5} \zeta_{11}   \notag
\end{align}
As a key virtue of the $f$-alphabet, the single-valued map of MZVs \cite{Schnetz:2013hqa, Brown:2013gia} takes the simple~form
\beq
{\rm sv}(f_2^N f_{i_1}\ldots f_{i_r}) = \delta_{N,0} \sum_{j=0}^r f_{i_j}\ldots f_{i_2} f_{i_1} \shuffle f_{i_{j+1}} f_{i_{j+2}} \ldots f_{i_{r}}
 \label{appA.00}
\eeq
By (\ref{appA.04}) and (\ref{appA.06}), this implies
\beq
{\rm sv}(\zeta_{2k+1}) = 2 \zeta_{2k+1} \, , \ \ \ \ \ \ 
{\rm sv}(\zeta_{2k}) = 0 \, , \ \ \ \ \ \ 
{\rm sv}(\zeta_{3,5}) = -10 \zeta_{3}\zeta_5 
 \label{appA.sv}
\eeq
and in particular identifies the first (conjecturally) indecomposable single-valued MZV beyond depth one to occur at weight 11:
\beq
{\rm sv}(\zeta_{3,3,5})  = 2 \zeta_{3,3,5} - 5 \zeta_3^2 \zeta_5 + 90 \zeta_2 \zeta_9
+ \frac{12}{5} \zeta_2^2 \zeta_7 - \frac{8}{7} \zeta_2^3 \zeta_5
 \label{sv335}
\eeq

%%%%%%%%%%%%%%%%%%%%%%%%%%%%%%%%%%%%%%%%%%%%%%%%%%
%%%%%%%%%%%%%%%%%%%%%%%%%%%%%%%%%%%%%%%%%%%%%%%%%%
\subsection{Kronecker-Eisenstein series and kernels}
\label{sec:A.3}

Following the normalization conventions
\beq
\theta_1(z,\tau) = q^{1/8} (e^{i\pi z} - e^{-i\pi z} ) \prod_{n=1}^\infty (1-q^n)(1-e^{2\pi i z} q^n)(1-e^{-2\pi i z} q^n)
\label{appA.09}
\eeq
for the odd Jacobi theta function
(with $z,\tau \in \mathbb C$ and $\Im \tau > 0$), we define the meromorphic Kronecker-Eisenstein series $F(z,\alpha,\tau)$ and the associated integration kernels
$g^{(n)}(z,\tau)$ by
\beq
F(z,\alpha,\tau) = \frac{ \theta'_1(0,\tau)  \theta_1(z{+}\alpha,\tau)  }{ \theta_1(z,\tau)  \theta_1(\alpha,\tau)  }
= \sum_{n=0}^\infty \alpha^{n-1} g^{(n)}(z,\tau)
\label{appA.10}
\eeq
leading, for instance, to $g^{(0)}(z,\tau) = 1$ and $g^{(1)}(z,\tau) = \partial_z \log \theta_1(z,\tau)$. In the CEE connection (\ref{defKcon}), the commuting indeterminate $\alpha \in \mathbb C \setminus \{0\}$ is promoted to the derivation ${\rm ad}_x$ at the level of the series expansion in (\ref{appA.10}).

The quasi-periodicity of $\theta_1$ gives rise to the following monodromies along the $A$- and $B$-cycles (for $n\geq 0$)
\begin{align}
g^{(n)}(z{+}1,\tau) &= g^{(n)}(z,\tau) 
\label{bmon} \\
g^{(n)}(z{+}\tau,\tau) &= g^{(n)}(z,\tau)  + \sum^n_{k=1}  \frac{ (-2\pi i)^k }{k!}  g^{(n-k)}(z,\tau)
\notag
\end{align}
As a consequence, $g^{(1)}(z,\tau)$ has simple poles at all $z \in \mathbb Z {+} \tau \mathbb Z$ with unit residue whereas $g^{(k)}(z,\tau)$ with $k \geq 2$ are regular at the origin but have simple poles
at $z = m\tau {+}n$ with $m,n\in \mathbb Z$ and $m\neq 0$ as mandated by the $B$-cycle monodromies in (\ref{bmon}).

A doubly-periodic but non-meromorphic completion of $F(z,\alpha,\tau)$ and $g^{(n)}(z,\tau)$ in (\ref{appA.10}) is achieved through an extra factor of $\exp (2\pi i \alpha  \frac{\Im z}{\Im \tau} )$ in the generating series
\beq
\Omega(z,\alpha,\tau) = \exp \bigg(2\pi i \alpha \, \frac{\Im z}{\Im \tau} \bigg) \, F(z,\alpha,\tau)
= \sum_{n=0}^\infty \alpha^{n-1} f^{(n)}(z,\tau)
\label{appA.11}
\eeq
such that, for instance, $f^{(0)}(z,\tau) = 1$ and 
$f^{(1)}(z,\tau) = \partial_z \log \theta_1(z,\tau) + 2\pi i \frac{\Im z}{\Im \tau}$. The non-meromorphicity of $\Omega$ and $f^{(n)}$ with $n\geq 1$ in $z$ is reflected by the derivatives
\begin{align}
\partial_{\bar z}  f^{(n)}(z,\tau)  &= -\frac{\pi}{\Im \tau} \, f^{(n-1)}(z,\tau) + \pi \delta_{n,1} \delta^2(z)
\notag \\
\partial_{\bar z} \Omega(z,\alpha,\tau) &= - \frac{\pi \alpha}{\Im \tau}\, \Omega(z,\alpha,\tau) + \pi  \delta^2(z)
\label{pbarf}
\end{align}
where the delta distribution arises from the simple pole of $ f^{(1)}(z,\tau) $ at the origin. The entire tower of $f^{(n)}(z,\tau) $ with $n \geq 2$ is regular throughout $z\in \mathbb C$. At fixed $z\in \mathbb C$, both $\Omega$ and $f^{(n)}$ in (\ref{appA.11}) have non-zero antiholomorphic derivatives under $\partial_{\bar \tau}$. However, when viewed as a function of $u,v \in \mathbb R$ instead of $z = u\tau{+}v \in \mathbb C$, the lattice-sum representation (\ref{latsumgf}) exposes that any $f^{(n)}(u\tau{+}v,\tau)$ with $n\geq 0$ is holomorphic in $\tau \in \mathbb C$ for $\Im \tau>0$.

The meromorphic and doubly-periodic integration kernels can be straightforwardly converted into each other via
\beq
 f^{(n)}(z,\tau) =  \sum_{k=0}^n  \frac{1}{k!} \, \big(  2\pi i u\big)^k \, g^{(n-k)}(z,\tau)
\label{appA.12}
\eeq
Note that the mixed heat equation
\beq
2\pi i \partial_\tau F(z,\eta,\tau) = \partial_z \partial_\eta F(z,\eta,\tau)
\label{appA.31}
\eeq
of the meromorphic Kronecker-Eisenstein series has the following doubly-periodic echo
\beq
2\pi i \partial_\tau \Omega(u\tau{+}v,\eta,\tau) = \partial_v \partial_\eta \Omega(u\tau{+}v,\eta,\tau)
\label{appA.32}
\eeq
when differentiating in the co-moving coordinate $v$.

%%%%%%%%%%%%%%%%%%%%%%%%%%%%%%%%%%%%%%%%%%%%%%%%%%
%%%%%%%%%%%%%%%%%%%%%%%%%%%%%%%%%%%%%%%%%%%%%%%%%%
\subsection{Path-ordered exponentials}
\label{sec:A.4}

Given a manifold $M$ and a Lie-algebra valued one-form connection $\conpl(t)$
depending on $t\in M$, the path-ordered exponential for a path between endpoints $a,b \in M$ is defined by
\begin{align}
{\rm Pexp}\bigg( \int^a_b \conpl(t) \bigg) &= 1 + \int^a_b \conpl(t) + \sum_{n=2}^\infty
 \int^a_b \conpl(t_1) \int^{t_1}_b \conpl(t_2) \ldots \int^{t_{n-1}}_b \conpl(t_n)
\notag \\
&= 1 + \int^a_b \conpl(t) + \int^a_b\conpl(t_1) \int^{t_1}_b \conpl(t_2) + \ldots
  \label{defpexp}
\end{align}
If the connection obeys the Maurer-Cartan equation or flatness condition
\beq
\dd_t \conpl(t) = \conpl(t) \wedge \conpl(t)
\label{fltness}
\eeq
then the path-ordered exponential (\ref{defpexp}) is homotopy invariant, i.e.\ only depends
on the endpoints $a,b$ and the homotopy class of the integration path on $M$.
In this case, the path-ordered exponential obeys the following differential equations
in the endpoints:
\begin{align}
\dd_a {\rm Pexp}\bigg( \int^a_b \conpl(t) \bigg) &= \conpl(a) {\rm Pexp}\bigg( \int^a_b \conpl(t) \bigg)  \notag \\
\dd_b {\rm Pexp}\bigg( \int^a_b \conpl(t) \bigg) &= - {\rm Pexp}\bigg( \int^a_b \conpl(t) \bigg) \conpl(b)   \label{notsec.01}
\end{align}
The composition-of-paths formula for iterated integrals implies that
\begin{align}
 {\rm Pexp}\bigg( \int^a_c \conpl(t) \bigg) &=   {\rm Pexp}\bigg( \int^a_b \conpl(t) \bigg)  {\rm Pexp}\bigg( \int^b_c \conpl(t) \bigg)
 \label{cpath}
\end{align}
for arbitrary points $a,b,c \in M$ as long as the homotopy classes of the paths on the right side are compatible with those on the left side. In particular, setting $a=c$ identifies the following inverse with respect to the concatenation product
\beq
 {\rm Pexp}\bigg( \int^a_b \conpl(t) \bigg)^{-1}  =  {\rm Pexp}\bigg( \int^b_a \conpl(t) \bigg) 
 \label{invpath}
\eeq

%%%%%%%%%%%%%%%%%%%%%%%%%%%%%%%%%%%%%%%%%%%%%%%%%%
%%%%%%%%%%%%%%%%%%%%%%%%%%%%%%%%%%%%%%%%%%%%%%%%%%
%%%%%%%%%%%%%%%%%%%%%%%%%%%%%%%%%%%%%%%%%%%%%%%%%%
%%%%%%%%%%%%%%%%%%%%%%%%%%%%%%%%%%%%%%%%%%%%%%%%%%
\section{Complex conjugations}
\label{sec:B}

This appendix gathers explicit expansion formulae and differential equations of complex conjugates of various generating series in the main text. The formulae of this appendix may come in handy for
concrete computations where detailed control on all minus signs is important.

\subsection{Complex conjugate generating series}
\label{sec:B.1}

The generating series (\ref{genKempl}) of eMPLs defined by (\ref{defempl}) can be complex conjugated to find
\begin{align}
{\rm Pexp} \biggl(\int^z_0
\overline{ \tilde{\mathbb K}_{x,-y}(z_1,\tau)}\biggr) &= 1- \sum_{n_1=0}^{\infty} (-2\pi i)^{1-n_1} \ad_x^{n_1}(y) \overline{ \tilde \Gamma\big( \smallmatrix n_1 \\ 0 \endsmallmatrix ;z,\tau\big) } +\ldots
\label{cctilK} \\
 &\quad + \sum_{n_1,n_2=0}^{\infty} (-2\pi i)^{2-n_1-n_2} \ad_x^{n_1}(y) \ad_x^{n_2}(y) \overline{\tilde \Gamma\big( \smallmatrix n_1 &n_2 \\ 0 &0 \endsmallmatrix ;z,\tau\big)} +\ldots
 \notag
\end{align}
The transposed complex conjugate of the generating series (\ref{notsec.18}) of iterated integrals (\ref{lieg1.35}) involving holomorphic Eisenstein series in its integration kernels $\dd \tau\, \tau^j \est_k(\tau)$ is given by
\begin{align}
\overline{ \mathbb I_{\ep^{\rm TS}}(\tau)^T } &=  1+\sum_{k_1=4}^\infty (k_1{-}1)  \sum_{j_1=0}^{k_1-2} \dfrac{ 1 }{j_1!} \, \overline{ \ee{j_1}{k_1}{\tau}} \epsilon_{k_1}^{(j_1) {\rm TS}} \label{cc.18} \\
&\quad
+\sum_{k_1,k_2=4}^\infty (k_1{-}1)(k_2{-}1) \sum_{j_1=0}^{k_1-2}  \sum_{j_2=0}^{k_2-2}\dfrac{ 1 }{j_1!j_2!}  \, \overline{ \ee{j_1&j_2}{k_1&k_2}{\tau}}  \epsilon_{k_2}^{(j_2) {\rm TS}} \epsilon_{k_1}^{(j_1) {\rm TS}} +\ldots
\notag
\end{align}
Next, we consider the transposed complex conjugate of the generating series (\ref{notsec.17}) of iterated integrals (\ref{lieg1.35}), (\ref{bsc.07}) and (\ref{absc.07}) involving Eisenstein series $\dd \tau\,\tau^j \est_k(\tau)$ and Kronecker-Eisenstein kernels via $\dd \tau\,\tau^j f^{(k)}(u\tau{+}v,\tau)$ as integration kernels,
 \begin{align}
\overline{ \mathbb I_{\ep,b}(u,v,\tau)^T }&= 1+\sum_{k_1=2}^\infty (k_1{-}1) \sum_{j_1=0}^{k_1-2} \dfrac{1}{j_1!}\biggl\{  \overline{ \ee{j_1}{k_1}{\tau} } \epsilon_{k_1}^{(j_1)}+ \overline{ \eez{j_1}{k_1}{z}{\tau} } b_{k_1}^{(j_1)}\biggr\} \label{cc.17}\\
    &\quad+\sum_{k_1,k_2=2}^\infty (k_1{-}1)(k_2{-}1)  \sum_{j_1=0}^{k_1-2}  \sum_{j_2=0}^{k_2-2}\dfrac{ 1}{j_1!j_2!}\biggl\{ \overline{ \ee{j_1&j_2}{k_1&k_2}{\tau}}\epsilon_{k_2}^{(j_2)}   \epsilon_{k_1}^{(j_1)} \notag\\
    &\qquad 
   + \overline{ \eez{j_1&j_2}{k_1&k_2}{z&z}{\tau} } b_{k_2}^{(j_2)}  b_{k_1}^{(j_1)}
    + \overline{ \eez{j_1&j_2}{k_1&k_2}{&z}{\tau} } b_{k_2}^{(j_2)} \epsilon_{k_1}^{(j_1)}
    + \overline{ \eez{j_1&j_2}{k_1&k_2}{z&}{\tau} } \epsilon_{k_2}^{(j_2)}  b_{k_1}^{(j_1)} \biggr\}+\dots
\notag
\end{align}
noting that we set $\epsilon_2 = 0$ by Proposition \ref{e2prop}.

The counterpart $\mathbb H_{\ep,b}(u,v,\tau)$ of $\mathbb I_{\ep,b}(u,v,\tau)$ in modular frame is introduced in section \ref{s:3.c.b} and decomposed into generating series $\mathbb{H}_{\ep,b}^{\pm}(u,v,\tau)$ in (\ref{intemgf.08}). 
While the expansion of $\mathbb{H}_{\ep,b}^\text{+}(u,v,\tau)$ is given by  (\ref{intemgf.10}) that of $\mathbb H^-_{\ep,b}(u,v,\tau)$ in terms of complex conjugate iterated $\tau$ integrals reads
\begin{align}
\mathbb{H}^{-}_{\ep,b}(u,v,\tau)&= 1+\sum_{k_1=2}^\infty (k_1{-}1) \sum_{j_1=0}^{k_1-2} \dfrac{1}{j_1!}\biggl\{   \bminus{j_1}{k_1}{\tau}  \epsilon_{k_1}^{(j_1)}+  \bminusz{j_1}{k_1}{z}{\tau}  b_{k_1}^{(j_1)}\biggr\}\label{ccH} \\
    &\quad+\sum_{k_1,k_2=2}^\infty (k_1{-}1)(k_2{-}1)  \sum_{j_1=0}^{k_1-2}  \sum_{j_2=0}^{k_2-2}\dfrac{ 1}{j_1!j_2!}\biggl\{  \bminus{j_1&j_2}{k_1&k_2}{\tau}\epsilon_{k_2}^{(j_2)}   \epsilon_{k_1}^{(j_1)} \notag \\
    &\qquad 
   +  \bminusz{j_1&j_2}{k_1&k_2}{z&z}{\tau}  b_{k_2}^{(j_2)}  b_{k_1}^{(j_1)}
    + \bminusz{j_1&j_2}{k_1&k_2}{&z}{\tau}  b_{k_2}^{(j_2)} \epsilon_{k_1}^{(j_1)}
    +  \bminusz{j_1&j_2}{k_1&k_2}{z&}{\tau}  \epsilon_{k_2}^{(j_2)}  b_{k_1}^{(j_1)} \biggr\}+\dots
    \notag
\end{align}

\subsection{Complex conjugate differential equations}
\label{secdifc}

For the path-ordered exponential $\mathbb L_{x,y,\ep}(u,v;i\infty,\tau)$ defined by (\ref{sc5.11}), the total differential can be found in (\ref{djplus}), see (\ref{notsec.03}) for the connection $\mathbb J_{x,y,\ep}(u,v,\tau)$. The total differential of its transposed complex conjugate is given by
\begin{align}
\dd \overline{\mathbb L_{x,y,\ep}(u,v;i\infty,\tau)^T} &=\overline{\mathbb L_{x,y,\ep}(u,v;i\infty,\tau)^T} \,\Big({-}x \, \dd u + 2\pi i T_{01}(u) \, \dd v \Big) \notag\\
&\quad -   \overline{\mathbb J_{x,y,\ep}(u,v,\tau)^T}\, \overline{\mathbb L_{x,y,\ep}(u,v;i\infty,\tau)^T}
\label{djminus} 
\end{align}
For its gauge transformed version $\mathbb{I}_{\ep,b}(u,v,\tau)$ of (\ref{notsec.17}) and (\ref{sc5.12}), the total differential can be found in section \ref{sec:4.1.3}, and its complex conjugate reads
\begin{align}
\dd \overline{\mathbb I_{\ep,b}(u,v,\tau)^T} &= \overline{\mathbb I_{\ep,b}(u,v,\tau)^T} \, \Big({-}x \, \dd u + 2\pi i T_{01}(u) \, \dd v \Big)  \label{dimin} \\
&\quad  -\overline{ \mathbb D_{\ep,b}(u,v,\tau)^T} \,\overline{\mathbb I_{\ep,b}(u,v,\tau)^T} 
+  (x - 2\pi i \bar \tau y) \, \overline{\mathbb I_{\ep,b}(u,v,\tau)^T}  \, \dd u
\notag \\
&\quad -   e^{-2\pi i \bar \tau\ep_0}\,
\ad_x \overline{ \Omega(x,\ad_{\frac{x}{2\pi i} },\tau) } \,y\, 
e^{2\pi i \bar \tau\ep_0}  \, \overline{\mathbb I_{\ep,b}(u,v,\tau)^T}\, (\dd v + \tau \, \dd u)
\notag 
\end{align}
%
%%%%%%%%%%%%%%%%%%%%%%%%%%%%%%%%%%%%%%%%%%%%%%%%%%
%%%%%%%%%%%%%%%%%%%%%%%%%%%%%%%%%%%%%%%%%%%%%%%%%%
%%%%%%%%%%%%%%%%%%%%%%%%%%%%%%%%%%%%%%%%%%%%%%%%%%
%%%%%%%%%%%%%%%%%%%%%%%%%%%%%%%%%%%%%%%%%%%%%%%%%%
%
Let us now consider the antiholomorphic derivatives of the iterated integrals $\bplusY{\dots\\\dots\\\dots}$ and $\bminusY{\dots\\\dots\\\dots}$ in (\ref{intemgf.04}) in analogy to their holomorphic $\tau$-derivatives in (\ref{dtauzbeta}). We shall express the following differential equations for $\bpmz{j_1 &\cdots &j_\ell}{k_1 &\cdots&k_\ell}{z &\cdots&z}{\tau}$ in terms of the Maa\ss{} operator $(\bar \tau{-}\tau)\partial_{\bar \tau}{+}w$
tailored to modular forms of anti-holomorphic weight $ w=\sum_{i=1}^{\ell}(k_i{-}2{-}2j_i)$:
%
%%%%%%%%%%%%%%%%%%%%%%%%%%%%%%%%%%%%%%%%%%%%%%%%%%
%%%%%%%%%%%%%%%%%%%%%%%%%%%%%%%%%%%%%%%%%%%%%%%%%%
%%%%%%%%%%%%%%%%%%%%%%%%%%%%%%%%%%%%%%%%%%%%%%%%%%
%%%%%%%%%%%%%%%%%%%%%%%%%%%%%%%%%%%%%%%%%%%%%%%%%%
\begin{align}
\bigg[(\bar \tau{-}\tau) \partial_{\bar\tau} {+}\sum_{i=1}^\ell(k_i{-}2j_i{-}2)\bigg] \bplusz{j_1 & \ldots &j_\ell}{k_1  &\ldots &k_\ell}{z&\ldots&z}{\tau} &= {-}4\pi \Im \tau \sum_{i=1}^\ell j_i \bplusz{j_1 & \ldots &j_i-1 &\ldots &j_\ell}{k_1 &\ldots &k_i &\ldots &k_\ell}{z&\ldots&z&\ldots&z}{\tau} \label{aholobpm}\\
\bigg[(\bar \tau{-}\tau) \partial_{\bar\tau} {+}\sum_{i=1}^\ell(k_i{-}2j_i{-}2)\bigg]  \bminusz{j_1 & \ldots &j_\ell}{k_1 &\ldots &k_\ell}{z&\ldots&z}{\tau}
&= {-}4\pi \Im \tau \sum_{i=1}^\ell j_i \bminusz{j_1 & \ldots &j_i-1 &\ldots &j_\ell}{k_1 &\ldots &k_i &\ldots &k_\ell}{z&\ldots&z&\ldots&z}{\tau} \notag\\
&\hspace{-16mm} + \delta_{j_\ell,0} \dfrac{4\pi \Im \tau }{(-2\pi i)^{k_\ell}} \overline{f^{(k_\ell)}(u\tau{+}v,\tau)} 
 \bminusz{j_1 &j_2& \ldots &j_{\ell-1}}{k_1 &k_2 &\ldots &k_{\ell-1}}{z&z&\ldots&z}{\tau} \notag
\end{align} 
When the rightmost column is $\smallmatrix j_\ell \\ k_\ell\endsmallmatrix$ rather than $\smallmatrix j_\ell \\ k_\ell \\ z \endsmallmatrix$ with an empty slot in the third line (corresponding to Eisenstein series as integration kernels), we proceed analogously to (\ref{dtaubeta}), where we replace $\overline{f^{(k_\ell)}(u\tau{+}v,\tau)}$ by $-\overline{\est_{k_\ell}(\tau)}$.

\section{Zeta generators}
\label{sec:C}

This appendix provides the explicit form of selected contributions to genus-one zeta generators $\sigma_w$ reviewed in section \ref{sec:2.4.2}. 
For ease of notation, we shall write
\beq
\boldsymbol{\epsilon}_k^{(j)} =  \epsilon_k^{(j) {\rm TS}}
\eeq
for the Tsunogai derivations throughout this appendix (the derivations $\ep_k = \boldsymbol{\epsilon}_k - b_k$ do not play any role in this appendix).

%%%%%%%%%%%%%%%%%%%%%%%%%%%%%%%%%%%%%%%%%%%%%%%%%%
%%%%%%%%%%%%%%%%%%%%%%%%%%%%%%%%%%%%%%%%%%%%%%%%%%
\subsection{Terms of modular depth two in all $\sigma_w$}
\label{sec:C.1}

For contributions $[\boldsymbol{\epsilon}_{k_1}^{(j_1)}, \boldsymbol{\epsilon}_{k_2}^{(j_2)}] = [\ep_{k_1}^{(j_1)\text{TS}}, \ep_{k_2}^{(j_2)\text{TS}}]$ to $\sigma_w$ of modular depth two, the following closed-form solution was proposed in \cite{Dorigoni:2024iyt},
\begin{align}
\label{extintr.10}
\sigma_w &= z_w - \frac{1}{(w{-}1)!} \boldsymbol{\epsilon}_{w+1}^{(w-1)} 
 \\
% %
&\quad -\frac{1}{2}\sum_{d=3}^{w-2} \frac{\BF_{d-1}}{\BF_{w-d+2}} \sum_{k=d+1}^{w-1} \BF_{k-d+1}\BF_{w-k+1}\BF_{w-k+1} s_{k,w-k+d}^{d} \notag  \\
&\quad -\sum_{d=5}^{w} \BF_{d-1}s_{d-1,w+1}^{d} -\frac{1}{2}\BF_{w+1}s_{w+1,w+1}^{w+2} \notag \\
&\quad+\sum_{k=w+3}^\infty \BF_k \sum_{j=0}^{w-2} \frac{(-1)^j\left(\begin{smallmatrix} k-2\\ j\end{smallmatrix}\right)^{-1}}{j!(w{-}2{-}j)!} \left[\boldsymbol{\epsilon}_{w+1}^{(w{-}2{-}j)},\boldsymbol{\epsilon}_k^{(j) } \right]+\dots
\notag
\end{align}
The ellipsis refers to terms of modular depth $\geq 3$, and we have used the shorthands 
\begin{align}
    \BF_k= \frac{ B_k}{k!}
    \label{defbf}
\end{align}
as well as
\begin{align}
\label{224}
    s_{p,q}^{d}= \frac{(d{-}2)!}{(p{-}2)!(q{-}2)!} \sum_{i=0}^{d-2} (-1)^{i} \left[
    \boldsymbol{\epsilon}_p^{(p-2-i) }, \boldsymbol{\epsilon}_{q}^{(q-d+i)}\right]\, 
\end{align}

%%%%%%%%%%%%%%%%%%%%%%%%%%%%%%%%%%%%%%%%%%%%%%%%%%
%%%%%%%%%%%%%%%%%%%%%%%%%%%%%%%%%%%%%%%%%%%%%%%%%%
\subsection{Examples of low-degree terms in $\sigma_{w\leq 7}$}
\label{sec:C.2}

The terms of degree $\leq 14$ in the expansion of $\sigma_3,\sigma_5,\sigma_7,\sigma_9$ are given by \cite{Dorigoni:2024iyt} (recall the shorthand $\boldsymbol{\epsilon}_k^{(j)} = \ep_k^{(j){\rm TS}} = \ep_k^{(j)}{+}b_k^{(j)}$ used within this appendix)
\begin{align}
\sigma_3 &= -\tfrac{1}{2} \boldsymbol{\epsilon}_4^{(2)}+z_3+\tfrac{1}{480} [\boldsymbol{\epsilon}_4,\boldsymbol{\epsilon}_4^{(1)}]+\tfrac{1}{30240} [\boldsymbol{\epsilon}_4^{(1)},\boldsymbol{\epsilon}_6] -\tfrac{1}{120960} [\boldsymbol{\epsilon}_4,\boldsymbol{\epsilon}_6^{(1)}]+\tfrac{1}{7257600} [\boldsymbol{\epsilon}_4,\boldsymbol{\epsilon}_8^{(1)}]
\label{firsts3} \\
&\quad -\tfrac{1}{1209600} [\boldsymbol{\epsilon}_4^{(1)},\boldsymbol{\epsilon}_8]
-\tfrac{1}{58060800} [\boldsymbol{\epsilon}_4,[\boldsymbol{\epsilon}_4,\boldsymbol{\epsilon}_6]]+\tfrac{1}{47900160} [\boldsymbol{\epsilon}_4^{(1)},\boldsymbol{\epsilon}_{10}]-\tfrac{1}{383201280} [\boldsymbol{\epsilon}_4,\boldsymbol{\epsilon}_{10}^{(1)}] + \ldots  
\notag \\
%%%%%%%
\sigma_5 &= -\tfrac{1}{24} \boldsymbol{\epsilon}_6^{(4)}
-\tfrac{5  }{48} [\boldsymbol{\epsilon}_4^{(1)},\boldsymbol{\epsilon}_4^{(2)}]
+z_5
+\tfrac{1}{5760} [\boldsymbol{\epsilon}_4,\boldsymbol{\epsilon}_6^{(3)}]
-\tfrac{1}{5760} [\boldsymbol{\epsilon}_4^{(1)},\boldsymbol{\epsilon}_6^{(2)}] +\tfrac{1}{5760} [\boldsymbol{\epsilon}_4^{(2)},\boldsymbol{\epsilon}_6^{(1)}]
\notag \\
&\quad
+\tfrac{1}{3456} [\boldsymbol{\epsilon}_4,[\boldsymbol{\epsilon}_4,\boldsymbol{\epsilon}_4^{(2)}]]
+\tfrac{1}{6912} [\boldsymbol{\epsilon}_4^{(1)},[\boldsymbol{\epsilon}_4^{(1)},\boldsymbol{\epsilon}_4]]
+\tfrac{1}{145152} [\boldsymbol{\epsilon}_6^{(1)},\boldsymbol{\epsilon}_6^{(2)}] -\tfrac{1}{145152} [\boldsymbol{\epsilon}_6,\boldsymbol{\epsilon}_6^{(3)}]
\notag \\
&\quad
-\tfrac{1}{2073600}
[\boldsymbol{\epsilon}_4,[\boldsymbol{\epsilon}_4,\boldsymbol{\epsilon}_6^{(2)}]]
+\tfrac{139 }{72576000} [\boldsymbol{\epsilon}_4^{(1)},[\boldsymbol{\epsilon}_4,\boldsymbol{\epsilon}_6^{(1)}]]
-\tfrac{23 }{24192000} [\boldsymbol{\epsilon}_4,[\boldsymbol{\epsilon}_4^{(1)},\boldsymbol{\epsilon}_6^{(1)}]]
\notag \\
&\quad
-\tfrac{1007 }{145152000} [\boldsymbol{\epsilon}_4^{(2)},[\boldsymbol{\epsilon}_4,\boldsymbol{\epsilon}_6]]
-\tfrac{1}{4147200}
[\boldsymbol{\epsilon}_4^{(1)},[\boldsymbol{\epsilon}_4^{(1)},\boldsymbol{\epsilon}_6]]
+\tfrac{289 }{48384000}
[\boldsymbol{\epsilon}_4,[\boldsymbol{\epsilon}_4^{(2)},\boldsymbol{\epsilon}_6]]
\notag \\
&\quad
+\tfrac{1}{145152000}
[\boldsymbol{\epsilon}_6,\boldsymbol{\epsilon}_8^{(3)}]
-\tfrac{1}{36288000} [\boldsymbol{\epsilon}_6^{(1)},\boldsymbol{\epsilon}_8^{(2)}]
+\tfrac{1}{14515200}
[\boldsymbol{\epsilon}_6^{(2)},\boldsymbol{\epsilon}_8^{(1)}]
-\tfrac{1}{7257600}
[\boldsymbol{\epsilon}_6^{(3)},\boldsymbol{\epsilon}_8]
+ \ldots \notag 
\end{align}
as well as
\begin{align}
\sigma_7 &=
-\tfrac{1}{720} \boldsymbol{\epsilon}_8^{(6)}
+\tfrac{7 }{1152} [\boldsymbol{\epsilon}_4^{(2)},\boldsymbol{\epsilon}_6^{(3)}]
-\tfrac{7 }{1152}
[\boldsymbol{\epsilon}_4^{(1)},\boldsymbol{\epsilon}_6^{(4)}]
-\tfrac{661 }{57600}
[\boldsymbol{\epsilon}_4^{(1)},[\boldsymbol{\epsilon}_4^{(1)},\boldsymbol{\epsilon}_4^{(2)}]]
-\tfrac{661 }{57600}
[\boldsymbol{\epsilon}_4^{(2)},[\boldsymbol{\epsilon}_4^{(2)},\boldsymbol{\epsilon}_4]]
\notag \\
&\quad
+\tfrac{1}{172800}
[\boldsymbol{\epsilon}_4,\boldsymbol{\epsilon}_8^{(5)}]
-\tfrac{1}{172800}
[\boldsymbol{\epsilon}_4^{(1)},\boldsymbol{\epsilon}_8^{(4)}]
+\tfrac{1}{172800}
[\boldsymbol{\epsilon}_4^{(2)},\boldsymbol{\epsilon}_8^{(3)}]
+\tfrac{1}{13824}
[\boldsymbol{\epsilon}_6^{(1)},\boldsymbol{\epsilon}_6^{(4)}]
-\tfrac{1}{13824}
[\boldsymbol{\epsilon}_6^{(2)},\boldsymbol{\epsilon}_6^{(3)}]
\notag \\
&\quad
+z_7
-\tfrac{1}{4354560}
[\boldsymbol{\epsilon}_6,\boldsymbol{\epsilon}_8^{(5)}]
+\tfrac{1}{4354560}
[\boldsymbol{\epsilon}_6^{(1)},\boldsymbol{\epsilon}_8^{(4)}]
-\tfrac{1}{4354560}
[\boldsymbol{\epsilon}_6^{(2)},\boldsymbol{\epsilon}_8^{(3)}]
+\tfrac{1}{4354560}
[\boldsymbol{\epsilon}_6^{(3)},\boldsymbol{\epsilon}_8^{(2)}]
\notag \\
&\quad
-\tfrac{1}{4354560}
[\boldsymbol{\epsilon}_6^{(4)},\boldsymbol{\epsilon}_8^{(1)}]
+\tfrac{7 }{552960}
[\boldsymbol{\epsilon}_4,[\boldsymbol{\epsilon}_4,\boldsymbol{\epsilon}_6^{(4)}]]
+\tfrac{7 }{552960}
[\boldsymbol{\epsilon}_4,[\boldsymbol{\epsilon}_4^{(1)},\boldsymbol{\epsilon}_6^{(3)}]]
+\tfrac{7 }{184320} [\boldsymbol{\epsilon}_4^{(1)},[\boldsymbol{\epsilon}_4^{(2)},\boldsymbol{\epsilon}_6^{(1)}]]
\notag \\
&\quad
+\tfrac{7 }{552960}
[\boldsymbol{\epsilon}_4^{(2)},[\boldsymbol{\epsilon}_4,\boldsymbol{\epsilon}_6^{(2)}]]
-\tfrac{7 }{184320}
[\boldsymbol{\epsilon}_4,[\boldsymbol{\epsilon}_4^{(2)},\boldsymbol{\epsilon}_6^{(2)}]]
-\tfrac{7 }{276480}
[\boldsymbol{\epsilon}_4^{(2)},[\boldsymbol{\epsilon}_4^{(2)},\boldsymbol{\epsilon}_6]]
\notag \\
&\quad
-\tfrac{7 }{552960}
[\boldsymbol{\epsilon}_4^{(1)},[\boldsymbol{\epsilon}_4,\boldsymbol{\epsilon}_6^{(3)}]]
-\tfrac{7 }{552960}
[\boldsymbol{\epsilon}_4^{(2)},[\boldsymbol{\epsilon}_4^{(1)},\boldsymbol{\epsilon}_6^{(1)}]]
+ \ldots \notag \\
\sigma_9 &=
-\tfrac{1}{40320} \boldsymbol{\epsilon}_{10}^{(8)}
-\tfrac{1}{5184} [\boldsymbol{\epsilon}_4^{(1)},\boldsymbol{\epsilon}_8^{(6)}] +\tfrac{1}{5184}
[\boldsymbol{\epsilon}_4^{(2)},\boldsymbol{\epsilon}_8^{(5)}]
-\tfrac{7 }{20736} [\boldsymbol{\epsilon}_6^{(3)},\boldsymbol{\epsilon}_6^{(4)}]
+\tfrac{1}{9676800}
[\boldsymbol{\epsilon}_4,\boldsymbol{\epsilon}_{10}^{(7)}]
\notag \\
&\quad
-\tfrac{1}{9676800}
[\boldsymbol{\epsilon}_4^{(1)},\boldsymbol{\epsilon}_{10}^{(6)}]
+\tfrac{1}{9676800}
[\boldsymbol{\epsilon}_4^{(2)},\boldsymbol{\epsilon}_{10}^{(5)}]
+\tfrac{7 }
{4147200}
[\boldsymbol{\epsilon}_6^{(1)},\boldsymbol{\epsilon}_8^{(6)}]
-\tfrac{7 }{4147200}
[\boldsymbol{\epsilon}_6^{(2)},\boldsymbol{\epsilon}_8^{(5)}]
\notag \\
&\quad
+\tfrac{7 }{4147200}
[\boldsymbol{\epsilon}_6^{(3)},\boldsymbol{\epsilon}_8^{(4)}]
-\tfrac{7 }{4147200}
[\boldsymbol{\epsilon}_6^{(4)},\boldsymbol{\epsilon}_8^{(3)}]
-\tfrac{529 }{691200}
[\boldsymbol{\epsilon}_4,[\boldsymbol{\epsilon}_4^{(2)},\boldsymbol{\epsilon}_6^{(4)}]]
+\tfrac{2959 }{2419200}
[\boldsymbol{\epsilon}_4^{(1)},[\boldsymbol{\epsilon}_4^{(2)},\boldsymbol{\epsilon}_6^{(3)}]]
\notag \\
&\quad
+\tfrac{5891 }{6220800}
[\boldsymbol{\epsilon}_4^{(2)},[\boldsymbol{\epsilon}_4,\boldsymbol{\epsilon}_6^{(4)}]]
-\tfrac{443 }{967680}
[\boldsymbol{\epsilon}_4^{(1)},[\boldsymbol{\epsilon}_4^{(1)},\boldsymbol{\epsilon}_6^{(4)}]]
-\tfrac{799 }{1088640}
[\boldsymbol{\epsilon}_4^{(2)},[\boldsymbol{\epsilon}_4^{(2)},\boldsymbol{\epsilon}_6^{(2)}]]
\notag\\
&\quad
-\tfrac{10651 }{21772800}
[\boldsymbol{\epsilon}_4^{(2)},[\boldsymbol{\epsilon}_4^{(1)},\boldsymbol{\epsilon}_6^{(3)}]] + \ldots
\label{othersts3}
\end{align}
also see the ancillary file of \cite{Dorigoni:2024oft} for terms of degree $\leq 20$ and modular depth $\leq 3$ in $\sigma_w$.

%%%%%%%%%%%%%%%%%%%%%%%%%%%%%%%%%%%%%%%%%%%%%%%%%%
%%%%%%%%%%%%%%%%%%%%%%%%%%%%%%%%%%%%%%%%%%%%%%%%%%
\subsection{The action of $z_5$ on $x,y$}
\label{sec:C.3}

The action of the arithmetic zeta generators $z_w$ on $x,y$ can be determined by imposing (\ref{dfprpsig}) degree by degree as detailed in \cite{Dorigoni:2024iyt}. While the expressions for $[z_3,x]$, $[z_3,y]$ in (\ref{z3onxy}) still fit into one line each,
the action of $z_5$ gives rise to the considerably longer combination of Lie polynomials $b_k^{(j)}$ in $x,y$ (see (\ref{lieg1.24}) and (\ref{lieg1.26}))
\begin{align}
    [z_5,x]&=\dfrac{1}{5760}[b_5, b_6^{(3)}] - \dfrac{1}{2880}[b_5^{(1)}, b_6^{(2)}]+ 
\dfrac{1}{1920} [b_5^{(2)}, b_6^{(1)}] - 
 \dfrac{1}{1440}[b_5^{(3)}, b_6] + \dfrac{13}{5760}[b_2, [b_4, b_5^{(2)}]]\notag\\
 &\quad- \dfrac{13}{
  2880}[b_2, [b_4^{(1)}, b_5^{(1)}]] + \dfrac{13}{1920}[b_2, [b_4^{(2)}, b_5]] + 
 \dfrac{1}{180} [b_3, [b_3, b_5^{(2)}]] + 
 \dfrac{1}{48} [b_3, [b_4, b_4^{(2)}]] \notag\\
 &\quad- 
 \dfrac{1}{90} [b_3, 
   [b_3^{(1)}, 
    b_5^{(1)}]] - \dfrac{17}{5760}[b_4, [b_2, b_5^{(2)}]] + 
 \dfrac{1}{96} [b_4 ,[b_3, b_4^{(2)}]] - 
 \dfrac{1}{90} [b_3^{(1)}, [b_3, b_5^{(1)}]] \notag\\
 &\quad- 
 \dfrac{5}{96}[b_3^{(1)}, [b_4, b_4^{(1)}]] + 
 \dfrac{1}{30} [b_3^{(1)}, 
   [b_3^{(1)}, 
    b_5]] + \dfrac{17}{2880}[b_4^{(1)}, [b_2, b_5^{(1)}]] - 
 \dfrac{1}{192} [b_4^{(1)}, 
   [b_3, 
    b_4^{(1)}]]\notag\\
    &\quad- \dfrac{17}{1920} [b_4^{(2)}, [b_2, b_5]] - 
 \dfrac{11}{72} [b_2, [b_2, [b_3, b_4^{(1)}]]] + 
\dfrac{11}{36} [b_2, [b_2, [b_3^{(1)}, b_4]]] + 
 \dfrac{7}{48} [b_2, [b_3, [b_2, b_4^{(1)}]]]\notag\\
 &\quad- 
 \dfrac{1}{36} [b_2, [b_3, [b_3, b_3^{(1)}]]] - 
 \dfrac{7}{24} [b_2, [b_3^{(1)}, [b_2, b_4]]]
- \dfrac{43}{288} [b_3, [b_2, [b_2, b_4^{(1)}]]]\notag\\
 &\quad - 
 \dfrac{47}{72} [b_3, [b_2, [b_3, b_3^{(1)}]]] + 
 \dfrac{43}{144} [b_3^{(1)}, [b_2, [b_2, b_4]]]- 
\dfrac{1}{16} [b_2, 
   [b_2, [b_2, [b_2, b_3]]]]
   \label{z5xexpr}
\end{align}
The analogous expression for $[z_5,y]$ can be obtained from the generalized reflection operator ${\cal Q}_\alpha$ in (\ref{lemq.01}) at $\alpha=1$: We have $[z_w,y] = -{\cal Q}_1^{-1} [z_w,x]{\cal Q}_1$ by (\ref{lemq.02}) and $\mathfrak{sl}_2$ invariance of $z_w$, and the conjugation of the $b_k^{(j)}$ on the right side of (\ref{z5xexpr}) by ${\cal Q}_1$ can be evaluated via (\ref{lemq.03}).

%%%%%%%%%%%%%%%%%%%%%%%%%%%%%%%%%%%%%%%%%%%%%%%%%%
%%%%%%%%%%%%%%%%%%%%%%%%%%%%%%%%%%%%%%%%%%%%%%%%%%
%%%%%%%%%%%%%%%%%%%%%%%%%%%%%%%%%%%%%%%%%%%%%%%%%%
%%%%%%%%%%%%%%%%%%%%%%%%%%%%%%%%%%%%%%%%%%%%%%%%%%
\section{The Lie algebra of the CEE connection and proof of Lemma \ref{braklem}}
\label{sec:D}

This appendix is dedicated to the bracket relations of the generators $x,y,\ep_k$ in the CEE connection reviewed in section \ref{sec:2.4.3} and provides a proof of Lemma \ref{braklem}.

%%%%%%%%%%%%%%%%%%%%%%%%%%%%%%%%%%%%%%%%%%%%%%%%%%
%%%%%%%%%%%%%%%%%%%%%%%%%%%%%%%%%%%%%%%%%%%%%%%%%%
\subsection{Brackets of $\ep^{(j)}_k ,b^{(j)}_k$ with $x,y$}
\label{sec:D.1}

The purpose of this section is to prove that the commutators $[x,b_k^{(j)}]$, $[y,b_k^{(j)}]$, $[x,\epsilon_k^{(j)}]$ and $[y,\epsilon_k^{(j)}]$ evaluate to Lie polynomials in $b_{k'}^{(j')}$ of degree $k{+}1$ as stated in item $(i)$ of Lemma \ref{braklem}.

\paragraph{Commutator $[y,b_k^{(j)}]$:}

We begin by recalling the definition $b_k^{(1)}=-\ad_{\epsilon_0}\ad_x^{k-1}y$ for $k\geq 3$ and find by repeated use of the Jacobi identity as well as $[\epsilon_0,x]=y$ and
$[\epsilon_0,[x,y]]=0$ that
\begin{align}
[\epsilon_0,\ad_x^{k-1} y]&=\sum_{\ell=0}^{k-3} \ad_x^\ell [y,\ad_x^{k-2-\ell}y]\notag\\
&=\sum_{\ell=0}^{k-3}\sum_{m=0}^\ell \binom{\ell}{m} [\ad_x^{\ell-m}y,\ad_x^{k-2-\ell+m}y]\notag\\
&=\sum_{\ell=0}^{k-3}\binom{k{-}2}{\ell}[\ad_x^{k-3-\ell}y,\ad_x^{\ell+1}y]
\end{align}
Translating back to $b_k=-\ad_x^{k-1} y$ we find that
\begin{align}
b_k^{(1)}&=(k{-}2)[y,b_{k-1}]-\sum_{\ell=0}^{k-4}\binom{k{-}2}{\ell}[b_{k-\ell-2},b_{\ell+2}] \,,\quad k\geq 3
\end{align}
which can now be solved for $[y,b_k]$ after replacing $k \rightarrow k{+}1$:
\beq 
[y,b_k]=\dfrac{1}{k{-}1}\biggl(b_{k+1}^{(1)}+\sum_{\ell=0}^{k-3}\binom{k{-}1}{\ell}[b_{j-\ell-1},b_{\ell+2}]\biggr) \,,\quad k\geq 2 
\eeq
Given that $[\epsilon_0,y]=0$ we can now act with $\ad_{\epsilon_0}^j$ on both sides to find:
\beq \label{Lie_1}
[y,b_k^{(j)}]=\dfrac{1}{k{-}1}\biggl(b_{k+1}^{(j+1)}+\sum_{m=0}^j\binom{j}{m}\sum_{\ell=0}^{k-3}\binom{k{-}1}{\ell}[b_{k-\ell-1}^{(j-m)},b_{\ell+2}^{(m)}]\biggr) \,,\quad k\geq 2
\eeq

\paragraph{Commutator $[x,b_k^{(j)}]$:}
From the defining property for $b_k=-\ad_x^{k-1} y$ we find 
\beq
[x,b_k]=b_{k+1} \,,\quad k\geq 2
\eeq
When considering adjoint actions of $\epsilon_0$ on $b_k^{(j)}=\ad_{\epsilon_0}^jb_k$ we find
\begin{align}\label{Lie_2}
[x,b_k^{(j)}]&= b_{k+1}^{(j)} -\sum_{\ell=0}^{j-1} \ad_{\epsilon_0}^\ell [y,b_k^{(j-1-\ell)}]\\
&= b_{k+1}^{(j)} -j  [y,b_k^{(j-1)}]\notag\\
&=
\dfrac{k{-}j{-}1}{k{-}1}b_{k+1}^{(j)}
-\dfrac{j}{k{-}1}\sum_{m=0}^{j-1}\binom{j{-}1}{m}\sum_{\ell=0}^{k-3}\binom{k{-}1}{\ell}[b_{k-\ell-1}^{(j-1-m)},b_{\ell+2}^{(m)}] \,,\quad k\geq 2
\notag
\end{align}
where in passing to the second and third line we used $[\epsilon_0,y]=0$ and (\ref{Lie_1}), respectively.

\paragraph{Commutator $[y,\epsilon_k^{(j)}]$:}

Given that $[y,\epsilon_0]=0$ and that $[y,\epsilon_k]$ is given by (\ref{epsxy}) we find 
\beq\label{Lie_3}
[y,\epsilon_k^{(j)}]=-\sum_{m=0}^j\binom{j}{m}\sum_{\ell=1}^{\tfrac{k}{2}-1}(-1)^j [b_{\ell+1}^{(j-m)},b_{k-\ell}^{(m)}]\,,\quad k\geq 2\ \text{even}
\eeq

\paragraph{Commutator $[x,\epsilon_k^{(j)}]$:}

Given that $[x,\epsilon_k]=0$ as in (\ref{epsxy}), it follows 
\begin{align}\label{Lie_4}
[x,\epsilon_k^{(j)}]&=-j[y,\epsilon_k^{(j-1)}]\\
&=j\sum_{m=0}^{j-1}\binom{j{-}1}{m}\sum_{\ell=1}^{\tfrac{k}{2}-1}(-1)^{j-1} [b_{\ell+1}^{(j-1-m)},b_{k-\ell}^{(m)}]\,,\quad k\geq 2\ \text{even}
\notag
\end{align}

Each term on the right sides of our expressions for $[x,b_k^{(j)}]$, $[y,b_k^{(j)}]$, $[x,\epsilon_k^{(j)}]$, $[y,\epsilon_k^{(j)}]$ in (\ref{Lie_1}), (\ref{Lie_2}), (\ref{Lie_3}), (\ref{Lie_4}) is a Lie polynomial in $b_{k'}^{(j')}$ of degree $k{+}1$ which completes our constructive proof of item $(i)$ of Lemma \ref{braklem}.

%%%%%%%%%%%%%%%%%%%%%%%%%%%%%%%%%%%%%%%%%%%%%%%%%%
%%%%%%%%%%%%%%%%%%%%%%%%%%%%%%%%%%%%%%%%%%%%%%%%%%
\subsection{Brackets $[ \ep^{(j_1)}_{k_1} , b^{(j_2)}_{k_2} ] $}
\label{sec:D.2}

We shall next demonstrate that brackets $[ \ep^{(j_1)}_{k_1} , b^{(j_2)}_{k_2} ] $ are expressible as Lie polynomials in $b_{k'}^{(j')}$ of degree $k_1{+}k_2$ as stated in item $(ii)$ of Lemma \ref{braklem}.

In a first step let us show that the bracket $[\epsilon_{k_1}^{(j_1)},b_{k_2}]$ in the case of $j_2=0$ is expressible as Lie polynomials in $b_k^{(j)}$, 
\begin{align}\label{Lie_5}
[\epsilon_{k_1}^{(j_1)},b_{k_2}]&=
[\epsilon_{k_1}^{(j_1)},\ad_{x}^{k_2-2} b_2]\\
%-[\epsilon_{k_1}^{(j_1)},\ad_{x}^{k_2-1} y]\notag\\
&=\sum_{\ell=0}^{k_2-3} \ad_x^{\ell} \big[[\epsilon_{k_1}^{(j_1)},x],\ad_x^{k_2-\ell-3}b_2 \big] + \ad_x^{k_2-2} [\ep_{k_1}^{(j_1)},b_2]\notag \\
&=\sum_{\ell=0}^{k_2-3} \ad_x^{\ell} \big[[\epsilon_{k_1}^{(j_1)},x],b_{k_2-\ell-1} \big]-\ad_x^{k_2-2} [b_{k_1}^{(j_1)},b_2]
\notag
\end{align}
The second line follows solely from iterations of the Leibniz rule of commutators. In passing to the third line, we have used the consequence $[\epsilon_k^{(j)},b_2]=-[b_k^{(j)},b_2]$ of 
$[\epsilon_k^{(j)\text{TS}},b_2]=0$ and $\epsilon_k^{(j)\text{TS}}=\epsilon_k^{(j)}+b_k^{(j)}$. Given (\ref{Lie_4}), (\ref{Lie_3}) and (\ref{Lie_2}), it is easy to see that (\ref{Lie_5}) is indeed a Lie polynomial in~$b_k^{(j)}$.

In order to identify $[ \ep^{(j_1)}_{k_1} , b^{(j_2)}_{k_2} ] $ as a Lie polynomial in $b_k^{(j)}$ for arbitrary $0\leq j_2 \leq k_2{-}2$, we proceed by induction in $j_2$, with (\ref{Lie_5}) as a base case at $j_2=0$. For the inductive step, we use the Leibniz rule of $\ad_{\ep_0}$ to rewrite
\beq
[\epsilon_{k_1}^{(j_1)},b_{k_2}^{(j_2)}] = 
\ad_{\epsilon_0}[\epsilon_{k_1}^{(j_1)},b_{k_2}^{(j_2-1)}]-[\epsilon_{k_1}^{(j_1+1)},b_{k_2}^{(j_2-1)}]
\label{indstep}
\eeq
Assuming the terms $[\epsilon_{k_1}^{(j_1)},b_{k_2}^{(j_2-1)}]$ and $[\epsilon_{k_1}^{(j_1+1)},b_{k_2}^{(j_2-1)}]$ on the right side to be Lie polynomials in $b_k^{(j)}$ by the inductive hypothesis, then the term $[\epsilon_{k_1}^{(j_1)},b_{k_2}^{(j_2)}]$ on the left side with a larger value of $j_2$ must be a Lie polynomial as well. This also relies on the fact that Lie polynomials in $b_k^{(j)}$ close under $\ad_{\epsilon_0}$. The instances of (\ref{indstep}) at $j_2=1,2,\ldots,k_2{-}2$ together with the fact that the degree $k_1{+}k_2$ is preserved by all steps of (\ref{Lie_5}) and (\ref{indstep}) then imply the statement of item $(ii)$ of Lemma \ref{braklem}. 

Note that the proof of this section is constructive: With the expressions for $[x,b_k^{(j)}]$, $[x,\epsilon_k^{(j)}]$ of section \ref{sec:D.1} at hand, both of $[\epsilon_{k_1}^{(j_1)},b_{k_2}]$ and $[\epsilon_{k_1}^{(j_1)},b_{k_2}^{(j_2)}]$ can be recursively computed via (\ref{Lie_5}) and (\ref{indstep}) to any desired order. The explicit form of all brackets $[\epsilon_{k_1}^{(j_1)},b_{k_2}^{(j_2)}]$ in terms of $b_k^{(j)}$ up to and including $k_1{+}k_2=10$ can be found in the ancillary {\tt Brackets.nb} file of the arXiv submission of this work.

%%%%%%%%%%%%%%%%%%%%%%%%%%%%%%%%%%%%%%%%%%%%%%%%%%
%%%%%%%%%%%%%%%%%%%%%%%%%%%%%%%%%%%%%%%%%%%%%%%%%%
\subsection{Proving items $(iii)$ to $(v)$ of the lemma}
\label{sec:D.7}

This section completes the proof of Lemma \ref{braklem} by proving items $(iii)$ to $(v)$

$(iii)$: Lie polynomials in $x,y$ can always be
written as a sequence of $\ad_x,\ad_y$ acting on $b_2 = [y,x]$. By item $(i)$ of the lemma, the
set of $b_k^{(j)}$ with $k\geq 2$ and $j=0,1,\ldots,k{-}2$ is closed under $\ad_x$ and $\ad_y$ which implies the statement of item $(iii)$.

$(iv)$: The extension lemma 2.1.2 of \cite{Schneps:2015mzv} extracts the brackets
$[\sigma_w,x]$, $[\sigma_w,y]$ from the expressions (\ref{dfprpsig})
for $[\sigma_w,t_{01}]$, $[\sigma_w,t_{12}]$ (also see section 5.3 of \cite{Dorigoni:2024iyt} for further details). By (\ref{dfprpsig}), every order in the series $[\sigma_w,t_{01}]$, $[\sigma_w,t_{12}]$ is a Lie polynomial in $x,y$ which propagates to the expressions for $[\sigma_w,x]$, $[\sigma_w,y]$ obtained from the extension lemma \cite{Schneps:2015mzv, Dorigoni:2024iyt}. Item $(iii)$ of the lemma then implies that all of $[\sigma_w,t_{01}]$, $[\sigma_w,t_{12}]$, $[\sigma_w,x]$ and $[\sigma_w,y]$ are in fact Lie series in $b^{(j)}_k$. Both of $[\sigma_w,x]$, $[\sigma_w,y]$ are Lie series in $b_k^{(j)}$ of degree $\geq w{+}2$ in $x,y$ since $\sigma_w$ are infinite series of derivations of degree $\geq w{+}1$.

$(v)$: Follows by specializing the statement of $(iv)$ to the degree-$(2w{+}1)$ part of $[\sigma_w,x]$, $[\sigma_w,y]$: The degree-$2w$ part of $\sigma_{w}$ produces Lie polynomials in $b^{(j)}_k$ of degree $2w{+}1$ via $[\sigma_w,x]$ and $[\sigma_w,y]$. It remains to show that this holds separately for $z_w$ and for the restriction of $\sigma_w{-}z_w$ to degree $2w$. For this purpose, we note that $\sigma_w{-}z_w$ is expressible via Tsunogai derivations, so item $(i)$ of the lemma implies that $[\sigma_w{-}z_w,x]$, $[\sigma_w{-}z_w,y]$ are Lie series in $b^{(j)}_k$ as well. On these grounds, the degree-$(2w{+}1)$ parts of $[\sigma_w{-}z_w,x]$, $[\sigma_w{-}z_w,y]$ and $[z_w,x]$, $[z_w,y]$ are separately Lie polynomials in $b^{(j)}_k$ of degree $(2w{+}1)$ as in the statement of item $(v)$ which completes the proof of the lemma.

%%%%%%%%%%%%%%%%%%%%%%%%%%%%%%%%%%%%%%%%%%%%%%%%%%
%%%%%%%%%%%%%%%%%%%%%%%%%%%%%%%%%%%%%%%%%%%%%%%%%%
\section{Proofs involving infinitesimal coactions of MZVs}
\label{genapp:thm21}

In this appendix, we review and apply a technique to prove identities between series in (motivic) MZVs. The key ingredient is the operation $\partial_{f_w}$ defined by
\beq
\partial_{f_w} f_{i_1} \ldots f_{i_{r-1}} f_{i_r}= \left\{ \begin{array}{cl} \delta_{w,i_r}  f_{i_1}\ldots f_{i_{r-1}} &, \ \ \ \ r\geq 1 \\ 
0 &, \ \ \ \ r = 0 
\end{array} \right.
\, , \ \ \ \ w,i_1,\ldots,i_r \geq 3 \ {\rm odd}
\label{clp.01}
\eeq
subject to $\partial_{f_w} f_2 = f_2 \partial_{f_w}$
which clips off the rightmost non-commuting letter in the $f$-alphabet representation of motivic MZVs (see appendix \ref{sec:A.2} for a brief review). This operation implements the infinitesimal version of the deconcatenation coaction in the $f$-alphabet and obeys the following Leibniz rule for shuffle products \cite{Brown:2011ik}\footnote{The derivations $\partial_{w}$ of the reference clip off the leftmost non-commuting letter as opposed to the rightmost one as in (\ref{clp.01}). This accounts for the fact that the $f$-alphabet of \cite{Brown:2011ik} is reversed in comparison to the one of the present work which in turn relates to the ordering conventions for the two entries of the motivic coaction.}
\beq
\partial_{f_w} (f_{i_1}\ldots f_{i_r}\shuffle f_{j_1}\ldots f_{j_s})=
( \partial_{f_w}  f_{i_1}\ldots f_{i_r})\shuffle f_{j_1}\ldots f_{j_s}
+f_{i_1}\ldots f_{i_r}\shuffle (  \partial_{f_w} f_{j_1}\ldots f_{j_s})
\label{clp.02}
\eeq
where $w,i_1,\ldots,i_r,j_1,\ldots, j_s \geq 3$ odd. The two proofs in this appendix both use the fact that two series in MZVs which obey the same ``differential equations'' with respect to $\partial_{f_w}$ and have the same coefficients of $f_2^n$ for all $n\geq 0$ as ``initial values'' (disregarding all contributions with at least one non-commuting letter $f_w$) must be identical.

%%%%%%%%%%%%%%%%%%%%%%%%%%%%%%%%%%%%%%%%%%%%%%%%%%
%%%%%%%%%%%%%%%%%%%%%%%%%%%%%%%%%%%%%%%%%%%%%%%%%%
\subsection{Completing the proof of Theorem \ref{2.thm:1}}
\label{app:thm21}

We shall here prove that the series $\mathbb I^{\rm eqv}_{\ep^{\rm TS}}(\tau)$ defined by (\ref{lieg1.51}) obeys the 
equivariance property (\ref{lieg1.42}) for any
$\gamma \in {\rm SL}_2(\mathbb Z)$ as stated in item $(ii)$ of Theorem \ref{2.thm:1}. Equivariance under the modular $T$ transformation was already shown in section \ref{sec:2.5.2}, and we shall here complete the proof by demonstrating equivariance under $S:\, \tau \rightarrow -\frac{1}{\tau}$, i.e.\ that
\begin{align}
\mathbb I^{\rm eqv}_{\ep^{\rm TS}}\big({-}\tfrac{1}{\tau} \big) &=
U_S^{-1} \, \mathbb I^{\rm eqv}_{\ep^{\rm TS}}( \tau)\, 
U_S
\label{clp.03}
\end{align}
see (\ref{lieg1.39}) and (\ref{lieg1.40}) for $U_S$. The first step is to evaluate the $S$-transformations of the meromorphic generating series $\mathbb I_{\ep^{\rm TS}}( \tau)$ and its complex conjugate according to (\ref{lieg1.41}) \cite{brown2017multiple, Dorigoni:2024oft},
\beq
\mathbb I_{\ep^{\rm TS}}\big({-}\tfrac{1}{\tau} \big) = \mathbb S_{\ep^{\rm TS}}  \, U_S^{-1}  \, \mathbb I_{\ep^{\rm TS}}( \tau) \,  U_S
\label{clp.04}
\eeq
where we introduced the following shorthand for the $S$ cocycle
\beq
\mathbb S_{\ep^{\rm TS}}  = 
{\rm Pexp} \bigg( \int^{0}_{i\infty}  U_\gamma^{-1}  \mathbb  D_{\ep^{\rm TS}}(\rho_1)  U_\gamma  \bigg) 
= {\rm Pexp} \bigg( \int^{i \infty}_{0}    \mathbb  D_{\ep^{\rm TS}}(\tau_1)   \bigg)
\label{clp.05}
\eeq
The reality properties of the connection $\mathbb  D_{\ep^{\rm TS}}(\tau_1)$ readily imply that $\overline{\mathbb S_{\ep^{\rm TS}}} = \mathbb S_{\ep^{\rm TS}}$ such that the $S$-transformation (\ref{clp.04}) propagates as follows to the series $\mathbb I^{\rm eqv}_{\ep^{\rm TS}}$ in (\ref{lieg1.51}):
\beq
\mathbb I^{\rm eqv}_{\ep^{\rm TS}}\big({-}\tfrac{1}{\tau} \big) =  U_S^{-1} \, (\mathbb M_z^{\rm sv})^{-1} \, \overline{ \mathbb I_{\ep^{\rm TS}}( \tau)^T }  \, U_S\, \mathbb S_{\ep^{\rm TS}}^T \, \mathbb M_\sigma^{\rm sv} \, \mathbb S_{\ep^{\rm TS}}  \, U_S^{-1}  \, \mathbb I_{\ep^{\rm TS}}( \tau) \,  U_S
\label{clp.06}
\eeq
In order to recover the desired conjugate of $\mathbb I^{\rm eqv}_{\ep^{\rm TS}}( \tau )$ on the right side of (\ref{clp.04}), we need the statement (iv) of the following lemma:

\begin{lemma}
\label{applem}
(equivalent to results in section 15 of \cite{brown2017multiple} as well as section 7 of \cite{Brown:2017qwo2} and in its present form conjectured in \cite{Kleinschmidt:2025dtk})
The $S$-cocycle $\mathbb S_{\ep^{\rm TS}} $ defined by (\ref{clp.05})
\begin{itemize}
\item[(i)] obeys the differential coaction 
\beq
\partial_{f_w} \mathbb \rho(\mathbb S_{\ep^{\rm TS}})  = \rho(\mathbb S_{\ep^{\rm TS}}) \, U_S^{-1} \, \sigma_w \, U_S - \sigma_w \, \rho( \mathbb S_{\ep^{\rm TS}}  )
\label{clp.07}
\eeq
involving zeta generators $\sigma_w$ under the operation in (\ref{clp.01});
\item[(ii)] can be written as
\beq
\mathbb S_{\ep^{\rm TS}}  = (\mathbb M_\sigma)^{-1}\, \mathbb X_{\ep^{\rm TS}}  \, U_S^{-1} \, \mathbb M_\sigma \, U_S 
\label{clp.08}
\eeq
where the ``initial value'' $\mathbb X_{\ep^{\rm TS}}$ to the ``differential equation'' (\ref{clp.07}) is a $\mathbb Q[\pi^2]$ series in Tsunogai's derivations $\ep_k^{(j){\rm TS}}$, and the series $\mathbb M_\sigma$ is closely related to $\mathbb M_\sigma^{\rm sv}$ in (\ref{not.05}); 
\beq
 \mathbb M_\sigma = \sum_{r=0}^{\infty} \sum_{i_1,\ldots,i_r \atop {\in 2\mathbb N+1} } \rho^{-1}(f_{i_1} \ldots f_{i_r}) \, \sigma_{i_1}\ldots  \sigma_{i_r}
\, , \ \ \ \ 
\mathbb M_\sigma^{\rm sv}
= \mathbb M_\sigma^{T} \mathbb M_\sigma
\label{clp.00}
\eeq
\item[(iii)] exhibits the following transposition properties in its ``initial value'' 
\beq
(\mathbb X_{\ep^{\rm TS}})^T = (\mathbb X_{\ep^{\rm TS}})^{-1}
\label{clp.09}
\eeq
\item[(iv)] satisfies the quadratic relation
\beq
\mathbb S_{\ep^{\rm TS}}^T \, \mathbb M_\sigma^{\rm sv}\, \mathbb S_{\ep^{\rm TS}} = U_S^{-1} \, \mathbb M_\sigma^{\rm sv} \, U_S 
\label{clp.10}
\eeq
\end{itemize}
\end{lemma}

Given that the quadratic relation (\ref{clp.10}) casts the $S$-transformation (\ref{clp.06}) into the form
\begin{align}
\mathbb I^{\rm eqv}_{\ep^{\rm TS}}\big({-}\tfrac{1}{\tau} \big) &=  U_S^{-1} \, (\mathbb M_z^{\rm sv})^{-1} \, \overline{ \mathbb I_{\ep^{\rm TS}}( \tau)^T }  \,  \mathbb M_\sigma^{\rm sv}   \, \mathbb I_{\ep^{\rm TS}}( \tau) \,  U_S \notag \\
&= U_S^{-1} \,\mathbb I^{\rm eqv}_{\ep^{\rm TS}}(\tau) \,  U_S
\label{clp.11}
\end{align}
the last step towards proving item $(ii)$ of Theorem \ref{2.thm:1} is the proof of Lemma \ref{applem}:

\begin{proof}
\phantom{x}

$(i)$: follows from translating the Galois action in section 15.3.3 of \cite{brown2017multiple} into the differential coaction (\ref{clp.07});

$(ii)$: One can easily check via $\partial_{f_w} 
\rho(U_S^{-1}\mathbb M_\sigma U_S) = \rho(U_S^{-1}\mathbb M_\sigma U_S)
U_S^{-1}\sigma_w U_S$ as well as $\partial_{f_w} \rho(\mathbb M^{-1}_\sigma) = - \sigma_w  \rho(\mathbb M^{-1}_\sigma)$ that (\ref{clp.08}) solves the ``differential equation'' (\ref{clp.07}). It remains to explain why the ``initial value'' $ \mathbb X_{\ep^{\rm TS}}$ obtained from discarding non-empty words in $f_w$ with $w\geq 3$ odd in the expansion of $ \rho(\mathbb S_{\ep^{\rm TS}} )$ is a $\mathbb Q[\pi^2]$ series in Tsunogai derivations. This follows from the fact that the multiple modular values in the expansion of $\mathbb S_{\ep^{\rm TS}}$ in $\ep_k^{(j)}$ are $\mathbb Q[\pi^{-2}]$-linear as opposed to $\mathbb Q[(2\pi i)^{\pm 1}]$-linear combinations of MZVs and all negative powers of $\pi^2$ can be traced back to $U_S^{-1} \sigma_w U_S$ along with (\ref{lieg1.40}). In order to see this, one can, for instance, determine the multiple modular values in $\mathbb S_{\ep^{\rm TS}}$ from the ratio of the $A$- and $B$-elliptic KZB associators in their $\tau \rightarrow i\infty$ degenerations \cite{KZB, EnriquezEllAss, Enriquez:Emzv} and use the transcendentality properties of $B$-elliptic MZVs derived in appendix C of \cite{Broedel:2018izr}. The link between multiple modular values and the degenerate $A$- and $B$-elliptic associators is also discussed in \cite{saad2020multiple}.

$(iii)$: When writing the expansion of the ``initial value'' in the form
\beq
\mathbb X_{\ep^{\rm TS}} = \sum_{r=0}^\infty \sum_{k_1,\ldots,k_r=4}^\infty
\sum_{j_1=0}^{k_1-2}\ldots \sum_{j_r=0}^{k_r-2}
\ep_{k_1}^{(j_1){\rm TS}} \ldots \ep_{k_r}^{(j_r){\rm TS}}
\xcomp{j_1 &\ldots &j_r}{k_1 &\ldots &k_r}
\label{clp.13}
\eeq
with coefficients $\xcomp{j_1 &\ldots &j_r}{k_1 &\ldots &k_r} \in \mathbb Q[\pi^2]$, then the transpose and inverse in (\ref{clp.09}) are given by
\begin{align}
(\mathbb X_{\ep^{\rm TS}})^T &=
\sum_{r=0}^\infty \sum_{k_1,\ldots,k_r=4}^\infty
\sum_{j_1=0}^{k_1-2}\ldots \sum_{j_r=0}^{k_r-2} (-1)^{j_1+\ldots +j_r}
\ep_{k_r}^{(j_r){\rm TS}} \ldots \ep_{k_1}^{(j_1){\rm TS}}
\xcomp{j_1 &\ldots &j_r}{k_1 &\ldots &k_r}
\notag \\
(\mathbb X_{\ep^{\rm TS}})^{-1} &= \sum_{r=0}^\infty (-1)^r\sum_{k_1,\ldots,k_r=4}^\infty
\sum_{j_1=0}^{k_1-2}\ldots \sum_{j_r=0}^{k_r-2}
\ep_{k_r}^{(j_r){\rm TS}} \ldots \ep_{k_1}^{(j_1){\rm TS}}
\xcomp{j_1 &\ldots &j_r}{k_1 &\ldots &k_r}
\label{clp.14}
\end{align}
The transcendentality properties of multiple modular values discussed in \cite{brown2017multiple, Broedel:2018izr, saad2020multiple} imply that the MZVs multiplying  $\ep_{k_1}^{(j_1){\rm TS}} \ldots \ep_{k_r}^{(j_r){\rm TS}}$ in $\mathbb S_{\ep^{\rm TS}}$ have weight $j_1{+}\ldots{+}j_r{+}r$ according to the $y$-degree of the derivations (inspection of (\ref{lieg1.04}) and (\ref{lieg1.05}) reveals that $\ep_{k}^{(j){\rm TS}} $ has $y$-degree $j{+}1$). Since this property holds for both of $(\mathbb M_\sigma)^{-1}$ and $ U_S^{-1} \mathbb M_\sigma U_S$ in (\ref{clp.08}), the same must be true for $\mathbb X_{\ep^{\rm TS}}$ such that the coefficients $\xcomp{j_1 &\ldots &j_r}{k_1 &\ldots &k_r}$ in (\ref{clp.13}) have weight $j_1{+}\ldots{+}j_r{+}r$. By the result $\xcomp{j_1 &\ldots &j_r}{k_1 &\ldots &k_r} \in \mathbb Q[\pi^2]$ of item $(ii)$, only even weights can occur, so all the $\ep_{k_1}^{(j_1){\rm TS}} \ldots \ep_{k_r}^{(j_r){\rm TS}}$ contributing to (\ref{clp.13}) have
\beq
j_1{+}\ldots{+}j_r{+}r \in 2\mathbb N \ \ \ \ \Longrightarrow \ \ \ \ (-1)^r = (-1)^{j_1+\ldots+j_r}
\label{clp.15}
\eeq
Note that this analysis is unaffected by Pollack relations among $\ep_{k}^{(j){\rm TS}} $ as they preserve the $y$ degree. In view of (\ref{clp.15}), the expansions of $(\mathbb X_{\ep^{\rm TS}})^T$ and $(\mathbb X_{\ep^{\rm TS}})^{-1}$ in (\ref{clp.14}) match and conclude the proof of (\ref{clp.09}).

$(iv)$: We can now prove (\ref{clp.10}) by direct computation, using (\ref{clp.08}), (\ref{clp.00}) and $U^T_S = U_S^{-1}$ in the first step,
\begin{align}
\mathbb S_{\ep^{\rm TS}}^T \, \mathbb M_\sigma^{\rm sv}\, \mathbb S_{\ep^{\rm TS}} &= 
U_S^{-1}  \, \mathbb M_\sigma^T \, U_S \, \mathbb X_{\ep^{\rm TS}}^T \,
(\mathbb M^T_\sigma)^{-1}  \, 
\mathbb M^T_\sigma\, \mathbb M_\sigma\,
(\mathbb M_\sigma)^{-1}\, \mathbb X_{\ep^{\rm TS}}  \, U_S^{-1} \, \mathbb M_\sigma \, U_S  \notag \\
&= U_S^{-1}  \, \mathbb M_\sigma^T \, U_S \, \mathbb X_{\ep^{\rm TS}}^T \, \mathbb X_{\ep^{\rm TS}}  \, U_S^{-1} \, \mathbb M_\sigma \, U_S  \notag \\
&= U_S^{-1} \, M_\sigma^T \, \mathbb M_\sigma \, U_S 
= U_S^{-1} \, \mathbb M_\sigma^{\rm sv} \, U_S 
\label{clp.16}
\end{align}
The last line is obtained from the statement (\ref{clp.09}) of item $(iii)$, and we have identified $\mathbb M_\sigma^{\rm sv} $ through another use of (\ref{clp.00}) which concludes the proof of item $(iv)$ and thus Lemma \ref{applem}.
\end{proof}
Since item $(iv)$ of Lemma \ref{applem} implies the desired $S$-equivariance (\ref{clp.11}) of $\mathbb I^{\rm eqv}_{\ep^{\rm TS}}$ and we have already proven $T$-equivariance in section \ref{sec:2.5.2}, this also concludes the proof of Theorem~\ref{2.thm:1}.
Note that the importance of the identity $\mathbb S_{\ep^{\rm TS}}^T  \mathbb M_\sigma^{\rm sv} \mathbb S_{\ep^{\rm TS}} =  U_S^{-1}  \mathbb M_\sigma^{\rm sv} U_S$ was stressed in section 4.2 of \cite{Dorigoni:2024oft}, and a variant without direct reference to zeta generators is the specialization of Theorem 7.2 of \cite{Brown:2017qwo2} to $\gamma=S$.

\subsection{Proof of Lemma \ref{sigphilem}}
\label{app:phi}

We shall here prove Lemma \ref{sigphilem} which  plays a crucial role in proving Theorem \ref{3.thm:1} and whose statement (\ref{tkpexp.07s}) takes the following form in the $f$-alphabet:
\beq
\rho \big( \Phinew( {-} t_{12}, {-} t_{01} ) \big) \shuffle  \rho\big( \mathbb M^{\rm sv}_{\Sigma(0)} \big) \shuffle   \rho \big( \Phinew^{-1}(  t_{12}, t_{01} ) \big)
=  \rho \big(\mathbb M^{\rm sv}_{\sigma} \big)
\label{tkpexp.07}
\eeq

\subsubsection{``Initial values''}
\label{app:phi.1}

Following the strategy in the preamble of this appendix, the first step of the proof is to show that both sides of (\ref{tkpexp.07}) have the same ``initial values'', i.e.\ terms without any non-commuting letter $f_w$ with $w\geq 3$ odd at arbitrary order $f_2^n, \ n\geq 0$ in the commutative letter $f_2$. The shuffle products act on the non-commuting letters of the $f$-alphabet, whereas the Lie algebra generators $t_{01}$, $t_{12}$ and $\Sigma_w(0)$ are concatenated from left to right.

The right side of (\ref{tkpexp.07}) does not involve any commuting letter $f_2$ by the definition (\ref{not.05}) of $\mathbb M^{\rm sv}_{\sigma}$ as a series in single-valued MZVs. The left side of (\ref{tkpexp.07}) in turn reduces to the concatenation product
$\rho( \Phinew( {-} t_{12}, {-} t_{01} )) \big|_{f_2}   \rho ( \Phinew^{-1}(  t_{12}, t_{01} ) )  \big|_{f_2} $ when isolating the ``initial value'' by discarding any $f_w$ with $w$ odd (the shuffle product acts trivially on powers of $f_2$). In the restriction of $\Phinew$ to a power series in $f_2$, only words with an even number of letters $t_{12}$, $t_{01}$ occur such that the first concatenation factor is insensitive to the minus signs of the arguments in $\rho( \Phinew( {-} t_{12}, {-} t_{01} )) \big|_{f_2} = \rho( \Phinew( t_{12},  t_{01} )) \big|_{f_2} $. Hence, the ``initial value'' for the left side of (\ref{tkpexp.07}) is $ \rho( \Phinew( t_{12},  t_{01} )) \big|_{f_2}  \rho( \Phinew^{-1}( t_{12},  t_{01} )) \big|_{f_2} = 1$, consistently with that of the right~side.

\subsubsection{``Differential equations''}
\label{app:phi.2}

It now remains to show that both sides of (\ref{tkpexp.07}) obey the same ``differential equation'' with respect to the operation $\partial_{f_w}$ in (\ref{clp.01}). The right side obeys
\beq
\partial_{f_w} \rho \big(\mathbb M^{\rm sv}_{\sigma} \big) = \sigma_w \, \rho \big(\mathbb M^{\rm sv}_{\sigma} \big) + \rho \big(\mathbb M^{\rm sv}_{\sigma} \big)\, \sigma_w
\label{tkpexp.09}
\eeq
as can be checked through the definition (\ref{not.05}) of $\mathbb M^{\rm sv}_{\sigma}$ and the Leibniz rule (\ref{clp.02}) of $\partial_{f_w}$ with respect to the shuffle product in the expression (\ref{appA.00}) for ${\rm sv}(f_{i_1}\ldots f_{i_r})$. 
The remaining task in the proof is to show that the same holds for the left side of (\ref{tkpexp.07}), i.e.\ that
\begin{align}
\partial_{f_w} \Psi &= \sigma_w \, \Psi + \Psi \, \sigma_w
\label{tkpexp.10}
\end{align}
for the following shorthand for the left side of  (\ref{tkpexp.07}):
\begin{align}
\Psi &= \rho \big( \Phinew( {-} t_{12}, {-} t_{01} ) \big) \shuffle  \rho\big( \mathbb M^{\rm sv}_{\Sigma(0)} \big) \shuffle   \rho \big( \Phinew^{-1}(  t_{12}, t_{01} ) \big)
\label{tkpexp.10psi}
\end{align}
In order to evaluate the $\partial_{f_w}$ action on the three factors of $\Psi$, we use (\ref{tkpexp.09}) with $\sigma_w \rightarrow \Sigma_w(0)$ together with the ``differential equations'' of the Drinfeld associator in \cite{FB:talk, Drummond:2013vz, privFBDKLS} and in section 5.1 of \cite{Frost:2025lre} (see (\ref{notsec.12}) for the Lie polynomial $P_w(t_{12},t_{01})= \Sigma_w(0) - \sigma_w$)
\begin{align}
\partial_{f_w} \rho \big( \Phinew( {-} t_{12}, {-} t_{01} ) \big) &= - \rho \big( \Phinew( {-} t_{12}, {-} t_{01} ) \big) \circ P_w(t_{12},t_{01})
\notag \\
\partial_{f_w} \rho \big( \Phinew^{-1}(  t_{12}, t_{01} ) \big) &= - P_w(t_{12},t_{01}) \;\tilde{\circ}\; \rho \big( \Phinew^{-1}(  t_{12}, t_{01} ) \big)
\label{tkpexp.11}
\end{align}
In the first place, the binary operations $\circ$ and $\tilde \circ$ on words $V$, $W$ in $t_{01}$, $t_{12}$ are defined by
\beq
V \circ W = V\, W - D_W(V) \, , \ \ \ \ \ \ V\; \tilde \circ \; W = V\, W + D_V(W)
\label{tkpexp.12}
\eeq
where $D_L$ are Ihara derivations associated with Lie polynomials $L$ that act in the following asymmetric way on $t_{12}$ and $t_{01}$ \cite{Ihara1989TheGR, Ihara:stable}:
\beq
D_L(t_{12}) = 0 \, , \ \ \ \ \ \ D_L(t_{01}) = [ t_{01}, L ]
\label{tkpexp.13}
\eeq
In both lines of (\ref{tkpexp.11}), the role of $L$ in (\ref{tkpexp.13}) is taken by the Lie polynomials $P_w(t_{12},t_{01})$. By the definition (\ref{dfprpsig}) of genus-one zeta generators $\sigma_w$ as Ihara derivations with respect to $P_w(t_{12},t_{01})$, we can replace $D_{P_w(t_{12},t_{01})}$ by commutators
\beq
D_{P_w(t_{12},t_{01})} V = [\sigma_w , V]
\label{tkpexp.14}
\eeq
for any word $V$ in $t_{01}$, $t_{12}$. This is obvious for one-letter words $V \in \{ t_{01},t_{12} \}$ by comparing (\ref{dfprpsig}) with (\ref{tkpexp.13}) and follows for words $V$ of arbitrary length $\geq 2$ by the fact that ${\rm ad}_{\sigma_w}$ and $D_{P_w(t_{12},t_{01})}$ obey the same Leibniz rule. This reduces the ``differential equations'' (\ref{tkpexp.11}) to the following concatenation products and commutators:
\begin{align}
\partial_{f_w} \rho \big( \Phinew( {-} t_{12}, {-} t_{01} ) \big) &= - \rho \big( \Phinew( {-} t_{12}, {-} t_{01} ) \big)  P_w(t_{12},t_{01}) + \big[\sigma_w,  \rho \big( \Phinew( {-} t_{12}, {-} t_{01} ) \big)   \big]
\notag \\
\partial_{f_w} \rho \big( \Phinew^{-1}(  t_{12}, t_{01} ) \big) &= - P_w(t_{12},t_{01})  \rho \big( \Phinew^{-1}(  t_{12}, t_{01} ) \big) - \big[ \sigma_w, \rho \big( \Phinew^{-1}(  t_{12}, t_{01} ) \big)\big]
\label{tkpexp.15}
\end{align}
With (\ref{tkpexp.15}) and (\ref{tkpexp.09}) at $\sigma_w \rightarrow \Sigma_w(0) = \sigma_w + P_w(t_{12},t_{01})$ in place, the Leibniz property of $\partial_{f_w}$ leads to the following action on the product $\Psi$ in (\ref{tkpexp.10psi})
\begin{align}
\rho^{-1}\partial_{f_w} \Psi &= 
\Big( {-}   \Phinew( {-} t_{12}, {-} t_{01} )  P_w(t_{12},t_{01}) + \big[\sigma_w,   \Phinew( {-} t_{12}, {-} t_{01} )   \big] \Big) \,  \mathbb M^{\rm sv}_{\Sigma(0)}  \,     \Phinew^{-1}(  t_{12}, t_{01} ) 
\notag \\
&\quad +
\Phinew( {-} t_{12}, {-} t_{01} )    \Big( \big(  P_w(t_{12},t_{01}){+}\sigma_w \big)\mathbb M^{\rm sv}_{\Sigma(0)}+\mathbb M^{\rm sv}_{\Sigma(0)} \big( P_w(t_{12},t_{01}){+}\sigma_w \big) \Big)     \Phinew^{-1}(  t_{12}, t_{01} ) 
\notag \\
& \quad -  \Phinew( {-} t_{12}, {-} t_{01} )  \, \mathbb M^{\rm sv}_{\Sigma(0)} \, \Big(
 P_w(t_{12},t_{01})    \Phinew^{-1}(  t_{12}, t_{01} )   + \big[ \sigma_w,   \Phinew^{-1}(  t_{12}, t_{01} ) \big]
\Big) \notag \\
&= \sigma_w\, \Phinew( {-} t_{12}, {-} t_{01} )  \, \mathbb M^{\rm sv}_{\Sigma(0)} \, \Phinew^{-1}(  t_{12}, t_{01} )
+  \Phinew( {-} t_{12}, {-} t_{01} )  \, \mathbb M^{\rm sv}_{\Sigma(0)} \, \Phinew^{-1}(  t_{12}, t_{01} ) \, \sigma_w
 \notag \\
 &= \sigma_w\, \rho^{-1}(\Psi)
 + \rho^{-1}(\Psi)\, \sigma_w
 \label{tkpexp.16}
\end{align}
reproducing the desired ``differential equation'' of $\Psi$ in (\ref{tkpexp.10}).

\subsubsection{Conclusion}
\label{app:phi.3}

In order to see why the statement (\ref{tkpexp.07}) of the lemma is implied by the above results on the ``initial conditions'' and ``differential equations'', we expand the difference between its left and right side in the $f$-alphabet:
\begin{align}
&\rho \big( \Phinew( {-} t_{12}, {-} t_{01} ) \big) \! \shuffle \!  \rho\big( \mathbb M^{\rm sv}_{\Sigma(0)} \big) \! \shuffle \!  \rho \big( \Phinew^{-1}(  t_{12}, t_{01} ) \big)
-  \rho \big(\mathbb M^{\rm sv}_{\sigma} \big) = \sum^\infty_{r=0} \sum_{i_1,\ldots,i_r\atop{\in 2\mathbb N+1}} f_{i_1}\ldots f_{i_r} \, \Upsilon_{i_1,\ldots,i_r}(f_2)
\label{tkpexp.18}
\end{align}
The coefficients $\Upsilon_{i_1,\ldots,i_r}(f_2)$ could in principle be arbitrary power series in $f_2$ with $\mathbb Q$-linear combinations of words in $t_{01},t_{12},\sigma_w$ at each order by inspection of the contributing series. The matching of ``initial values'' in section \ref{app:phi.1} translates into
\beq
\Upsilon_{\emptyset}(f_2) = 0
\label{tkpexp.19}
\eeq
for the coefficient of the empty word in the non-commuting letters $f_w$ with $w$ odd. The matching of ``differential equations'' in section \ref{app:phi.2} translates into the recursion
\beq
\Upsilon_{i_1,\ldots,i_r,w}(f_2)
= \sigma_w \,\Upsilon_{i_1,\ldots,i_r}(f_2)+\Upsilon_{i_1,\ldots,i_r}(f_2)\, \sigma_w
\label{tkpexp.20}
\eeq
in the number $r$ of non-commuting letters in (\ref{tkpexp.18}). With the vanishing (\ref{tkpexp.19}) as a base case and (\ref{tkpexp.20}) as an inductive step, it follows that all the $\Upsilon_{i_1,\ldots,i_r}(f_2)$ in  (\ref{tkpexp.18}) are identically zero for arbitrary $r\geq 0$. This implies the statement (\ref{tkpexp.07s}) of the lemma and concludes its proof.

%%%%%%%%%%%%%%%%%%%%%%%%%%%%%%%%%%%%%%%%%%%%%%%%%%
%%%%%%%%%%%%%%%%%%%%%%%%%%%%%%%%%%%%%%%%%%%%%%%%%%
\section{Augmented zeta generators $\Sigma_w(u)$}
\label{sec:aug}

In this appendix, we shall spell out selected contributions to the expansion of augmented zeta generators $\Sigma_w(u)$ defined by (\ref{not.06}). 

%%%%%%%%%%%%%%%%%%%%%%%%%%%%%%%%%%%%%%%%%%%%%%%%%%
%%%%%%%%%%%%%%%%%%%%%%%%%%%%%%%%%%%%%%%%%%%%%%%%%%
\subsection{Expansions of $\Sigma_w(u)$}
\label{sec:C.4}

The expansions of the simplest augmented zeta generators $\Sigma_w(u)$ at $w=3,5,7$ up to and including degree 8 are given by
\begin{align}
\Sigma_3(u) &= - \frac{1}{2}\, \ep_4^{(2)} + B_1(u) \,[b_2, b_3^{(1)}] + \frac{1}{2}\,B_2(u) \, [b_3,b_3^{(1)}] + \frac{1}{4}\,B_2(u) \, [b_2, b_4^{(1)}] + z_3 - \frac{1}{4}\, [b_3,b_3^{(1)}]\notag \\
&\quad + \frac{1}{18}\, B_3(u) \, [b_2, b_5^{(1)}]
+ \frac{1}{6}\, B_3(u)\, \big( [b_3,b_4^{(1)}] {-} [b^{(1)}_3,b_4]\big)
+ \bigg(\frac{B_1(u)}{12}  {-} \frac{B_3(u)}{18}  \bigg) \, \big[b_2, [b_2, b_3] \big]
\notag \\
&\quad +\frac{1}{480} \,[\ep_4, \ep_4^{(1)}] {+}\frac{1}{96}\, B_4(u) \,[b_2,b_6^{(1)}]
{+} \frac{1}{24}\, B_4(u)\, \big( [b_3,b_5^{(1)}] {-} [b^{(1)}_3,b_5]\big) {+}\frac{1}{16}\, B_4(u) \,[b_4,b_4^{(1)}]
\notag \\
&\quad + \bigg(\frac{1}{1440}{+} \frac{B_2(u)}{24}  {-} \frac{ B_4(u)}{32} \bigg) \, \big[b_2, [b_2, b_4] \big]
+ \bigg(\frac{1}{720}{-} \frac{B_2(u)}{24} {+} \frac{B_4(u)}{24}  \bigg) \, \big[b_3, [b_3, b_2] \big]+
\ldots
\notag \\
\Sigma_5(u) &= -\frac{1}{24} \, \ep_6^{(4)} - \frac{1}{2}\, B_1(u) \, [b_3^{(1)},b_4^{(2)}] + \frac{1}{6} \, B_1(u) \, [b_2, b_5^{(3)}] 
- \frac{5}{48} \, [\ep_4^{(1)}, \ep_4^{(2)}]
 \notag \\
&\quad - \frac{1}{8} \,B_2(u) [b^{(1)}_4, b_4^{(2)}]  + \frac{1}{12}\,B_2(u) \big( [b_3, b_5^{(3)}] - [b_3^{(1)},b_5^{(2)}] \big) + \frac{1}{48}\, B_2(u)\, [b_2, b_6^{(3)}] \notag\\
&\quad -\frac{1}{16}\, \big(1+ 5\,B_2(u) \big)\, \big[ b_2,[b_2, b_4^{(2)}]\big] + \frac{1}{6}\, \big(1+2B_2(u) \big) \big[ b_3^{(1)}, [b_2,b_3^{(1)}]\big] +\ldots \notag \\
\Sigma_7(u) &= -\frac{1}{6!} \, \ep_8^{(6)}  + \ldots
\label{explSIG}
\end{align}
with infinite series in $\ep_k^{(j)}$, $b_k^{(j)}$ (and $z_5$, $z_7$ in case of $\Sigma_5(u) $, $\Sigma_7(u)$) of degree $\geq 9$ in the ellipsis.

%%%%%%%%%%%%%%%%%%%%%%%%%%%%%%%%%%%%%%%%%%%%%%%%%%
%%%%%%%%%%%%%%%%%%%%%%%%%%%%%%%%%%%%%%%%%%%%%%%%%%
\subsection{Source terms of $\Sigma_w(u)$}
\label{sec:C.5}

Given that the augmented version (\ref{not.06}) of the genus-one zeta generators only depend on $u$ through the conjugation with $e^{ux}$, their expansion can be
written as
\beq
\Sigma_w(u) = \sum_{k=0}^\infty \frac{B_k(u)}{k!} \, {\rm ad}_x^k (\Sigma_w^{\rm sc})
\eeq
with $\Sigma_w^{\rm sc}$ denoting a $u$-independent ``source term'' which is
different from $\Sigma_w(u=0)$ but by itself an infinite series in $\ep_k^{(j)}$ and $b_k^{(j)}$. The nested commutators of $x$ with the contributing $\ep_k^{(j)}$ and $b_k^{(j)}$ can be evaluated using the methods of section \ref{sec:D.1}. The expansion of $\Sigma_w^{\rm sc}$ at $w=3,5,7$
to up to and including degree 10 is given as follows and illustrates that these source terms offer a very economic way of packing the information of the more explicit expansions in (\ref{explSIG}):

\begin{align}
    \Sigma^{\rm sc}_3&=-\dfrac{1}{2}\epsilon_4^{(2)}+z_3 - \dfrac{1}{4} [b_3, b_3^{(1)}] + 
 \dfrac{1}{480} [\epsilon_4, \epsilon_4^{(1)}] +\dfrac{1}{1440}[b_2, [b_2, b_4]] - 
 \dfrac{1}{720} [b_3, [b_2, b_3]]   \notag\\
 &\quad + \dfrac{1}{30240} [\ep_4^{(1)},\ep_6] -\dfrac{1}{120960} [\ep_4,\ep_6^{(1)}]
 -  \dfrac{1}{210080}[b_2, [b_3, b_5]] +\dfrac{1}{210080} 
  [b_3, [b_2, b_5]]\notag\\
  &\quad  -\dfrac{1}{120960}[b_2, [b_2, b_6]]+ \dfrac{1}{ 3 20160}
  [b_3, [b_3, b_4]]- \dfrac{1}{2 20160}
  [b_4, [b_2, b_4]] + \ldots
\notag  \\
  %%%%%%%
    \Sigma^{\rm sc}_5&=-\dfrac{1}{4!}\epsilon_6^{(4)} -\dfrac{5}{48} [\epsilon_4^{(1)}, \epsilon_4^{(2)}] -\dfrac{1}{16} [b_2, [b_2, b_4^{(2)}]]+\dfrac{1}{6} [b_3^{(1)}, [b_2, b_3^{(1)}]]\notag\\
 &\quad  
 +  \dfrac{1}{5760} \big( [\ep_4, \ep^{(3)}_6] -  [\ep^{(1)}_4, \ep^{(2)}_6] +[\ep^{(2)}_4,\ep^{(1)}_6] \big)
 +z_5 + 
 \dfrac{1}{1440} \big( [b_5, b^{(3)}_5]  - 
 [b^{(1)}_5, b^{(2)}_5] \big) \notag\\
&\quad
+ \dfrac{1}{17280}[b_2,[b_2, b^{(2)}_6]]+ 
 \dfrac{1}{4320}[b_2, [b_3, b^{(2)}_5]] - 
 \dfrac{7}{1920}[b_2, [b_4, b^{(2)}_4]]\notag\\
 &\quad- 
 \dfrac{1}{2160}[b_2, [b^{(1)}_3, b^{(1)}_5]] - 
 \dfrac{1}{2160}[b_3, [b_2, b^{(2)}_5]] + \dfrac{19}{5760} [b^{(1)}_4, [b_2, b^{(1)}_4]] \notag\\
 &\quad+ 
 \dfrac{17}{960} [b_3, [b_3, b^{(2)}_4]] - 
 \dfrac{47}{2880} [b_3, [b^{(1)}_3, b^{(1)}_4]]  + 
 \dfrac{1}{2160}[b^{(1)}_3, [b_2, b^{(1)}_5]]\notag\\
 &\quad - 
 \dfrac{17}{960} [b^{(1)}_3, [b_3, b_4^{(1)}]] + 
 \dfrac{49}{1440} [b^{(1)}_3, [b^{(1)}_3, b_4]] - 
 \dfrac{19}{2880} [b^{(2)}_4, [b_2, b_4]]\notag\\
 &\quad 
+  \dfrac{37}{240} [b^{(1)}_3, [b_2, [b_2, b_3]]]
 - \dfrac{331}{2160} [b_3, [b_2, [b_2, b^{(1)}_3]]] 
 \notag\\
 &\quad
-  \dfrac{361}{2160} [b_2, [b_2, [b_3, b^{(1)}_3]]]  - 
\dfrac{1}{1152}[b_2, [b_2, [b_2, b^{(1)}_4]]] + \ldots
\notag \\
 %%%%%%%
    \Sigma^{\rm sc}_7&=-\dfrac{1}{6!}\epsilon^{(6)}_8
    +\dfrac{7}{1152} \big([\epsilon_4^{(2)}, \epsilon_6^{(3)}] - [\epsilon_4^{(1)}, \epsilon_6^{(4)}] \big)
    -\dfrac{5}{96}[b_3^{(1)}, [b_4^{(2)}, b_3^{(1)}]]-\dfrac{17}{384}[b_4^{(2)}, [b_2, b_4^{(2)}]] \notag\\
    &\quad+\dfrac{1}{144}[b_3^{(1)}, [b_2, b_5^{(3)}]]+\dfrac{5}{96}[b_5^{(3)}, [b_2, b_3^{(1)}]]-\dfrac{13}{2304}[b_2, [b_2, b_6^{(4)}]] + \ldots
    \label{src.7}
\end{align}
with terms of degree $\geq 11$ in the ellipsis.

%%%%%%%%%%%%%%%%%%%%%%%%%%%%%%%%%%%%%%%%%%%%%%%%%%
%%%%%%%%%%%%%%%%%%%%%%%%%%%%%%%%%%%%%%%%%%%%%%%%%%
%%%%%%%%%%%%%%%%%%%%%%%%%%%%%%%%%%%%%%%%%%%%%%%%%%
%%%%%%%%%%%%%%%%%%%%%%%%%%%%%%%%%%%%%%%%%%%%%%%%%%
\section{Proof of Lemma \ref{3.lem:1}}
\label{app:gmod}

This appendix provides the proof of Lemma \ref{3.lem:1} which states the modular properties of the series $\mathbbm{\Gamma}_{x,y}(z,\tau)$ of Brown-Levin eMPLs defined by (\ref{notsec.19}). As a first step, we show in section \ref{appG.1} that the connection $\mathbb J^{\rm BL}_{x- 2\pi i \tau y,y}(z,\tau)$ in the integrand of $\mathbbm{\Gamma}_{x,y}(z,\tau)$ transforms  equivariantly in the sense of Definition \ref{def:eqv}. In a second step, we demonstrate in section \ref{appG.2} that the equivariance of the connection does not immediately propagate to $\mathbbm{\Gamma}_{x,y}(z,\tau)$ and work out correction terms from the regularization prescription of endpoint divergences. The third step in section \ref{appG.3} cancels these correction terms against contributions from the factor of $\exp\big( \ee{0}{2}{\tau} b_2\big)$ in (\ref{notsec.19}) and pinpoints the origin of the monodromy phases $e^{i\pi b_2/6}$ and $e^{-i\pi b_2/2}$ in the statement (\ref{modgser}) of the lemma.

\subsection{Equivariance of $\mathbb J^{\rm BL}_{x- 2\pi i \tau y,y}(z,\tau)$}
\label{appG.1}

The modular transformation of the Brown-Levin connection in the original alphabet of (\ref{notsec.06}) is given by
\beq
\mathbb J^{\rm BL}_{x,y}\bigg(\frac{z}{c\tau{+}d},\frac{a\tau{+}b}{c\tau{+}d} \bigg) = 
\mathbb J^{\rm BL}_{(c\tau{+}d)x,\frac{y}{c\tau{+}d} - \frac{cx}{2\pi i}}(z,\tau) \, , \ \ \ \ \ \ ( \smallmatrix a &b \\ c &d \endsmallmatrix ) \in {\rm SL}_2(\mathbb Z)
\label{prfG.01}
\eeq
where the inhomogeneity $- \frac{cx}{2\pi i}$ in the second letter can for instance be understood from its link (\ref{notsec.07}) to the modular connection $\mathbb J^{\rm mod}_{x,y}(z,\tau)$ in (\ref{jmodtrf}). The $\tau$-dependent letter in the variant $\mathbb J^{\rm BL}_{x- 2\pi i \tau y,y}(z,\tau)$ of the Brown-Levin connection in the integrand of $\mathbbm{\Gamma}_{x,y}(z,\tau)$ further modifies the modular $T$- and $S$-transformation resulting from (\ref{prfG.01}) to
\begin{align}
{\cal T} \cdot \big[\mathbb J^{\rm BL}_{x- 2\pi i \tau y,y}(z,\tau) \big] &=
\mathbb J^{\rm BL}_{x- 2\pi i (\tau{+}1) y ,y}(z,\tau)
= U_T^{-1}
\mathbb J^{\rm BL}_{x- 2\pi i \tau y,y}(z,\tau)
U_T
\notag \\
{\cal S} \cdot \big[\mathbb J^{\rm BL}_{x- 2\pi i \tau y,y}(z,\tau) \big]&= \, \!
\mathbb J^{\rm BL}_{\tau x+ 2\pi i y,-\frac{x}{2\pi i}}(z,\tau) \,
= U_S^{-1} \mathbb J^{\rm BL}_{x- 2\pi i \tau y,y}(z,\tau) U_S
\label{prfG.02}
\end{align}
where the action of
\beq
{\cal T}\cdot (z,\tau) = (z,\tau{+}1) \, , \ \ \ \ \ \
{\cal S}\cdot (z,\tau) = \bigg(\frac{z}{\tau},-\frac{1}{\tau} \bigg)
\label{prfG.03}
\eeq
does not affect the generators $x,y$. In the second step of (\ref{prfG.02}), we have identified the letters $(x- 2\pi i (\tau{+}1) y ,y)$ and 
$(\tau x+ 2\pi i y,-\frac{x}{2\pi i})$ of the intermediate expressions as the images of the letters in $\mathbb J^{\rm BL}_{x- 2\pi i \tau y,y}(z,\tau)$ under the $\mathfrak{sl}_2$ action (\ref{uschoice.00}). Hence, we have established in (\ref{prfG.02}) that
\beq
\gamma \cdot \big[\mathbb J^{\rm BL}_{x- 2\pi i \tau y,y}(z,\tau) \big]
= U_\gamma^{-1} \mathbb J^{\rm BL}_{x- 2\pi i \tau y,y}(z,\tau) U_\gamma \ \forall \ \gamma \in {\rm SL}_2(\mathbb Z)
\label{prfG.04}
\eeq
by showing that it holds for both generators (\ref{prfG.03}) of the modular group.

\subsection{Implications for $\mathbbm{\Gamma}_{x,y}(z,\tau)$}
\label{appG.2}

In our second step of proving Lemma \ref{3.lem:1}, we try to uplift the equivariance (\ref{prfG.04}) of the Brown-Levin connection to its path-ordered exponential in the definition (\ref{notsec.19}) of $\mathbbm{\Gamma}_{x,y}(z,\tau)$. For generic endpoints $w,z\notin \mathbb Z{+} \tau \mathbb Z$ away from the singular points of $ \mathbb J^{\rm BL}_{x- 2\pi i \tau y,y}(z,\tau)$, equivariance would straightforwardly carry over to ${\rm Pexp} (\int_z^w
{\mathbb J}^{\rm BL}_{x-2\pi i \tau y,y}(z_1,\tau)
) $. However, the endpoint $w=0$ encountered in (\ref{notsec.19}) requires regularization, and we shall illustrate through a cutoff prescription using some $\varepsilon \ll 1$ as in \cite{Broedel:2014vla, Broedel:2018iwv, Broedel:2019tlz} that the modular $S$ transformation (\ref{prfG.03}) introduces a correction term $\tau^{-b_2}$ through the rescaling of $\varepsilon$ by $\tau$. For this purpose, we split the Brown-Levin connection into 
\beq
\mathbb J^{\rm BL}_{x- 2\pi i \tau y,y}(z,\tau) = \frac{\dd z}{z} [x,y] + \mathbb J^{\rm reg}_{x- 2\pi i \tau y,y}(z,\tau)
\label{prfG.05}
\eeq
where the second part is non-singular as $z \rightarrow 0$ and apply the
composition-of-paths formula, the cutoff-regularized path ordered exponential.
\begin{align}
{\cal S} \cdot \bigg[ {\rm Pexp} \biggl(\int_{z}^{ \varepsilon}
  \mathbb J^{\rm BL}_{x- 2\pi i \tau y,y}(z_1,\tau)\biggr) \bigg]&= {\rm Pexp} \biggl(\int_{z/\tau}^\varepsilon
\bigg[ \frac{\dd z_1}{z_1} [x,y] + \mathbb J^{\rm reg}_{x+ 2\pi i  y/\tau,y}(z_1,-\tfrac{1}{\tau}) \bigg]\biggr) \notag \\
&= {\rm Pexp} \biggl(\int_{z}^{\tau \varepsilon}
\bigg[ \frac{\dd z_1}{z_1} [x,y] + \mathbb J^{\rm reg}_{x+ 2\pi i y/ \tau,y}(\tfrac{z_1}{\tau},-\tfrac{1}{\tau}) \bigg]\biggr)  \notag \\
&= {\rm Pexp} \biggl({-} b_2 \int_{\varepsilon}^{\tau \varepsilon}
 \frac{\dd z_1}{z_1} \bigg)
{\rm Pexp} \biggl(\int_{z}^{ \varepsilon}
\mathbb J^{\rm BL}_{x+ 2\pi i  y/\tau,y} \big(\tfrac{z_1}{\tau},-\tfrac{1}{\tau} \big) \biggr)  \notag \\
&= \tau^{-b_2 } U_S^{-1}
{\rm Pexp} \biggl(\int_{z}^{ \varepsilon}
  \mathbb J^{\rm BL}_{x- 2\pi i \tau y,y}(z_1,\tau)\biggr) U_S
  \label{prfG.06}
\end{align}
In passing to the third line, we have used that the regular part $\mathbb J^{\rm reg}_{x- 2\pi i \tau y,y}$ does not contribute to the Pexp of the infinitesimal path $\int_{\varepsilon}^{\tau \varepsilon}$. In the last step, the first Pexp has been evaluated by resumming logarithms, and the second Pexp has inherited the $S$ equivariance (\ref{prfG.02}) of the Brown-Levin connection since the integration paths $\int_{z}^{ \varepsilon}$ avoids its singular points.

Since the modular $T$ transformation (\ref{prfG.03}) does not act on the endpoint $z$ of the integration path, there is no analogue of the correction term $\tau^{-b_2 }$ in
\begin{align}
{\cal T} \cdot \bigg[ {\rm Pexp} \biggl(\int_{z}^{ \varepsilon}
  \mathbb J^{\rm BL}_{x- 2\pi i \tau y,y}(z_1,\tau)\biggr) \bigg]
&= U_T^{-1}
{\rm Pexp} \biggl(\int_{z}^{ \varepsilon}
  \mathbb J^{\rm BL}_{x- 2\pi i \tau y,y}(z_1,\tau)\biggr) U_T
  \label{prfG.07}
\end{align}

\subsection{Simplifications from the additional integral over ${\rm G}_2$}
\label{appG.3}

In the last step of proving Lemma \ref{3.lem:1}, the correction factor $ \tau^{-b_2 }$ in (\ref{prfG.06}) will be seen to conspire with the modular transformation of the
factor of $\exp\big( \ee{0}{2}{\tau} b_2\big)$ in the definition (\ref{notsec.19}) of $\mathbbm{\Gamma}_{x,y}(z,\tau)$. By the quasi-modular transformation ${\rm G}_2(-\tfrac{1}{\tau}) = \tau^2 {\rm G}_2(\tau) -2\pi i \tau$ and $ {\rm G}_2(\tau{+}1)=  {\rm G}_2(\tau)$ of the Eisenstein series in the integrand, its primitive acquires logarithms in $\tau$ in its $S$ image,
\begin{align}
\ee{0}{2}{\tau{+}1} &= \ee{0}{2}{\tau}+\frac{i\pi}{6}
\label{prfG.08} \\
\ee{0}{2}{-\tfrac{1}{\tau}} &= \ee{0}{2}{\tau} 
+ \int^\tau_0 \frac{\dd \rho}{\rho}
+ \int^0_{i\infty} \frac{\dd \rho}{2\pi i} \,  {\rm G}_2(\rho)
\notag \\
&= \ee{0}{2}{\tau} + \log(\tau) -\frac{i\pi}{2}
\notag
\end{align}
The constant $-\frac{i\pi}{2}$ in the last line can be verified numerically from the series expansion of $\ee{0}{2}{\tau} = \frac{i \pi \tau}{6} - 2 \sum_{m,n=1}^\infty \frac{1}{n} q^{mn}$ and analytically from the degeneration formulae for elliptic associators in \cite{KZB, EnriquezEllAss}.

By inserting the modular transformations (\ref{prfG.08}) into the factor of $\exp\big( \ee{0}{2}{\tau} b_2\big)$ entering $\mathbbm{\Gamma}_{x,y}(z,\tau)$, we find the inverse of the correction factor $ \tau^{-b_2 }$ in (\ref{prfG.06}),
\beq
\exp\big( \ee{0}{2}{-\tfrac{1}{\tau}}  b_2 \big) = \exp \big( \ee{0}{2}{\tau}  b_2 -\tfrac{i\pi}{2}  b_2 \big) \tau^{b_2}
\eeq
as well as the monodromy phases $e^{i\pi b_2/6}$ and $e^{-i\pi b_2/2}$ in the statement (\ref{modgser}) of the lemma, thus completing its proof.

%%%%%%%%%%%%%%%%%%%%%%%%%%%%%%%%%%%%%%%%%%%%%%%%%%
%%%%%%%%%%%%%%%%%%%%%%%%%%%%%%%%%%%%%%%%%%%%%%%%%%
%%%%%%%%%%%%%%%%%%%%%%%%%%%%%%%%%%%%%%%%%%%%%%%%%%
%%%%%%%%%%%%%%%%%%%%%%%%%%%%%%%%%%%%%%%%%%%%%%%%%%

\section{Further examples of equivariant iterated integrals}
\label{apbeqv}

This appendix lists further instances of the equivariant
iterated integrals $\beqvY{j_1&j_2\\k_1&k_2\\z&z}$ in modular frame accompanying the word $b_{k_1}^{(j_1)} b_{k_2}^{(j_2)}$ in the presentation (\ref{sec5.1_1}) of the generating series $\mathbb H^{\rm eqv}_{\ep,b}(u,v,\tau)$: three at degree $k_1{+}k_2= 7$
\begin{align}
\beqvY{0&2\\3&4\\z&z}&=\bplusY{0&2\\3&4\\z&z}+\bminusY{2&0\\4&3\\z&z}+\bminusY{0\\3\\z}\bplusY{2\\4\\z}\notag\\
&\quad  -\bplusY{0&2\\3&4\\z&}-\bminusY{2&0\\4&3\\&z}-\bminusY{0\\3\\z}\bplusY{2\\4\\}\notag\\
&\quad+\dfrac{2}{3}\zeta_3\biggl( \bminusY{0\\3\\z}-\dfrac{B_3(u)}{6}(-2\pi i\bar{\tau})\biggr)\notag\\
&\quad+\dfrac{1}{6\pi \Im \tau}\zeta_3\biggl(-\dfrac{B_3(u)}{12}(-2\pi i \bar{\tau})^2\biggr)+\dfrac{\zeta_5}{12\pi \Im \tau}B_1(u)
\notag \\  
  %%%%%
\beqvY{2&0\\4&3\\z&z}&=\bplusY{2&0\\4&3\\z&z}+\bminusY{0&2\\3&4\\z&z}+\bminusY{2\\4\\z}\bplusY{0\\3\\z}\notag\\
&\quad  +\bplusY{0&2\\3&4\\z&}+\bminusY{2&0\\4&3\\&z}+\bminusY{0\\3\\z}\bplusY{2\\4\\}\notag\\
&\quad-\dfrac{2}{3}\zeta_3\biggl( \bminusY{0\\3\\z}-\dfrac{B_3(u)}{6}(-2\pi i\bar{\tau})\biggr)\notag\\
&\quad-\dfrac{1}{6\pi \Im \tau}\zeta_3\biggl(-\dfrac{B_3(u)}{12}(-2\pi i \bar{\tau})^2\biggr)-\dfrac{\zeta_5}{12\pi \Im \tau}B_1(u)
\notag \\
\beqvY{2&1\\4&3\\z&z}&=\bplusY{2&1\\4&3\\z&z}+\bminusY{1&2\\3&4\\z&z}+\bminusY{2\\4\\z}\bplusY{1\\3\\z}\notag\\
&\quad +\bplusY{1&2\\3&4\\z&}+\bminusY{2&1\\4&3\\&z}+\bminusY{1\\3\\z}\bplusY{2\\4\\}\notag\\
&\quad-\dfrac{2}{3}\zeta_3\biggl(\bminusY{1\\3\\z}-\dfrac{B_3(u)}{12}(-2\pi i \bar{\tau})^2\biggr)-\dfrac{\zeta_5}{3}B_1(u)
\label{wt7exbs}
\end{align}
%%%%%%
and another two at degree $k_1{+}k_2= 8$
%%%%%%
\begin{align}
\beqvY{2&1\\4&4\\z&z}&=\bplusY{2&1\\4&4\\z&z}+\bminusY{1&2\\4&4\\z&z}+\bplusY{1\\4\\z}\bminusY{2\\4\\z}\notag\\
&\quad+\biggl(\bplusY{1&2\\4&4\\z&}+\bminusY{2&1\\4&4\\&z}+\bplusY{2\\4}\bminusY{1\\4\\z}\biggr)\notag\\
&\quad-\biggl(\bplusY{2&1\\4&4\\z&}+\bminusY{1&2\\4&4\\&z}+\bplusY{1\\4\\}\bminusY{2\\4\\z}\biggr)\notag\\
  &\quad-\dfrac{2}{3}\zeta_3\bminusY{1\\4\\z}+\dfrac{1}{6 \pi \Im \tau }\zeta_3\bminusY{2\\4\\z}\notag\\
  &\quad+\dfrac{1}{72}(-2\pi i\bar{\tau})^2B_4(u)\zeta_3-\dfrac{1}{18}B_2(u)\zeta_5
  \notag \\
\beqvY{2&1\\5&3\\z&z}&=\bplusY{2&1\\5&3\\z&z}+\bminusY{1&2\\3&5\\z&z}+\bplusY{1\\3\\z}\bminusY{2\\5\\z}\notag\\
    &\quad+\dfrac{3}{2}\biggl(\bplusY{2&1\\4&4\\z&}+\bminusY{1&2\\4&4\\&z}+\bplusY{1\\4}\bminusY{2\\4\\z}\biggr)\notag\\
    &\quad+\dfrac{3}{4}\biggr(\bplusY{1&2\\4&4\\z&}+\bminusY{2&1\\4&4\\&z}+\bplusY{2\\4}\bminusY{1\\4\\z}\biggr)\notag\\
    &\quad-\dfrac{1}{2}\zeta_3\bminusY{1\\4\\z}-\dfrac{1}{4y}\zeta_3\bminusY{2\\4\\z}
    -\dfrac{1}{24}\zeta_5B_2(u) 
    \notag\\
    &\quad+\dfrac{1}{96}(-2\pi i\bar{\tau})^2B_4(u)\zeta_3+\dfrac{1}{192 \pi \Im \tau}(-2\pi i\bar{\tau})^3B_4(u)\zeta_3
    \label{wt8exbs}
\end{align}

%%%%%%%%%%%%%%%%%%%%%%%%%%%%%%%%%%%%%%%%%%%%%%%%%%
%%%%%%%%%%%%%%%%%%%%%%%%%%%%%%%%%%%%%%%%%%%%%%%%%%
%%%%%%%%%%%%%%%%%%%%%%%%%%%%%%%%%%%%%%%%%%%%%%%%%%
%%%%%%%%%%%%%%%%%%%%%%%%%%%%%%%%%%%%%%%%%%%%%%%%%%

\bibliographystyle{JHEP}
%\bibliography{refs.bib}

\providecommand{\href}[2]{#2}\begingroup\raggedright\endgroup

\end{document}